%% file: ys-202607-012-quasi-spin-bcs.tex
\documentclass[11pt,a4paper]{article}

\usepackage{fontspec}
\usepackage[titletoc,title,header]{appendix}

\usepackage{amsmath}
\usepackage{amssymb}
\usepackage{amsthm}
\usepackage{mathtools}
\usepackage{xparse}    

\usepackage[margin=1in]{geometry}

\usepackage{graphicx}
\usepackage{xcolor}

\usepackage{hyperref}
\hypersetup{
  colorlinks=true,
  linkcolor=blue,
  citecolor=blue,
  urlcolor=blue,
  pdfauthor={Yoshitsugu Sekine},
  pdftitle={An Operator-Algebraic Exposition of the Thirring--Wehrl Theory of the Quasi-Spin BCS Model}
}

\usepackage[]{natbib}
\theoremstyle{plain}
\newtheorem{thm}{Theorem}[section]
\AtEndEnvironment{thm}{\qed}
\newtheorem{lem}[thm]{Lemma}
\AtEndEnvironment{lem}{\qed}
\newtheorem{prop}[thm]{Proposition}
\AtEndEnvironment{prop}{\qed}
\newtheorem{cor}[thm]{Corollary}
\AtEndEnvironment{cor}{\qed}

\AtEndEnvironment{conj}{\qed}

\AtEndEnvironment{fact}{\qed}

\theoremstyle{definition}
\newtheorem{defn}[thm]{Definition}
\AtEndEnvironment{defn}{\qed}
\newtheorem{ex}[thm]{Example}
\AtEndEnvironment{ex}{\qed}

\theoremstyle{remark}
\newtheorem{rem}[thm]{Remark}
\AtEndEnvironment{rem}{\qed}

\input{mycommands.tex}

\title{An Operator-Algebraic Exposition of the Thirring--Wehrl Theory of the Quasi-Spin BCS Model\\\vspace{0.5em}{\large Infinite Tensor Products, the Bogoliubov--Haag Method, and the Direct Integral Decomposition of Equilibrium States}}

\author{%
Yoshitsugu Sekine\\{\small\texttt{4429sekine@gmail.com}}%
}

\date{\today}

\begin{document}

\maketitle

\begin{abstract}
This expository review reformulates the BCS analyses of Haag, Emch--Guenin, and Thirring--Wehrl, together with subsequent operator-algebraic mean-field results of Bóna, Raggio--Werner, and Bru--de Siqueira Pedra, in the language of quasi-local \(\oacstar\)-algebras and state decompositions \cite{RudolfHaag005,EmchGuenin001,ThirringWehrl001,ThirringWalter001,PavelBona001,PavelBona002,RaggioWerner002,BruPedra001}. The quasi-spin model is realized on the UHF algebra of type \(2^{\infty}\), with product sectors described by von Neumann's incomplete infinite tensor products \cite{VonNeumannJohn001}. This framework distinguishes strong limits of intensive observables, domain-restricted limits of extensive observables, and sectorwise Bogoliubov--Haag limits, including their operator criterion, gap equation, and limiting dynamics. In the degenerate model, the ground-state and thermal gauge averages admit central direct-integral decompositions over the gauge circle, and the thermal decomposition converges to the ground-state decomposition at zero temperature. Adjoining the gauge phase as a central classical variable combines the phase-dependent Bogoliubov dynamics into one automorphism group on \(\fun{\conti}{\mansphere^{1}, \oa{A}}\).
\end{abstract}

\setcounter{tocdepth}{3}
\tableofcontents

\section{Introduction}\label{introduction}

The quasi-spin BCS model \eqref{eq:bcs-hamiltonian} makes the thermodynamic limit inseparable from the choice of representation and observable algebra. The Bogoliubov approximation replaces the pairing interaction in \eqref{eq:bcs-hamiltonian} by the self-consistent bilinear Hamiltonian \eqref{eq:effective-hamiltonian} \cite{BogoliubovNikolai002,BogoliubovTolmachevShirkov001}. The infinite-volume analysis of \cite{RudolfHaag005} interprets this replacement sectorwise. The gauge-invariant formulation of \cite{EmchGuenin001} organizes the mutually inequivalent symmetry-breaking vacua over the gauge orbit and represents the gauge-invariant nonpure vacuum through a direct integral. The product-sector analysis of \cite{ThirringWehrl001} places the quasi-spin Hamiltonian of \cite{BaumannEderSexlThirring001} in von Neumann's infinite tensor product spaces \cite{VonNeumannJohn001}, and \cite{ThirringWalter001} treats the degenerate model at finite temperature. The present note is an expository reconstruction of these results in one operator-algebraic notation. No claim of originality is made for the ground-state, equilibrium-state, or phase-decomposition mechanisms.

Subsequent operator-algebraic work places these conclusions in the general theory of quantum mean-field systems. Integral decompositions of limiting Gibbs states are treated for the full quasi-spin BCS model in \cite[Section IV.2]{RaggioWerner002}, while \cite[Section 6.2]{BruPedra001} treats equilibrium and ground states of the strong-coupling BCS--Hubbard model.

The main text works on the UHF algebra \(\oa{A}\) of type \(2^{\infty}\) from Definition \ref{def:quasilocal} and distinguishes three kinds of thermodynamic limit. The mean spin has the strong limit of Corollary \ref{cor:mean-spin}. The relative number operator and the centered interaction converge only in the restricted senses of Propositions \ref{prop:number} and \ref{prop:lemma3}. The shifted Hamiltonians have the sectorwise form limit \eqref{eq:sectorwise-bogoliubov-limit} after the subtraction \eqref{eq:e-omega}. Theorem \ref{thm:bogoliubov} represents this form by a self-adjoint operator exactly when the flip amplitudes \eqref{eq:flip-amplitude} are square summable. Alignment gives its diagonal Bogoliubov--Haag realization and the gap equation \eqref{eq:gap-equation}. Theorem \ref{thm:gns-identification} identifies the product representations with von Neumann's incomplete tensor products, and Propositions \ref{prop:disjointness} and \ref{prop:gauge-implementable} describe the sector separation and gauge implementability.

The algebra also fixes the scope of the dynamics. The finite-volume dynamics \eqref{eq:finite-volume-dynamics} converges as in Theorem \ref{thm:dynamics} only after a product representation and its mean profile have been fixed. Proposition \ref{prop:ground-dynamics-non-descent} excludes a phase-independent point-norm limit on \(\oa{A}\) that reproduces every fiber dynamics. For the gauge orbit considered here, the phase-extended algebra \eqref{eq:phase-extended-algebra} carries the single automorphism group \eqref{eq:phase-extended-ground-dynamics} and the representation \eqref{eq:phase-extended-ground-representation}. This specializes the enlarged-algebra construction of \cite{PavelBona001}, whose equilibrium and ground states for the strong-coupling quasi-spin BCS model are analyzed in \cite{PavelBona002}.

The ground and thermal decompositions are likewise indexed by fiber-dependent dynamics. The ground-state orbit of Definition \ref{def:main-ground-orbit} gives the gauge average \eqref{eq:ground-state-average} and its central decomposition \eqref{eq:ground-central-decomposition} into pure product ground states for the fiber dynamics \eqref{eq:ground-fiber-bogoliubov-dynamics}. The positive-temperature family of Definition \ref{def:main-thermal-phase} gives the limiting Gibbs state \eqref{eq:limit-state} and the central decomposition of Theorem \ref{thm:central-decomposition}. Its fibers are mutually disjoint factor states and satisfy the fiber KMS condition of Proposition \ref{prop:fiber-kms}. The thermal Green-function limit is Theorem \ref{thm:green}, and Corollary \ref{cor:zero-temperature-central-decomposition} connects the thermal decomposition to the product ground-state sectors.

The original proofs refer several essential ingredients to the literature. The appendices supply the required infinite tensor product results from \cite{VonNeumannJohn001}, finite-spin representation formulas from \cite{EugeneWigner002}, and \(\Gamma\)-function estimates based on \cite{WatsonWhittaker001}. The quasi-local UHF construction is cited to \cite[Section 2.6]{BratteliRobinson003}.

The organization follows the mathematical dependencies. Sections \ref{sec:algebra}--\ref{sec:extensive} define the model and separate the intensive and extensive limits. Sections \ref{sec:bogoliubov} and \ref{sec:time} construct the sectorwise Hamiltonians and dynamics. Sections \ref{sec:thermal}--\ref{sec:green} derive the limiting Gibbs state and Green functions. Section \ref{sec:decomposition} identifies the central decomposition of the limiting Gibbs state and its zero-temperature limit. Appendices \ref{sec:functional-analytic-appendix}, \ref{sec:oa-appendix}, \ref{sec:bloch}, \ref{sec:analytic}, \ref{sec:itp}, and \ref{sec:spin} contain the functional-analytic lemmas, the quasi-local algebra, the one-site Bloch calculus, the analytic estimates, the infinite tensor product theory, and the finite-spin representation theory.

Throughout, \(\monnat\) denotes the natural numbers including \(0\), and \(\semigrposint\) denotes the positive integers. An operator-algebraic state is denoted by \(\oastate[\psi]\), and its GNS cyclic vector by \(\oagnsvector[\Psi]\). More generally, the state \(\oastate[\psi_{\alpha}]\) has GNS cyclic vector \(\oagnsvector[\Psi_{\alpha}]\). The symbol \(\Psi\) without \(\oagnsvector\) denotes a general Hilbert-space vector. Every Hilbert-space inner product \(\bkt{\Phi}{\Psi}\) is conjugate-linear in \(\Phi\) and linear in \(\Psi\). Every sesquilinear form in this paper follows the same convention.

\section{Main Results}\label{sec:main-results}

The common notation fixes the data needed to state the principal results. The subsequent sections give the detailed constructions and proofs.

\subsection{Common quasi-spin setting}\label{common-quasi-spin-setting}

The observable algebra is the UHF algebra \[\begin{aligned}
\oa{A}
=
\gtclos{\bigcup_{\Omega
\in
\semigrposint} \oa{A}_{\Omega}},
\quad
\oa{A}_{\Omega}
=
\bigotimes_{p = 1}^{\Omega}
\spmat{2}{\fldcmp}.
\end{aligned}\] The closure is taken in the unique \(\oacstar\)-norm. The matrices \(\sigma^{x}, \sigma^{y}, \sigma^{z}\) are the Pauli matrices, \(\sigma^{\pm}
=
\frac{1}{2} \rbk{\sigma^{x} \pm \imunit \sigma^{y}}\), and a subscript \(p\) denotes the copy in the tensor factor \(p\).

\begin{defn}[quasi-spin BCS data for the main results]\label{def:main-bcs-setting}
Let $\seq{\varepsilon_{p}}{p
\in
\semigrposint}$ be a bounded real sequence and let
$\sminvtemperature_{c}
>
0$.
For $\Omega
\in
\semigrposint$ define
$$\begin{aligned}
S^{\pm}_{\Omega}
=
\sum_{p
=
1}^{\Omega} \sigma^{\pm}_{p},
\quad
\physvec{M}_{\Omega}
=
\frac{1}{\Omega} \sum_{p = 1}^{\Omega} \physvec{\sigma}_{p},
\quad
\physham_{\Omega}
=
- \sum_{p = 1}^{\Omega} \varepsilon_{p} \sigma^{z}_{p}
- \frac{2}{\sminvtemperature_{c} \Omega} S^{+}_{\Omega} S^{-}_{\Omega}.
\end{aligned}$$
The finite-volume Heisenberg dynamics is
$$\fun{\tau^{\Omega}_{t}}{A}
=
\fnexp{\imunit t \physham_{\Omega}} A \fnexp{- \imunit t \physham_{\Omega}},$$
for $A
\in
\oa{A}_{\Omega}$,
and the Gibbs state at inverse temperature $\sminvtemperature
>
0$ is
$$\fun{\oastate[\psi_{\sminvtemperature,\Omega}]}{A}
=
\frac{\sqfun{\trace}{\fnexp{- \sminvtemperature \physham_{\Omega}} A}}
{\sqfun{\trace}{\fnexp{- \sminvtemperature \physham_{\Omega}}}}.$$
Its extension to $\oa{A}$ is defined by
\begin{equation}\label{expedition0025007}
\begin{aligned}
\oastate[\widehat{\psi}_{\sminvtemperature,\Omega}]
=
\oastate[\psi_{\sminvtemperature,\Omega}] \otimes
\oastate[\psi^{\mathrm{tr}}_{> \Omega}],
\end{aligned}
\end{equation}
where $\oastate[\psi^{\mathrm{tr}}_{> \Omega}]$ is the product of the normalized trace states on the tensor factors with indices greater than $\Omega$.
The gauge automorphisms are determined by
$$\begin{aligned}
\fun{\mathfrak{g}_{\vartheta}}{\sigma^{\pm}_{p}}
=
\fnexp{\pm \imunit \vartheta} \sigma^{\pm}_{p},
\quad
\fun{\mathfrak{g}_{\vartheta}}{\sigma^{z}_{p}}
=
\sigma^{z}_{p}.
\end{aligned}$$
\end{defn}

The model is degenerate when \(\varepsilon_{p}
=
\varepsilon\) for every \(p\). Definitions \ref{def:gauge} and \ref{def:degenerate-thermal-setting} give the detailed versions of the last two conventions.

\subsection{Product sectors and self-consistent fields}\label{product-sectors-and-self-consistent-fields}

The sectorwise results require the asymptotic correlation between the one-particle energies and the spin directions. The following definition collects the required data.

\begin{defn}[product-sector data for the main results]\label{def:main-sector-data}
A spin configuration is a sequence
$\physvec{\omega}
=
\seq{\physvec{u}_{p}}{p
\in
\semigrposint}$
in $\mansphere^{2}$.
It has a mean profile if the empirical measures
$\frac{1}{\Omega} \sum_{p
=
1}^{\Omega} \delta_{\physvec{u}_{p}}$
converge weakly on $\mansphere^{2}$ to a probability measure
$\lambda_{\physvec{\omega}}$.
It has a joint mean profile if the empirical measures
$\frac{1}{\Omega} \sum_{p
=
1}^{\Omega} \delta_{\rbk{\varepsilon_{p}, \physvec{u}_{p}}}$
converge weakly on
$\closedinterval{-E}{E} \times \mansphere^{2}$,
where
$E
=
\sup_{p
\in
\semigrposint} \abs{\varepsilon_{p}}$.
Its mean polarization is
$$\begin{aligned}
\eta_{\physvec{\omega}} \bar{\physvec{u}}_{\physvec{\omega}}
=
\lim_{\Omega \to \infty}
\frac{1}{\Omega} \sum_{p
=
1}^{\Omega} \physvec{u}_{p},
\quad
0
\leq
\eta_{\physvec{\omega}}
\leq
1
\end{aligned},$$
where $\bar{\physvec{u}}_{\physvec{\omega}}$ is a unit vector when
$\eta_{\physvec{\omega}}
>
0$.

Choose unit vectors $\xi^{\pm}_{\physvec{\omega},p}$ satisfying
$$\rbk{\physvec{\sigma}_{p} \cdot \physvec{u}_{p}} \xi^{\pm}_{\physvec{\omega},p}
=
\pm \xi^{\pm}_{\physvec{\omega},p}.$$
The product state with reference sequence
$\seq{\xi^{+}_{\physvec{\omega},p}}{p
\in
\semigrposint}$
is denoted by $\oastate[\psi_{\physvec{\omega}}]$,
and its product representation is
$\oarepn_{\physvec{\omega}}$
on $\sphilb{H}_{\physvec{\omega}}$.
The flip domain $\sphilb{D}_{\physvec{\omega}}$ is the linear span of the product vectors
$\xi^{F}_{\physvec{\omega}}$
obtained by replacing $\xi^{+}_{\physvec{\omega},p}$ with
$\xi^{-}_{\physvec{\omega},p}$ at the sites of a finite set
$F
\subset
\semigrposint$.

The relative number cutoffs and their diagonal action on the flip domain are
$$\begin{aligned}
N_{\Omega,\physvec{\omega}}
=
\sum_{p = 1}^{\Omega}
\frac{1}{2} \rbk{1 - \physvec{\sigma}_{p} \cdot \physvec{u}_{p}},
\quad
N^{0}_{\physvec{\omega}} \xi^{F}_{\physvec{\omega}}
=
\abscard{F} \xi^{F}_{\physvec{\omega}}.
\end{aligned}$$
For the centered interaction define
$$\begin{aligned}
m^{\pm}_{\physvec{\omega},p}
=
\frac{1}{2} \rbk{u^{x}_{p} \pm \imunit u^{y}_{p}},
\quad
d^{\pm}_{\physvec{\omega},p}
=
\sigma^{\pm}_{p} - m^{\pm}_{\physvec{\omega},p},
\quad
D^{\pm}_{\physvec{\omega},\Omega}
=
\sum_{p = 1}^{\Omega} d^{\pm}_{\physvec{\omega},p},
\quad
R_{\physvec{\omega},\Omega}
=
\frac{1}{\Omega}
D^{+}_{\physvec{\omega},\Omega} D^{-}_{\physvec{\omega},\Omega}
\end{aligned}$$
and
$$\begin{aligned}
c_{\physvec{\omega}}
&=
\int_{\mansphere^{2}}
\frac{\rbk{1 + u^{z}}^{2}}{4}
\opdmsr{\fun{\lambda_{\physvec{\omega}}}{\physvec{u}}}.
\end{aligned}$$

Define the effective field and one-site Hamiltonian by
$$\begin{aligned}
\physvec{b}_{\physvec{\omega},p}
=
\rbk{\frac{\eta_{\physvec{\omega}}}{\sminvtemperature_{c}} \bar{u}^{x}_{\physvec{\omega}},
\frac{\eta_{\physvec{\omega}}}{\sminvtemperature_{c}} \bar{u}^{y}_{\physvec{\omega}},
\varepsilon_{p}},
\quad
\physham[h]_{\physvec{\omega},p}
=
- \physvec{b}_{\physvec{\omega},p} \cdot \physvec{\sigma}_{p}.
\end{aligned}$$
The flip amplitude $t_{\physvec{\omega},p}$ is the coefficient determined by
$$\physham[h]_{\physvec{\omega},p} \xi^{+}_{\physvec{\omega},p}
=
- \rbk{\physvec{b}_{\physvec{\omega},p} \cdot \physvec{u}_{p}} \xi^{+}_{\physvec{\omega},p}
+ t_{\physvec{\omega},p} \xi^{-}_{\physvec{\omega},p}.$$
The configuration is aligned if there are signs
$\eta_{\physvec{\omega},p}
\in
\setone{+1, -1}$
such that
$$\begin{aligned}
\physvec{b}_{\physvec{\omega},p}
\neq
0,
\quad
\physvec{u}_{p}
=
\eta_{\physvec{\omega},p}
\frac{\physvec{b}_{\physvec{\omega},p}}
{\abs{\physvec{b}_{\physvec{\omega},p}}}.
\end{aligned}$$
Finally,
set
$E_{\physvec{\omega},\Omega}
=
\fun{\oastate[\psi_{\physvec{\omega}}]}{\physham_{\Omega}}$.
\end{defn}

The detailed definitions of the mean profile, the product representation, the joint mean profile, and alignment are Definitions \ref{def:mean-profile}, \ref{def:spin-product-representation}, \ref{def:joint-mean-profile}, and \ref{def:aligned-configuration}.

\subsection{Scale-dependent limits in product sectors}\label{scale-dependent-limits-in-product-sectors}

The scale-dependent theorem distinguishes the strong limits of intensive observables from the domain and matrix-element limits of extensive observables. Note that the limits are taken on the representated space.

\begin{thm}[intensive and extensive observables]\label{thm:main-observable-scales}
Let $\physvec{\omega}$ have a mean profile and use Definition \ref{def:main-sector-data}.
The following conclusions hold.
\begin{enumerate}
\item For each $\gamma
\in
\setone{x, y, z}$,
the mean spin converges strongly:
$$\slim_{\Omega \to \infty}
\fun{\oarepn_{\physvec{\omega}}}{M^{\gamma}_{\Omega}}
=
\eta_{\physvec{\omega}} \bar{u}^{\gamma}_{\physvec{\omega}}.$$

\item The operator $N^{0}_{\physvec{\omega}}$ is essentially self-adjoint.
If $N_{\physvec{\omega}}$ denotes its closure,
then for every $\Psi
\in
\dom N_{\physvec{\omega}}$,
$$\lim_{\Omega \to \infty}
\norm{
\fun{\oarepn_{\physvec{\omega}}}{N_{\Omega,\physvec{\omega}}} \Psi
- N_{\physvec{\omega}} \Psi}
=
0,$$
and for every $t
\in
\fldreal$,
$$\slim_{\Omega \to \infty}
\fnexp{
\imunit t
\fun{\oarepn_{\physvec{\omega}}}{N_{\Omega,\physvec{\omega}}}}
=
\fnexp{\imunit t N_{\physvec{\omega}}}.$$

\item For every $\Phi, \Psi
\in
\sphilb{D}_{\physvec{\omega}}$,
the centered interaction has the matrix-element limit
$$\lim_{\Omega \to \infty}
\bkt{\Phi}{
\fun{\oarepn_{\physvec{\omega}}}{R_{\physvec{\omega},\Omega}} \Psi}
=
c_{\physvec{\omega}} \bkt{\Phi}{\Psi}.$$
For a constant configuration with $\abs{u^{z}} < 1$ this convergence is not strong,
and for every $\Omega
\geq
8$, it holds that
$$\norm{\fun{\oarepn_{\physvec{\omega}}}{R_{\physvec{\omega},\Omega}}}
\geq
\frac{\Omega}{324}.$$

\item If
$\sum_{p \in \semigrposint}
\rbk{1 - \rbk{u^{z}_{p}}^{2}}
<
\infty$,
then every gauge automorphism is unitarily implementable in
$\oarepn_{\physvec{\omega}}$.

If instead
$\sum_{p \in \semigrposint}
\rbk{1 - \rbk{u^{z}_{p}}^{2}}
=
\infty$,
then for every $\vartheta
\notin
2 \pi \ringratint$ the representations associated with
$\oastate[\psi_{\physvec{\omega}}]$
and
$\oastate[\psi_{\physvec{\omega}}] \circ \mathfrak{g}_{\vartheta}$
are disjoint.
In particular $\mathfrak{g}_{\vartheta}$ is not unitarily implementable in this sector.
\end{enumerate}
\end{thm}

The strong law in Theorem \ref{thm:main-observable-scales}(1) is proved in Theorem \ref{thm:lln} and Corollary \ref{cor:mean-spin}. The assertions about relative number, the centered interaction, and gauge implementability are Propositions \ref{prop:number}, \ref{prop:lemma3}, and \ref{prop:gauge-implementable}, respectively.

\subsection{Sectorwise Bogoliubov--Haag theory}\label{sectorwise-bogoliubovhaag-theory}

The sectorwise theorem separates form convergence, operator representability, self-consistency, and dynamical convergence.

\begin{thm}[sectorwise Bogoliubov--Haag limit and dynamics]\label{thm:main-sectorwise-results}
Let $\physvec{\omega}$ have a mean profile and use the product-sector data of
Definition \ref{def:main-sector-data}.
The following conclusions hold.
\begin{enumerate}
\item The matrix elements of the centered Hamiltonians converge on the flip domain:
$$\fun{\opform{Q}_{\physvec{\omega}}}{\Phi, \Psi}
=
\lim_{\Omega \to \infty}
\bkt{\Phi}{
\fun{\oarepn_{\physvec{\omega}}}{
\physham_{\Omega} - E_{\physvec{\omega},\Omega}} \Psi},$$
for all $\Phi, \Psi
\in
\sphilb{D}_{\physvec{\omega}}$.

\item The form $\opform{Q}_{\physvec{\omega}}$ is represented on
$\sphilb{D}_{\physvec{\omega}}$
by an operator if and only if
$$\sum_{p
\in
\semigrposint} \abs{t_{\physvec{\omega},p}}^{2}
<
\infty.$$
In that case it has a unique self-adjoint realization
$\physham_{\txtbogoliubov,\physvec{\omega}}$
for which $\sphilb{D}_{\physvec{\omega}}$ is a core.

\item If $\physvec{\omega}$ is aligned,
then $t_{\physvec{\omega},p}
=
0$ for every $p$ and
$$\physham_{\txtbogoliubov,\physvec{\omega}} \xi^{F}_{\physvec{\omega}}
=
\rbk{
\sum_{p
\in
F} 2 \eta_{\physvec{\omega},p} \abs{\physvec{b}_{\physvec{\omega},p}}
} \xi^{F}_{\physvec{\omega}}.$$
If its transverse mean polarization is nonzero,
the gap equation is
$$\lim_{\Omega \to \infty} \frac{1}{\Omega} \sum_{p
=
1}^{\Omega}
\frac{\eta_{\physvec{\omega},p}}
{\sqrt{
\varepsilon_{p}^{2}
+\frac{\eta_{\physvec{\omega}}^{2}}{\sminvtemperature_{c}^{2}}
\rbk{
\rbk{\bar{u}^{x}_{\physvec{\omega}}}^{2}
+\rbk{\bar{u}^{y}_{\physvec{\omega}}}^{2}
}}}
=
\sminvtemperature_{c}.$$

\item If $\physvec{\omega}$ is aligned and has a joint mean profile,
then for every $A
\in
\oa{A}$ and $t
\in
\fldreal$,
$$\slim_{\Omega \to \infty}
\fun{\oarepn_{\physvec{\omega}}}{\fun{\tau^{\Omega}_{t}}{A}}
=
\fnexp{\imunit t \physham_{\txtbogoliubov,\physvec{\omega}}}
\fun{\oarepn_{\physvec{\omega}}}{A}
\fnexp{- \imunit t \physham_{\txtbogoliubov,\physvec{\omega}}}.$$
\end{enumerate}
\end{thm}

The form limit, the representability criterion, the diagonal realization, and the gap equation are Theorem \ref{thm:bogoliubov}. The strong dynamical limit is Theorem \ref{thm:dynamics}.

The convergence in Theorem \ref{thm:main-sectorwise-results}(4) concerns represented local observables. It does not assert convergence of the Hamiltonians or their implementing evolution operators; Proposition \ref{prop:unitary-failure} gives an aligned counterexample to weak convergence of the latter.

\subsection{Strong-coupling ground-state orbit}\label{strong-coupling-ground-state-orbit}

The degenerate strong-coupling setting produces a circle of mutually disjoint pure product sectors.

\begin{defn}[strong-coupling ground-state orbit]\label{def:main-ground-orbit}
Assume
$\varepsilon_{p}
=
\varepsilon$
and
$\sminvtemperature_{c} \abs{\varepsilon}
<
1$.
Set
$$\begin{aligned}
r
=
\sqrt{1 - \sminvtemperature_{c}^{2} \varepsilon^{2}},
\quad
\fun{\physvec{u}}{\phi}
=
\vecbk{
r\fun{\cos}{\phi},
r\fun{\sin}{\phi},
\sminvtemperature_{c} \varepsilon}.
\end{aligned}$$
The angle satisfies $\phi
\in
\rightopeninterval{0}{2 \pi}$.
Let $\fun{\physvec{\omega}}{\phi}$ be the constant configuration with value
$\fun{\physvec{u}}{\phi}$,
let $\oastate[\psi_{\fun{\physvec{\omega}}{\phi}}]$ be its pure product state,
and let
$$\oastate[\psi_{\txtgs}]
=
\int_{\mansphere^{1}}
\oastate[\psi_{\fun{\physvec{\omega}}{\phi}}]
\opdmsr{\msrprb_{\mansphere^{1}}}(\phi)$$
be its gauge average with respect to normalized Haar measure.
The associated phase-extended algebra is
\begin{equation}\label{eq:phase-extended-algebra}
\oa{C}
=
\fun{\conti}{\mansphere^{1}, \oa{A}}
\cong
\fun{\conti}{\mansphere^{1}} \otimes_{\min} \oa{A}.
\end{equation}
\end{defn}

\begin{thm}[central decomposition of the product ground-state sectors]\label{thm:main-ground-decomposition}
Under Definition \ref{def:main-ground-orbit},
the product representations at distinct phases are mutually disjoint and
$$\oastate[\psi_{\txtgs}]
=
\int_{\mansphere^{1}}
\oastate[\psi_{\fun{\physvec{\omega}}{\phi}}]
\opdmsr{\msrprb_{\mansphere^{1}}}(\phi)$$
is their central decomposition.
Its GNS von Neumann algebra has center
$\fun{\lp^{\infty}}{\mansphere^{1}, \msrprb_{\mansphere^{1}}}$.
The direct integral of the fiber Bogoliubov--Haag Hamiltonians is
$$\physham_{\txtgs}
=
\int_{\mansphere^{1}}^{\oplus}
\physham_{\txtbogoliubov,\fun{\physvec{\omega}}{\phi}}
\opdmsr{\msrprb_{\mansphere^{1}}}(\phi).$$
It is nonnegative and has
$$\begin{aligned}
\Ker\physham_{\txtgs}
\cong
\fun{\lp^{2}}{\mansphere^{1}, \msrprb_{\mansphere^{1}}},
\quad
\opspec{\physham_{\txtgs}}
=
\set{\frac{2 n}{\sminvtemperature_{c}}}{n
\in
\monnat}.
\end{aligned}$$
Thus its excitation gap is
$2/\sminvtemperature_{c}$.
The decomposable fiber dynamics does not descend to a point-norm continuous automorphism group of $\oa{A}$.
After adjoining the phase as a central classical variable,
the fiber dynamics defines one point-norm continuous automorphism group on
$\fun{\conti}{\mansphere^{1}, \oa{A}}
\cong
\fun{\conti}{\mansphere^{1}} \otimes_{\min} \oa{A}$.
\end{thm}

The aligned gauge orbit and disjointness are Corollary \ref{cor:strong-coupling}. The central decomposition and the spectral assertions are Theorem \ref{thm:ground-state-direct-integral}. Non-descent is Proposition \ref{prop:ground-dynamics-non-descent}. The phase-extended dynamics is Proposition \ref{prop:phase-extended-ground-dynamics}.

\subsection{Superconducting thermal phase}\label{superconducting-thermal-phase}

The finite-temperature results use the same degenerate Hamiltonian but a mixed product state in every phase sector.

\begin{defn}[superconducting thermal data for the main results]\label{def:main-thermal-phase}
Assume
$\varepsilon_{p}
=
\varepsilon$,
$\sminvtemperature
>
\sminvtemperature_{c}$,
and let $\eta_{0}
\in
\rbk{0, 1}$ be the positive solution of
$$\eta_{0}
=
\fun{\tanh}{\frac{\sminvtemperature}{\sminvtemperature_{c}} \eta_{0}}.$$
The superconducting condition is
$\sminvtemperature_{c} \abs{\varepsilon}
<
\eta_{0}$.
Set
$$\begin{aligned}
z_{0}
=
\sminvtemperature_{c} \varepsilon,
\quad
r_{0}
=
\sqrt{\eta_{0}^{2} - z_{0}^{2}},
\quad
\physvec{m}_{\phi}
=
\vecbk{
r_{0}\fun{\cos}{\phi},
r_{0}\fun{\sin}{\phi},
z_{0}}
\quad
\rbk{\phi
\in
\rightopeninterval{0}{2 \pi}}.
\end{aligned}$$
Let $\oastate[\psi^{\rbk{1}}_{\sminvtemperature,\phi}]$ be the one-site state with density matrix
$$D_{\phi}
=
\frac{1}{2} \rbk{1 + \physvec{m}_{\phi} \cdot \physvec{\sigma}},$$
and define the product state and its self-consistent one-site Hamiltonian by
$$\begin{aligned}
\oastate[\psi_{\sminvtemperature,\phi}]
=
\bigotimes_{p \in \semigrposint}
\oastate[\psi^{\rbk{1}}_{\sminvtemperature,\phi}],
\quad
\physham[h]_{\phi}
=
-\frac{\eta_{0}}{\sminvtemperature_{c}}
\physvec{\sigma} \cdot \frac{\physvec{m}_{\phi}}{\eta_{0}}.
\end{aligned}$$
The associated Bogoliubov product dynamics is generated locally by the copies
$\physham[h]_{\phi,p}$.
Let
$\rbk{\sphilb{H}_{0}, \oarepn_{0}, \oagnsvector[\Psi_{0}]}$
be the GNS triple of $\oastate[\psi_{\sminvtemperature,0}]$ and define
$$\begin{aligned}
\oarepn_{\phi}
=
\oarepn_{0} \circ \mathfrak{g}_{\phi},
\quad
\sphilb{H}_{\sminvtemperature}
=
\fun{\lp^{2}}{\mansphere^{1}, \msrprb_{\mansphere^{1}}; \sphilb{H}_{0}},
\quad
\oarepn_{\sminvtemperature}
=
\int_{\mansphere^{1}}^{\oplus}
\oarepn_{\phi}
\opdmsr{\msrprb_{\mansphere^{1}}}(\phi).
\end{aligned}$$
\end{defn}

Definition \ref{def:superconducting-thermal-regime} gives the detailed thermal regime, and \eqref{eq:phase-product-state} defines the Bogoliubov product states.

\subsection{Equilibrium state, Green functions, and central decomposition}\label{equilibrium-state-green-functions-and-central-decomposition}

The thermal main theorem combines the thermodynamic limit, the Green-function limit, and the decomposition of the limiting state.

\begin{thm}[equilibrium and Green-function limits]\label{thm:main-equilibrium-results}
Assume Definition \ref{def:main-thermal-phase}.
The following conclusions hold.
\begin{enumerate}
\item For every $A
\in
\oa{A}$,
$$\fun{\oastate[\psi_{\sminvtemperature}]}{A}
=
\lim_{\Omega \to \infty}
\fun{\oastate[\widehat{\psi}_{\sminvtemperature,\Omega}]}{A}
=
\int_{\mansphere^{1}}
\fun{\oastate[\psi_{\sminvtemperature,\phi}]}{A}
\opdmsr{\msrprb_{\mansphere^{1}}}(\phi),$$
where the states $\oastate[\widehat{\psi}_{\sminvtemperature,\Omega}]$ are defined
in \eqref{expedition0025007}.

\item For local $B_{1}, \dotsc, B_{k}$ and real
$t_{1}, \dotsc, t_{k}$,
$$\begin{aligned}
\lim_{\Omega \to \infty}
\fun{\oastate[\widehat{\psi}_{\sminvtemperature,\Omega}]}{
\fun{\tau^{\Omega}_{t_{1}}}{B_{1}} \dotsm
\fun{\tau^{\Omega}_{t_{k}}}{B_{k}}}
=
\int_{\mansphere^{1}}
\fun{\oastate[\psi_{\sminvtemperature,\phi}]}{
\fun{\tau^{\txtbogoliubov,\phi}_{t_{1}}}{B_{1}} \dotsm
\fun{\tau^{\txtbogoliubov,\phi}_{t_{k}}}{B_{k}}}
\opdmsr{\msrprb_{\mansphere^{1}}}(\phi).
\end{aligned}$$

\item The states $\oastate[\psi_{\sminvtemperature,\phi}]$ are mutually disjoint factor states and each is the unique KMS state at inverse temperature
$\sminvtemperature$
for their respective dynamics
$\tau^{\txtbogoliubov,\phi}$.
In particular each fiber state is extremal in its KMS simplex.
On the phase-extended algebra
$\oa{C}$
of \eqref{eq:phase-extended-algebra},
the fiber dynamics forms one point-norm continuous automorphism group and the fiberwise gauge average is a KMS state for that group.
The decomposition in (1) is the central decomposition of
$\oastate[\psi_{\sminvtemperature}]$.
Its GNS representation is the direct integral of the fiber representations and the center of the represented von Neumann algebra is
$\fun{\lp^{\infty}}{\mansphere^{1}, \msrprb_{\mansphere^{1}}}$.
It is generated by the phase of the order parameter,
represented by the strong limit
$$\slim_{\Omega \to \infty}
\fun{\oarepn_{\sminvtemperature}}{\frac{2}{r_{0}} M^{+}_{\Omega}}
=
M_{\fnexp{\imunit \phi}} \otimes \id_{\sphilb{H}_{0}},$$
where $M_{\fnexp{\imunit \phi}}$ is multiplication by the coordinate function
$\phi
\mapsto
\fnexp{\imunit \phi}$
on the gauge circle.

\item The gauge group fixes the averaged state and acts transitively on the fibers:
$$\begin{aligned}
\oastate[\psi_{\sminvtemperature,\phi}] \circ \mathfrak{g}_{\vartheta}
=
\oastate[\psi_{\sminvtemperature,\phi + \vartheta}],
\quad
\fun{\oastate[\psi_{\sminvtemperature,\phi}]}{\sigma^{+}_{p}}
=
\frac{r_{0}}{2} \fnexp{\imunit \phi}.
\end{aligned}$$

\item If in addition
$\sminvtemperature_{c} \abs{\varepsilon}
<
1$,
then as $\sminvtemperature
\to
\infty$ the thermal fiber states converge pointwise on $\oa{A}$,
uniformly in $\phi$,
to the pure product states of Definition \ref{def:main-ground-orbit}.
The limit Gibbs states converge to $\oastate[\psi_{\txtgs}]$ and their central decomposition measures converge weakly to the ground-state central decomposition measure.
\end{enumerate}
\end{thm}

The state limit is Theorem \ref{thm:limit-state}, and the Green-function limit is Theorem \ref{thm:green}. The fiber properties, the GNS direct integral, the center, and gauge covariance are collected in Theorem \ref{thm:central-decomposition}. The zero-temperature conclusions are Corollary \ref{cor:zero-temperature-central-decomposition}.

For \(\sminvtemperature
\leq
\sminvtemperature_{c}\), Theorem \ref{thm:central-decomposition}(5) gives a single factor product state and a trivial center.

\section{The Quasi-Spin Algebra and the BCS Hamiltonian}\label{sec:algebra}

The observable algebra used throughout the paper is the UHF algebra \(\oa{A}\) of type \(2^{\infty}\), with local subalgebras \(\oa{A}_{\Lambda}
=
\bigotimes_{p
\in
\Lambda} \spmat{2}{\fldcmp}\) for finite \(\Lambda
\subset
\semigrposint\). The operator-algebraic construction and the simplicity fact used in the main text are collected in Appendix \ref{sec:oa-appendix}. The one-site Bloch sphere identities and finite-dimensional spin identities are collected in Appendix \ref{sec:bloch}.

\subsection{Quasi-spin observables and gauge symmetry}\label{quasi-spin-observables-and-gauge-symmetry}

The Pauli matrices generate the local quasi-spin observables, from which the collective spin, the BCS Hamiltonian, and the gauge automorphisms are constructed. The formulas below also fix the finite-volume permutation symmetry used in the thermal analysis.

The single-site Pauli matrices are \[
\sigma^{x}
=
\begin{pmatrix} 0 & 1 \\ 1 & 0 \end{pmatrix},
\quad
\sigma^{y}
=
\begin{pmatrix} 0 & -\imunit \\ \imunit & 0 \end{pmatrix},
\quad
\sigma^{z}
=
\begin{pmatrix} 1 & 0 \\ 0 & -1 \end{pmatrix}.
\] Set \(\physvec{\sigma}
=
\rbk{\sigma^{x}, \sigma^{y}, \sigma^{z}}\) and \[\sigma^{\pm}
=
\frac{1}{2} \rbk{\sigma^{x} \pm \imunit \sigma^{y}}.\] Their copies in the tensor factor \(p\) are written \(\sigma_{p}^{x}\), \(\sigma_{p}^{y}\), \(\sigma_{p}^{z}\), \(\physvec{\sigma}_{p}\), and \(\sigma_{p}^{\pm}\). For \(\Omega
\in
\semigrposint\), the finite integer interval of site labels is denoted by \(\intint{1..\Omega}
=
\setone{1, \cdots, \Omega}\). The collective operators over the first \(\Omega\) modes are \begin{equation}\label{eq:collective}
S_{\Omega}^{\gamma}
=
\frac{1}{2} \sum_{p
=
1}^{\Omega} \sigma_{p}^{\gamma}
\quad \rbk{\gamma
=
x, y, z},
\quad
S_{\Omega}^{\pm}
=
\sum_{p
=
1}^{\Omega} \sigma_{p}^{\pm},
\quad
\physvec{M}_{\Omega}
=
\frac{1}{\Omega} \sum_{p
=
1}^{\Omega} \physvec{\sigma}_{p}.
\end{equation} Thus \(S_{\Omega}^{\pm}
=
S_{\Omega}^{x} \pm \imunit S_{\Omega}^{y}\) and \(\physvec{M}_{\Omega}
=
\frac{2}{\Omega} \physvec{S}_{\Omega}\). The components of the mean magnetization \(\physvec{M}_{\Omega}\) are written \(M_{\Omega}^{\gamma}\) and \(M_{\Omega}^{\pm}
=
\frac{1}{\Omega} S_{\Omega}^{\pm}\). In the electron picture \(\sigma_{p}^{+}\) creates and \(\sigma_{p}^{-}\) annihilates the electron pair in mode \(p\), and \(\frac{1}{2} \rbk{1 + \sigma_{p}^{z}}\) is the pair occupation number.

The quasi-spin BCS Hamiltonian couples every pair mode to every other with equal strength. It is a mean-field Hamiltonian.

\begin{defn}[BCS Hamiltonian]
Let $\seq{\varepsilon_{p}}{p
\in
\semigrposint}$ be a bounded sequence of real numbers and define its uniform bound by
\begin{equation}\label{eq:energy-bound}
E
=
\sup_{p \in \semigrposint} \abs{\varepsilon_{p}}.
\end{equation}
Let $\sminvtemperature_{c} > 0$.
For $\Omega
\in
\semigrposint$,
the BCS Hamiltonian on the first $\Omega$ pair modes is the self-adjoint element of $\oa{A}_{\Omega}$
\begin{equation}\label{eq:bcs-hamiltonian}
\physham_{\Omega}
=
- \sum_{p
=
1}^{\Omega} \varepsilon_{p} \sigma_{p}^{z} - \frac{2 }{\sminvtemperature_{c} \Omega} S_{\Omega}^{+} S_{\Omega}^{-}
=
- \sum_{p
=
1}^{\Omega} \varepsilon_{p} \sigma_{p}^{z} - \frac{2}{\sminvtemperature_{c} \Omega} \sum_{p
=
1}^{\Omega} \sum_{p'
=
1}^{\Omega} \sigma_{p}^{+} \sigma_{p'}^{-}.
\end{equation}
The model is called degenerate when $\varepsilon_{p}
=
\varepsilon$ for all $p$.
\end{defn}

The convention \eqref{eq:bcs-hamiltonian} is that of \cite[Eq. (2)]{ThirringWalter001}. The Hamiltonian of \cite[Eq. (1)]{ThirringWehrl001} is \[\sum_{p
\leq
\Omega}
\varepsilon_{p}
\rbk{1 - \sigma_{p}^{z}} - \frac{2}{\sminvtemperature_{c} \Omega} S_{\Omega}^{-} S_{\Omega}^{+},\] which differs from \eqref{eq:bcs-hamiltonian} by the constant \(\sum_{p
\leq
\Omega} \varepsilon_{p}\) and by the operator \(\frac{2}{\sminvtemperature_{c} \Omega} \commutator{S_{\Omega}^{+}}{S_{\Omega}^{-}}
=
\frac{4}{\sminvtemperature_{c} \Omega} S_{\Omega}^{z}\). The constant is immaterial throughout, and the commutator term is a mean-field one-body term of norm \(\frac{2}{\sminvtemperature_{c}}\) which can be absorbed by the replacement \(\varepsilon_{p}
\mapsto
\varepsilon_{p} + \frac{2}{\sminvtemperature_{c} \Omega}\) and changes none of the limit statements, as will be evident from the proofs. In the electron picture \(\varepsilon_{p}\) is the single-electron kinetic energy measured from the chemical potential and \(\sminvtemperature_{c}\) is, in suitable units, the critical inverse temperature. The coupling scale in \eqref{eq:bcs-hamiltonian} is its inverse. Both facts are taken from \cite{BardeenCooperSchrieffer001,BogoliubovNikolai002} and play no logical role in the arguments of this paper.

A gauge transformation of the electrons multiplies each pair creation operator by a phase. It therefore acts on the quasi-spins as a rotation around the \(z\)-axis.

\begin{defn}[gauge automorphisms]\label{def:gauge}
For $\vartheta
\in
\fldreal$ let $u_{\vartheta}
=
\fnexp{\imunit \vartheta \sigma^{z} / 2}
\in
\spmat{2}{\fldcmp}$.
For every finite $\Lambda
\subset
\semigrposint$,
the local gauge unitary is
\begin{equation}\label{eq:local-gauge-unitary}
u_{\Lambda,\vartheta}
=
\bigotimes_{p
\in
\Lambda} u_{\vartheta}
\in
\oa{A}_{\Lambda}.
\end{equation}
The gauge automorphism $\mathfrak{g}_{\vartheta}$ of $\oa{A}$ is the unique $\ast$-automorphism with
$$\fun{\mathfrak{g}_{\vartheta}}{A}
=
u_{\Lambda,\vartheta} A \faadj{u_{\Lambda,\vartheta}},
\quad A
\in
\oa{A}_{\Lambda}.$$
The gauge group is
$\dggrgauge
=
\set{\mathfrak{g}_{\vartheta}}{\vartheta
\in
\fldreal}$.
\end{defn}

The local formulas are compatible with the embeddings \(\oa{A}_{\Lambda}
\subset
\oa{A}_{\Lambda'}\) because the conjugating unitaries differ by a unitary of \(\oa{A}_{\Lambda' \setminus \Lambda}\) which acts trivially on \(\oa{A}_{\Lambda}\). It follows that \(\mathfrak{g}_{\vartheta}\) is well defined and isometric on \(\oa{A}_{\txtloc}\) and extends to \(\oa{A}\) by continuity. The action on the generators is \begin{equation}\label{eq:gauge-action}
\fun{\mathfrak{g}_{\vartheta}}{\sigma_{p}^{\pm}}
=
\fnexp{\pm \imunit \vartheta} \sigma_{p}^{\pm},
\quad
\fun{\mathfrak{g}_{\vartheta}}{\sigma_{p}^{z}}
=
\sigma_{p}^{z},
\end{equation} since \(u_{\vartheta} \sigma^{+} \faadj{u_{\vartheta}}
=
\fnexp{\imunit \vartheta} \sigma^{+}\) by direct matrix multiplication. The family \(\rbk{\mathfrak{g}_{\vartheta}}_{\vartheta
\in
\fldreal}\) is a one-parameter group with \(\mathfrak{g}_{\vartheta + 2 \pi}
=
\mathfrak{g}_{\vartheta}\). For local \(A\), norm continuity of \(\vartheta
\mapsto
\fun{\mathfrak{g}_{\vartheta}}{A}\) follows from the finite tensor-product unitary \(u_{\Lambda,\vartheta}\) defined by \eqref{eq:local-gauge-unitary}. The density of \(\oa{A}_{\txtloc}\) and the identity \(\norm{\mathfrak{g}_{\vartheta}}
=
1\) extend it to every \(A
\in
\oa{A}\) by an \(\frac{\epsilon}{3}\)-argument.

The gauge action \eqref{eq:gauge-action} leaves the finite-volume Hamiltonian invariant: \begin{equation}\label{eq:gauge-invariance}
\fun{\mathfrak{g}_{\vartheta}}{\physham_{\Omega}}
=
\physham_{\Omega}
\quad \rbk{\vartheta
\in
\fldreal},
\end{equation} since \(\sigma_{p}^{+} \sigma_{p'}^{-}\) and \(\sigma_{p}^{z}\) are gauge invariant, and in the degenerate case \(\physham_{\Omega}\) is also invariant under all permutations of the modes \(1, \dotsc, \Omega\).

For the degenerate model the total spin gives the finite-volume spectrum. Proposition \ref{prop:degenerate-total-spin} is \cite[Eqs. (2)--(5)]{ThirringWehrl001} in the present convention.

\begin{prop}[degenerate total-spin diagonalization]\label{prop:degenerate-total-spin}
Assume $\varepsilon_{p}
=
\varepsilon$ for all $p$.
Then $\physham_{\Omega}$ is the function of the total spin operators:
\begin{equation}\label{eq:hamiltonian-total-spin}
\physham_{\Omega}
=
- 2 \varepsilon S_{\Omega}^{z} - \frac{2}{\sminvtemperature_{c} \Omega} \rbk{\physvec{S}_{\Omega}^{2} - \rbk{S_{\Omega}^{z}}^{2} + S_{\Omega}^{z}}
\end{equation}
and $\physham_{\Omega}$ is diagonal on the joint eigenspaces of $\physvec{S}_{\Omega}^{2}$ and $S_{\Omega}^{z}$.
On the joint eigenspace with quantum numbers $\fun{S}{S + 1}$ and $S_{z}$,
described in Appendix \ref{sec:spin},
the eigenvalue is
\begin{equation}\label{eq:eigenvalues}
\fun{E_{\Omega}}{S, S_{z}}
=
- 2 \varepsilon S_{z} - \frac{2}{\sminvtemperature_{c} \Omega} \rbk{\fun{S}{S + 1} - \fun{S_{z}}{S_{z} - 1}},
\end{equation}
and the two characteristic frequencies of the model are the differences
\begin{equation}\label{eq:frequencies}
\begin{aligned}
\fun{E_{\Omega}}{S, S_{z}} - \fun{E_{\Omega}}{S - 1, S_{z}}
&=
- \frac{4}{\sminvtemperature_{c} \Omega} S,
\\ 
\fun{E_{\Omega}}{S, S_{z}} - \fun{E_{\Omega}}{S, S_{z} - 1}
&=
- 2 \varepsilon + \frac{4}{\sminvtemperature_{c} \Omega} \rbk{S_{z} - 1}
\end{aligned}
\end{equation}
respectively.
\end{prop}

\begin{proof}
The spin operators $S_{\Omega}^{\gamma}$ satisfy the angular momentum relations $\commutator{S_{\Omega}^{x}}{S_{\Omega}^{y}}
=
\imunit S_{\Omega}^{z}$ and cyclic.
The raising and lowering product is
$$S_{\Omega}^{+} S_{\Omega}^{-}
=
\rbk{S_{\Omega}^{x}}^{2} + \rbk{S_{\Omega}^{y}}^{2} + \imunit \commutator{S_{\Omega}^{y}}{S_{\Omega}^{x}}
=
\physvec{S}_{\Omega}^{2} - \rbk{S_{\Omega}^{z}}^{2} + S_{\Omega}^{z}.$$
Substituting this identity into \eqref{eq:bcs-hamiltonian} gives \eqref{eq:hamiltonian-total-spin}.
The right side of \eqref{eq:hamiltonian-total-spin} is a polynomial in the commuting self-adjoint operators $\physvec{S}_{\Omega}^{2}$ and $S_{\Omega}^{z}$.
It follows that it is diagonal on their joint eigenspaces.
Substituting the joint eigenvalues $\fun{S}{S + 1}$ and $S_{z}$ gives \eqref{eq:eigenvalues}.
Finally,
\eqref{eq:frequencies} is obtained by subtracting the two corresponding values in \eqref{eq:eigenvalues}.
\end{proof}

For \(S, S_{z}\) of order \(\Omega\) both frequencies are of order one. The first is the gap frequency \(- \Delta\) and the second is twice the effective chemical potential of \cite[Eqs. (4)--(5)]{ThirringWehrl001}. The interplay between these two frequencies governs the entire representation theory of the dynamics: the Bogoliubov Hamiltonian reproduces the first frequency but freezes the second, and the second vanishes exactly on the spectral region singled out by the gap equation. This mechanism is proved in Sections \ref{sec:bogoliubov} and \ref{sec:time}.

\section{The Law of Large Numbers for Intensive Observables}\label{sec:intensive}

Intensive observables are means of one-site observables over the first \(\Omega\) modes. In product representations they obey a law of large numbers in the strong operator topology. This is the operator-algebraic form of \cite[Lemma 1]{ThirringWehrl001} and is used throughout, both in pure product representations and in the thermal product representations. The elementary strong-convergence lemmas used in this section and in later convergence arguments are collected in Appendix \ref{sec:functional-analytic-appendix}.

\subsection{Strong convergence in product representations}\label{strong-convergence-in-product-representations}

The product-state law of large numbers proved below gives strong limits for intensive observables and for their continuous functional calculus.

\begin{thm}[law of large numbers]\label{thm:lln}
Let $\oastate[\psi]
=
\bigotimes_{p \in \semigrposint} \oastate[\psi_{p}]$ be a product state of $\oa{A}$,
either pure or with invertible density matrices,
and let $\rbk{\sphilb{H}_{\psi}, \oarepn_{\psi}, \oagnsvector[\Psi]}$ be its GNS triple.
Let $X_{p}
\in
\oa{A}_{\setone{p}}$ be one-site elements with $C
=
\sup_{p \in \semigrposint} \norm{X_{p}} < \infty$ and means $x_{p}
=
\fun{\oastate[\psi_{p}]}{X_{p}}$,
and assume $\bar{x}_{\Omega}
=
\frac{1}{\Omega} \sum_{p
=
1}^{\Omega} x_{p}
\to
\bar{x}$ as $\Omega
\to
\infty$.
The averaged observables have the following strong limit on $\sphilb{H}_{\psi}$:
$$\slim_{\Omega \to \infty}
\frac{1}{\Omega} \sum_{p
=
1}^{\Omega} \fun{\oarepn_{\psi}}{X_{p}}
=
\bar{x}.$$
\end{thm}

\begin{proof}
Write $Y_{\Omega}
=
\frac{1}{\Omega} \sum_{p
\leq
\Omega} \fun{\oarepn_{\psi}}{X_{p}}$,
uniformly bounded by $C$.
Since vectors of the form $\fun{\oarepn_{\psi}}{B} \oagnsvector[\Psi]$ with $B
\in
\oa{A}_{\txtloc}$ are dense and the family $Y_{\Omega} - \bar{x}$ is uniformly bounded,
it suffices to prove $\lim_{\Omega \to \infty}
\norm{\rbk{Y_{\Omega} - \bar{x}} \fun{\oarepn_{\psi}}{B} \oagnsvector[\Psi]}
=
0$ for local $B$,
say $B
\in
\oa{A}_{\Lambda}$.
Split
$$\rbk{Y_{\Omega} - \bar{x}} \fun{\oarepn_{\psi}}{B} \oagnsvector[\Psi]
=
\commutator{Y_{\Omega}}{\fun{\oarepn_{\psi}}{B}} \oagnsvector[\Psi] + \fun{\oarepn_{\psi}}{B} \rbk{Y_{\Omega} - \bar{x}_{\Omega}} \oagnsvector[\Psi] + \rbk{\bar{x}_{\Omega} - \bar{x}} \fun{\oarepn_{\psi}}{B} \oagnsvector[\Psi].$$
The commutator involves only the sites of $\Lambda$.
It follows that $$\norm{\commutator{Y_{\Omega}}{\fun{\oarepn_{\psi}}{B}}}
\leq
\frac{1}{\Omega} \sum_{p
\in
\Lambda} 2 \norm{X_{p}} \norm{B}
\leq
\frac{2 C \abscard{\Lambda} \norm{B}}{\Omega}
\to
0
\quad
\rbk{\Omega \to \infty},$$
where the last term tends to $0$ by hypothesis.
For the middle term,
with $\tilde{X}_{p}
=
X_{p} - x_{p}$,
$$\norm{\rbk{Y_{\Omega} - \bar{x}_{\Omega}} \oagnsvector[\Psi]}^{2}
=
\frac{1}{\Omega^{2}} \sum_{p, q
\leq
\Omega} \fun{\oastate[\psi]}{\faadj{\tilde{X}_{p}} \tilde{X}_{q}}.$$
For $p
\neq
q$,
the two centered operators are supported on the disjoint one-site algebras $\oa{A}_{\setone{p}}$ and $\oa{A}_{\setone{q}}$.
The defining factorization property of $\oastate[\psi]
=
\bigotimes_{r
\in
\semigrposint} \oastate[\psi_{r}]$ therefore gives
$$\begin{aligned}
\fun{\oastate[\psi]}{\faadj{\tilde{X}_{p}} \tilde{X}_{q}}
=
\fun{\oastate[\psi_{p}]}{\faadj{\tilde{X}_{p}}}
\fun{\oastate[\psi_{q}]}{\tilde{X}_{q}}
=
\rbk{\fun{\oastate[\psi_{p}]}{\faadj{X_{p}}} - \cmpconj{x_{p}}}
\rbk{\fun{\oastate[\psi_{q}]}{X_{q}} - x_{q}}
=
0.
\end{aligned}$$
For $p
=
q$,
positivity of the state yields the variance estimate
$$\begin{aligned}
\fun{\oastate[\psi]}{\faadj{\tilde{X}_{p}} \tilde{X}_{p}}
=
\fun{\oastate[\psi_{p}]}{\faadj{X_{p}} X_{p}} - \abs{x_{p}}^{2}
\leq
\fun{\oastate[\psi_{p}]}{\faadj{X_{p}} X_{p}}
\leq
\norm{X_{p}}^{2}
\leq
C^{2}.
\end{aligned}$$
Thus only the $\Omega$ diagonal terms remain,
and the squared norm is at most $\frac{C^{2}}{\Omega}
\to
0$ as $\Omega \to \infty$.
\end{proof}

For a real unit vector \(\physvec{u}
\in
\mansphere^{2}\) set \[P_{\physvec{u}}
=
\frac{1}{2} \rbk{1 + \physvec{\sigma} \cdot \physvec{u}},
\quad
P_{- \physvec{u}}
=
1 - P_{\physvec{u}}.\] Choose unit vectors \(\xi_{\physvec{u}}^{\pm}
\in
\Ran P_{\pm \physvec{u}}\). For every \(\physvec{u}\) that occurs, also fix a right-handed orthonormal frame \(\rbk{\physvec{e}_{1}, \physvec{e}_{2}, \physvec{u}}\) and set \[\physvec{u}^{\pm}
=
\frac{1}{2} \rbk{\physvec{e}_{1} \pm \imunit \physvec{e}_{2}},
\quad
f^{\pm}
=
\physvec{\sigma} \cdot \physvec{u}^{\pm}.\] A spin configuration is a sequence \(\physvec{\omega}
=
\seq{\physvec{u}_{p}}{p
\in
\semigrposint}\). It determines the one-site vectors \[\xi_{\physvec{\omega},p}^{\pm}
=
\xi_{\physvec{u}_{p}}^{\pm}.\]

The spin configuration specifies a product-state background. The Bloch vector \(\physvec{u}_{p}\) gives its one-site expectations: the longitudinal component fixes pair occupation, and the transverse part fixes the pair amplitude.

\begin{defn}[mean profile and mean polarization]\label{def:mean-profile}
The spin configuration $\physvec{\omega}$ has a mean profile if the empirical measures
$\lambda_{\physvec{\omega},\Omega}
=
\frac{1}{\Omega} \sum_{p
=
1}^{\Omega} \delta_{\physvec{u}_{p}}$
converge weakly to a probability measure $\lambda_{\physvec{\omega}}$ on $\mansphere^{2}$.
The mean polarization of such a configuration is
\begin{equation}\label{eq:mean-polarization}
\eta_{\physvec{\omega}} \bar{\physvec{u}}_{\physvec{\omega}}
=
\int_{\mansphere^{2}}
\physvec{u}
\opdmsr{\fun{\lambda_{\physvec{\omega}}}{\physvec{u}}}
=
\lim_{\Omega \to \infty}
\frac{1}{\Omega}
\sum_{p
=
1}^{\Omega} \physvec{u}_{p},
\quad 0
\leq
\eta_{\physvec{\omega}}
\leq
1,
\end{equation}
with $\bar{\physvec{u}}_{\physvec{\omega}}$ a unit vector when $\eta_{\physvec{\omega}} > 0$.
\end{defn}

The measure \(\lambda_{\physvec{\omega}}\) records the asymptotic distribution of Bloch directions, with the mode ordering discarded as appropriate for unweighted mean-field averages. Its first moment gives the pair-density imbalance through its longitudinal component and the Cooper-pair amplitude through its transverse component. The full profile retains higher distributional information, but it does not determine the product sector: configurations with the same profile may have nonsummable tail differences and hence disjoint representations by Lemma \ref{lem:geometric-equivalence} and Proposition \ref{prop:disjointness}.

For every continuous \(g\) on \(\mansphere^{2}\), Definition \ref{def:mean-profile} gives \[\frac{1}{\Omega} \sum_{p
\leq
\Omega} \fun{g}{\physvec{u}_{p}}
\to
\int_{\mansphere^{2}}
\fun{g}{\physvec{u}}
\opdmsr{\fun{\lambda_{\physvec{\omega}}}{\physvec{u}}}
\quad
\rbk{\Omega
\to
\infty}.\] Appendix \ref{sec:bloch} proves the one-site identities associated with the choices above.

\begin{defn}[product representation of a spin configuration]\label{def:spin-product-representation}
The reference sequence of $\physvec{\omega}$ is
$$\xi_{\physvec{\omega}}
=
\seq{\xi_{\physvec{\omega},p}^{+}}{p
\in
\semigrposint}.$$
The product representation associated with $\physvec{\omega}$ is the representation
$\oarepn_{\physvec{\omega}}$ of $\oa{A}$ on the incomplete tensor product space
$\sphilb{H}_{\physvec{\omega}}$ with reference sequence $\xi_{\physvec{\omega}}$.
The corresponding product state $\oastate[\psi_{\physvec{\omega}}]$ is defined on local elementary tensors by
\begin{equation}\label{eq:spin-product-state}
\fun{\oastate[\psi_{\physvec{\omega}}]}{\bigotimes_{p
\in
\Lambda} a_{p}}
=
\prod_{p
\in
\Lambda}
\bkt{\xi_{\physvec{\omega},p}^{+}}{a_{p} \xi_{\physvec{\omega},p}^{+}}
\quad
\rbk{\Lambda
\subset
\semigrposint \text{ finite},
\ a_{p}
\in
\spmat{2}{\fldcmp}}.
\end{equation}
\end{defn}

This representation describes fluctuations above \(\physvec{\omega}\): local observables generate finite flips and the Hilbert-space completion their superpositions. Thus \(\physvec{\omega}\) fixes the boundary condition at infinity and hence the GNS sector.

The detailed construction is recalled in Definitions \ref{def:itps} and \ref{def:product-rep}.

\begin{cor}[mean spin in a product representation]\label{cor:mean-spin}
Let $\physvec{\omega}$ be a spin configuration with mean profile
and mean polarization $\eta_{\physvec{\omega}} \bar{\physvec{u}}_{\physvec{\omega}}$.
For each $\gamma$ it holds that
$\fun{\oarepn_{\physvec{\omega}}}{M_{\Omega}^{\gamma}}
\to
\eta_{\physvec{\omega}} \bar{u}_{\physvec{\omega}}^{\gamma}$ strongly as $\Omega \to \infty$,
and for every $\physvec{b}
\in
\fldreal^{3}$ and $t
\in
\fldreal$,
$\fun{\oarepn_{\physvec{\omega}}}{\fnexp{\imunit t \physvec{M}_{\Omega} \cdot \physvec{b}}}
\to
\fnexp{\imunit t \eta_{\physvec{\omega}} \bar{\physvec{u}}_{\physvec{\omega}} \cdot \physvec{b}}$ strongly as $\Omega \to \infty$,
uniformly on $t$-compacts.
If in addition $\lim_{\Omega \to \infty}
\frac{1}{\Omega} \sum_{p
\leq
\Omega} \varepsilon_{p} \physvec{u}_{p}
=
\physvec{k}_{\physvec{\omega}}$,
then the normalized Hamiltonians converge strongly:
$$\slim_{\Omega \to \infty}
\fun{\oarepn_{\physvec{\omega}}}{\frac{\physham_{\Omega}}{\Omega}}
=
-k_{\physvec{\omega}}^{z}
-\frac{\eta_{\physvec{\omega}}^{2}}{2 \sminvtemperature_{c}}
\rbk{\rbk{\bar{u}_{\physvec{\omega}}^{x}}^{2} + \rbk{\bar{u}_{\physvec{\omega}}^{y}}^{2}}.$$
\end{cor}

\begin{proof}
The one-site state at slot $p$ is pure with Bloch vector $\physvec{u}_{p}$.
It follows that $\fun{\oastate[\psi_{p}]}{\sigma_{p}^{\gamma}}
=
u_{p}^{\gamma}$ and Theorem \ref{thm:lln} applies to $X_{p}
=
\sigma_{p}^{\gamma}$ with $\bar{x}
=
\eta_{\physvec{\omega}} \bar{u}_{\physvec{\omega}}^{\gamma}$.
This proves the first limit.
Lemma \ref{lem:strong-exponentials} gives the second.
Its generators $\physvec{M}_{\Omega} \cdot \physvec{b}$ are self-adjoint and uniformly bounded by $\abs{\physvec{b}}$-multiples.
For the Hamiltonian,
$\frac{\physham_{\Omega}}{\Omega}
=
- \frac{1}{\Omega} \sum_{p
\leq
\Omega} \varepsilon_{p} \sigma_{p}^{z}
- \frac{2}{\sminvtemperature_{c}} M_{\Omega}^{+} M_{\Omega}^{-}$
by \eqref{eq:bcs-hamiltonian} and \eqref{eq:collective}.
The first term converges strongly to $- k_{\physvec{\omega}}^{z}$ by Theorem \ref{thm:lln} with $X_{p}
=
\varepsilon_{p} \sigma_{p}^{z}$.
For the interaction term,
$M_{\Omega}^{\pm}
=
\frac{1}{2} \rbk{M^{x}_{\Omega} \pm \imunit M^{y}_{\Omega}}
\to
\frac{\eta_{\physvec{\omega}}}{2} \rbk{\bar{u}_{\physvec{\omega}}^{x} \pm \imunit \bar{u}_{\physvec{\omega}}^{y}}$ strongly as $\Omega \to \infty$ with uniform bounds.
It follows that Lemma \ref{lem:strong-products} gives $M^{+}_{\Omega} M^{-}_{\Omega}
\to
\frac{\eta_{\physvec{\omega}}^{2}}{4} \rbk{\rbk{\bar{u}_{\physvec{\omega}}^{x}}^{2} + \rbk{\bar{u}_{\physvec{\omega}}^{y}}^{2}}$ as $\Omega
\to
\infty$.
\end{proof}

The law of large numbers cannot be improved to convergence in norm. This is the remark following \cite[Lemma 1]{ThirringWehrl001}.

\begin{prop}[no uniform convergence]
For every unit vector $\physvec{u}
\in
\mansphere^{2}$,
every $c
\in
\fldcmp$,
every spin configuration $\physvec{\omega}$,
and every $\Omega$,
it follows that
$$\norm{\fun{\oarepn_{\physvec{\omega}}}{\physvec{M}_{\Omega} \cdot \physvec{u}} - c}
\geq
1.$$
\end{prop}

\begin{proof}
By the factorization of Proposition \ref{prop:factorization} the representation $\oarepn_{\physvec{\omega}}$
restricted to $\oa{A}_{\Omega}$ is unitarily equivalent to $A
\mapsto
A \otimes 1$.
It follows that it preserves norms and spectra of elements of $\oa{A}_{\Omega}$.
The operator
$\physvec{M}_{\Omega} \cdot \physvec{u}
=
\frac{1}{\Omega} \sum_{p
\leq
\Omega} \physvec{\sigma}_{p} \cdot \physvec{u}$
has the product vectors
$\rbk{\xi_{\physvec{u}}^{+}}^{\otimes \Omega}$ and $\rbk{\xi_{\physvec{u}}^{-}}^{\otimes \Omega}$ as eigenvectors with eigenvalues $\pm 1$.
It follows that its spectrum contains $+1$ and $-1$.
It follows that
$\norm{\physvec{M}_{\Omega} \cdot \physvec{u} - c}
\geq
\max \rbk{\abs{1 - c}, \abs{1 + c}}
\geq
1$
for every $c
\in
\fldcmp$.
\end{proof}

This proposition distinguishes the sectorwise strong convergence of Corollary \ref{cor:mean-spin} from norm convergence. The corresponding later results are Theorem \ref{thm:bogoliubov}(1) for shifted extensive Hamiltonians, Theorem \ref{thm:dynamics} for the Heisenberg dynamics, and Theorem \ref{thm:central-decomposition}(2) for the direct-integral representation.

\section{Extensive Observables}\label{sec:extensive}

Extensive observables grow with \(\Omega\). After centering they converge at best on restricted domains, and the residual objects, the relative number operator and the gauge generator, exist only relative to an equivalence class. The results below recover \cite[Lemmas 2, 3]{ThirringWehrl001} and the remark on von Neumann's theorem in \cite[Section 3]{ThirringWehrl001}. Fix a spin configuration \(\physvec{\omega}
= \seq{\physvec{u}_p}{p \in \semigrposint}\). Definition \ref{def:spin-product-representation} supplies the reference sequence \(\xi_{\physvec{\omega}}\), the product representation \(\oarepn_{\physvec{\omega}}\), and its incomplete tensor product space \(\sphilb{H}_{\physvec{\omega}}\). Proposition \ref{prop:flip-basis} constructs the flip basis \(\set{\xi_{\physvec{\omega}}^{F}}{F \text{ is a finite subset of } \semigrposint}\) from the one-site bases \(\setone{\xi_{\physvec{\omega},p}^{+}, \xi_{\physvec{\omega},p}^{-}}\). Let \(\sphilb{D}_{\physvec{\omega}}\) be the dense subspace of finite linear combinations of these flip vectors.

\subsection{The relative number operator}\label{the-relative-number-operator}

The relative number operator constructed below counts flips from a fixed spin configuration and has a natural diagonal domain and spectrum.

The physical pair number over the first \(\Omega\) modes is \[\sum_{p
=
1}^{\Omega}
\frac{1}{2}
\rbk{1 + \sigma^{z}_{p}}.\] It is extensive and generates finite-volume gauge rotations up to a scalar phase. By contrast, \(N_{\Omega,\physvec{\omega}}\) counts reversals relative to the reference background: it counts pairs for \(\physvec{u}_{p}=\vecbk{0,0,-1}\), pair holes for \(\physvec{u}_{p}=\vecbk{0,0,1}\), and rotated Bogoliubov excitations for a transverse reference.

\begin{prop}[relative number operator]\label{prop:number}
Let $N_{\Omega,\physvec{\omega}}
=
\sum_{p
=
1}^{\Omega} \frac{1}{2} \rbk{1 - \physvec{\sigma}_{p} \cdot \physvec{u}_{p}}
\in
\oa{A}_{\Omega}$.
Define $N_{\physvec{\omega}}^{0}$ on $\sphilb{D}_{\physvec{\omega}}$ by $N_{\physvec{\omega}}^{0} \xi_{\physvec{\omega}}^{F}
=
\abscard{F} \xi_{\physvec{\omega}}^{F}$.
Then $N_{\physvec{\omega}}^{0}$ is essentially self-adjoint.
Let $N_{\physvec{\omega}}$ be its closure.
The number cutoffs converge on the domain of the class-relative closure:
for any $\Psi
\in
\dom N_{\physvec{\omega}}$ it holds that
$$\fun{\oarepn_{\physvec{\omega}}}{N_{\Omega,\physvec{\omega}}} \Psi
\to
N_{\physvec{\omega}} \Psi
\quad \rbk{\Omega \to \infty},$$
and $\fnexp{\imunit t \fun{\oarepn_{\physvec{\omega}}}{N_{\Omega,\physvec{\omega}}}}
\to
\fnexp{\imunit t N_{\physvec{\omega}}}$ strongly as $\Omega \to \infty$ for every $t
\in
\fldreal$.
\end{prop}

\begin{proof}
By Lemma \ref{lem:frame} the operator $\frac{1}{2} \rbk{1 - \physvec{\sigma}_{p} \cdot \physvec{u}_{p}}
=
P_{- \physvec{u}_{p}}$ acts on the slot $p$ as the projection onto the flipped basis vector,
and the flipped-basis action of its representation is
\begin{equation}\label{eq:number-action}
\fun{\oarepn_{\physvec{\omega}}}{N_{\Omega,\physvec{\omega}}} \xi_{\physvec{\omega}}^{F}
=
\abscard{F \cap \intint{1..\Omega}} \xi_{\physvec{\omega}}^{F}.
\end{equation}
The operator $N_{\physvec{\omega}}^{0}$ is symmetric and diagonal in an orthonormal basis with real eigenvalues.
The ranges of $N_{\physvec{\omega}}^{0} \pm \imunit$ contain all $\rbk{\abscard{F} \pm \imunit} \xi_{\physvec{\omega}}^{F}$.
It follows that these ranges are dense.
It follows that $N_{\physvec{\omega}}^{0}$ is essentially self-adjoint by the basic range criterion for essential self-adjointness
\cite[Theorem VIII.3 and its corollary]{ReedSimon001}.
Concretely,
the closure is the multiplication operator by $F
\mapsto
\abscard{F}$ on the coefficient space $\ell^{2}$,
which is self-adjoint on its maximal domain.
For $\Psi
=
\sum_{F \subset \semigrposint, \abscard{F} < \infty} c_{F} \xi_{\physvec{\omega}}^{F}
\in
\dom N_{\physvec{\omega}}$,
that is,
with $\sum_{F \subset \semigrposint, \abscard{F} < \infty} \abscard{F}^{2} \abs{c_{F}}^{2} < \infty$,
the norm difference satisfies
$$\norm{\rbk{\fun{\oarepn_{\physvec{\omega}}}{N_{\Omega,\physvec{\omega}}} - N_{\physvec{\omega}}} \Psi}^{2}
=
\sum_{F \subset \semigrposint, \abscard{F} < \infty} \rbk{\abscard{F \cap \intint{1..\Omega}} - \abscard{F}}^{2} \abs{c_{F}}^{2}
\to
0
\quad \rbk{\Omega \to \infty}.$$
Each term tends to $0$ and is dominated by $\abscard{F}^{2} \abs{c_{F}}^{2}$.
Dominated convergence proves the limit.
For the unitaries it holds that
$$\fnexp{\imunit t \fun{\oarepn_{\physvec{\omega}}}{N_{\Omega,\physvec{\omega}}}} \xi_{\physvec{\omega}}^{F}
=
\fnexp{\imunit t \abscard{F \cap \intint{1..\Omega}}} \xi_{\physvec{\omega}}^{F}
\to
\fnexp{\imunit t \abscard{F}} \xi_{\physvec{\omega}}^{F}
=
\fnexp{\imunit t N_{\physvec{\omega}}} \xi_{\physvec{\omega}}^{F}
\quad
\rbk{\Omega
\to
\infty},$$
and uniform boundedness extends the convergence from the basis to the whole space.
\end{proof}

The operator \(N_{\physvec{\omega}}\) counts flipped spins relative to the configuration \(\physvec{\omega}\). Its vacuum is \(\xi_{\physvec{\omega}}^{\emptyset}\), and eigenvalue \(n\) counts \(n\) flips. It has meaning only in the class of \(\physvec{\omega}\); outside that class the reference vacuum contains infinitely many relative excitations, as proved next. The failure of the finite-volume cutoffs outside this class is expressed by the following proposition.

\begin{prop}[divergence of relative-number cutoffs]\label{prop:number-divergence}
Let $\physvec{\omega}'$ be a spin configuration such that
$\sum_{p
\in
\semigrposint}
\rbk{1 - \physvec{u}_{p} \cdot \physvec{u}'_{p}}
=
\infty$.
The expectations in the reference vector of $\oarepn_{\physvec{\omega}'}$ satisfy
$$\bkt{\xi_{\physvec{\omega}'}^{\emptyset}}
{\fun{\oarepn_{\physvec{\omega}'}}{N_{\Omega,\physvec{\omega}}} \xi_{\physvec{\omega}'}^{\emptyset}}
=
\sum_{p
\leq
\Omega}
\frac{1 - \physvec{u}_{p} \cdot \physvec{u}'_{p}}{2}
\to
\infty
\quad \rbk{\Omega \to \infty},$$
and the corresponding vectors have diverging norms:
$$\lim_{\Omega \to \infty}
\norm{\fun{\oarepn_{\physvec{\omega}'}}{N_{\Omega,\physvec{\omega}}} \xi_{\physvec{\omega}'}^{\emptyset}}
=
\infty.$$
In particular,
the vector sequence
$\seq{\fun{\oarepn_{\physvec{\omega}'}}{N_{\Omega,\physvec{\omega}}} \xi_{\physvec{\omega}'}^{\emptyset}}{\Omega
\in
\semigrposint}$
does not converge in $\sphilb{H}_{\physvec{\omega}'}$.
\end{prop}

\begin{proof}
Lemma \ref{lem:overlap} gives the expectation of the projection at slot $p$:
$$\bkt{\xi_{\physvec{\omega}',p}^{+}}{P_{- \physvec{u}_{p}} \xi_{\physvec{\omega}',p}^{+}}
=
1 - \abs{\bkt{\xi_{\physvec{u}'_{p}}^{+}}{\xi_{\physvec{u}_{p}}^{+}}}^{2}
=
\frac{1 - \physvec{u}_{p} \cdot \physvec{u}'_{p}}{2}.$$
Summation over $p
\leq
\Omega$ proves the expectation formula and its divergence.
The operator $\fun{\oarepn_{\physvec{\omega}'}}{N_{\Omega,\physvec{\omega}}}$ is positive,
and $\xi_{\physvec{\omega}'}^{\emptyset}$ is a unit vector.
The Cauchy--Schwarz inequality gives
$$\bkt{\xi_{\physvec{\omega}'}^{\emptyset}}{\fun{\oarepn_{\physvec{\omega}'}}{N_{\Omega,\physvec{\omega}}} \xi_{\physvec{\omega}'}^{\emptyset}}
\leq
\norm{\fun{\oarepn_{\physvec{\omega}'}}{N_{\Omega,\physvec{\omega}}} \xi_{\physvec{\omega}'}^{\emptyset}}.$$
The norm divergence follows from the expectation divergence.
\end{proof}

Proposition \ref{prop:number-divergence} gives the precise content of the statement of \cite[Section 3]{ThirringWehrl001} that the formal number expression \(\sum_{p \in \semigrposint}
\frac{1}{2}
\rbk{1 - \physvec{\sigma}_{p} \cdot \physvec{u}_{p}}\) has no meaning in other equivalence classes.

\subsection{Gauge unitaries and the breakdown of Stone's theorem}\label{gauge-unitaries-and-the-breakdown-of-stones-theorem}

The gauge automorphisms act continuously on the quasi-local algebra but need not admit a strongly continuous unitary implementation in a fixed product representation. The summability criterion below distinguishes implementable rotations from rotations into disjoint product representations and yields the failure of Stone's theorem.

The gauge automorphisms of Definition \ref{def:gauge} rotate every Bloch vector about the \(z\)-axis. Let \(R_{\vartheta}\) denote the rotation of \(\fldreal^{3}\) characterized by \(b^{x} + \imunit b^{y}
=
\fnexp{- \imunit \vartheta} \rbk{a^{x} + \imunit a^{y}}\), \(b^{z}
=
a^{z}\), for \(\physvec{b}
=
R_{\vartheta} \physvec{a}\). Thus \(R_{\vartheta}\) is the rotation through the angle \(- \vartheta\). Linearity of \eqref{eq:gauge-action} in the components gives \begin{equation}\label{eq:gauge-rotation}
\fun{\mathfrak{g}_{\vartheta}}{\physvec{\sigma}_{p} \cdot \physvec{a}}
=
\physvec{\sigma}_{p} \cdot R_{\vartheta} \physvec{a}
\quad \rbk{\physvec{a}
\in
\fldreal^{3}}.
\end{equation} Evaluating \eqref{eq:gauge-rotation} in the pure product state of a configuration \(\physvec{\omega}\) gives \begin{equation}\label{eq:gauge-on-states}
\oastate[\psi_{\physvec{\omega}}] \circ \mathfrak{g}_{\vartheta}
=
\oastate[\psi_{R_{- \vartheta} \physvec{\omega}}],
\quad
R_{- \vartheta} \physvec{\omega}
=
\seq{R_{- \vartheta} \physvec{u}_{p}}{p \in \semigrposint},
\end{equation} verified on the generators: \(\fun{\oastate[\psi_{\physvec{\omega}}]}{\fun{\mathfrak{g}_{\vartheta}}{\physvec{\sigma}_{p} \cdot \physvec{a}}}
=
\physvec{u}_{p} \cdot R_{\vartheta} \physvec{a}
=
\rbk{R_{- \vartheta} \physvec{u}_{p}} \cdot \physvec{a}\), and extended to all of \(\oa{A}\) by multiplicativity of product states over slots and density.

\begin{prop}[implementability of the gauge group]\label{prop:gauge-implementable}
Let $\physvec{\omega}
=
\seq{\physvec{u}_{p}}{p \in \semigrposint}$ be a spin configuration and $s_{p}
=
1 - \rbk{u_{p}^{z}}^{2}$.
\begin{enumerate}
\item If $\sum_{p \in \semigrposint} s_{p} < \infty$,
then for every $\vartheta$ there is a unitary $V_{\vartheta}$ on $\sphilb{H}_{\physvec{\omega}}$
with
$V_{\vartheta} \fun{\oarepn_{\physvec{\omega}}}{A} \faadj{V_{\vartheta}}
=
\fun{\oarepn_{\physvec{\omega}}}{\fun{\mathfrak{g}_{\vartheta}}{A}}$
for all $A
\in
\oa{A}$.
\item If the transverse profile satisfies
\begin{equation}\label{eq:gauge-transverse-divergence}
\sum_{p
\in
\semigrposint} \rbk{1 - \rbk{u_{p}^{z}}^{2}}
=
\infty,
\end{equation}
then for every $\vartheta
\notin
2 \pi \ringratint$ the GNS representations of $\oastate[\psi_{\physvec{\omega}}]$ and of $\oastate[\psi_{\physvec{\omega}}] \circ \mathfrak{g}_{\vartheta}$ are disjoint.
No implementing unitary exists,
and the gauge symmetry is spontaneously broken in the product representation $\oarepn_{\physvec{\omega}}$.
\end{enumerate}
\end{prop}

The summable alternative is the sector-preserving side of the distinction later described by Theorem \ref{thm:vn-commutant}. The divergent alternative is the hypothesis of Proposition \ref{prop:stone-fails}. The same disjoint-sector consequence appears in Proposition \ref{prop:fiber-disjoint} and enters Theorem \ref{thm:central-decomposition} for the superconducting gauge orbit.

\begin{proof}
Note that $\physvec{\sigma} \cdot \rbk{R_{- \vartheta} \physvec{u}_{p}}
=
\faadj{u_{\vartheta}} \rbk{\physvec{\sigma} \cdot \physvec{u}_{p}} u_{\vartheta}$
has $\faadj{u_{\vartheta}} \xi_{\physvec{\omega},p}^{+}$ as $+1$-eigenvector.
By \eqref{eq:gauge-rotation} and \eqref{eq:gauge-on-states} the GNS representation of
$\oastate[\psi_{\physvec{\omega}}] \circ \mathfrak{g}_{\vartheta}$ is the product representation of the rotated configuration
$R_{- \vartheta} \physvec{\omega}$,
realized by the one-site vectors $\faadj{u_{\vartheta}} \xi_{\physvec{\omega},p}^{+}$.
The one-site overlap is computed from
$\faadj{u_{\vartheta}}
=
\fun{\cos}{\frac{\vartheta}{2}} - \imunit \fun{\sin}{\frac{\vartheta}{2}} \sigma^{z}$
and $\bkt{\xi_{\physvec{\omega},p}^{+}}{\sigma^{z} \xi_{\physvec{\omega},p}^{+}}
=
u^{z}_{p}$.
The resulting overlap identity is
\begin{equation}\label{eq:gauge-overlap}
\abs{\bkt{\xi_{\physvec{\omega},p}^{+}}{\faadj{u_{\vartheta}} \xi_{\physvec{\omega},p}^{+}}}^{2}
=
\abs{\fun{\cos}{\frac{\vartheta}{2}} - \imunit \fun{\sin}{\frac{\vartheta}{2}} u_{p}^{z}}^{2}
=
1 - s_{p} \sin^{2} \frac{\vartheta}{2}.
\end{equation}

(1)
Suppose that $\sum_{p
\in
\semigrposint} s_{p}
<
\infty$.
Since $1 - \abs{w}
\leq
1 - \abs{w}^{2}$ for $\abs{w}
\leq
1$,
the sequences $\seq{\xi_{\physvec{\omega},p}^{+}}{p
\in
\semigrposint}$ and $\seq{\faadj{u_{\vartheta}} \xi_{\physvec{\omega},p}^{+}}{p
\in
\semigrposint}$ are weakly equivalent.
It follows that Proposition \ref{prop:disjointness}(1) provides a unitary $W$ between the two product representations.
Composing with the identification of the GNS representation of $\oastate[\psi_{\physvec{\omega}}] \circ \mathfrak{g}_{\vartheta}$ gives $V_{\vartheta}$.

(2)
Suppose instead that \eqref{eq:gauge-transverse-divergence} holds and $\vartheta
\notin
2 \pi \ringratint$.
Using \eqref{eq:gauge-overlap} and $1 - \abs{w}
\geq
\frac{1}{2} \rbk{1 - \abs{w}^{2}}$,
\begin{equation}\label{eq:gauge-divergence}
\sum_{p \in \semigrposint}
\rbk{1 - \abs{\bkt{\xi_{\physvec{\omega},p}^{+}}{\faadj{u_{\vartheta}} \xi_{\physvec{\omega},p}^{+}}}}
\geq
\frac{\sin^{2} \frac{\vartheta}{2}}{2} \sum_{p \in \semigrposint} s_{p}
=
\infty,
\end{equation}
It follows that the two sequences are not weakly equivalent and Proposition \ref{prop:disjointness}(2) gives disjointness.
\end{proof}

The infinite tensor product notation uses the following equivalence relation on sequences of vectors. Let \(\seq{h_{p}}{p
\geq
1}\) be a sequence of separable Hilbert spaces with \(2
\leq
\dim h_{p}
\leq
\infty\). In Proposition \ref{prop:stone-fails} the spaces are \(h_{p}
=
\fldcmp^{2}\).

\begin{defn}[adapted sequences and equivalence]\label{def:c-sequences}
A sequence $\xi
=
\seq{\xi_{p}}{p
\geq
1}$ with $\xi_{p}
\in
h_{p}$ is called adapted when $\norm{\xi_{p}}
=
1$ for all but finitely many $p$.
An adapted sequence with unit entries is an adapted sequence $\xi$
such that $\norm{\xi_{p}}
=
1$ for every $p$.
Two adapted sequences $\xi, \eta$ are equivalent,
written $\xi
\approx
\eta$,
when the absolute-defect sum is finite:
$$\sum_{p \in \semigrposint} \abs{1 - \bkt{\xi_{p}}{\eta_{p}}}
<
\infty,$$
and weakly equivalent,
written $\xi
\approx
_{w} \eta$,
when the phase-free defect sum is finite:
$$\sum_{p \in \semigrposint} \abs{1 - \abs{\bkt{\xi_{p}}{\eta_{p}}}}
<
\infty.$$
The equivalence class of $\xi$ is written $\fun{C}{\xi}$
and the weak equivalence class $\fun{C_{w}}{\xi}$.
\end{defn}

Unit entries remove irrelevant one-site normalizations without losing product states; finitely many nonunit entries instead describe vectors inside a fixed sector. Strong equivalence keeps precisely the summable tail changes that define one incomplete tensor product, whereas a nonsummable tail change gives a different boundary condition inaccessible to quasi-local observables. Weak equivalence additionally removes slotwise phases: it classifies product representations up to unitary equivalence, with disjointness and the phase relation given by Proposition \ref{prop:disjointness} and Lemma \ref{lem:phase-unitary}.

The global gauge unitaries require the complete product space, not only one incomplete tensor product sector. The incomplete tensor product space is defined in Definition \ref{def:itps}.

\begin{defn}[complete product space]\label{def:complete}
Let $\mathcal{W}$ be a set of adapted sequences with unit entries
in the sense of Definition \ref{def:c-sequences}.
Require that if $\xi
\in
\mathcal{W}$ and $\eta$ is an adapted sequence with unit entries satisfying $\eta
\approx
_{w} \xi$,
then $\eta
\in
\mathcal{W}$.
The associated complete product space is the orthogonal direct sum
$$\widehat{\sphilb{H}}_{\mathcal{W}}
=
\bigoplus_{C
\in
\mathcal{W} / {\approx}} \sphilb{H}_{C},$$
where $C$ runs over the equivalence classes contained in $\mathcal{W}$
and $\sphilb{H}_{C}$ is the incomplete tensor product space of any reference sequence of the class $C$.
A product vector $\bigotimes_{p
\in
\semigrposint} \eta_{p}$
for
$\eta
\in
\mathcal{W}$
is regarded as an element of the summand of its class.
\end{defn}

Unit sequences suffice to index physical backgrounds; finite nonunit changes already belong to their summands. Weak closure adds the phase-shifted strong classes required by gauge rotations, without introducing sectors based only on products of norms.

On the complete product space the gauge group is implemented for every configuration, but not continuously. This is the operator-theoretic reason why a global number operator cannot exist, and it resolves the apparent conflict with von Neumann's theorem discussed in \cite[Section 3]{ThirringWehrl001}.

\begin{prop}[global gauge unitaries are not weakly continuous]\label{prop:stone-fails}
Let $\widehat{\sphilb{H}}$ be the complete product space over all adapted unit sequences in $\prod_{p \in \semigrposint} \fldcmp^{2}$,
and let $\widehat{U}_{\vartheta}$ be the unitary of $\widehat{\sphilb{H}}$ acting slotwise by $\faadj{u_{\vartheta}}$.
Let $\physvec{\omega}$ be a configuration satisfying \eqref{eq:gauge-transverse-divergence}.
Then for every product vector $\eta$ in the class of $\xi_{\physvec{\omega}}$ and every $\vartheta
\notin
2 \pi \ringratint$,
$$\bkt{\eta}{\widehat{U}_{\vartheta} \eta}
=
0,$$
while $\bkt{\eta}{\widehat{U}_{0} \eta}
=
1$.
The map $\vartheta
\mapsto
\widehat{U}_{\vartheta}$ is not weakly continuous,
and there is no self-adjoint $\widehat{N}$ with $\widehat{U}_{\vartheta}
=
\fnexp{\imunit \vartheta \widehat{N}}$.
\end{prop}

\begin{proof}
Applying the fixed unitary $\faadj{u_{\vartheta}}$ in every slot maps adapted sequences to adapted sequences and preserves all slotwise inner products.
It therefore preserves equivalence,
weak equivalence,
and all products of inner products.
By Definition \ref{def:complete} it therefore induces a unitary $\widehat{U}_{\vartheta}$ of $\widehat{\sphilb{H}}$ permuting the class summands,
and $\vartheta
\mapsto
\widehat{U}_{\vartheta}$ is a one-parameter group by the group property of $u_{\vartheta}$.
Fix $\vartheta
\notin
2 \pi \ringratint$ and a product vector $\eta$ with $\eta
\approx
\xi_{\physvec{\omega}}$.
Write
$$\faadj{u_{\vartheta}} \xi_{\physvec{\omega}}
=
\seq{\faadj{u_{\vartheta}} \xi_{\physvec{\omega},p}^{+}}{p
\in
\semigrposint}.$$
The slotwise images satisfy $\faadj{u_{\vartheta}} \eta
\approx
\faadj{u_{\vartheta}} \xi_{\physvec{\omega}}$,
because slotwise unitaries preserve the defining sums of Definition \ref{def:c-sequences}.
If the classes $\fun{C}{\eta}$ and $\fun{C}{\faadj{u_{\vartheta}} \eta}$ were equal,
then in particular $\eta
\approx
_{w} \faadj{u_{\vartheta}} \eta$.
Lemma \ref{lem:transitivity} would then give the chain
$\xi_{\physvec{\omega}}
\approx
_{w} \eta
\approx
_{w} \faadj{u_{\vartheta}} \eta
\approx
_{w} \faadj{u_{\vartheta}} \xi_{\physvec{\omega}}$.
This chain implies weak equivalence of $\xi_{\physvec{\omega}}$ and $\faadj{u_{\vartheta}} \xi_{\physvec{\omega}}$.
This contradicts \eqref{eq:gauge-divergence}.
It follows that the two classes differ.
The vectors $\eta$ and $\widehat{U}_{\vartheta} \eta$ therefore lie in orthogonal summands of $\widehat{\sphilb{H}}$,
and $\bkt{\eta}{\widehat{U}_{\vartheta} \eta}
=
0$ by Proposition \ref{prop:complete-inner}.
At $\vartheta
=
0$ the inner product is $\norm{\eta}^{2}
=
1$.
It follows that weak continuity fails at $0$.
By Stone's theorem \cite[Theorem VIII.8]{ReedSimon001} a generator would imply weak continuity.
\end{proof}

\begin{cor}[absence of a gauge generator]
Suppose that \eqref{eq:gauge-transverse-divergence} holds.
No unitary on the fixed GNS space $\sphilb{H}_{\physvec{\omega}}$ implements a nontrivial gauge rotation.
On the complete product space,
the slotwise rotations define the unitary representation
$\fldreal
\ni
\vartheta
\mapsto
\widehat{U}_{\vartheta}$.
This representation is not weakly continuous and has no self-adjoint Stone generator.
\end{cor}

\begin{proof}
The first assertion is Proposition \ref{prop:gauge-implementable}(2).
The second assertion is Proposition \ref{prop:stone-fails}.
\end{proof}

\begin{rem}[fermionic interpretation of the gauge obstruction]\label{rem:fermionic-gauge-obstruction}
Here gauge is the global particle-number $\fun{\liegr{U}}{1}$ symmetry of reduced BCS,
not local electromagnetic redundancy;
the quasi-spin algebra also omits singly occupied blocked states.
Although the finite-volume Hamiltonian is gauge invariant by \eqref{eq:gauge-invariance},
a nontrivial phase rotation connects disjoint broken-symmetry sectors by
Propositions \ref{prop:gauge-implementable} and \ref{prop:stone-fails}.
Hence the all-sector action has no Stone generator,
while Proposition \ref{prop:number} concerns excitations within one phase
\cite{BardeenCooperSchrieffer001,RudolfHaag005,EmchGuenin001,ThirringWehrl001}.
\end{rem}

The weak-closure formulation appears in Theorem \ref{thm:vn-commutant}. On the direct sum associated with a fixed weak equivalence class, the weak closure \(\oadoublecommutant{\fun{\oarepn}{\oa{A}}}\) preserves every class summand \(\sphilb{H}_{C}\). A gauge unitary carrying \(\sphilb{H}_{C}\) to an orthogonal class summand therefore cannot be a weak-operator limit of represented quasi-spin observables. This obstruction is distinct from Proposition \ref{prop:number}. That proposition constructs the class-relative limit \(N_{\physvec{\omega}}\) on \(\sphilb{H}_{\physvec{\omega}}\), not a generator of the discontinuous global gauge representation.

\subsection{The centered interaction}\label{the-centered-interaction}

The interaction of \eqref{eq:bcs-hamiltonian} is extensive. After centering and normalization, it converges in matrix elements on the flip domain, but the convergence is not strong and its operator norm grows with \(\Omega\). Fix a spin configuration \(\physvec{\omega}\) and the product state \(\oastate[\psi_{\physvec{\omega}}]\) defined by \eqref{eq:spin-product-state}. The resulting estimates prove \cite[Lemma 3]{ThirringWehrl001} together with the two remarks following it. Define, for the configuration \(\physvec{\omega}\), \begin{equation}\label{eq:centered-quantities}
\begin{gathered}
m^{\pm}_{\physvec{\omega},p}
=
\frac{1}{2} \rbk{u^{x}_{p} \pm \imunit u^{y}_{p}},
\quad
d^{\gamma}_{\physvec{\omega},p}
=
\sigma^{\gamma}_{p} - u^{\gamma}_{p}
\quad
\rbk{\gamma
\in
\setone{x, y, z}},
\quad
d^{\pm}_{\physvec{\omega},p}
=
\sigma^{\pm}_{p} - m^{\pm}_{\physvec{\omega},p},
\\ 
D^{\pm}_{\physvec{\omega},\Omega}
=
\sum_{p
=
1}^{\Omega} d^{\pm}_{\physvec{\omega},p},
\quad
R_{\physvec{\omega},\Omega}
=
\frac{1}{\Omega} D^{+}_{\physvec{\omega},\Omega} D^{-}_{\physvec{\omega},\Omega},
\end{gathered}
\end{equation} The definitions give \(m^{\pm}_{\physvec{\omega},p}
=
\fun{\oastate[\psi_{\physvec{\omega}}]}{\sigma^{\pm}_{p}}\) and \(\fun{\oastate[\psi_{\physvec{\omega}}]}{d^{\pm}_{\physvec{\omega},p}}
=
0\). Define the one-site flip coefficients by \[\upsilon_{\physvec{\omega},p}^{-}
=
\rbk{u^{+}_{p}}^{x} - \imunit \rbk{u^{+}_{p}}^{y},
\quad
\upsilon_{\physvec{\omega},p}^{+}
=
\rbk{u^{+}_{p}}^{x} + \imunit \rbk{u^{+}_{p}}^{y}.\] Equation \eqref{eq:centered-action} gives their action on the reference slot vectors: \begin{equation}\label{eq:centered-amplitudes}
d^{-}_{\physvec{\omega},p} \xi_{\physvec{\omega},p}^{+}
=
\upsilon_{\physvec{\omega},p}^{-} \xi_{\physvec{\omega},p}^{-},
\quad
d^{+}_{\physvec{\omega},p} \xi_{\physvec{\omega},p}^{+}
=
\upsilon_{\physvec{\omega},p}^{+} \xi_{\physvec{\omega},p}^{-},
\quad
\abs{\upsilon_{\physvec{\omega},p}^{-}}
=
\frac{1 + u^{z}_{p}}{2},
\quad
\abs{\upsilon_{\physvec{\omega},p}^{+}}
=
\frac{1 - u^{z}_{p}}{2}.
\end{equation} These identities follow from \(d^{\pm}_{\physvec{\omega},p}
=
\frac{1}{2} \rbk{d^{x}_{\physvec{\omega},p} \pm \imunit d^{y}_{\physvec{\omega},p}}\) and \eqref{eq:centered-action}. For an explicit evaluation, write \(\theta_{\physvec{\omega},p}\) and \(\phi_{\physvec{\omega},p}\) for the spherical coordinates of \(\physvec{u}_{p}\), and choose the corresponding frame \[\begin{aligned}
\physvec{e}_{\physvec{\omega},p,1}
&=
\vecbk{
\cos \theta_{\physvec{\omega},p} \cos \phi_{\physvec{\omega},p},
\cos \theta_{\physvec{\omega},p} \sin \phi_{\physvec{\omega},p},
- \sin \theta_{\physvec{\omega},p}
},
\\ 
\physvec{e}_{\physvec{\omega},p,2}
&=
\vecbk{- \sin \phi_{\physvec{\omega},p}, \cos \phi_{\physvec{\omega},p}, 0},
\\ 
\physvec{u}_{p}
&=
\rbk{
\sin \theta_{\physvec{\omega},p} \cos \phi_{\physvec{\omega},p},
\sin \theta_{\physvec{\omega},p} \sin \phi_{\physvec{\omega},p},
\cos \theta_{\physvec{\omega},p}
}
\end{aligned}\] Substitution into the preceding identities gives \[\begin{aligned}
\upsilon_{\physvec{\omega},p}^{-}
=
\frac{1 + \cos \theta_{\physvec{\omega},p}}{2} \fnexp{- \imunit \phi_{\physvec{\omega},p}},
\quad
\upsilon_{\physvec{\omega},p}^{+}
=
- \frac{1 - \cos \theta_{\physvec{\omega},p}}{2} \fnexp{\imunit \phi_{\physvec{\omega},p}}.
\end{aligned}\] Since a change of frame multiplies \(\physvec{u}^{+}_{p}\) by a phase, the moduli are frame independent. In particular the one-site variance is \begin{equation}\label{eq:one-site-variance}
\fun{\oastate[\psi_{\physvec{\omega}}]}{d^{+}_{\physvec{\omega},p} d^{-}_{\physvec{\omega},p}}
=
\norm{d^{-}_{\physvec{\omega},p} \xi_{\physvec{\omega},p}^{+}}^{2}
=
\frac{\rbk{1 + u^{z}_{p}}^{2}}{4},
\end{equation} the quantity on the right of \cite[Eq. (30)]{ThirringWehrl001}.

\begin{prop}[generalized weak limit of the centered interaction]\label{prop:lemma3}
Let $\physvec{\omega}$ have a mean profile in the sense of Definition \ref{def:mean-profile},
and define
$$c_{\physvec{\omega}}
=
\int_{\mansphere^{2}}
\frac{\rbk{1 + u^{z}}^{2}}{4}
\opdmsr{\fun{\lambda_{\physvec{\omega}}}{\physvec{u}}}.$$
\begin{enumerate}
\item The diagonal part converges strongly on $\sphilb{H}_{\physvec{\omega}}$:
$$\slim_{\omega \to \infty}
\fun{\oarepn_{\physvec{\omega}}}{\frac{1}{\Omega}
\sum_{p
\leq
\Omega} d^{+}_{\physvec{\omega},p} d^{-}_{\physvec{\omega},p}}
=
c_{\physvec{\omega}}.$$

\item The off-diagonal part is
$$O_{\physvec{\omega},\Omega}
=
\frac{1}{\Omega} \sum_{\substack{p, q \leq \Omega \\ p \neq q}} d^{+}_{\physvec{\omega},p} d^{-}_{\physvec{\omega},q}.$$
For every $\Phi,
\Psi
\in
\sphilb{D}_{\physvec{\omega}}$,
its matrix elements satisfy
$$\lim_{\Omega \to \infty}
\bkt{\Phi}{\fun{\oarepn_{\physvec{\omega}}}{O_{\physvec{\omega},\Omega}} \Psi}
=
0.$$

\item For every $\Phi,
\Psi
\in
\sphilb{D}_{\physvec{\omega}}$,
the full centered interaction satisfies
$$\lim_{\Omega \to \infty}
\bkt{\Phi}{\fun{\oarepn_{\physvec{\omega}}}{R_{\physvec{\omega},\Omega}} \Psi}
=
c_{\physvec{\omega}} \bkt{\Phi}{\Psi}.$$

\item
Let $\physvec{\omega}$ be the constant configuration determined by
$\physvec{u}_{p}
=
\physvec{u}$ for every $p
\in
\semigrposint$,
where $\abs{u^{z}} < 1$.
the reference-vector norm satisfies
$$\lim_{\Omega \to \infty}
\norm{\fun{\oarepn_{\physvec{\omega}}}{O_{\physvec{\omega},\Omega}} \xi_{\physvec{\omega}}^{\emptyset}}^{2}
=
\frac{\rbk{1 - \rbk{u^{z}}^{2}}^{2}}{8} > 0.$$
In particular the convergence in (3) is not strong.
\item For every $\Omega
\geq
8$,
the operator norm obeys
$$\norm{\fun{\oarepn_{\physvec{\omega}}}{R_{\physvec{\omega},\Omega}}}
\geq
\frac{\Omega}{324}.$$
\end{enumerate}
\end{prop}

\begin{proof}

(1)
The operators $X_{\physvec{\omega},p}
=
d^{+}_{\physvec{\omega},p} d^{-}_{\physvec{\omega},p}$ act on one site and satisfy
$\norm{X_{\physvec{\omega},p}}
\leq
\rbk{1 + \frac{1}{2}}^{2}
\leq
4$.
Their expectations are given by \eqref{eq:one-site-variance}.
Applying the mean profile hypothesis to the continuous function
$\physvec{u}
\mapsto
\frac{\rbk{1 + u^{z}}^{2}}{4}$
shows that the Cesàro averages of these expectations converge to
$c_{\physvec{\omega}}$,
and Theorem \ref{thm:lln} gives the asserted strong convergence.

(2)
It suffices to take flip vectors $\Phi
=
\xi_{\physvec{\omega}}^{F}$
and
$\Psi
=
\xi_{\physvec{\omega}}^{G}$.
The operator $d^{+}_{\physvec{\omega},p} d^{-}_{\physvec{\omega},q}$,
$p
\neq
q$,
acts nontrivially only on the slots $p$ and $q$.
It follows that $\bkt{\xi_{\physvec{\omega}}^{F}}{\fun{\oarepn_{\physvec{\omega}}}{d^{+}_{\physvec{\omega},p} d^{-}_{\physvec{\omega},q}} \xi_{\physvec{\omega}}^{G}}
=
0$ unless $\fun{\supp}{F} \triangle \fun{\supp}{G}
\subset
\setone{p, q}$.
Suppose that a slot $r
\in
\setone{p, q}$ lies outside $\fun{\supp}{F} \cup \fun{\supp}{G}$.
The corresponding one-site factor in the matrix element is
$\bkt{\xi_{\physvec{\omega},r}^{+}}{d^{\pm}_{\physvec{\omega},r} \xi_{\physvec{\omega},r}^{+}}
=
0$ by centering.
Hence a nonzero matrix element requires $p, q
\in
\fun{\supp}{F} \cup \fun{\supp}{G}$.
There are at most $\rbk{\abscard{F} + \abscard{G}}^{2}$ such ordered pairs,
and each matrix element is bounded by
$\norm{d^{+}_{\physvec{\omega},p}} \norm{d^{-}_{\physvec{\omega},q}}
\leq
4$.
The number of contributing ordered pairs and the uniform bound on each matrix element give
$$\abs{\bkt{\xi_{\physvec{\omega}}^{F}}{\fun{\oarepn_{\physvec{\omega}}}{O_{\physvec{\omega},\Omega}} \xi_{\physvec{\omega}}^{G}}}
\leq
\frac{4 \rbk{\abscard{F} + \abscard{G}}^{2}}{\Omega}
\to
0
\quad \rbk{\Omega \to \infty}.$$

(3)
Applying (1) and (2) to the decomposition
$R_{\physvec{\omega},\Omega}
=
O_{\physvec{\omega},\Omega} + \frac{1}{\Omega} \sum_{p
\leq
\Omega} d^{+}_{\physvec{\omega},p} d^{-}_{\physvec{\omega},p}$
gives the claimed matrix-element limit.

(4)
Choose a common frame for the constant configuration,
in which $\upsilon_{\physvec{\omega},p}^{-}
=
\upsilon_{\physvec{\omega}}^{-}$ and $\upsilon_{\physvec{\omega},p}^{+}
=
\upsilon_{\physvec{\omega}}^{+}$ for all $p
\in
\semigrposint$.
Their moduli are $\abs{\upsilon_{\physvec{\omega}}^{-}}
=
\frac{1 + u^{z}}{2}$,
$\abs{\upsilon_{\physvec{\omega}}^{+}}
=
\frac{1 - u^{z}}{2}$.
For $p
\neq
q$,
\eqref{eq:centered-amplitudes} gives $\fun{\oarepn_{\physvec{\omega}}}{d^{+}_{\physvec{\omega},p} d^{-}_{\physvec{\omega},q}} \xi_{\physvec{\omega}}^{\emptyset}
=
\upsilon_{\physvec{\omega},p}^{+} \upsilon_{\physvec{\omega},q}^{-} \xi_{\physvec{\omega}}^{\setone{p, q}}$.
Equation \eqref{eq:centered-amplitudes} gives the action of the off-diagonal interaction on the reference vector:
$$\fun{\oarepn_{\physvec{\omega}}}{O_{\physvec{\omega},\Omega}} \xi_{\physvec{\omega}}^{\emptyset}
=
\frac{2 \upsilon_{\physvec{\omega}}^{-} \upsilon_{\physvec{\omega}}^{+}}{\Omega} \sum_{\setone{p, q}
\subset
\intint{1..\Omega}, p
\neq
q} \xi_{\physvec{\omega}}^{\setone{p, q}},$$
where each unordered pair receives two ordered contributions.
Orthonormality of the flip vectors gives
$$\norm{\fun{\oarepn_{\physvec{\omega}}}{O_{\physvec{\omega},\Omega}} \xi_{\physvec{\omega}}^{\emptyset}}^{2}
=
\frac{4 \abs{\upsilon_{\physvec{\omega}}^{-} \upsilon_{\physvec{\omega}}^{+}}^{2}}{\Omega^{2}} \binom{\Omega}{2}
\to
2 \abs{\upsilon_{\physvec{\omega}}^{-} \upsilon_{\physvec{\omega}}^{+}}^{2}
=
\frac{\rbk{1 - \rbk{u^{z}}^{2}}^{2}}{8}
\quad \rbk{\Omega \to \infty}.$$

(5)
The Bloch parametrization \eqref{eq:bloch-state} gives
$$\set{
\bkt{\chi}{\sigma^{-} \chi}
}{
\chi
\in
\fldcmp^{2},
\norm{\chi}
=
1
}
=
\set{
z
\in
\fldcmp
}{
\abs{z}
\leq
\frac{1}{2}
},
\quad
\abs{m^{-}_{\physvec{\omega},p}}
\leq
\frac{1}{2}.$$
Choose $\beta_{\physvec{\omega},p}
\in
[0, 2 \pi)$ such that
$$m^{-}_{\physvec{\omega},p}
=
-\abs{m^{-}_{\physvec{\omega},p}}
\fnexp{\imunit \beta_{\physvec{\omega},p}}
\quad
\rbk{p
\leq
\Omega},$$
where $\beta_{\physvec{\omega},p}$ is arbitrary when
$m^{-}_{\physvec{\omega},p}
=
0$.
Partition the angles into eight half-open arcs of length $\frac{\pi}{4}$.
For one arc with midpoint $\beta_{\physvec{\omega},\Omega,0}$,
the corresponding index set $P_{\physvec{\omega},\Omega,0}$ satisfies
$$\abscard{P_{\physvec{\omega},\Omega,0}}
\geq
\frac{\Omega}{8},
\quad
\abs{\beta_{\physvec{\omega},p} - \beta_{\physvec{\omega},\Omega,0}}
\leq
\frac{\pi}{8}
\quad
\rbk{p
\in
P_{\physvec{\omega},\Omega,0}},$$
after choosing representatives in that arc.

The disc identity permits unit vectors
$\chi_{\physvec{\omega},p}
\in
\fldcmp^{2}$ satisfying
$$\bkt{\chi_{\physvec{\omega},p}}{\sigma^{-}_{p} \chi_{\physvec{\omega},p}}
=
\begin{cases}
\frac{1}{2} \fnexp{\imunit \beta_{\physvec{\omega},p}},
&
p
\in
P_{\physvec{\omega},\Omega,0},
\\
m^{-}_{\physvec{\omega},p},
&
p
\notin
P_{\physvec{\omega},\Omega,0}.
\end{cases}$$
Define
$$\widetilde{\chi}_{\physvec{\omega},\Omega}
=
\rbk{\bigotimes_{p
=
1}^{\Omega} \chi_{\physvec{\omega},p}}
\otimes
\rbk{\bigotimes_{p
>
\Omega} \xi_{\physvec{\omega},p}^{+}},
\quad
z_{\physvec{\omega},p}
=
\bkt{\chi_{\physvec{\omega},p}}{\sigma^{-}_{p} \chi_{\physvec{\omega},p}}
- m^{-}_{\physvec{\omega},p}.$$
The one-site choices give
$$\begin{aligned}
\bkt{\widetilde{\chi}_{\physvec{\omega},\Omega}}{
\fun{\oarepn_{\physvec{\omega}}}{D^{-}_{\physvec{\omega},\Omega}}
\widetilde{\chi}_{\physvec{\omega},\Omega}}
=
\sum_{p
\leq
\Omega} z_{\physvec{\omega},p},
\quad
z_{\physvec{\omega},p}
=
\begin{cases}
\fnexp{\imunit \beta_{\physvec{\omega},p}}
\rbk{\frac{1}{2} + \abs{m^{-}_{\physvec{\omega},p}}},
&
p
\in
P_{\physvec{\omega},\Omega,0},
\\ 
0,
&
p
\notin
P_{\physvec{\omega},\Omega,0}.
\end{cases}
\end{aligned}$$
Projection onto the midpoint direction yields
$$\begin{aligned}
\abs{\sum_{p
\leq
\Omega} z_{\physvec{\omega},p}}
&\geq
\fun{\opreal}{
\fnexp{- \imunit \beta_{\physvec{\omega},\Omega,0}}
\sum_{p
\leq
\Omega} z_{\physvec{\omega},p}
}
\\
&=
\sum_{p
\in
P_{\physvec{\omega},\Omega,0}}
\rbk{\frac{1}{2} + \abs{m^{-}_{\physvec{\omega},p}}}
\fun{\cos}{
\beta_{\physvec{\omega},p}
-\beta_{\physvec{\omega},\Omega,0}
}
\\
&\geq
\frac{1}{2}
\fun{\cos}{\frac{\pi}{8}}
\abscard{P_{\physvec{\omega},\Omega,0}}
\geq
\frac{\Omega}{18}.
\end{aligned}$$
Finally,
$R_{\physvec{\omega},\Omega}
=
\frac{1}{\Omega}
\faadj{\rbk{D^{-}_{\physvec{\omega},\Omega}}}
D^{-}_{\physvec{\omega},\Omega}$
and the unit norm of
$\widetilde{\chi}_{\physvec{\omega},\Omega}$ give
$$\begin{aligned}
\norm{\fun{\oarepn_{\physvec{\omega}}}{R_{\physvec{\omega},\Omega}}}
=
\frac{1}{\Omega}
\norm{\fun{\oarepn_{\physvec{\omega}}}{D^{-}_{\physvec{\omega},\Omega}}}^{2}
\geq
\frac{1}{\Omega}
\abs{\bkt{\widetilde{\chi}_{\physvec{\omega},\Omega}}{
\fun{\oarepn_{\physvec{\omega}}}{D^{-}_{\physvec{\omega},\Omega}}
\widetilde{\chi}_{\physvec{\omega},\Omega}}}^{2}
\geq
\frac{\Omega}{324}.
\end{aligned}$$
\end{proof}

The generalized weak limit of \cite[Section 3]{ThirringWehrl001} means the following sectorwise matrix-element limit. For a fixed configuration \(\physvec{\omega}\), part (3) states that \[\bkt{\Phi}{
\fun{\oarepn_{\physvec{\omega}}}{R_{\physvec{\omega},\Omega}}
\Psi}
\to
c_{\physvec{\omega}}
\bkt{\Phi}{\Psi}
\quad
\rbk{\Omega
\to
\infty}
\quad
\rbk{\Phi, \Psi
\in
\sphilb{D}_{\physvec{\omega}}}.\] Thus the limiting form on \(\sphilb{D}_{\physvec{\omega}}\) is \[\fun{\opform{q}_{\physvec{\omega}}}{\Phi, \Psi}
=
c_{\physvec{\omega}}
\bkt{\Phi}{\Psi}.\] This bounded form extends continuously to \(\sphilb{H}_{\physvec{\omega}}\) and is represented there by the scalar operator \(c_{\physvec{\omega}}\). This is the \(c\)-number described in \cite[p. 310]{ThirringWehrl001}.

The matrix-element limit on \(\sphilb{D}_{\physvec{\omega}}\) is not weak operator convergence on \(\sphilb{H}_{\physvec{\omega}}\). For a constant configuration with \(\abs{u^{z}}
<
1\), parts (1) and (4) give \[\lim_{\Omega
\to
\infty}
\norm{
\rbk{
\fun{\oarepn_{\physvec{\omega}}}{R_{\physvec{\omega},\Omega}}
- c_{\physvec{\omega}}
}
\xi_{\physvec{\omega}}^{\emptyset}
}
=
\frac{1 - \rbk{u^{z}}^{2}}{2 \sqrt{2}}
>
0,\] which excludes strong convergence to \(c_{\physvec{\omega}}\). Part (5) gives \[\norm{
\fun{\oarepn_{\physvec{\omega}}}{R_{\physvec{\omega},\Omega}}
}
\geq
\frac{\Omega}{324}.\] The norms are therefore not uniformly bounded, which excludes weak operator convergence to any bounded operator. For a different configuration \(\physvec{\omega}'\), part (3), when applicable, defines a separate form \(\opform{q}_{\physvec{\omega}'}\) on \(\sphilb{D}_{\physvec{\omega}'} \subset \sphilb{H}_{\physvec{\omega}'}\). No identification or convergence between the Hilbert spaces of distinct configurations is asserted.

\section{The Bogoliubov--Haag Hamiltonian and the Gap Equation}\label{sec:bogoliubov}

The sectorwise Bogoliubov--Haag construction separates form convergence from operator representability, diagonal realization, and the gap equation. The four parts of Theorem \ref{thm:bogoliubov} establish these statements in that order. They recover \cite[Theorem 1]{ThirringWehrl001} and make the operator domain and self-adjoint realization explicit.

\subsection{Sectorwise forms and the centered expansion}\label{sectorwise-forms-and-the-centered-expansion}

Fix a spin configuration \(\physvec{\omega}\) with a mean profile in the sense of Definition \ref{def:mean-profile}, and let \(\oarepn_{\physvec{\omega}}\) be its product representation from Definition \ref{def:spin-product-representation}. The matrix-element limit of \(\fun{\oarepn_{\physvec{\omega}}}{
\physham_{\Omega} - E_{\physvec{\omega},\Omega}}\) is defined in \eqref{eq:sectorwise-bogoliubov-limit}, with the configuration-dependent scalar \(E_{\physvec{\omega},\Omega}\) given by \eqref{eq:e-omega}. The limit in \eqref{eq:sectorwise-bogoliubov-limit} is taken on \(\sphilb{D}_{\physvec{\omega}}\), in the sectorwise sense specified before the present section. The sesquilinear form \(\opform{Q}_{\physvec{\omega}}\) defined by \eqref{eq:sectorwise-bogoliubov-limit} need not be represented by an operator. The mean polarization of \(\physvec{\omega}\) is denoted by \(\eta_{\physvec{\omega}} \bar{\physvec{u}}_{\physvec{\omega}}\) as in \eqref{eq:mean-polarization}. Using the one-site means \(m^{\pm}_{\physvec{\omega},p}\) from \eqref{eq:centered-quantities}, define their finite-volume averages and limiting transverse means by \begin{equation}\label{eq:mean-m}
\bar{M}^{\pm}_{\physvec{\omega},\Omega}
=
\frac{1}{\Omega} \sum_{p
=
1}^{\Omega} m^{\pm}_{\physvec{\omega},p},
\quad
\bar{m}^{\pm}_{\physvec{\omega}}
=
\lim_{\Omega \to \infty} \bar{M}^{\pm}_{\physvec{\omega},\Omega}
=
\frac{\eta_{\physvec{\omega}}}{2} \rbk{\bar{u}_{\physvec{\omega}}^{x} \pm \imunit \bar{u}_{\physvec{\omega}}^{y}}.
\end{equation} The centered operators \(d^{\gamma}_{\physvec{\omega},p}\), \(d^{\pm}_{\physvec{\omega},p}\), \(D^{\pm}_{\physvec{\omega},\Omega}\), and \(R_{\physvec{\omega},\Omega}\) are also those of \eqref{eq:centered-quantities}. The product-state expectation defines the configuration-dependent subtraction as \begin{equation}\label{eq:e-omega}
\begin{aligned}
E_{\physvec{\omega},\Omega}
=
\fun{\oastate[\psi_{\physvec{\omega}}]}{\physham_{\Omega}}
=
- \sum_{p
=
1}^{\Omega} \varepsilon_{p} u^{z}_{p}
- \frac{2}{\sminvtemperature_{c} \Omega} \sum_{p
=
1}^{\Omega}
\fun{\oastate[\psi_{\physvec{\omega}}]}{
d^{+}_{\physvec{\omega},p} d^{-}_{\physvec{\omega},p}}
- \frac{2}{\sminvtemperature_{c}}
\Omega \bar{M}^{+}_{\physvec{\omega},\Omega} \bar{M}^{-}_{\physvec{\omega},\Omega}.
\end{aligned}
\end{equation} The second equality in \eqref{eq:e-omega} uses \(\fun{\oastate[\psi_{\physvec{\omega}}]}{d^{\pm}_{\physvec{\omega},p}}
=
0\). Factorization of \(\oastate[\psi_{\physvec{\omega}}]\) over distinct slots removes the cross terms in \(\fun{\oastate[\psi_{\physvec{\omega}}]}{
D^{+}_{\physvec{\omega},\Omega} D^{-}_{\physvec{\omega},\Omega}}\).

\begin{lem}[centered expansion]\label{lem:centered-expansion}
With $d^{\gamma}_{\physvec{\omega},p}$,
$d^{\pm}_{\physvec{\omega},p}$,
$D^{\pm}_{\physvec{\omega},\Omega}$,
and $R_{\physvec{\omega},\Omega}$ defined by \eqref{eq:centered-quantities}
and the notation above,
the centered expansion in $\oa{A}_{\Omega}$ is
$$\begin{aligned}
\physham_{\Omega} - E_{\physvec{\omega},\Omega}
=
- \sum_{p
=
1}^{\Omega}
\sqbk{\varepsilon_{p} d^{z}_{\physvec{\omega},p}
+ \frac{2}{\sminvtemperature_{c}}
\rbk{\bar{M}^{-}_{\physvec{\omega},\Omega} d^{+}_{\physvec{\omega},p}
+ \bar{M}^{+}_{\physvec{\omega},\Omega} d^{-}_{\physvec{\omega},p}}}
- \frac{2}{\sminvtemperature_{c}}
\rbk{R_{\physvec{\omega},\Omega} - \fun{\oastate[\psi_{\physvec{\omega}}]}{R_{\physvec{\omega},\Omega}}}.
\end{aligned}$$
\end{lem}

\begin{proof}
Substitute $\sigma^{z}_{p}
=
u^{z}_{p} + d^{z}_{\physvec{\omega},p}$ and $S^{\pm}_{\Omega}
=
D^{\pm}_{\physvec{\omega},\Omega} + \Omega \bar{M}^{\pm}_{\physvec{\omega},\Omega}$ into \eqref{eq:bcs-hamiltonian}:
$$\begin{aligned}
\physham_{\Omega}
&=
- \sum_{p
\leq
\Omega} \varepsilon_{p} u^{z}_{p}
- \sum_{p
\leq
\Omega} \varepsilon_{p} d^{z}_{\physvec{\omega},p}
- \frac{2}{\sminvtemperature_{c} \Omega} D^{+}_{\physvec{\omega},\Omega} D^{-}_{\physvec{\omega},\Omega}
\\ 
&\quad
- \frac{2}{\sminvtemperature_{c}} \bar{M}^{-}_{\physvec{\omega},\Omega} D^{+}_{\physvec{\omega},\Omega}
- \frac{2}{\sminvtemperature_{c}} \bar{M}^{+}_{\physvec{\omega},\Omega} D^{-}_{\physvec{\omega},\Omega}
- \frac{2}{\sminvtemperature_{c}} \Omega \bar{M}^{+}_{\physvec{\omega},\Omega} \bar{M}^{-}_{\physvec{\omega},\Omega}.
\end{aligned}$$
Subtracting \eqref{eq:e-omega} and using $R_{\physvec{\omega},\Omega}
=
\frac{1}{\Omega} D^{+}_{\physvec{\omega},\Omega} D^{-}_{\physvec{\omega},\Omega}$ with $\fun{\oastate[\psi_{\physvec{\omega}}]}{R_{\physvec{\omega},\Omega}}
=
\frac{1}{\Omega} \sum_{p
\leq
\Omega} \fun{\oastate[\psi_{\physvec{\omega}}]}{d^{+}_{\physvec{\omega},p} d^{-}_{\physvec{\omega},p}}$ gives the identity.
\end{proof}

The one-body part of Lemma \ref{lem:centered-expansion} is, up to the replacement of \(\bar{M}^{\pm}_{\physvec{\omega},\Omega}\) by their limits, a sum of centered one-site operators. The effective one-site Hamiltonian associated with the configuration \(\physvec{\omega}\) is \begin{equation}\label{eq:effective-hamiltonian}
\physham[h]_{\physvec{\omega},p}
=
- \varepsilon_{p} \sigma^{z}_{p} - \frac{2}{\sminvtemperature_{c}} \rbk{\bar{m}^{-}_{\physvec{\omega}} \sigma^{+}_{p} + \bar{m}^{+}_{\physvec{\omega}} \sigma^{-}_{p}}
=
- \physvec{b}_{\physvec{\omega},p} \cdot \physvec{\sigma}_{p},
\quad
\physvec{b}_{\physvec{\omega},p}
=
\rbk{\frac{1}{\sminvtemperature_{c}} \eta_{\physvec{\omega}} \bar{u}_{\physvec{\omega}}^{x}, \frac{1}{\sminvtemperature_{c}} \eta_{\physvec{\omega}} \bar{u}_{\physvec{\omega}}^{y}, \varepsilon_{p}},
\end{equation} where the second form follows from \(2 \rbk{\bar{m}^{-}_{\physvec{\omega}} \sigma^{+} + \bar{m}^{+}_{\physvec{\omega}} \sigma^{-}}
=
\eta_{\physvec{\omega}} \rbk{\bar{u}_{\physvec{\omega}}^{x} \sigma^{x} + \bar{u}_{\physvec{\omega}}^{y} \sigma^{y}}\), expanded from \(\sigma^{\pm}
=
\frac{1}{2} \rbk{\sigma^{x} \pm \imunit \sigma^{y}}\). Its matrix elements in the flip frame of the slot \(p\) are read off the decomposition of Lemma \ref{lem:frame} applied to \(\physvec{\sigma} \cdot \physvec{b}_{\physvec{\omega},p}\): \begin{equation}\label{eq:flip-amplitude}
\physham[h]_{\physvec{\omega},p} \xi_{\physvec{\omega},p}^{+}
=
- \rbk{\physvec{b}_{\physvec{\omega},p} \cdot \physvec{u}_{p}} \xi_{\physvec{\omega},p}^{+} + t_{\physvec{\omega},p} \xi_{\physvec{\omega},p}^{-},
\quad
t_{\physvec{\omega},p}
=
- 2 \physvec{b}_{\physvec{\omega},p} \cdot \physvec{u}^{+}_{p},
\end{equation} and the flip amplitude \(t_{\physvec{\omega},p}\) vanishes exactly when \(\physvec{b}_{\physvec{\omega},p}\) is parallel to \(\physvec{u}_{p}\), because \(\physvec{u}^{+}_{p}\) spans, together with its conjugate, the complexified orthogonal complement of \(\physvec{u}_{p}\).

\subsection{Operator representability and the gap equation}\label{operator-representability-and-the-gap-equation}

The limiting form \eqref{eq:sectorwise-bogoliubov-limit} does not automatically determine an operator on \(\sphilb{H}_{\physvec{\omega}}\). Its representability is controlled by the square summability of the configuration-dependent flip amplitudes \(\seq{t_{\physvec{\omega},p}}{p
\in
\semigrposint}\) from \eqref{eq:flip-amplitude}. When this condition holds, the form determines a unique self-adjoint Hamiltonian \(\physham_{\txtbogoliubov,\physvec{\omega}}\) for which the flip domain is a core. Alignment eliminates the flip amplitudes, makes this Hamiltonian diagonal, and leaves the gap equation as the self-consistency condition. Thus \(\opform{Q}_{\physvec{\omega}}\) denotes the limiting form, whereas \(\physham_{\txtbogoliubov,\physvec{\omega}}\) denotes its self-adjoint realization. The operator initially defined on the flip domain is used only in the construction and receives no separate symbol. The lowercase bold symbol \(\physvec{b}_{\physvec{\omega},p}\) remains the one-site effective field of \eqref{eq:effective-hamiltonian} and does not denote another Hamiltonian. The aligned conclusion and the gap equation below recast \cite[Theorem 1 and Eqs. (33)--(36)]{ThirringWehrl001} in the sign conventions of \eqref{eq:bcs-hamiltonian}. Part (2) specifies the self-adjoint realization and its core.

\begin{defn}[aligned spin configuration]\label{def:aligned-configuration}
Let $\physvec{\omega}$ be a spin configuration with a mean profile in the sense of
Definition \ref{def:mean-profile},
and let $\physvec{b}_{\physvec{\omega},p}$ be defined by
\eqref{eq:effective-hamiltonian} for every $p
\in
\semigrposint$.
The configuration $\physvec{\omega}$ is aligned if there exists a sign sequence
$\seq{\eta_{\physvec{\omega},p}}{p
\in
\semigrposint}$ such that
$$\eta_{\physvec{\omega},p}
\in
\setone{+1, -1},
\quad
\physvec{b}_{\physvec{\omega},p}
\neq
0,
\quad
\physvec{u}_{p}
=
\eta_{\physvec{\omega},p} \frac{\physvec{b}_{\physvec{\omega},p}}{\abs{\physvec{b}_{\physvec{\omega},p}}}
\quad
\rbk{p
\in
\semigrposint}.$$
\end{defn}

The vector \(\physvec{b}_{\physvec{\omega},p}\) is the self-consistent one-site field. Alignment is the zero-torque condition \(\physvec{b}_{\physvec{\omega},p}\times\physvec{u}_{p}=0\), so the reference product state is stationary. The sign \(\eta_{\physvec{\omega},p}=+1\) selects the lower local eigenstate of \(-\physvec{\sigma}_{p}\cdot\physvec{b}_{\physvec{\omega},p}\), whereas \(-1\) selects the inverted state; the all-positive branch is the product ground-state branch. Alignment also removes \eqref{eq:flip-amplitude}, diagonalizes the Bogoliubov--Haag Hamiltonian in relative excitations, and reduces transverse self-consistency to the gap equation in Theorem \ref{thm:bogoliubov}(4).

\begin{thm}[sectorwise Bogoliubov--Haag limit and the gap equation]\label{thm:bogoliubov}
Let $\physvec{\omega}$ be a spin configuration with a mean profile in the sense of
Definition \ref{def:mean-profile}.
Let $\oarepn_{\physvec{\omega}}$ act on $\sphilb{H}_{\physvec{\omega}}$ as in Definition \ref{def:spin-product-representation}.
\begin{enumerate}
\item For all $\Phi, \Psi
\in
\sphilb{D}_{\physvec{\omega}}$ the limit
\begin{equation}\label{eq:sectorwise-bogoliubov-limit}
\fun{\opform{Q}_{\physvec{\omega}}}{\Phi, \Psi}
=
\lim_{\Omega \to \infty}
\bkt{\Phi}{
\fun{\oarepn_{\physvec{\omega}}}{
\physham_{\Omega} - E_{\physvec{\omega},\Omega}} \Psi}
\end{equation}
exists and equals $\sum_{p \in \semigrposint} \bkt{\Phi}{\fun{\oarepn_{\physvec{\omega}}}{\physham[h]_{\physvec{\omega},p} - \fun{\oastate[\psi_{\physvec{\omega}}]}{\physham[h]_{\physvec{\omega},p}}} \Psi}$,
a sum with finitely many nonzero terms for each pair of flip vectors.
\item The form $\opform{Q}_{\physvec{\omega}}$ is represented on $\sphilb{D}_{\physvec{\omega}}$ by an operator if and only if
$\sum_{p
\in
\semigrposint} \abs{t_{\physvec{\omega},p}}^{2}
<
\infty$.
In that case there is a unique self-adjoint operator
$\physham_{\txtbogoliubov,\physvec{\omega}}$ on $\sphilb{H}_{\physvec{\omega}}$
such that $\sphilb{D}_{\physvec{\omega}}$ is a core for
$\physham_{\txtbogoliubov,\physvec{\omega}}$ and
$\bkt{\Phi}{\physham_{\txtbogoliubov,\physvec{\omega}} \Psi}
=
\fun{\opform{Q}_{\physvec{\omega}}}{\Phi, \Psi}$
for all $\Phi,
\Psi
\in
\sphilb{D}_{\physvec{\omega}}$.
This operator is the Bogoliubov--Haag Hamiltonian in the sector
$\physvec{\omega}$.
\item If $\physvec{\omega}$ is aligned in the sense of
Definition \ref{def:aligned-configuration},
then $t_{\physvec{\omega},p}
=
0$ for all $p$.
The restriction of the operator in (2) to
$\sphilb{D}_{\physvec{\omega}}$ is the diagonal sum
$$\left.\physham_{\txtbogoliubov,\physvec{\omega}}\right|_{\sphilb{D}_{\physvec{\omega}}}
=
\sum_{p \in \semigrposint} \fun{\oarepn_{\physvec{\omega}}}{\eta_{\physvec{\omega},p} \abs{\physvec{b}_{\physvec{\omega},p}} \rbk{1 - \physvec{\sigma}_{p} \cdot \physvec{u}_{p}}}.$$
On the flip vectors it acts by
$\physham_{\txtbogoliubov,\physvec{\omega}} \xi_{\physvec{\omega}}^{F}
=
\rbk{\sum_{p
\in
F} 2 \eta_{\physvec{\omega},p} \abs{\physvec{b}_{\physvec{\omega},p}}} \xi_{\physvec{\omega}}^{F}$.
\item If $\physvec{\omega}$ is aligned in the sense of
Definition \ref{def:aligned-configuration} and
$\eta_{\physvec{\omega}} \rbk{\bar{u}_{\physvec{\omega}}^{x}, \bar{u}_{\physvec{\omega}}^{y}}
\neq
0$,
the gap equation holds:
\begin{equation}\label{eq:gap-equation}
\lim_{\Omega \to \infty} \frac{1}{\Omega} \sum_{p
=
1}^{\Omega} \frac{\eta_{\physvec{\omega},p}}{\sqrt{\varepsilon_{p}^{2} + \frac{1}{\sminvtemperature_{c}^{2}} \eta_{\physvec{\omega}}^{2} \rbk{\rbk{\bar{u}_{\physvec{\omega}}^{x}}^{2} + \rbk{\bar{u}_{\physvec{\omega}}^{y}}^{2}}}}
=
\sminvtemperature_{c}.
\end{equation}
\end{enumerate}
\end{thm}

The hypotheses of Theorem \ref{thm:bogoliubov} determine the limit in \eqref{eq:gap-equation}, but they do not determine the summands separately. In the degenerate strong-coupling specialization, Corollary \ref{cor:strong-coupling} gives \(\eta_{\fun{\physvec{\omega}}{\phi},p}
=
1\), \(\eta_{\fun{\physvec{\omega}}{\phi}}
=
1\), and \[\rbk{\bar{u}_{\fun{\physvec{\omega}}{\phi}}^{x}}^{2}
+
\rbk{\bar{u}_{\fun{\physvec{\omega}}{\phi}}^{y}}^{2}
=
r^{2}
=
1 - \sminvtemperature_{c}^{2} \varepsilon^{2}.\] The limit in \eqref{eq:gap-equation} is then evaluated directly as \[\begin{aligned}
\lim_{\Omega \to \infty}
\frac{1}{\Omega}
\sum_{p = 1}^{\Omega}
\frac{\eta_{\fun{\physvec{\omega}}{\phi},p}}
{\sqrt{\varepsilon^{2}
+\frac{1}{\sminvtemperature_{c}^{2}}
\eta_{\fun{\physvec{\omega}}{\phi}}^{2}
r^{2}}}
=
\frac{1}{\sqrt{\varepsilon^{2} + \frac{1}{\sminvtemperature_{c}^{2}}
\rbk{1 - \sminvtemperature_{c}^{2} \varepsilon^{2}}}}
=
\sminvtemperature_{c}.
\end{aligned}\]

\begin{proof}

(1)
Lemma \ref{lem:centered-expansion} decomposes the finite-volume matrix element into its one-body and centered-interaction contributions:
$$\begin{aligned}
&\bkt{\Phi}{\fun{\oarepn_{\physvec{\omega}}}{
\physham_{\Omega} - E_{\physvec{\omega},\Omega}} \Psi}
\\ 
&=
- \sum_{p
\leq
\Omega} \bkt{\Phi}{\fun{\oarepn_{\physvec{\omega}}}{\sqbk{\varepsilon_{p} d^{z}_{\physvec{\omega},p} + \frac{2}{\sminvtemperature_{c}} \rbk{\bar{M}^{-}_{\physvec{\omega},\Omega} d^{+}_{\physvec{\omega},p} + \bar{M}^{+}_{\physvec{\omega},\Omega} d^{-}_{\physvec{\omega},p}}}} \Psi}
\\ 
&\quad
- \frac{2}{\sminvtemperature_{c}} \bkt{\Phi}{\fun{\oarepn_{\physvec{\omega}}}{R_{\physvec{\omega},\Omega} - \fun{\oastate[\psi_{\physvec{\omega}}]}{R_{\physvec{\omega},\Omega}}} \Psi}.
\end{aligned}$$
Proposition \ref{prop:lemma3}(3),
\eqref{eq:one-site-variance},
and the mean-profile condition in Definition \ref{def:mean-profile} give the two limits
$$\begin{aligned}
\lim_{\Omega \to \infty}
\bkt{\Phi}{\fun{\oarepn_{\physvec{\omega}}}{R_{\physvec{\omega},\Omega}} \Psi}
=
c_{\physvec{\omega}} \bkt{\Phi}{\Psi},
\quad
\lim_{\Omega \to \infty}
\fun{\oastate[\psi_{\physvec{\omega}}]}{R_{\physvec{\omega},\Omega}}
=
\lim_{\Omega \to \infty}
\frac{1}{\Omega}
\sum_{p \leq \Omega}
\frac{\rbk{1 + u^{z}_{p}}^{2}}{4}
=
c_{\physvec{\omega}}.
\end{aligned}$$
By linearity,
take $\Phi
=
\xi_{\physvec{\omega}}^{F}$ and $\Psi
=
\xi_{\physvec{\omega}}^{G}$,
and set
$\Lambda_{F,G}
=
\fun{\supp}{F} \cup \fun{\supp}{G}$.
The centered-interaction contribution satisfies
$$\begin{aligned}
\lim_{\Omega \to \infty}
\bkt{\Phi}{\fun{\oarepn_{\physvec{\omega}}}{R_{\physvec{\omega},\Omega}} \Psi}
- \fun{\oastate[\psi_{\physvec{\omega}}]}{R_{\physvec{\omega},\Omega}}
\bkt{\Phi}{\Psi}
=
c_{\physvec{\omega}} \bkt{\Phi}{\Psi}
-c_{\physvec{\omega}} \bkt{\Phi}{\Psi}
=
0.
\end{aligned}$$

For every $p
\notin
\Lambda_{F,G}$,
centering and factorization give
$\bkt{\Phi}{\fun{\oarepn_{\physvec{\omega}}}{d^{\pm}_{\physvec{\omega},p}} \Psi}
=
0$.
Since $\norm{d^{\pm}_{\physvec{\omega},p}}
\leq
2$,
the finite sums obey
$$\abs{
\sum_{p
\leq
\Omega}
\bkt{\Phi}{
\fun{\oarepn_{\physvec{\omega}}}{d^{\pm}_{\physvec{\omega},p}}
\Psi}
}
\leq
2 \abscard{\Lambda_{F,G}}.$$
Equation \eqref{eq:effective-hamiltonian} gives
$$\physham[h]_{\physvec{\omega},p}
-
\fun{\oastate[\psi_{\physvec{\omega}}]}{
\physham[h]_{\physvec{\omega},p}}
=
- \varepsilon_{p} d^{z}_{\physvec{\omega},p}
-
\frac{2}{\sminvtemperature_{c}}
\rbk{
\bar{m}^{-}_{\physvec{\omega}} d^{+}_{\physvec{\omega},p}
+
\bar{m}^{+}_{\physvec{\omega}} d^{-}_{\physvec{\omega},p}
}.$$
The error made by replacing
$\bar{M}^{\pm}_{\physvec{\omega},\Omega}$ with
$\bar{m}^{\pm}_{\physvec{\omega}}$ is bounded by
$$\begin{aligned}
&\abs{
-\sum_{p
\leq
\Omega}
\bkt{\Phi}{
\fun{\oarepn_{\physvec{\omega}}}{
\varepsilon_{p} d^{z}_{\physvec{\omega},p}
+
\frac{2}{\sminvtemperature_{c}}
\rbk{
\bar{M}^{-}_{\physvec{\omega},\Omega} d^{+}_{\physvec{\omega},p}
+
\bar{M}^{+}_{\physvec{\omega},\Omega} d^{-}_{\physvec{\omega},p}
}}
\Psi}
-
\sum_{p
\leq
\Omega}
\bkt{\Phi}{
\fun{\oarepn_{\physvec{\omega}}}{
\physham[h]_{\physvec{\omega},p}
-
\fun{\oastate[\psi_{\physvec{\omega}}]}{
\physham[h]_{\physvec{\omega},p}}
}
\Psi}
}
\\
&\leq
\frac{4 \abscard{\Lambda_{F,G}}}{\sminvtemperature_{c}}
\rbk{
\abs{\bar{M}^{-}_{\physvec{\omega},\Omega}
- \bar{m}^{-}_{\physvec{\omega}}}
+
\abs{\bar{M}^{+}_{\physvec{\omega},\Omega}
- \bar{m}^{+}_{\physvec{\omega}}}
}
\to
0
\quad
\rbk{\Omega
\to
\infty}.
\end{aligned}$$

For $\Lambda_{F,G}
\subset
\intint{1..\Omega}$,
the remaining one-site sum is already constant in $\Omega$:
$$\begin{aligned}
\sum_{p \leq \Omega}
\bkt{\xi_{\physvec{\omega}}^{F}}
{\fun{\oarepn_{\physvec{\omega}}}{\physham[h]_{\physvec{\omega},p}
-\fun{\oastate[\psi_{\physvec{\omega}}]}{\physham[h]_{\physvec{\omega},p}}}
\xi_{\physvec{\omega}}^{G}}
=
\begin{cases}
\displaystyle
\sum_{p \in F}
\bkt{\xi_{\physvec{\omega}}^{F}}
{\fun{\oarepn_{\physvec{\omega}}}{\physham[h]_{\physvec{\omega},p}
-\fun{\oastate[\psi_{\physvec{\omega}}]}{\physham[h]_{\physvec{\omega},p}}}
\xi_{\physvec{\omega}}^{F}},
&
F
=
G,
\\
\displaystyle
\bkt{\xi_{\physvec{\omega}}^{F}}{
\fun{\oarepn_{\physvec{\omega}}}{
\physham[h]_{\physvec{\omega},p_{0}}
-
\fun{\oastate[\psi_{\physvec{\omega}}]}{
\physham[h]_{\physvec{\omega},p_{0}}}
}
\xi_{\physvec{\omega}}^{G}},
&
F \triangle G
=
\setone{p_{0}},
\\
0,
&
\abscard{F \triangle G}
\geq
2.
\end{cases}
\end{aligned}$$
This proves the limit in (1).

(2)
Write $F
\sim
_{p} G$ when $F \triangle G
=
\setone{p}$.
From \eqref{eq:flip-amplitude} and its conjugate,
the nonzero matrix elements of $\opform{Q}_{\physvec{\omega}}$ on the flip basis are
$$\fun{\opform{Q}_{\physvec{\omega}}}{\xi_{\physvec{\omega}}^{F}, \xi_{\physvec{\omega}}^{F}}
=
\sum_{p
\in
F} 2 \eta'_{\physvec{\omega},p},
\quad
\fun{\opform{Q}_{\physvec{\omega}}}{\xi_{\physvec{\omega}}^{F}, \xi_{\physvec{\omega}}^{G}}
=
t_{\physvec{\omega},p} \text{ or } \cmpconj{t_{\physvec{\omega},p}} \quad \rbk{F
\sim
_{p} G},$$
where
$2 \eta'_{\physvec{\omega},p}
=
\bkt{\xi_{\physvec{\omega},p}^{-}}{\physham[h]_{\physvec{\omega},p} \xi_{\physvec{\omega},p}^{-}}
- \bkt{\xi_{\physvec{\omega},p}^{+}}{\physham[h]_{\physvec{\omega},p} \xi_{\physvec{\omega},p}^{+}}
=
2 \physvec{b}_{\physvec{\omega},p} \cdot \physvec{u}_{p}$
is the diagonal gap of the slot $p$.
If an operator represents $\opform{Q}_{\physvec{\omega}}$ on $\sphilb{D}_{\physvec{\omega}}$,
the functional
$\Phi
\mapsto
\fun{\opform{Q}_{\physvec{\omega}}}{\Phi,\xi_{\physvec{\omega}}^{\emptyset}}$
must be represented by a vector in $\sphilb{H}_{\physvec{\omega}}$.
The expansion coefficients of this vector in the flip basis are
$\fun{\opform{Q}_{\physvec{\omega}}}{\xi_{\physvec{\omega}}^{F}, \xi_{\physvec{\omega}}^{\emptyset}}$.
Their square sum is
$\sum_{p \in \semigrposint} \abs{t_{\physvec{\omega},p}}^{2}$,
which must be finite.
Conversely,
if $\sum_{p \in \semigrposint} \abs{t_{\physvec{\omega},p}}^{2} < \infty$,
define an operator on the flip basis by
$$\xi_{\physvec{\omega}}^{G}
\mapsto
\rbk{\sum_{p
\in
G} 2 \physvec{b}_{\physvec{\omega},p} \cdot \physvec{u}_{p}} \xi_{\physvec{\omega}}^{G} + \sum_{p
\notin
G} t_{\physvec{\omega},p} \xi_{\physvec{\omega}}^{G \cup \setone{p}} + \sum_{p
\in
G} \cmpconj{t_{\physvec{\omega},p}} \xi_{\physvec{\omega}}^{G \setminus \setone{p}},$$
with the middle series norm convergent.
Equation \eqref{eq:flip-amplitude} fixes the assignment of
$t_{\physvec{\omega},p}$ versus $\cmpconj{t_{\physvec{\omega},p}}$.
This assignment extends linearly to an operator on
$\sphilb{D}_{\physvec{\omega}}$.
Its matrix is Hermitian,
and the operator is symmetric.
Comparison with the matrix elements of $\opform{Q}_{\physvec{\omega}}$ computed above shows that
it realizes the form.

Apply Lemma \ref{lem:centered-tensor-sum} to
$a_{\physvec{\omega},p}
=
\physham[h]_{\physvec{\omega},p}
- \fun{\oastate[\psi_{\physvec{\omega}}]}{\physham[h]_{\physvec{\omega},p}}$.
Equation \eqref{eq:flip-amplitude} gives
$\norm{a_{\physvec{\omega},p}\xi_{\physvec{\omega},p}^{+}}
=
\abs{t_{\physvec{\omega},p}}$.
Thus the square-summability condition is exactly the hypothesis of that lemma.
The operator just constructed is therefore essentially self-adjoint.
Its self-adjoint closure is the operator
$\physham_{\txtbogoliubov,\physvec{\omega}}$ asserted in (2),
and $\sphilb{D}_{\physvec{\omega}}$ is a core.
If another self-adjoint operator has the same core and represents
$\opform{Q}_{\physvec{\omega}}$ there,
the two operators agree on that common core and hence are equal.

(3)
Definition \ref{def:aligned-configuration} and the definition of
$\physvec{u}^{+}_{p}$ in \eqref{eq:flip-vectors} give
$$\physvec{b}_{\physvec{\omega},p} \cdot \physvec{u}^{+}_{p}
=
\eta_{\physvec{\omega},p} \abs{\physvec{b}_{\physvec{\omega},p}} \physvec{u}_{p} \cdot \physvec{u}^{+}_{p}
=
\frac{\eta_{\physvec{\omega},p} \abs{\physvec{b}_{\physvec{\omega},p}}}{2}
\physvec{u}_{p} \cdot
\rbk{
\physvec{e}_{\physvec{\omega},p,1}
+
\imunit \physvec{e}_{\physvec{\omega},p,2}
}
=
0.$$
Equation \eqref{eq:flip-amplitude} now gives
$t_{\physvec{\omega},p}
=
0$.
The aligned one-site Hamiltonian and its centered version are
$$\physham[h]_{\physvec{\omega},p}
=
- \eta_{\physvec{\omega},p} \abs{\physvec{b}_{\physvec{\omega},p}} \physvec{\sigma}_{p} \cdot \physvec{u}_{p},
\quad
\physham[h]_{\physvec{\omega},p} - \fun{\oastate[\psi_{\physvec{\omega}}]}{\physham[h]_{\physvec{\omega},p}}
=
\eta_{\physvec{\omega},p} \abs{\physvec{b}_{\physvec{\omega},p}} \rbk{1 - \physvec{\sigma}_{p} \cdot \physvec{u}_{p}}.$$
By Lemma \ref{lem:frame},
its image under $\oarepn_{\physvec{\omega}}$ acts on $\xi_{\physvec{\omega}}^{F}$ as
$2 \eta_{\physvec{\omega},p} \abs{\physvec{b}_{\physvec{\omega},p}}$
times the indicator of $p
\in
F$.
The restriction of $\physham_{\txtbogoliubov,\physvec{\omega}}$ to
$\sphilb{D}_{\physvec{\omega}}$ is therefore the diagonal sum stated in (3).

(4)
Alignment determines the transverse components $\rbk{u^{x}_{p}, u^{y}_{p}}
=
\frac{\eta_{\physvec{\omega},p} \eta_{\physvec{\omega}}}{\sminvtemperature_{c} \abs{\physvec{b}_{\physvec{\omega},p}}} \rbk{\bar{u}_{\physvec{\omega}}^{x}, \bar{u}_{\physvec{\omega}}^{y}}$.
Average this identity over $p
\leq
\Omega$.
Equation \eqref{eq:mean-polarization} identifies the limit of the left side.
The averaged transverse gap equation is
$$\eta_{\physvec{\omega}} \rbk{\bar{u}_{\physvec{\omega}}^{x}, \bar{u}_{\physvec{\omega}}^{y}}
=
\frac{1}{\sminvtemperature_{c}} \eta_{\physvec{\omega}} \rbk{\bar{u}_{\physvec{\omega}}^{x}, \bar{u}_{\physvec{\omega}}^{y}} \lim_{\Omega \to \infty} \frac{1}{\Omega} \sum_{p
\leq
\Omega} \frac{\eta_{\physvec{\omega},p}}{\abs{\physvec{b}_{\physvec{\omega},p}}}.$$
The limit on the right exists because the left side converges and
$\frac{1}{\sminvtemperature_{c}} \eta_{\physvec{\omega}} \rbk{\bar{u}_{\physvec{\omega}}^{x}, \bar{u}_{\physvec{\omega}}^{y}}
\neq
0$.
Dividing by $\eta_{\physvec{\omega}} \rbk{\bar{u}_{\physvec{\omega}}^{x}, \bar{u}_{\physvec{\omega}}^{y}}
\neq
0$ and inserting $\abs{\physvec{b}_{\physvec{\omega},p}}$ from \eqref{eq:effective-hamiltonian} gives \eqref{eq:gap-equation}.
\end{proof}

The energy renormalization \eqref{eq:e-omega} reproduces \cite[Eq. (35)]{ThirringWehrl001}, up to constant rearrangements caused by the different operator ordering of the interaction.

\begin{rem}
The transverse mean of the configuration defines the gap parameter
$$\Delta_{\physvec{\omega}}
=
\frac{\eta_{\physvec{\omega}}}{\sminvtemperature_{c}}
\sqrt{\rbk{\bar{u}_{\physvec{\omega}}^{x}}^{2} + \rbk{\bar{u}_{\physvec{\omega}}^{y}}^{2}}.$$
Equation \eqref{eq:gap-equation} is the self-consistency equation for
$\Delta_{\physvec{\omega}}$.
The effective-field norm in \eqref{eq:effective-hamiltonian} becomes
$$\abs{\physvec{b}_{\physvec{\omega},p}}
=
\sqrt{\varepsilon_{p}^{2} + \Delta_{\physvec{\omega}}^{2}}.$$
In the aligned ground-state sector
$\eta_{\physvec{\omega},p}
\equiv
1$,
Theorem \ref{thm:bogoliubov}(3) gives the diagonal action
$$\physham_{\txtbogoliubov,\physvec{\omega}}
\xi_{\physvec{\omega}}^{F}
=
\rbk{
\sum_{p
\in
F}
2 \sqrt{\varepsilon_{p}^{2} + \Delta_{\physvec{\omega}}^{2}}
}
\xi_{\physvec{\omega}}^{F}.$$
Thus $\physham_{\txtbogoliubov,\physvec{\omega}}$ is nonnegative,
and a one-flip vector at the slot $p$ has excitation energy
$2 \sqrt{\varepsilon_{p}^{2} + \Delta_{\physvec{\omega}}^{2}}$.
If $\Delta_{\physvec{\omega}}
>
0$,
the excitation energy of a mode with $\varepsilon_{p}
=
0$ is $2 \Delta_{\physvec{\omega}}$.
This positive spectral separation is the energy gap,
which explains the name gap equation.
\end{rem}

\subsection{Gauge orbit of aligned product ground states}\label{gauge-orbit-of-aligned-product-ground-states}

The degenerate strong-coupling specialization of Theorem \ref{thm:bogoliubov} constructs constant aligned configurations indexed by the gauge angle. Their product vectors are ground states of the corresponding Bogoliubov--Haag Hamiltonians, and their product representations are mutually disjoint.

\begin{cor}[strong coupling and the gauge orbit of product ground states]\label{cor:strong-coupling}
Let the model be degenerate,
$\varepsilon_{p}
=
\varepsilon$ with $\sminvtemperature_{c} \abs{\varepsilon} < 1$,
and set $r
=
\sqrt{1 - \sminvtemperature_{c}^{2} \varepsilon^{2}}$.
For each $\phi
\in
[0, 2 \pi)$ define the unit vector
$$\fun{\physvec{u}}{\phi}
=
\rbk{r \fun{\cos}{\phi}, r \fun{\sin}{\phi}, \sminvtemperature_{c} \varepsilon}
\in
\mansphere^{2},$$
and let $\fun{\physvec{\omega}}{\phi}$ be the constant configuration determined by
$\physvec{u}_{p}
=
\fun{\physvec{u}}{\phi}$ for every $p
\in
\semigrposint$.
The configuration $\fun{\physvec{\omega}}{\phi}$
is aligned in the sense of Definition \ref{def:aligned-configuration},
with $\eta_{\fun{\physvec{\omega}}{\phi},p}
\equiv
1$,
$\eta_{\fun{\physvec{\omega}}{\phi}}
=
1$,
$\abs{\physvec{b}_{\fun{\physvec{\omega}}{\phi},p}}
=
\frac{1}{\sminvtemperature_{c}}$,
and satisfies the gap equation \eqref{eq:gap-equation}.
The associated Hamiltonian $$\physham_{\txtbogoliubov,\fun{\physvec{\omega}}{\phi}}
=
\frac{1}{\sminvtemperature_{c}}
\sum_{p \in \semigrposint}
\fun{\oarepn_{\fun{\physvec{\omega}}{\phi}}}{1 - \physvec{\sigma}_{p} \cdot \fun{\physvec{u}}{\phi}}
\geq
0$$ has the product vector $\xi_{\fun{\physvec{\omega}}{\phi}}^{\emptyset}$ as its ground state.
The gauge automorphisms permute these product states,
$\oastate[\psi_{\fun{\physvec{\omega}}{\phi}}] \circ \mathfrak{g}_{\vartheta}
=
\oastate[\psi_{\fun{\physvec{\omega}}{\phi + \vartheta}}]$,
and the representations $\oarepn_{\fun{\physvec{\omega}}{\phi}}$ for distinct $\phi
\in
[0, 2 \pi)$ are mutually disjoint.
\end{cor}

This result gives a continuum of mutually disjoint, gauge-related ground-state product representations. Their gauge average and its central direct integral are constructed in Theorem \ref{thm:ground-state-direct-integral}. Their thermal analogues are the factor representations in the direct integral constructed in Section \ref{sec:decomposition}.

\begin{proof}
For the constant configuration the mean polarization is $\eta_{\fun{\physvec{\omega}}{\phi}} \bar{\physvec{u}}_{\fun{\physvec{\omega}}{\phi}}
=
\fun{\physvec{u}}{\phi}$.
It follows that $\eta_{\fun{\physvec{\omega}}{\phi}}
=
1$ and
$$\physvec{b}_{\fun{\physvec{\omega}}{\phi},p}
=
\vecbk{\frac{1}{\sminvtemperature_{c}} r \fun{\cos}{\phi},
\frac{1}{\sminvtemperature_{c}} r \fun{\sin}{\phi},
\varepsilon},
\quad
\abs{\physvec{b}_{\fun{\physvec{\omega}}{\phi},p}}
=
\sqrt{\frac{1}{\sminvtemperature_{c}^{2}} r^{2} + \varepsilon^{2}}
=
\frac{1}{\sminvtemperature_{c}},
\quad
\frac{\physvec{b}_{\fun{\physvec{\omega}}{\phi},p}}{\abs{\physvec{b}_{\fun{\physvec{\omega}}{\phi},p}}}
=
\fun{\physvec{u}}{\phi}.$$
Definition \ref{def:aligned-configuration} shows that the configuration is
aligned with $\eta_{\fun{\physvec{\omega}}{\phi},p}
=
1$,
and the gap equation \eqref{eq:gap-equation} is satisfied because every summand equals
$\sminvtemperature_{c}$.
The form of $\physham_{\txtbogoliubov,\fun{\physvec{\omega}}{\phi}}$ and its ground state are Theorem \ref{thm:bogoliubov}(3).
The gauge action on the states is \eqref{eq:gauge-on-states},
and $R_{- \vartheta}$ rotates through the angle $+ \vartheta$.
It follows that $R_{- \vartheta} \fun{\physvec{u}}{\phi}
=
\fun{\physvec{u}}{\phi + \vartheta}$.
Directly,
$\fun{\oastate[\psi_{\fun{\physvec{\omega}}{\phi}}]}{\fun{\mathfrak{g}_{\vartheta}}{\sigma^{+}_{p}}}
=
\fnexp{\imunit \vartheta} \frac{r}{2} \fnexp{\imunit \phi}$ by \eqref{eq:gauge-action}.
For $\phi
\neq
\phi'$ the constant configurations differ by a fixed positive amount at every site,
$\abs{\fun{\physvec{u}}{\phi} - \fun{\physvec{u}}{\phi'}}^{2}
=
2 r^{2} \rbk{1 - \fun{\cos}{\phi - \phi'}} > 0$.
It follows that $\sum_{p \in \semigrposint} \abs{\physvec{u}_{p} - \physvec{u}'_{p}}^{2}
=
\infty$ and Lemma \ref{lem:geometric-equivalence} with Proposition \ref{prop:disjointness}(2) give disjointness.
\end{proof}

\subsection{Direct integral of the product ground-state sectors}\label{direct-integral-of-the-product-ground-state-sectors}

The gauge orbit in Corollary \ref{cor:strong-coupling} admits a gauge-invariant average independently of the positive-temperature Gibbs-state limit. Its central decomposition consists of the pure product ground states on that orbit. The normalized Haar probability measure on the gauge circle is \begin{equation}\label{eq:gauge-haar-measure}
\msrprb_{\mansphere^{1}}
=
\frac{\opdmsr{\phi}}{2 \pi},
\quad
\phi
\in
\rightopeninterval{0}{2 \pi}.
\end{equation} The weak-\(\ast\) continuity of \(\phi
\mapsto
\oastate[\psi_{\fun{\physvec{\omega}}{\phi}}]\) follows from the gauge covariance in Corollary \ref{cor:strong-coupling}. For \(A
\in
\oa{A}\), it defines the gauge-averaged ground-sector state \begin{equation}\label{eq:ground-state-average}
\fun{\oastate[\psi_{\txtgs}]}{A}
=
\int_{\mansphere^{1}}
\fun{\oastate[\psi_{\fun{\physvec{\omega}}{\phi}}]}{A}
\opdmsr{\msrprb_{\mansphere^{1}}}(\phi).
\end{equation}

Fix the GNS triple \(\rbk{\sphilb{H}_{\txtgs,0},
\oarepn_{\txtgs,0},
\oagnsvector[\Psi_{\txtgs,0}]}\) of \(\oastate[\psi_{\fun{\physvec{\omega}}{0}}]\) and define \begin{equation}\label{eq:ground-fiber-representations}
\oarepn_{\txtgs,\phi}
=
\oarepn_{\txtgs,0} \circ \mathfrak{g}_{\phi}.
\end{equation} The triple \(\rbk{\sphilb{H}_{\txtgs,0},
\oarepn_{\txtgs,\phi},
\oagnsvector[\Psi_{\txtgs,0}]}\) is a GNS triple of \(\oastate[\psi_{\fun{\physvec{\omega}}{\phi}}]\). In this subsection, \(\physham_{\txtbogoliubov,\fun{\physvec{\omega}}{\phi}}\) denotes the realization of the associated Bogoliubov--Haag Hamiltonian on this GNS triple. For a finite \(\Lambda
\subset
\semigrposint\), define the ground-sector one-site Hamiltonians and product dynamics on \(\oa{A}_{\Lambda}\) by \begin{equation}\label{eq:ground-fiber-bogoliubov-dynamics}
\begin{aligned}
\physham[h]_{\txtgs,\phi,p}
&=
\frac{1}{\sminvtemperature_{c}}
\rbk{
1
-
\physvec{\sigma}_{p}
\cdot
\fun{\physvec{u}}{\phi}
},
\\
\fun{\tau^{\txtbogoliubov,\phi}_{t}}{A}
&=
\fnexp{
\imunit t
\sum_{p
\in
\Lambda}
\physham[h]_{\txtgs,\phi,p}
}
A
\fnexp{
-
\imunit t
\sum_{p
\in
\Lambda}
\physham[h]_{\txtgs,\phi,p}
}.
\end{aligned}
\end{equation} The compatible local actions extend by density to a point-norm continuous product dynamics on \(\oa{A}\) for each fixed \(\phi\). Define \begin{equation}\label{eq:ground-direct-integral-representation}
\begin{aligned}
\sphilb{H}_{\txtgs}
&=
\int_{\mansphere^{1}}^{\oplus}
\sphilb{H}_{\txtgs,0}
\opdmsr{\msrprb_{\mansphere^{1}}}(\phi),\\
\oarepn_{\txtgs}
&=
\int_{\mansphere^{1}}^{\oplus}
\oarepn_{\txtgs,\phi}
\opdmsr{\msrprb_{\mansphere^{1}}}(\phi),\\
\fun{\oagnsvector[\Psi_{\txtgs}]}{\phi}
&=
\oagnsvector[\Psi_{\txtgs,0}].
\end{aligned}
\end{equation}

\begin{thm}[direct integral of product ground-state sectors]\label{thm:ground-state-direct-integral}
Under the hypotheses of Corollary \ref{cor:strong-coupling},
the construction \eqref{eq:ground-direct-integral-representation} has the following properties.
\begin{enumerate}
\item The triple
$\rbk{\sphilb{H}_{\txtgs},
\oarepn_{\txtgs},
\oagnsvector[\Psi_{\txtgs}]}$
is the GNS triple of the gauge-averaged state
\eqref{eq:ground-state-average}.
\item The decomposition
\begin{equation}\label{eq:ground-central-decomposition}
\oastate[\psi_{\txtgs}]
=
\int_{\mansphere^{1}}
\oastate[\psi_{\fun{\physvec{\omega}}{\phi}}]
\opdmsr{\msrprb_{\mansphere^{1}}}(\phi)
\end{equation}
is the central decomposition of
$\oastate[\psi_{\txtgs}]$ into mutually disjoint pure product states.
For
\begin{equation}\label{eq:ground-von-neumann-algebra}
\oa{M}_{\txtgs}
=
\oadoublecommutant{\fun{\oarepn_{\txtgs}}{\oa{A}}},
\end{equation}
the center is
\begin{equation}\label{eq:ground-center}
\oa{M}_{\txtgs}
\cap
\oacommutant{\oa{M}_{\txtgs}}
=
\fun{\lp^{\infty}}{\mansphere^{1},\msrprb_{\mansphere^{1}}}.
\end{equation}
\item The fiber Hamiltonians have the positive self-adjoint direct integral
\begin{equation}\label{eq:ground-direct-integral-hamiltonian}
\physham_{\txtgs}
=
\int_{\mansphere^{1}}^{\oplus}
\physham_{\txtbogoliubov,\fun{\physvec{\omega}}{\phi}}
\opdmsr{\msrprb_{\mansphere^{1}}}(\phi).
\end{equation}
The operator \eqref{eq:ground-direct-integral-hamiltonian} satisfies
\begin{equation}\label{eq:ground-direct-integral-spectrum}
\begin{aligned}
\physham_{\txtgs}
&\geq
0,
\quad
\physham_{\txtgs}\oagnsvector[\Psi_{\txtgs}]
=
0,
\\ 
\Ker \physham_{\txtgs}
&=
\set{\phi \mapsto \fun{f}{\phi}\oagnsvector[\Psi_{\txtgs,0}]}
{f \in \fun{\lp^{2}}{\mansphere^{1},\msrprb_{\mansphere^{1}}}},
\\ 
\opspec{\physham_{\txtgs}}
&=
\set{\frac{2 n}{\sminvtemperature_{c}}}
{n \in \monnat}.
\end{aligned}
\end{equation}
In particular,
the ground space has the gauge-circle multiplicity and the spectral gap above it is
$2 / \sminvtemperature_{c}$.
\item The unitary group generated by
\eqref{eq:ground-direct-integral-hamiltonian} defines a normal
$\oawstar$-dynamics on $\oa{M}_{\txtgs}$.
The vector state induced by
$\oagnsvector[\Psi_{\txtgs}]$ is a ground state of this dynamics,
and its fiber dynamics is the Bogoliubov dynamics associated with
$\fun{\physvec{\omega}}{\phi}$.
\end{enumerate}
\end{thm}

\begin{proof}

(1)
The expectation of the direct-integral vector is
$$\bkt{\oagnsvector[\Psi_{\txtgs}]}{
\fun{\oarepn_{\txtgs}}{A}
\oagnsvector[\Psi_{\txtgs}]}
=
\int_{\mansphere^{1}}
\fun{\oastate[\psi_{\fun{\physvec{\omega}}{\phi}}]}{A}
\opdmsr{\msrprb_{\mansphere^{1}}}(\phi)
=
\fun{\oastate[\psi_{\txtgs}]}{A}.$$
It remains to verify cyclicity.
Theorem \ref{thm:lln} and
\eqref{eq:ground-fiber-representations} give the fiberwise strong limit
$$\slim_{\Omega \to \infty}
\fun{\oarepn_{\txtgs,\phi}}{\frac{2}{r} M_{\Omega}^{+}}
=
\fnexp{\imunit \phi}.$$
The operators on the left are uniformly bounded.
Dominated convergence on the direct integral gives
\begin{equation}\label{eq:ground-order-parameter-limit}
\slim_{\Omega \to \infty}
\fun{\oarepn_{\txtgs}}{
\frac{2}{r} M_{\Omega}^{+}}
=
M_{\fnexp{\imunit\phi}}
\otimes
\id_{\sphilb{H}_{\txtgs,0}}.
\end{equation}
Powers of \eqref{eq:ground-order-parameter-limit} and its adjoint show that the closure of
$\fun{\oarepn_{\txtgs}}{\oa{A}}
\oagnsvector[\Psi_{\txtgs}]$
is invariant under multiplication by every trigonometric polynomial in $\phi$.

Continuous $\sphilb{H}_{\txtgs,0}$-valued sections are dense in
$\sphilb{H}_{\txtgs}$.
For a continuous section $\xi$,
the cyclicity of each fiber vector supplies,
locally in $\phi$,
an element $A
\in
\oa{A}$ for which
$\fun{\oarepn_{\txtgs,\phi}}{A}
\oagnsvector[\Psi_{\txtgs,0}]$
approximates $\fun{\xi}{\phi}$.
Norm continuity of the gauge action extends the approximation to an open arc.
A finite partition of unity subordinate to these arcs gives a finite sum of such sections with continuous scalar coefficients.
Uniform approximation of those coefficients by trigonometric polynomials and
\eqref{eq:ground-order-parameter-limit} place $\xi$ in the closure of
$\fun{\oarepn_{\txtgs}}{\oa{A}}
\oagnsvector[\Psi_{\txtgs}]$.
This proves (1).

(2)
Equation \eqref{eq:ground-order-parameter-limit} places the multiplication algebra
$\fun{\lp^{\infty}}{\mansphere^{1},\msrprb_{\mansphere^{1}}}$
inside $\oa{M}_{\txtgs}$.
An operator in
$\oacommutant{\oa{M}_{\txtgs}}$
commutes with this multiplication algebra and is therefore decomposable by
Lemma \ref{lem:diagonal-commutant}.
Its fibers commute with
$\fun{\oarepn_{\txtgs,\phi}}{\oa{A}}$ almost everywhere.
Proposition \ref{prop:irreducible} makes each fiber scalar,
and Corollary \ref{cor:strong-coupling} identifies the distinct fibers as mutually disjoint sectors.
Consequently,
\begin{equation}\label{eq:ground-commutant}
\oacommutant{\oa{M}_{\txtgs}}
=
\fun{\lp^{\infty}}{\mansphere^{1},\msrprb_{\mansphere^{1}}},
\end{equation}
and the bicommutant is the algebra of all decomposable
$\opspbddlin{\sphilb{H}_{\txtgs,0}}$-valued operators.
Equations \eqref{eq:ground-center} and
\eqref{eq:ground-central-decomposition} follow.
This proves (2).

(3)
Equation \eqref{eq:gauge-action} and the formula for
$\fun{\physvec{u}}{\phi}$ in Corollary \ref{cor:strong-coupling} give
$$\fun{\oarepn_{\txtgs,\phi}}{
1
-
\physvec{\sigma}_{p}
\cdot
\fun{\physvec{u}}{\phi}}
=
\fun{\oarepn_{\txtgs,0}}{
1
-
\physvec{\sigma}_{p}
\cdot
\fun{\physvec{u}}{0}}.$$
The represented fiber Hamiltonians are therefore the same positive self-adjoint operator on
$\sphilb{H}_{\txtgs,0}$.
The diagonal action in Theorem \ref{thm:bogoliubov}(3) is
$$\physham_{\txtbogoliubov,\fun{\physvec{\omega}}{0}}
\xi_{\fun{\physvec{\omega}}{0}}^{F}
=
\frac{2 \abscard{F}}{\sminvtemperature_{c}}
\xi_{\fun{\physvec{\omega}}{0}}^{F}.$$
Proposition \ref{prop:flip-basis} shows that these vectors form an orthonormal basis.
The kernel and spectrum in
\eqref{eq:ground-direct-integral-spectrum} follow fiberwise.
This proves (3).

(4)
The unitaries
$\fnexp{\imunit t\physham_{\txtgs}}$
are decomposable and preserve the decomposable von Neumann algebra
$\oa{M}_{\txtgs}$ by conjugation.
Their positive generator annihilates
$\oagnsvector[\Psi_{\txtgs}]$.
The restriction to the fiber at $\phi$ is generated by
$\physham_{\txtbogoliubov,\fun{\physvec{\omega}}{\phi}}$,
which proves (4).
\end{proof}

Theorem \ref{thm:ground-state-direct-integral} is a central decomposition of the gauge-averaged state into states that are ground states of the covariant family of fiber Bogoliubov dynamics. The precise distinction from a ground-state decomposition for one fixed dynamics is recorded next.

\subsection{The phase-extended algebra}\label{sec:phase-extended-algebra}

The obstruction on the bare quasi-local algebra disappears when the gauge phase is retained as a central classical variable. This construction is the restriction to the gauge orbit of the physical-algebra method of \cite{PavelBona001,PavelBona002}.

Bóna's physical algebra is not generally \(\fun{\conti}{\mansphere^{1}} \otimes \oa{A}\). In the norm-continuous setting of \cite{PavelBona001}, it has the form \(\oa{A} \otimes N\), where \(N\) is a commutative \(\oacstar\)-algebra of macroscopic classical observables. If \(K\) is the Gelfand spectrum of \(N\), then \(N
\cong
\fun{\conti}{K}\). A closed invariant gauge orbit \(\mathcal{O}
\subset
K\) gives the quotient \begin{equation}\label{eq:bona-orbit-quotient}
\rbk{\oa{A} \otimes_{\min} N}
\big/
\rbk{\oa{A} \otimes_{\min} I_{\mathcal{O}}}
\cong
\oa{A} \otimes_{\min} \fun{\conti}{\mathcal{O}},
\end{equation} where \(I_{\mathcal{O}}\) is the ideal of functions in \(\fun{\conti}{K}\) vanishing on \(\mathcal{O}\). The quotient identification uses the nuclearity of the UHF algebra \(\oa{A}\). For the ground-state orbit of Definition \ref{def:main-ground-orbit}, \(\mathcal{O}
=
\mansphere^{1}\). For this orbit, the quotient algebra is \(\oa{C}\) in \eqref{eq:phase-extended-algebra}. The dynamics on Bóna's full algebra generally combines a classical flow on \(K\) with the quantum dynamics in its fibers. Every point of the aligned ground-state orbit is stationary for the self-consistent classical flow. Stationarity makes the classical pullback trivial on this orbit. The extended dynamics reduces to the fiberwise formula below. The copy of the bare algebra in \(\oa{C}\) is the constant-section embedding \begin{equation}\label{eq:constant-section-embedding}
\fun{\iota}{A}(\phi)
=
A
\quad
\rbk{A
\in
\oa{A}}.
\end{equation}

\begin{prop}[phase-extended ground-state dynamics]\label{prop:phase-extended-ground-dynamics}
The formula
\begin{equation}\label{eq:phase-extended-ground-dynamics}
\funrbk{\fun{\widetilde{\alpha}^{\txtgs}_{t}}{F}}{\phi}
=
\fun{\tau^{\txtbogoliubov,\phi}_{t}}{\fun{F}{\phi}}
\end{equation}
defines a point-norm continuous automorphism group on
$\oa{C}$.
The state
\begin{equation}\label{eq:phase-extended-ground-state}
\fun{\oastate[\widetilde{\psi}_{\txtgs}]}{F}
=
\int_{\mansphere^{1}}
\fun{\oastate[\psi_{\fun{\physvec{\omega}}{\phi}}]}{\fun{F}{\phi}}
\opdmsr{\msrprb_{\mansphere^{1}}}(\phi)
\end{equation}
is a ground state of
$\rbk{\oa{C}, \widetilde{\alpha}^{\txtgs}}$.
The representation
\begin{equation}\label{eq:phase-extended-ground-representation}
\funrbk{\fun{\oarepn_{\txtgs,\txtexternal}}{F}\xi}{\phi}
=
\fun{\oarepn_{\txtgs,\phi}}{\fun{F}{\phi}}
\fun{\xi}{\phi}
\end{equation}
satisfies
\begin{equation}\label{eq:phase-extended-ground-covariance}
\fun{\alpha^{\txtgs}_{t}}{
\fun{\oarepn_{\txtgs,\txtexternal}}{F}
}
=
\fun{\oarepn_{\txtgs,\txtexternal}}{
\fun{\widetilde{\alpha}^{\txtgs}_{t}}{F}
}.
\end{equation}
If
$r
>
0$,
the constant-section algebra
$\fun{\iota}{\oa{A}}$
is not invariant under
$\widetilde{\alpha}^{\txtgs}$.
\end{prop}

\begin{proof}
The algebraic tensor product
$\fun{\conti}{\mansphere^{1}} \odot \oa{A}_{\txtloc}$
is dense in
$\oa{C}$.
For every local
$A$,
the map
\begin{equation*}
\rbk{t,\phi}
\longmapsto
\fun{\tau^{\txtbogoliubov,\phi}_{t}}{A}
\end{equation*}
is norm continuous.
Its norm-continuity is uniform on compact subsets of
$\fldreal \times \mansphere^{1}$,
because the local product dynamics is generated by finitely many one-site Hamiltonians depending continuously on
$\phi$.
The fiberwise group law and the preceding uniform continuity show that
\eqref{eq:phase-extended-ground-dynamics} defines a point-norm continuous group of isometric
$\ast$-automorphisms on the dense algebraic tensor product.
The group extends uniquely to
$\oa{C}$.

For
$F
\in
\oa{C}$,
the right side of
\eqref{eq:phase-extended-ground-representation}
is a bounded measurable operator field with essential supremum at most
$\norm{F}$.
It defines the representation
$\oarepn_{\txtgs,\txtexternal}$.
Equations \eqref{eq:ground-wstar-dynamics-fibers} and
\eqref{eq:phase-extended-ground-representation} give
\eqref{eq:phase-extended-ground-covariance} fiberwise.
The constant section
$\oagnsvector[\Psi_{\txtgs}]$
represents the state
\eqref{eq:phase-extended-ground-state}.
Theorem \ref{thm:ground-state-direct-integral}(4) makes its vector state a ground state of
$\alpha^{\txtgs}$.
Covariance \eqref{eq:phase-extended-ground-covariance} identifies
\eqref{eq:phase-extended-ground-state}
with the restriction of that ground state to the invariant represented copy of
$\oa{C}$.

For
$A
\in
\oa{A}$,
the evolved constant section is
\begin{equation}\label{eq:evolved-constant-section}
\funrbk{
\fun{\widetilde{\alpha}^{\txtgs}_{t}}{\fun{\iota}{A}}
}{\phi}
=
\fun{\tau^{\txtbogoliubov,\phi}_{t}}{A}.
\end{equation}
For the one-site observable
$\sigma^{z}_{p}$,
the derivative of the section in
\eqref{eq:evolved-constant-section} at
$t
=
0$
is
\begin{equation}\label{eq:phase-extended-generator-obstruction}
\left.
\frac{\partial}{\partial t}
\fun{\tau^{\txtbogoliubov,\phi}_{t}}{\sigma^{z}_{p}}
\right|_{t = 0}
=
\frac{2 r}{\sminvtemperature_{c}}
\rbk{
\fun{\sin}{\phi}\sigma^{x}_{p}
-
\fun{\cos}{\phi}\sigma^{y}_{p}
}.
\end{equation}
This section depends nontrivially on
$\phi$
when
$r
>
0$.
Equation \eqref{eq:phase-extended-generator-obstruction} shows that
\eqref{eq:evolved-constant-section} is not generally constant.
\end{proof}

\begin{rem}[fiber ground states and descent of the dynamics]\label{rem:ground-dynamics-descent}
The direct-integral Hamiltonian
\eqref{eq:ground-direct-integral-hamiltonian} defines the normal
$\oawstar$-dynamics whose action on
$X
\in
\oa{M}_{\txtgs}$
is
\begin{equation}\label{eq:ground-wstar-dynamics}
\fun{\alpha^{\txtgs}_{t}}{X}
=
\fnexp{\imunit t\physham_{\txtgs}}
X
\fnexp{-\imunit t\physham_{\txtgs}}.
\end{equation}
Its action on a represented quasi-local observable is the decomposable operator whose fiber is
\begin{equation}\label{eq:ground-wstar-dynamics-fibers}
\funrbk{
\fun{\alpha^{\txtgs}_{t}}{
\fun{\oarepn_{\txtgs}}{A}
}
\xi
}{\phi}
=
\fun{\oarepn_{\txtgs,\phi}}{
\fun{\tau^{\txtbogoliubov,\phi}_{t}}{A}
}
\fun{\xi}{\phi}.
\end{equation}
Equation \eqref{eq:ground-wstar-dynamics} combines the
$\phi$-dependent Bogoliubov dynamics after the passage to the von Neumann algebra
$\oa{M}_{\txtgs}$.
\end{rem}

\begin{prop}[non-descent to the quasi-local algebra]\label{prop:ground-dynamics-non-descent}
Suppose that
$r
>
0$.
There is no group homomorphism
$\tau
\colon
\fldreal
\to
\fun{\grauto}{\oa{A}}$ such that
$t
\mapsto
\fun{\tau_{t}}{A}$
is norm continuous for every
$A
\in
\oa{A}$ and
\begin{equation}\label{eq:ground-dynamics-descent}
\fun{\alpha^{\txtgs}_{t}}{
\fun{\oarepn_{\txtgs}}{A}
}
=
\fun{\oarepn_{\txtgs}}{
\fun{\tau_{t}}{A}
}.
\end{equation}
The represented quasi-local algebra
$\fun{\oarepn_{\txtgs}}{\oa{A}}$
is not invariant under
$\alpha^{\txtgs}$.
\end{prop}

\begin{proof}
Suppose that such a group homomorphism exists.
Lemma \ref{lem:simple} shows that the UHF algebra
$\oa{A}$ is simple.
The nonzero representation
$\oarepn_{\txtgs}$ is therefore faithful.
For each fixed $t$,
the algebraic part of
\eqref{eq:ground-dynamics-descent} is equivalent to
\begin{equation}\label{eq:ground-represented-algebra-invariance}
\fun{\alpha^{\txtgs}_{t}}{
\fun{\oarepn_{\txtgs}}{\oa{A}}
}
=
\fun{\oarepn_{\txtgs}}{\oa{A}}.
\end{equation}
If \eqref{eq:ground-represented-algebra-invariance} held,
it would determine
$\tau_{t}$
uniquely by conjugating the restriction of
$\alpha^{\txtgs}_{t}$
with the faithful representation
$\oarepn_{\txtgs}$.
Point-norm continuity of the resulting group would remain an additional requirement.
Equivalently,
\eqref{eq:ground-dynamics-descent} would require
\begin{equation}\label{eq:ground-dynamics-descent-fibers}
\fun{\oarepn_{\txtgs,\phi}}{
\fun{\tau_{t}}{A}
}
=
\fun{\oarepn_{\txtgs,\phi}}{
\fun{\tau^{\txtbogoliubov,\phi}_{t}}{A}
}
\end{equation}
for almost every $\phi$ with the same element
$\fun{\tau_{t}}{A}
\in
\oa{A}$ on every fiber.
Faithfulness of
$\oarepn_{\txtgs,\phi}$
would force
$\tau_{t}
=
\tau^{\txtbogoliubov,\phi}_{t}$
for almost every $\phi$.
The norm continuity of the right side in $\phi$ would extend this equality to every $\phi$.
This is impossible when
$r
>
0$.
The one-site generator in
\eqref{eq:phase-extended-generator-obstruction}
depends nontrivially on
$\phi$.
No phase-independent group on
$\oa{A}$
can satisfy
\eqref{eq:ground-dynamics-descent-fibers}.
In particular,
\eqref{eq:ground-wstar-dynamics} does not preserve
$\fun{\oarepn_{\txtgs}}{\oa{A}}$.
\end{proof}

\begin{rem}[ground-state interpretation]
This obstruction agrees with the absence of a norm limit of the finite-volume dynamics
\eqref{eq:finite-volume-dynamics} on $\oa{A}$.
Theorem \ref{thm:dynamics} instead gives a strong limit only after a product representation has been fixed.

For a fixed
$\oacstar$-dynamical system
$\rbk{\oa{A},\tau}$,
a decomposition inside its ground-state convex set requires almost every component state to be a ground state for that same group
$\tau$.
Equation \eqref{eq:ground-central-decomposition} instead has components
$\oastate[\psi_{\fun{\physvec{\omega}}{\phi}}]$
that are ground states of their respective groups
$\tau^{\txtbogoliubov,\phi}$.
It is therefore a central decomposition of
$\oastate[\psi_{\txtgs}]$
and a direct integral of fiber ground-state representations,
but it cannot be a ground-state decomposition for a single dynamics on
$\oa{A}$
that reproduces all of the fiber Bogoliubov dynamics.
A common fiberwise dynamics is obtained in
Proposition \ref{prop:phase-extended-ground-dynamics}
after adjoining the phase as a central classical variable.
This is the gauge-orbit restriction of the physical-algebra construction introduced for quantum mean-field systems in
\cite{PavelBona001}.
The equilibrium and ground-state analysis of \cite{PavelBona002} includes the strong-coupling quasi-spin BCS model in that framework.
Proposition \ref{prop:ground-dynamics-non-descent} does not contradict those results,
because it asserts non-descent only to the bare UHF algebra $\oa{A}$.
Theorem \ref{thm:ground-state-direct-integral} gives the corresponding von Neumann direct-integral realization.
\end{rem}

\section{Time Evolution in Product Sectors}\label{sec:time}

The thermodynamic-limit Heisenberg dynamics is derived inside each aligned product sector. Uniform estimates for noncommutative polynomials and a classical substitution produce the limiting dynamics. The final subsections isolate the failures that occur without the gap equation or at the level of implementing unitaries.

For the finite-volume Hamiltonian \eqref{eq:bcs-hamiltonian}, the Heisenberg dynamics is defined by \begin{equation}\label{eq:finite-volume-dynamics}
\fun{\tau^{\Omega}_{t}}{A}
=
\fnexp{\imunit t \physham_{\Omega}} A \fnexp{- \imunit t \physham_{\Omega}}.
\end{equation} The dynamics \eqref{eq:finite-volume-dynamics} does not converge in norm on \(\oa{A}\), because the mean-field interaction is not quasi-local.

\begin{defn}[frozen mean-field dynamics]\label{def:frozen-mean-field-dynamics}
Let $\physvec{\omega}$ be a spin configuration with a mean polarization.
Its frozen mean field is the time-independent family
$\seq{\physvec{b}_{\physvec{\omega},p}}{p \in \semigrposint}$
of \eqref{eq:effective-hamiltonian},
evaluated from $\physvec{\omega}$ once and held fixed for all times.
The associated frozen mean-field dynamics is the product automorphism group generated by
$\physham[h]_{\physvec{\omega},p}
=
- \physvec{b}_{\physvec{\omega},p} \cdot \physvec{\sigma}_{p}$.
For $A
\in
\oa{A}_{\Lambda}$,
its time-$t$ automorphism is
$$A
\mapsto
\rbk{\bigotimes_{p \in \Lambda}
\fnexp{\imunit t \physham[h]_{\physvec{\omega},p}}}
A
\rbk{\bigotimes_{p \in \Lambda}
\fnexp{- \imunit t \physham[h]_{\physvec{\omega},p}}}.$$
\end{defn}

Thus frozen means that the field is not recomputed from the time-evolved mean polarization. Nevertheless, in every product representation associated with a configuration aligned in the sense of Definition \ref{def:aligned-configuration}, it converges strongly to the Bogoliubov dynamics, whereas Proposition \ref{prop:no-frozen-mean-field-dynamics} gives an exact failure mechanism for the dynamics of Definition \ref{def:frozen-mean-field-dynamics} in degenerate constant non-aligned sectors. The analysis proves the aligned convergence statement of \cite[Section 4]{ThirringWehrl001}, establishes the counterexample, and prepares the algebraic machinery reused for the thermal Green functions in Section \ref{sec:green}. The symbols \(\sphilb{H}_{\physvec{\omega}}\) and \(\oarepn_{\physvec{\omega}}\) retain the meanings fixed in Definition \ref{def:spin-product-representation}.

\subsection{Mean-field polynomials and uniform derivation bounds}\label{sec:polynomials}

Noncommutative polynomials in the actual finite-volume observables combine finitely supported quasi-spin observables with the weighted mean spins. These polynomials are stable under commutation with the finite-volume Hamiltonian, and the uniform bounds proved below control the Taylor series for the represented dynamics.

The commutator with \(\physham_{\Omega}\) maps local observables into polynomials in local observables and the \(\varepsilon\)-weighted means \begin{equation}\label{eq:weighted-means}
M^{\gamma}_{n, \Omega}
=
\frac{1}{\Omega} \sum_{p
=
1}^{\Omega} \varepsilon_{p}^{n} \sigma^{\gamma}_{p},
\quad
M^{\pm}_{n, \Omega}
=
\frac{1}{\Omega} \sum_{p
=
1}^{\Omega} \varepsilon_{p}^{n} \sigma^{\pm}_{p}
\quad \rbk{n
\geq
0},
\end{equation} which for \(n
=
0\) are the mean magnetizations of \eqref{eq:collective}. The operators \(\frac{1}{2 \Omega} \sum_{p
\leq
\Omega} \varepsilon^{n}_{p} \physvec{\sigma}_{p}\) are the quantities \(s_{n}\) of \cite[Eq. (38)]{ThirringWehrl001}.

\begin{defn}[mean-field polynomial family]\label{def:mean-field-polynomial-family}
Fix an integer $\Omega_{0}
\geq
1$ and a finite set
$\Lambda_{0}
\subset
\intint{1..\Omega_{0}}$.
A mean-field polynomial family is a sequence
$\seq{F_{\Omega}}{\Omega \geq \Omega_{0}}$ with
$F_{\Omega}
\in
\oa{A}_{\Omega}$ of the form
$$F_{\Omega}
=
\sum_{a
\in
I} c_{a}
\prod_{j
=
1}^{\ell_{a}} X_{a,j,\Omega},$$
where $I$ is finite and the coefficients $c_{a}
\in
\fldcmp$ do not depend on $\Omega$.
Each factor $X_{a,j,\Omega}$ is either a fixed local operator
$\sigma^{\gamma}_{p}$ with $p
\in
\Lambda_{0}$ and $\gamma
\in
\setone{+, -, z}$,
or one of the operators $M^{\gamma}_{n,\Omega}$ in
\eqref{eq:weighted-means}.
The choice and order of all factors are independent of $\Omega$.
A family with one summand is called a monomial family.
If its unique coefficient is $c$ and its mean-field factors have indices
$n_{1}, \dotsc, n_{r}$,
set
$$\fun{\nu}{F}
=
\abs{c} E^{\sum_{j
=
1}^{r} n_{j}}.$$
\end{defn}

For a monomial family, the local and mean-field factor bounds give the uniform estimate \(\norm{F_{\Omega}}
\leq
\fun{\nu}{F}\).

All commutators in the finite-volume derivation are ordinary operator commutators in \(\oa{A}_{\Omega}\). The one-site commutators may equivalently be computed first in \(\spmat{2}{\fldcmp}\) and then placed at the site \(p\).

\begin{lem}[exact finite-volume derivation]\label{lem:exact-derivation}
For $\Omega
\geq
\Omega_{0}$,
define
$\fun{\delta_{\Omega}}{A}
=
\imunit \commutator{\physham_{\Omega}}{A}$ for every $A
\in
\oa{A}_{\Omega}$.
Its values on the factors in Definition
\ref{def:mean-field-polynomial-family} are
\begin{align*}
\fun{\delta_{\Omega}}{\sigma^{\gamma}_{p}} &
=
\imunit \rbk{- \varepsilon_{p} \commutator{\sigma^{z}_{p}}{\sigma^{\gamma}_{p}} - \frac{2}{\sminvtemperature_{c}} M^{+}_{0,\Omega} \commutator{\sigma^{-}_{p}}{\sigma^{\gamma}_{p}} - \frac{2}{\sminvtemperature_{c}} \commutator{\sigma^{+}_{p}}{\sigma^{\gamma}_{p}} M^{-}_{0,\Omega}}, \\
\fun{\delta_{\Omega}}{M^{+}_{n,\Omega}} &
=
\imunit \rbk{- 2 M^{+}_{n + 1,\Omega} + \frac{2}{\sminvtemperature_{c}} M^{+}_{0,\Omega} M^{z}_{n,\Omega}}, \\
\fun{\delta_{\Omega}}{M^{-}_{n,\Omega}} &
=
\imunit \rbk{2 M^{-}_{n + 1,\Omega} - \frac{2}{\sminvtemperature_{c}} M^{z}_{n,\Omega} M^{-}_{0,\Omega}}, \\
\fun{\delta_{\Omega}}{M^{z}_{n,\Omega}} &
=
- \frac{4 \imunit}{\sminvtemperature_{c}} \rbk{M^{+}_{0,\Omega} M^{-}_{n,\Omega} - M^{+}_{n,\Omega} M^{-}_{0,\Omega}}.
\end{align*}
For every mean-field polynomial family
$\seq{F_{\Omega}}{\Omega \geq \Omega_{0}}$,
$\seq{\fun{\delta_{\Omega}}{F_{\Omega}}}{\Omega \geq \Omega_{0}}$
is again a mean-field polynomial family.
\end{lem}

\begin{proof}
For the factor $\sigma^{\gamma}_{p}$,
the kinetic part contributes $\imunit \commutator{- \varepsilon_{p} \sigma^{z}_{p}}{\sigma^{\gamma}_{p}}$,
and the interaction part
$$\begin{aligned}
- \frac{2 \imunit}{\sminvtemperature_{c} \Omega}
\commutator{S^{+}_{\Omega} S^{-}_{\Omega}}{\sigma^{\gamma}_{p}}
&=
- \frac{2 \imunit}{\sminvtemperature_{c} \Omega}
\rbk{S^{+}_{\Omega}
\commutator{S^{-}_{\Omega}}{\sigma^{\gamma}_{p}}
+ \commutator{S^{+}_{\Omega}}{\sigma^{\gamma}_{p}} S^{-}_{\Omega}}
\\ 
&=
- \frac{2 \imunit}{\sminvtemperature_{c}}
\rbk{M^{+}_{0, \Omega}
\commutator{\sigma^{-}_{p}}{\sigma^{\gamma}_{p}}
+ \commutator{\sigma^{+}_{p}}{\sigma^{\gamma}_{p}} M^{-}_{0, \Omega}},
\end{aligned}$$
because $\commutator{S^{\pm}_{\Omega}}{\sigma^{\gamma}_{p}}
=
\commutator{\sigma^{\pm}_{p}}{\sigma^{\gamma}_{p}}$ is one-site.

For the factor $M^{+}_{n,\Omega}$,
the kinetic part gives $$\frac{\imunit}{\Omega}
\sum_{p
\leq
\Omega}
\varepsilon^{n}_{p}
\commutator{- \varepsilon_{p} \sigma^{z}_{p}}{\sigma^{+}_{p}}
=
- 2 \imunit M^{+}_{n + 1, \Omega}$$
from $\commutator{\sigma^{z}}{\sigma^{+}}
=
2 \sigma^{+}.$
The interaction part gives,
with $\commutator{S^{-}_{\Omega}}{\sigma^{+}_{p}}
=
- \sigma^{z}_{p}$ and $\commutator{S^{+}_{\Omega}}{\sigma^{+}_{p}}
=
0$,
$$- \frac{2 \imunit}{\sminvtemperature_{c} \Omega} \cdot \frac{1}{\Omega} \sum_{p
\leq
\Omega} \varepsilon^{n}_{p} S^{+}_{\Omega} \rbk{- \sigma^{z}_{p}}
=
\frac{2 \imunit}{\sminvtemperature_{c}} M^{+}_{0, \Omega} M^{z}_{n, \Omega}.$$

For the factor $M^{-}_{n,\Omega}$,
$\commutator{\sigma^{z}}{\sigma^{-}}
=
- 2 \sigma^{-}$ gives $+ 2 \imunit M^{-}_{n + 1, \Omega}$,
and $\commutator{S^{+}_{\Omega}}{\sigma^{-}_{p}}
=
\sigma^{z}_{p}$,
$\commutator{S^{-}_{\Omega}}{\sigma^{-}_{p}}
=
0$ give $- \frac{2 \imunit}{\sminvtemperature_{c}} M^{z}_{n, \Omega} M^{-}_{0, \Omega}$.

For the factor $M^{z}_{n,\Omega}$,
the kinetic part vanishes,
and $\commutator{S^{-}_{\Omega}}{\sigma^{z}_{p}}
=
2 \sigma^{-}_{p}$,
$\commutator{S^{+}_{\Omega}}{\sigma^{z}_{p}}
=
- 2 \sigma^{+}_{p}$ give
$$- \frac{2 \imunit}{\sminvtemperature_{c} \Omega} \cdot \frac{1}{\Omega} \sum_{p
\leq
\Omega} \varepsilon^{n}_{p} \rbk{2 S^{+}_{\Omega} \sigma^{-}_{p} - 2 \sigma^{+}_{p} S^{-}_{\Omega}}
=
- \frac{4 \imunit}{\sminvtemperature_{c}} \rbk{M^{+}_{0, \Omega} M^{-}_{n, \Omega} - M^{+}_{n, \Omega} M^{-}_{0, \Omega}}.$$

All identities are exact for every finite $\Omega$.
No error terms arise because the commutators of the collective operators with one-site operators are themselves one-site operators.

On a product of factors,
the Leibniz rule expands $\fun{\delta_{\Omega}}{F_{\Omega}}$ into a finite sum obtained by replacing one factor at a time by its corresponding generator formula.
Every summand is a product of the factors allowed in Definition
\ref{def:mean-field-polynomial-family}.
This proves the asserted closure of mean-field polynomial families under
$\delta_{\Omega}$.
\end{proof}

\begin{lem}[uniform derivation bounds]\label{lem:derivation-bounds}
Let $a
=
2 E + \frac{8}{\sminvtemperature_{c}}$.
For every monomial family
$\seq{W_{\Omega}}{\Omega \geq \Omega_{0}}$ of degree $\ell
\geq
1$ and every $N
\geq
0$,
the iterated derivation satisfies uniformly in $\Omega$
$$\norm{\fun{\delta_{\Omega}^{N}}{W_{\Omega}}}
\leq
a^{N} \frac{\rbk{\ell + N - 1}!}{\rbk{\ell - 1}!} \fun{\nu}{W}.$$
For $B
\in
\oa{A}_{\Lambda_{0}}$ with $\abscard{\Lambda_{0}}
=
k$,
writing $\fun{\delta_{\Omega}}{B}
=
\imunit \commutator{\physham_{\Omega}}{B}$,
the preceding estimate gives
$$\norm{\fun{\delta_{\Omega}^{N}}{B}}
\leq
a^{N} \frac{\rbk{k + N - 1}!}{\rbk{k - 1}!} \cdot 6^{k} \norm{B}.$$
The Heisenberg series
$\fun{\tau^{\Omega}_{t}}{B}
=
\sum_{N
\geq
0} \frac{t^{N}}{N!} \fun{\delta^{N}_{\Omega}}{B}$
converges in norm for $\abs{t}
<
\frac{1}{a}$,
uniformly in $\Omega$.
It extends the Heisenberg dynamics analytically to the strip
$\abs{\opimag z}
<
\frac{1}{a}$.
On this strip,
the uniform bound is
$\norm{\fun{\tau^{\Omega}_{t + \imunit s}}{B}}
\leq
6^{k} \norm{B} \rbk{1 - a \abs{s}}^{- k}$ for all real $t$ and $\abs{s} < \frac{1}{a}$.
\end{lem}

\begin{proof}
The generator formulas in Lemma \ref{lem:exact-derivation} give a uniform estimate for one factor.
The one-site commutators
$\commutator{\sigma^{\pm}_{p}}{\sigma^{\gamma}_{p}}$
are scalar multiples of at most one Pauli factor,
and their scalar coefficients have modulus at most $2$.
The total $\nu$-weights of all terms produced from each type of factor satisfy
$$\begin{aligned}
\sigma^{\gamma}_{p} \colon\quad
&2 E + \frac{4}{\sminvtemperature_{c}} + \frac{4}{\sminvtemperature_{c}}
\leq
a,
\\ 
M^{\pm}_{n,\Omega} \colon\quad
&2 E^{n + 1} + \frac{2}{\sminvtemperature_{c}} E^{n}
\leq
a E^{n},
\\ 
M^{z}_{n,\Omega} \colon\quad
&\frac{8}{\sminvtemperature_{c}} E^{n}
\leq
a E^{n}.
\end{aligned}$$
Each resulting monomial has degree at most $2$.

Consider a monomial of degree $\ell$.
The Leibniz rule replaces one of its $\ell$ factors by one of the terms estimated in the one-factor calculation.
The resulting monomials have degree at most $\ell + 1$,
and the sum of their $\nu$-weights is at most $a \ell$ times the original weight.
After $j$ iterations the degree is at most $\ell + j$.
The case $N
=
0$ is immediate.
Assume $N
\geq
1$.
The product of the successive weight bounds is
$$a^{N}
\prod_{j
=
0}^{N - 1} \rbk{\ell + j}
\fun{\nu}{W}
=
a^{N}
\frac{\rbk{\ell + N - 1}!}{\rbk{\ell - 1}!}
\fun{\nu}{W}.$$
The triangle inequality converts this total weight estimate into the asserted operator-norm estimate for
$\fun{\delta_{\Omega}^{N}}{W_{\Omega}}$.

An element $B
\in
\oa{A}_{\Lambda_{0}}$ expands over the tensor products of the one-site family $1, \sigma^{x}, \sigma^{y}, \sigma^{z}$,
which is orthonormal for the normalized trace state $\oastate[\psi_{\mathrm{tr}}]$ of $\oa{A}_{\Lambda_{0}}$.
Write such a tensor product as
$P
=
\bigotimes_{p
\in
\Lambda_{0}} P_{p}$,
where
$P_{p}
\in
\setone{1, \sigma^{x}, \sigma^{y}, \sigma^{z}}$,
and define
$$\fun{r}{P}
=
\abscard{\set{p \in \Lambda_{0}}{P_{p} \in \setone{\sigma^{x}, \sigma^{y}}}}.$$
For this tensor product,
its coefficient and norm bound are
$$c_{P}
=
\fun{\oastate[\psi_{\mathrm{tr}}]}{\faadj{P} B},
\quad
\abs{c_{P}}
=
\abs{\fun{\oastate[\psi_{\mathrm{tr}}]}{\faadj{P} B}}
\leq
\norm{B}.$$
The identities
$\sigma^{x}
=
\sigma^{+} + \sigma^{-}$ and
$\sigma^{y}
=
- \imunit \rbk{\sigma^{+} - \sigma^{-}}$
expand $P$ into $2^{\fun{r}{P}}$ monomials,
each with a scalar coefficient of modulus $\abs{c_{P}}$.
For fixed $r
\in
\setone{0, \dotsc, k}$,
the number of possible tensor products $P$ is
$\binom{k}{r}
2^{r}
2^{k - r}$.
The three factors count the choices of the $r$ sites,
the choices between $\sigma^{x}$ and $\sigma^{y}$ at these sites,
and the choices between $1$ and $\sigma^{z}$ at the remaining sites.
The sum of the coefficient moduli before collecting equal monomials is bounded by
$$\begin{aligned}
\sum_{\substack{P = \bigotimes_{p \in \Lambda_{0}} P_{p} \\
P_{p} \in \setone{1, \sigma^{x}, \sigma^{y}, \sigma^{z}}}}
2^{\fun{r}{P}}
\abs{c_{P}}
\leq
\norm{B}
\sum_{r
=
0}^{k}
\binom{k}{r}
2^{r}
2^{k - r}
2^{r}
=
2^{k} \norm{B}
\sum_{r
=
0}^{k}
\binom{k}{r} 2^{r}
=
6^{k} \norm{B}.
\end{aligned}$$
The resulting expansion has a finite index set $I_{B}$ and monomials
$W_{i}$ of degree at most $k$ such that
$$B
=
\sum_{i
\in
I_{B}} c_{i} W_{i},
\quad
\sum_{i
\in
I_{B}} \abs{c_{i}}
\leq
6^{k} \norm{B}.$$
Each $W_{i}$ has unit weight,
and the monomial estimate with $\ell
\leq
k$ gives the asserted bound for
$\fun{\delta_{\Omega}^{N}}{B}$.

The coefficient in the Heisenberg series is the binomial coefficient
$$\frac{1}{N!}
\frac{\rbk{k + N - 1}!}{\rbk{k - 1}!}
=
\frac{\rbk{k + N - 1}!}{N! \rbk{k - 1}!}
=
\binom{k + N - 1}{N}.$$
For
$\abs{x}
<
1$,
termwise differentiation of the geometric series gives
$$\begin{aligned}
&\sum_{N \geq 0}
\binom{k + N - 1}{N} x^{N}
=
\sum_{m
\geq
k - 1}
\binom{m}{k - 1} x^{m - k + 1}
\\ 
&=
\frac{1}{\rbk{k - 1}!}
\rbk{\opod{x}}^{k - 1}
\sum_{m
\geq
0} x^{m}
=
\frac{1}{\rbk{k - 1}!}
\rbk{\opod{x}}^{k - 1}
\frac{1}{1 - x}
=
\rbk{1 - x}^{-k}.
\end{aligned}$$
For
$a \abs{t}
<
1$,
substitution of $x=a\abs{t}$ yields
$$\sum_{N
\geq
0}
\frac{\rbk{a \abs{t}}^{N}}{N!}
\frac{\rbk{k + N - 1}!}{\rbk{k - 1}!}
=
\rbk{1 - a \abs{t}}^{-k}.$$
This identity proves uniform norm convergence of the Heisenberg series for
$\abs{t}
<
\frac{1}{a}$.

For fixed finite $\Omega$ the map $t
\mapsto
\fun{\tau^{\Omega}_{t}}{B}$ is entire in norm,
with $N$-th derivative $\fun{\tau^{\Omega}_{t}}{\fun{\delta^{N}_{\Omega}}{B}}$,
because $\physham_{\Omega}$ is bounded.
For real $t$,
the isometry of $\tau^{\Omega}_{t}$ and its commutation with
$\delta_{\Omega}$ ensure that the Taylor coefficients at $t$ obey the same bounds.
The Taylor series at $t$ with increment $\imunit s$ satisfies
$$\norm{\fun{\tau^{\Omega}_{t + \imunit s}}{B}}
\leq
6^{k} \norm{B}
\rbk{1 - a \abs{s}}^{- k},
\quad
t
\in
\fldreal,
\quad
\abs{s}
<
\frac{1}{a}.$$
The estimate is uniform in $\Omega$ and proves the analytic extension on the stated strip.
\end{proof}

\subsection{Classical substitution and stationarity}\label{classical-substitution-and-stationarity}

Classical substitution replaces the mean-field generators by their sectorwise scalar limits. The resulting stationary derivation generates the candidate limiting dynamics.

\begin{defn}[joint mean profile]\label{def:joint-mean-profile}
A spin configuration $\physvec{\omega}$ has a joint mean profile when the empirical measures
$\frac{1}{\Omega} \sum_{p
\leq
\Omega} \delta_{\rbk{\varepsilon_{p}, \physvec{u}_{p}}}$
on $\closedinterval{- E}{E} \times \mansphere^{2}$ converge weakly.
For such a configuration define
\begin{equation}\label{eq:mu-moments}
\mu^{\gamma}_{\physvec{\omega},k}
=
\lim_{\Omega \to \infty} \frac{1}{\Omega} \sum_{p
=
1}^{\Omega} \varepsilon^{k}_{p} u^{\gamma}_{p}
\quad \rbk{\gamma
=
x, y, z},
\quad
\mu^{\pm}_{\physvec{\omega},k}
=
\frac{1}{2} \rbk{\mu^{x}_{\physvec{\omega},k} \pm \imunit \mu^{y}_{\physvec{\omega},k}}
=
\lim_{\Omega \to \infty} \frac{1}{\Omega} \sum_{p
=
1}^{\Omega} \varepsilon^{k}_{p} m^{\pm}_{\physvec{\omega},p}.
\end{equation}
\end{defn}

The joint mean profile retains the correlation between one-particle energy and Bloch direction that the ordinary profile discards. Its moments \(\mu^{\gamma}_{\physvec{\omega},k}\) are energy-weighted polarizations. They enter because commutation with the kinetic Hamiltonian raises the energy power, so the full joint profile supplies the limits required by repeated Heisenberg commutators.

Weak convergence and the energy bound \eqref{eq:energy-bound} give \(\abs{\mu^{\gamma}_{\physvec{\omega},k}}
\leq
E^{k}\). The definitions \eqref{eq:mu-moments} and \eqref{eq:mean-m} identify \(\mu^{\pm}_{\physvec{\omega},0}
=
\bar{m}^{\pm}_{\physvec{\omega}}\).

Let \(\seq{F_{\Omega}}{\Omega \geq \Omega_{0}}\) be a mean-field polynomial family of Definition \ref{def:mean-field-polynomial-family}. Replacing every factor \(M^{\gamma}_{k,\Omega}\) in its fixed polynomial expression by \(\mu^{\gamma}_{\physvec{\omega},k}\) produces an element \(F_{\physvec{\omega}}
\in
\oa{A}_{\Lambda_{0}}\). On \(\oa{A}_{\Lambda_{0}}\) define \[\fun{\delta_{\txtbogoliubov,\physvec{\omega}}}{B}
=
\imunit
\sum_{p
\in
\Lambda_{0}}
\commutator{\physham[h]_{\physvec{\omega},p}}{B},\] where \(\physham[h]_{\physvec{\omega},p}\) is the effective one-site Hamiltonian \eqref{eq:effective-hamiltonian}. On a local factor this derivation is \[\fun{\delta_{\txtbogoliubov,\physvec{\omega}}}{\sigma^{\gamma}_{p}}
=
\imunit
\commutator{\physham[h]_{\physvec{\omega},p}}{\sigma^{\gamma}_{p}}.\]

\begin{lem}[stationarity of the means]\label{lem:stationarity}
Let $\physvec{\omega}$ have a joint mean profile and be aligned in the sense of
Definition \ref{def:aligned-configuration}.
For every $k
\geq
0$,
substitution of the moments \eqref{eq:mu-moments} into the three mean-field formulas of Lemma
\ref{lem:exact-derivation} gives
\begin{align*}
\imunit \rbk{- 2 \mu^{+}_{\physvec{\omega},k + 1} + \frac{2}{\sminvtemperature_{c}} \mu^{+}_{\physvec{\omega},0} \mu^{z}_{\physvec{\omega},k}} &
=
0, \\
\imunit \rbk{2 \mu^{-}_{\physvec{\omega},k + 1} - \frac{2}{\sminvtemperature_{c}} \mu^{z}_{\physvec{\omega},k} \mu^{-}_{\physvec{\omega},0}} &
=
0, \\
- \frac{4 \imunit}{\sminvtemperature_{c}} \rbk{\mu^{+}_{\physvec{\omega},0} \mu^{-}_{\physvec{\omega},k} - \mu^{+}_{\physvec{\omega},k} \mu^{-}_{\physvec{\omega},0}} &
=
0.
\end{align*}
For every mean-field polynomial family,
substitution in the polynomial expression for
$\fun{\delta_{\Omega}}{F_{\Omega}}$ satisfies
$$\rbk{\fun{\delta_{\Omega}}{F_{\Omega}}}_{\physvec{\omega}}
=
\fun{\delta_{\txtbogoliubov,\physvec{\omega}}}{F_{\physvec{\omega}}}.$$
\end{lem}

\begin{proof}
If $\bar{m}^{\pm}_{\physvec{\omega}}
=
0$,
alignment forces $\physvec{b}_{\physvec{\omega},p}
=
\rbk{0, 0, \varepsilon_{p}}$ and $\physvec{u}_{p}
=
\pm \hat{\physvec{z}}$ wherever $\varepsilon_{p}
\neq
0$.
It follows that $m^{\pm}_{\physvec{\omega},p}
=
0$ there,
and at sites with $\varepsilon_{p}
=
0$ alignment requires $\physvec{b}_{\physvec{\omega},p}
\neq
0$.
It follows that such sites are absent.
These observations give $\mu^{\pm}_{\physvec{\omega},k}
=
0$ for all $k$.
Every term obtained by substituting the moments into the mean-field formulas of Lemma
\ref{lem:exact-derivation}
contains one of
$\mu^{\pm}_{\physvec{\omega},0}$,
$\mu^{\pm}_{\physvec{\omega},k + 1}$,
or $\mu^{\pm}_{\physvec{\omega},k}$.
All three scalar expressions vanish.

Suppose that $\bar{m}^{\pm}_{\physvec{\omega}}
\neq
0$.
Alignment gives
$$m^{\pm}_{\physvec{\omega},p}
=
\frac{\eta_{\physvec{\omega},p} \frac{1}{\sminvtemperature_{c}}}{\abs{\physvec{b}_{\physvec{\omega},p}}} \bar{m}^{\pm}_{\physvec{\omega}},
\quad
u^{z}_{p}
=
\frac{\eta_{\physvec{\omega},p} \varepsilon_{p}}{\abs{\physvec{b}_{\physvec{\omega},p}}}.$$
The transverse component of $\physvec{b}_{\physvec{\omega},p}$ is
$$\frac{1}{\sminvtemperature_{c}} \eta_{\physvec{\omega}} \rbk{\bar{u}_{\physvec{\omega}}^{x}, \bar{u}_{\physvec{\omega}}^{y}}
=
\frac{2}{\sminvtemperature_{c}} \rbk{\opreal \bar{m}^{+}_{\physvec{\omega}}, \opimag \bar{m}^{+}_{\physvec{\omega}}}.$$
Define
$$\kappa_{\physvec{\omega},k}
=
\lim_{\Omega \to \infty} \frac{1}{\Omega} \sum_{p
\leq
\Omega} \frac{\eta_{\physvec{\omega},p} \varepsilon^{k}_{p}}{\abs{\physvec{b}_{\physvec{\omega},p}}}.$$
The limit exists by \eqref{eq:mu-moments} and the condition
$\bar{m}^{\pm}_{\physvec{\omega}}
\neq
0$.
The moments are
$$\mu^{\pm}_{\physvec{\omega},k}
=
\frac{1}{\sminvtemperature_{c}} \kappa_{\physvec{\omega},k} \bar{m}^{\pm}_{\physvec{\omega}},
\quad
\mu^{z}_{\physvec{\omega},k}
=
\kappa_{\physvec{\omega},k + 1}.$$
Substitution into the formulas of Lemma \ref{lem:exact-derivation} gives
\begin{align*}
\imunit \rbk{- \frac{2}{\sminvtemperature_{c}} \kappa_{\physvec{\omega},k + 1} \bar{m}^{+}_{\physvec{\omega}} + \frac{2}{\sminvtemperature_{c}} \bar{m}^{+}_{\physvec{\omega}} \kappa_{\physvec{\omega},k + 1}} &
=
0, \\
\imunit \rbk{\frac{2}{\sminvtemperature_{c}} \kappa_{\physvec{\omega},k + 1} \bar{m}^{-}_{\physvec{\omega}} - \frac{2}{\sminvtemperature_{c}} \kappa_{\physvec{\omega},k + 1} \bar{m}^{-}_{\physvec{\omega}}} &
=
0, \\
- \frac{4 \imunit}{\sminvtemperature_{c}} \rbk{\bar{m}^{+}_{\physvec{\omega}} \cdot \frac{\kappa_{\physvec{\omega},k}}{\sminvtemperature_{c}} \bar{m}^{-}_{\physvec{\omega}} - \frac{\kappa_{\physvec{\omega},k}}{\sminvtemperature_{c}} \bar{m}^{+}_{\physvec{\omega}} \cdot \bar{m}^{-}_{\physvec{\omega}}} &
=
0.
\end{align*}

The Leibniz rule extends the substitution identity to every polynomial family.
For local factors it is the definition of $\delta_{\txtbogoliubov,\physvec{\omega}}$,
and for mean-field factors both sides vanish by the three identities just proved.
\end{proof}

\begin{lem}[strong limits of mean-field polynomials]\label{lem:polynomial-limits}
Let $\physvec{\omega}$ have a joint mean profile.
For every mean-field polynomial family
$\seq{F_{\Omega}}{\Omega \geq \Omega_{0}}$ of Definition
\ref{def:mean-field-polynomial-family},
it holds that
$$\slim_{\Omega \to \infty}
\fun{\oarepn_{\physvec{\omega}}}{F_{\Omega}}
=
\fun{\oarepn_{\physvec{\omega}}}{F_{\physvec{\omega}}}.$$
\end{lem}

\begin{proof}
The $\sigma$-generators are constant in $\Omega$.
For a mean-field generator,
Theorem \ref{thm:lln} gives
$$\slim_{\Omega \to \infty}
\fun{\oarepn_{\physvec{\omega}}}{M^{\gamma}_{k, \Omega}}
=
\mu^{\gamma}_{\physvec{\omega},k},$$
where the one-site variables are
$X_{p}
=
\varepsilon^{k}_{p} \sigma^{\gamma}_{p}$,
with $\norm{X_{p}}
\leq
E^{k}$.
Their means are $\varepsilon^{k}_{p} u^{\gamma}_{p}$ or
$\varepsilon^{k}_{p} m^{\pm}_{\physvec{\omega},p}$.

All factors are uniformly bounded.
It follows that Lemma \ref{lem:strong-products},
applied repeatedly to the factors of each monomial,
gives the claim.
\end{proof}

\subsection{Convergence of the Heisenberg dynamics}\label{convergence-of-the-heisenberg-dynamics}

The preceding polynomial estimates and stationarity identities are now combined to prove sectorwise strong convergence of the represented Heisenberg dynamics.

\begin{thm}[Bogoliubov dynamics in aligned phases]\label{thm:dynamics}
Let $\physvec{\omega}$ have a joint mean profile and be aligned in the sense of
Definition \ref{def:aligned-configuration}.
Let $\oarepn_{\physvec{\omega}}$ be the product representation of Definition \ref{def:spin-product-representation}.
Let $\physham_{\txtbogoliubov,\physvec{\omega}}$ be the self-adjoint Bogoliubov--Haag Hamiltonian of Theorem \ref{thm:bogoliubov}(3) on
$\sphilb{H}_{\physvec{\omega}}$.
Define
$\fun{\hat{\tau}_{\physvec{\omega},t}}{X}
=
\fnexp{\imunit t \physham_{\txtbogoliubov,\physvec{\omega}}} X \fnexp{- \imunit t \physham_{\txtbogoliubov,\physvec{\omega}}}$.
Then for every $A
\in
\oa{A}$ and $t
\in
\fldreal$,
the strong limit on $\sphilb{H}_{\physvec{\omega}}$ is
$$\slim_{\Omega \to \infty}
\fun{\oarepn_{\physvec{\omega}}}{\fun{\tau^{\Omega}_{t}}{A}}
=
\fun{\hat{\tau}_{\physvec{\omega},t}}{\fun{\oarepn_{\physvec{\omega}}}{A}}.$$
Moreover $\hat{\tau}_{\physvec{\omega},t}$ implements on $\sphilb{H}_{\physvec{\omega}}$ the product automorphism group
$$\tau^{\txtbogoliubov}_{\physvec{\omega},t}
=
\bigotimes_{p \in \semigrposint} \fun{\Ad}{\fnexp{\imunit t \physham[h]^{\txteff}_{\physvec{\omega},p}}},
\quad
\physham[h]^{\txteff}_{\physvec{\omega},p}
=
\eta_{\physvec{\omega},p} \abs{\physvec{b}_{\physvec{\omega},p}} \rbk{1 - \physvec{\sigma}_{p} \cdot \physvec{u}_{p}},$$
of $\oa{A}$:
$\fun{\hat{\tau}_{\physvec{\omega},t}}{\fun{\oarepn_{\physvec{\omega}}}{A}}
=
\fun{\oarepn_{\physvec{\omega}}}{\fun{\tau^{\txtbogoliubov}_{\physvec{\omega},t}}{A}}$.
\end{thm}

\begin{proof}
The local automorphism group gives the implementation statement.
The prescription $\tau^{\txtbogoliubov}_{\physvec{\omega},t}$ on $\oa{A}_{\Lambda}$ is conjugation by the local unitary $\bigotimes_{p
\in
\Lambda} \fnexp{\imunit t \physham[h]^{\txteff}_{\physvec{\omega},p}}$.
As in Definition \ref{def:gauge} it defines a one-parameter automorphism group of $\oa{A}$,
norm continuous in $t$ on each $\oa{A}_{\Lambda}$ with generator $\fun{\delta_{\txtbogoliubov,\physvec{\omega}}}{B}
=
\imunit \sum_{p
\in
\Lambda} \commutator{\physham[h]^{\txteff}_{\physvec{\omega},p}}{B}$ there.
Alignment gives the local Hamiltonian identity
$$\begin{aligned}
\physham[h]_{\physvec{\omega},p}
=
- \eta_{\physvec{\omega},p}
\abs{\physvec{b}_{\physvec{\omega},p}}
\physvec{\sigma}_{p} \cdot \physvec{u}_{p}
=
\physham[h]^{\txteff}_{\physvec{\omega},p}
-
\eta_{\physvec{\omega},p}
\abs{\physvec{b}_{\physvec{\omega},p}}.
\end{aligned}$$
The two local Hamiltonians differ by a scalar,
and therefore generate the same derivation $\delta_{\txtbogoliubov,\physvec{\omega}}$.
Theorem \ref{thm:bogoliubov}(3) realizes the operator on $\sphilb{H}_{\physvec{\omega}}$ as
$$\physham_{\txtbogoliubov,\physvec{\omega}}
=
\sum_{p \in \semigrposint}
\fun{\oarepn_{\physvec{\omega}}}{\physham[h]^{\txteff}_{\physvec{\omega},p}}.$$
This operator acts diagonally on the flip basis.
Under the factorization $U_{\Lambda}$ of Proposition \ref{prop:factorization},
the Hamiltonian splits as
$$\physham_{\txtbogoliubov,\physvec{\omega}}
=
\rbk{\sum_{p
\in
\Lambda}
\physham[h]^{\txteff}_{\physvec{\omega},p}}
\otimes 1
+
1 \otimes \physham_{\txtbogoliubov,\physvec{\omega}}^{\mathrm{tail}}.$$
The two summands commute,
and the tail part is diagonal in the tail flip basis.
The exponential and a local observable $B
\in
\oa{A}_{\Lambda}$ factorize as
$$\begin{aligned}
\fnexp{\imunit t \physham_{\txtbogoliubov,\physvec{\omega}}}
&=
\rbk{\bigotimes_{p
\in
\Lambda}
\fnexp{\imunit t \physham[h]^{\txteff}_{\physvec{\omega},p}}}
\otimes
\fnexp{\imunit t \physham_{\txtbogoliubov,\physvec{\omega}}^{\mathrm{tail}}},
\\
\fun{\oarepn_{\physvec{\omega}}}{B}
&=
U_{\Lambda}
\rbk{B \otimes 1}
\faadj{U_{\Lambda}}.
\end{aligned}$$
Conjugation therefore gives
$$\fun{\hat{\tau}_{\physvec{\omega},t}}{
\fun{\oarepn_{\physvec{\omega}}}{B}}
=
\fun{\oarepn_{\physvec{\omega}}}{
\fun{\tau^{\txtbogoliubov}_{\physvec{\omega},t}}{B}}.$$
Density extends this to $\oa{A}$.

For the convergence,
let $\fun{S_{\physvec{\omega}}}{t}$ denote the following statement.
For every mean-field polynomial family
$\seq{F_{\Omega}}{\Omega \geq \Omega_{0}}$,
it follows that
$$\slim_{\Omega \to \infty}
\fun{\oarepn_{\physvec{\omega}}}{\fun{\tau^{\Omega}_{t}}{F_{\Omega}}}
=
\fun{\hat{\tau}_{\physvec{\omega},t}}{\fun{\oarepn_{\physvec{\omega}}}{F_{\physvec{\omega}}}}.$$
Every local element gives a constant family with
$F_{\physvec{\omega}}
=
F_{\Omega}$.
The operators in the convergence statement are uniformly bounded,
so the convergence extends from local elements to $\oa{A}$ by density.
Together with $\fun{S_{\physvec{\omega}}}{t}$ for every $t$,
this proves the theorem.

We first establish $\fun{S_{\physvec{\omega}}}{t}$ for $\abs{t} < \frac{1}{a}$,
where $a$ is the constant from Lemma \ref{lem:derivation-bounds}.
By finite linearity,
Lemma \ref{lem:derivation-bounds} gives the uniformly norm-convergent Taylor expansion
$$\fun{\tau^{\Omega}_{t}}{F_{\Omega}}
=
\sum_{N
\geq
0}
\frac{t^{N}}{N!}
\fun{\delta_{\Omega}^{N}}{F_{\Omega}},$$
norm convergent uniformly in $\Omega$.
For fixed $N$,
Lemma \ref{lem:polynomial-limits} and Lemma \ref{lem:stationarity} give
$$\slim_{\Omega \to \infty}
\fun{\oarepn_{\physvec{\omega}}}{\fun{\delta_{\Omega}^{N}}{F_{\Omega}}}
=
\fun{\oarepn_{\physvec{\omega}}}{\fun{\delta_{\txtbogoliubov,\physvec{\omega}}^{N}}{F_{\physvec{\omega}}}}.$$
The uniform tail bounds of Lemma \ref{lem:derivation-bounds} justify interchanging the limit and the sum on any fixed vector.
The resulting series is
$$\sum_{N \geq 0}
\frac{t^{N}}{N!}
\fun{\delta_{\txtbogoliubov,\physvec{\omega}}^{N}}{F_{\physvec{\omega}}}
=
\fun{\tau^{\txtbogoliubov}_{\physvec{\omega},t}}{F_{\physvec{\omega}}}.$$
It converges in norm because $\delta_{\txtbogoliubov,\physvec{\omega}}$ preserves supports.
Set
$\bar{b}_{\physvec{\omega}}
=
\sup_{p \in \semigrposint} \abs{\physvec{b}_{\physvec{\omega},p}}
<
\infty$.
If $C$ is supported on $k$ sites,
the scalar-shift identity above and
$\norm{\physham[h]_{\physvec{\omega},p}}
=
\abs{\physvec{b}_{\physvec{\omega},p}}$
give
$$\begin{aligned}
\norm{\fun{\delta_{\txtbogoliubov,\physvec{\omega}}}{C}}
\leq
\sum_{p \in \fun{\supp}{C}}
\norm{\commutator{\physham[h]_{\physvec{\omega},p}}{C}}
\leq
2 k \bar{b}_{\physvec{\omega}} \norm{C}.
\end{aligned}$$
The derivation preserves the support of $C$,
and induction therefore gives
$$\norm{\fun{\delta^{N}_{\txtbogoliubov,\physvec{\omega}}}{C}}
\leq
\rbk{2 k \bar{b}_{\physvec{\omega}}}^{N} \norm{C}.$$
The implementation identity converts the represented sum into
$\fun{\hat{\tau}_{\physvec{\omega},t}}{\fun{\oarepn_{\physvec{\omega}}}{F_{\physvec{\omega}}}}$.

The group property propagates the short-time convergence.

If $\fun{S_{\physvec{\omega}}}{t}$ holds and $\abs{s} < \frac{1}{a}$,
then $\fun{S_{\physvec{\omega}}}{t + s}$ holds.
The group property and the uniformly convergent expansion in $s$ give
$$\fun{\oarepn_{\physvec{\omega}}}{\fun{\tau^{\Omega}_{t + s}}{F_{\Omega}}}
=
\sum_{N
\geq
0}
\frac{s^{N}}{N!}
\fun{\oarepn_{\physvec{\omega}}}{\fun{\tau^{\Omega}_{t}}{\fun{\delta_{\Omega}^{N}}{F_{\Omega}}}}.$$
By $\fun{S_{\physvec{\omega}}}{t}$ applied to the polynomial family
$\seq{\fun{\delta_{\Omega}^{N}}{F_{\Omega}}}{\Omega \geq \Omega_{0}}$,
the strong limit of the $N$-th summand is
$$\slim_{\Omega \to \infty}
\frac{s^{N}}{N!}
\fun{\oarepn_{\physvec{\omega}}}{
\fun{\tau^{\Omega}_{t}}{
\fun{\delta_{\Omega}^{N}}{F_{\Omega}}}}
=
\frac{s^{N}}{N!}
\fun{\hat{\tau}_{\physvec{\omega},t}}{
\fun{\oarepn_{\physvec{\omega}}}{
\fun{\delta_{\txtbogoliubov,\physvec{\omega}}^{N}}{F_{\physvec{\omega}}}}}.$$
The interchange of the limit $\Omega
\to
\infty$ with the $N$-summation is again justified by the uniform tail bounds.
The stationarity identity of Lemma \ref{lem:stationarity} and the continuity of $\hat{\tau}_{\physvec{\omega},t}$ identify the limiting series as
$$\sum_{N \geq 0}
\frac{s^{N}}{N!}
\fun{\hat{\tau}_{\physvec{\omega},t}}{
\fun{\oarepn_{\physvec{\omega}}}{
\fun{\delta_{\txtbogoliubov,\physvec{\omega}}^{N}}{F_{\physvec{\omega}}}}}
=
\fun{\hat{\tau}_{\physvec{\omega},t}}{
\fun{\oarepn_{\physvec{\omega}}}{
\fun{\tau^{\txtbogoliubov}_{\physvec{\omega},s}}{F_{\physvec{\omega}}}}}
=
\fun{\hat{\tau}_{\physvec{\omega},t + s}}{
\fun{\oarepn_{\physvec{\omega}}}{F_{\physvec{\omega}}}},$$
the last step by the implementation identity and the group property of $\hat{\tau}_{\physvec{\omega}}$.

For any $t
\in
\fldreal$,
choose $n
\in
\semigrposint$ such that $\frac{\abs{t}}{n}
<
\frac{1}{a}$.
The short-time result gives $\fun{S_{\physvec{\omega}}}{t/n}$,
and repeated application of the propagation implication with $s
=
t/n$ gives $\fun{S_{\physvec{\omega}}}{j t/n}$ for $j
=
1,\ldots,n$.
In particular,
$\fun{S_{\physvec{\omega}}}{t}$ holds for every $t
\in
\fldreal$.
\end{proof}

\begin{cor}[asymptotic constancy of the means]\label{cor:means-constant}
Under the hypotheses of Theorem \ref{thm:dynamics},
for every integer $k
\geq
0$,
$\gamma
\in
\setone{+, -, z}$,
and $t
\in
\fldreal$,
$$\slim_{\Omega \to \infty}
\fun{\oarepn_{\physvec{\omega}}}{\fun{\tau^{\Omega}_{t}}{M^{\gamma}_{k, \Omega}}}
=
\mu^{\gamma}_{\physvec{\omega},k}.$$
\end{cor}

\begin{proof}
Apply $\fun{S_{\physvec{\omega}}}{t}$ to the mean-field polynomial family
$F_{\Omega}
=
M^{\gamma}_{k,\Omega}$.
Its substituted element is the scalar
$F_{\physvec{\omega}}
=
\mu^{\gamma}_{\physvec{\omega},k}$,
fixed by $\hat{\tau}_{\physvec{\omega},t}$.
\end{proof}

Corollary \ref{cor:means-constant} contains \cite[Eqs. (42)--(45)]{ThirringWehrl001}: in representations satisfying the gap equation, all time derivatives of the weighted means vanish in the limit. Their vanishing makes the weighted means constants of the limiting motion. Theorem \ref{thm:dynamics} gives the precise form of \cite[Eqs. (46)--(47)]{ThirringWehrl001}.

\begin{rem}[physical status of the sector dynamics]\label{rem:physical-sector-dynamics}
The finite-volume BCS dynamics is gauge invariant and number conserving.
For an aligned phase, Theorem \ref{thm:dynamics} gives convergence to its self-consistent Bogoliubov dynamics
\cite{RudolfHaag005,ThirringWehrl001}.
Proposition \ref{prop:no-frozen-mean-field-dynamics} excludes only the frozen mean-field dynamics of
Definition \ref{def:frozen-mean-field-dynamics} for a non-aligned configuration.
\end{rem}

\subsection{Failure of the Bogoliubov dynamics without the gap equation}\label{failure-of-the-bogoliubov-dynamics-without-the-gap-equation}

The degenerate model permits an exact comparison between finite-volume precession and the proposed Bogoliubov dynamics. The comparison proves that a constant configuration has the required limiting dynamics exactly when it is aligned, or equivalently when it satisfies the gap equation in this setting. The proof starts from the commutation relation \(\commutator{\physham_{\Omega}}{S^{z}_{\Omega}}
=
0\) and the frequency formula \eqref{eq:frequencies}. On the joint eigenvectors of \(\physvec{S}^{2}_{\Omega}\) and \(S^{z}_{\Omega}\), these identities give \begin{equation}\label{eq:exact-precession}
\fun{\tau^{\Omega}_{t}}{M^{+}_{0, \Omega}}
=
M^{+}_{0, \Omega} \fnexp{\imunit t \rbk{- 2 \varepsilon + \frac{2}{\sminvtemperature_{c}} M^{z}_{0, \Omega}}},
\quad
\fun{\tau^{\Omega}_{t}}{M^{z}_{0, \Omega}}
=
M^{z}_{0, \Omega}.
\end{equation} The joint eigenvectors form a basis, which makes \eqref{eq:exact-precession} an operator identity. The term second frequency refers to the second spectral difference in \eqref{eq:frequencies}. It is the energy change at fixed \(S\) under the transition \(S_{z}
=
m \mapsto m + 1\) generated by \(S^{+}_{\Omega}\): \begin{equation}\label{eq:second-frequency-transition}
\fun{E_{\Omega}}{S, m + 1}
-
\fun{E_{\Omega}}{S, m}
=
- 2 \varepsilon
+
\frac{4 m}{\sminvtemperature_{c} \Omega}.
\end{equation} Since \(M^{z}_{0, \Omega}
=
\frac{2}{\Omega} S^{z}_{\Omega}\), the exponent in \eqref{eq:exact-precession} is precisely the spectral difference \eqref{eq:second-frequency-transition} on the joint eigenspace labeled by \(\rbk{S, m}\).

\begin{prop}[precession and the second frequency]\label{prop:precession}
Let the model be degenerate,
let $\physvec{u}
\in
\mansphere^{2}$,
and let $\physvec{\omega}$ be the constant configuration determined by
$\physvec{u}_{p}
=
\physvec{u}$ for every $p
\in
\semigrposint$.
Suppose that $m^{+}
=
\frac{1}{2} \rbk{u^{x} + \imunit u^{y}}
\neq
0$.
Set
$$2 \tilde{\mu}
=
2 \varepsilon - \frac{2}{\sminvtemperature_{c}} u^{z}
=
\frac{2}{\sminvtemperature_{c}}
\rbk{\sminvtemperature_{c} \varepsilon - u^{z}}.$$
The operator form of the spectral difference \eqref{eq:second-frequency-transition} has the strong limit
$$\slim_{\Omega \to \infty}
\fun{\oarepn_{\physvec{\omega}}}{
- 2 \varepsilon
+
\frac{2}{\sminvtemperature_{c}} M^{z}_{0, \Omega}}
=
- 2 \tilde{\mu}.$$
The transverse mean and the longitudinal mean have the strong limit
$$\slim_{\Omega \to \infty}
\fun{\oarepn_{\physvec{\omega}}}{\fun{\tau^{\Omega}_{t}}{M^{+}_{0, \Omega}}}
=
m^{+} \fnexp{- \imunit t 2 \tilde{\mu}},
\quad
\slim_{\Omega \to \infty}
\fun{\oarepn_{\physvec{\omega}}}{\fun{\tau^{\Omega}_{t}}{M^{z}_{0, \Omega}}}
=
u^{z}.$$
The limiting order parameter precesses about the $z$-axis with the frequency $2 \tilde{\mu}$,
which vanishes exactly when $u^{z}
=
\sminvtemperature_{c} \varepsilon$,
that is,
exactly when the configuration is aligned in the sense of
Definition \ref{def:aligned-configuration}.
\end{prop}

\begin{proof}
The transverse limit are given
due to the equation \eqref{eq:exact-precession},
Corollary \ref{cor:mean-spin} for the factors,
Lemma \ref{lem:strong-exponentials} for the exponential factor with generator $\frac{2}{\sminvtemperature_{c}} M^{z}_{0, \Omega}$,
and Lemma \ref{lem:strong-products} for the product.
The limits of the factors are $m^{+}$ and $u^{z}$.
Alignment of the constant configuration with $m^{+}
\neq
0$ means $\physvec{u} \parallel \physvec{b}
=
\rbk{\frac{1}{\sminvtemperature_{c}} u^{x}, \frac{1}{\sminvtemperature_{c}} u^{y}, \varepsilon}$,
which for the transverse components is automatic.
The $z$-component of the same alignment condition reads $u^{z}
=
\sminvtemperature_{c} \varepsilon$.
$\physvec{u} \parallel \physvec{b}$ with equal transverse parts up to the factor $\frac{1}{\sminvtemperature_{c} \abs{\physvec{b}}}$ forces $\abs{\physvec{b}}
=
\frac{1}{\sminvtemperature_{c}}$.
It follows that $u^{z}
=
\sminvtemperature_{c} \varepsilon$.
Conversely $u^{z}
=
\sminvtemperature_{c} \varepsilon$ gives $\physvec{b}
=
\frac{1}{\sminvtemperature_{c}} \physvec{u}$.
\end{proof}

The signed thermodynamic limit of the second characteristic spectral difference in \eqref{eq:frequencies} is \(-2 \tilde{\mu}\). The parameter \(2 \tilde{\mu}\) is the corresponding phase-precession frequency in Proposition \ref{prop:precession}. The chemical potential adjustment of \cite[Section 1]{ThirringWehrl001} replaces \(\physham_{\Omega}\) by \(\physham_{\Omega} + \tilde{\mu} \sum_{p
\leq
\Omega} \sigma^{z}_{p}\). Equivalently, the energy \(\varepsilon\) is measured from the effective chemical potential and is shifted according to \[\varepsilon
\mapsto
\varepsilon - \tilde{\mu}
=
\frac{1}{\sminvtemperature_{c}} u^{z}.\] After this adjustment, the configuration is aligned in the sense of Definition \ref{def:aligned-configuration}.

\begin{prop}[no frozen mean-field dynamics in non-aligned classes]\label{prop:no-frozen-mean-field-dynamics}
Let the model be degenerate,
let $\physvec{u}
\in
\mansphere^{2}$,
and let $\physvec{\omega}$ be the constant configuration determined by
$\physvec{u}_{p}
=
\physvec{u}$ for every $p
\in
\semigrposint$.
Suppose that $\rbk{u^{x}, u^{y}}
\neq
0$ and $u^{z}
\neq
\sminvtemperature_{c} \varepsilon$.
This configuration is not aligned in the sense of Definition
\ref{def:aligned-configuration}.
Let $\tau^{\physham[h]}$ be the frozen mean-field dynamics of Definition \ref{def:frozen-mean-field-dynamics} for this configuration.
Equation \eqref{eq:effective-hamiltonian} applied to the mean polarization of the constant configuration gives the time-independent field
$$\physvec{b}
=
\vecbk{\frac{1}{\sminvtemperature_{c}} u^{x},
\frac{1}{\sminvtemperature_{c}} u^{y},
\varepsilon}
=
\frac{1}{\sminvtemperature_{c}}
\vecbk{u^{x},
u^{y},
\sminvtemperature_{c} \varepsilon}.$$
The one-site generator of $\tau^{\physham[h]}$ is
$\physham[h]
=
- \physvec{b} \cdot \physvec{\sigma}$ and 
the frozen dynamics has the following longitudinal strong limit:
$$\begin{aligned}
\slim_{\Omega \to \infty}
\fun{\oarepn_{\physvec{\omega}}}{\fun{\tau^{\physham[h]}_{t}}{M^{z}_{0, \Omega}}}
=
u^{z} + A \rbk{1 - \fun{\cos}{2 \abs{\physvec{b}} t}},
\quad
A
=
\frac{\rbk{\rbk{u^{x}}^{2} + \rbk{u^{y}}^{2}} \rbk{\varepsilon - \frac{u^{z}}{\sminvtemperature_{c}}}}{\sminvtemperature_{c} \abs{\physvec{b}}^{2}}
\neq
0.
\end{aligned}$$
The finite-volume dynamics satisfies
$\fun{\oarepn_{\physvec{\omega}}}{\fun{\tau^{\Omega}_{t}}{M^{z}_{0, \Omega}}}
\to
u^{z}$ as $\Omega
\to
\infty$ for every $t$.
The frozen mean-field dynamics $\tau^{\physham[h]}$ does not reproduce the limiting Heisenberg dynamics of the non-aligned class.
\end{prop}

\begin{proof}
The one-site Heisenberg evolution under $\physham[h]
=
- \physvec{b} \cdot \physvec{\sigma}$ rotates the Bloch sphere.
Lemma \ref{lem:pauli-product} gives the Heisenberg differential equation
$$\frac{d}{d t} \fun{\tau^{\physham[h]}_{t}}{\physvec{\sigma} \cdot \physvec{a}}
=
\fun{\tau^{\physham[h]}_{t}}{\imunit \commutator{\physham[h]}{\physvec{\sigma} \cdot \physvec{a}}}
=
\fun{\tau^{\physham[h]}_{t}}{\physvec{\sigma} \cdot \rbk{2 \physvec{b} \times \physvec{a}}},$$
since the product identity applied in both orders gives $\commutator{\physvec{b} \cdot \physvec{\sigma}}{\physvec{a} \cdot \physvec{\sigma}}
=
2 \imunit \physvec{\sigma} \cdot \rbk{\physvec{b} \times \physvec{a}}$.
It follows that $\imunit \commutator{- \physvec{b} \cdot \physvec{\sigma}}{\physvec{a} \cdot \physvec{\sigma}}
=
2 \physvec{\sigma} \cdot \rbk{\physvec{b} \times \physvec{a}}$.
It follows that $\fun{\tau^{\physham[h]}_{t}}{\physvec{\sigma} \cdot \physvec{a}}
=
\physvec{\sigma} \cdot \fun{R^{\physvec{b}}_{t}}{\physvec{a}}$,
where $R^{\physvec{b}}_{t}
=
\fnexp{2 t \abs{\physvec{b}} K}$ with $K \physvec{a}
=
\hat{\physvec{b}} \times \physvec{a}$ is the rotation about $\hat{\physvec{b}}
=
\frac{\physvec{b}}{\abs{\physvec{b}}}$ by the angle $2 \abs{\physvec{b}} t$.
The evolved one-site observables $\fun{\tau^{\physham[h]}_{t}}{\sigma^{z}_{p}}
=
\physvec{\sigma}_{p} \cdot \fun{R^{\physvec{b}}_{t}}{\hat{\physvec{z}}}$ are one-site with means $\physvec{u} \cdot \fun{R^{\physvec{b}}_{t}}{\hat{\physvec{z}}}$.
Theorem \ref{thm:lln} gives
$$\slim_{\Omega \to \infty}
\fun{\oarepn_{\physvec{\omega}}}{
\fun{\tau^{\physham[h]}_{t}}{M^{z}_{0, \Omega}}}
=
\fun{u^{z}}{t}.$$
Lemma \ref{lem:rodrigues},
applied with $\physvec{q}
=
\hat{\physvec{b}}$,
$\physvec{a}
=
\hat{\physvec{z}}$,
and $\theta
=
2 \abs{\physvec{b}} t$,
gives
$$\begin{aligned}
&\fun{u^{z}}{t}
=
\physvec{u} \cdot \fun{R^{\physvec{b}}_{t}}{\hat{\physvec{z}}}
\\ 
&=
\rbk{\physvec{u} \cdot \hat{\physvec{b}}} \rbk{\hat{\physvec{b}} \cdot \hat{\physvec{z}}}
+\sqbk{u^{z} - \rbk{\physvec{u} \cdot \hat{\physvec{b}}} \rbk{\hat{\physvec{b}} \cdot \hat{\physvec{z}}}}
\fun{\cos}{2 \abs{\physvec{b}} t}
+\physvec{u} \cdot \rbk{\hat{\physvec{b}} \times \hat{\physvec{z}}} \fun{\sin}{2 \abs{\physvec{b}} t}.
\end{aligned}$$
The scalar triple-product term vanishes:
$$\physvec{u} \cdot \rbk{\hat{\physvec{b}} \times \hat{\physvec{z}}}
=
u^{x} \hat{b}^{y} - u^{y} \hat{b}^{x}
=
\frac{1}{\sminvtemperature_{c} \abs{\physvec{b}}} \rbk{u^{x} u^{y} - u^{y} u^{x}}
=
0.$$
The longitudinal part of \eqref{eq:rodrigues-rotation} is
$$\begin{aligned}
&\rbk{\physvec{u} \cdot \hat{\physvec{b}}} \hat{b}^{z} - u^{z}
=
\frac{\rbk{\frac{1}{\sminvtemperature_{c}} \rbk{\rbk{u^{x}}^{2} + \rbk{u^{y}}^{2}} + \varepsilon u^{z}} \varepsilon - u^{z} \rbk{\frac{1}{\sminvtemperature_{c}^{2}} \rbk{\rbk{u^{x}}^{2} + \rbk{u^{y}}^{2}} + \varepsilon^{2}}}
{\abs{\physvec{b}}^{2}}
\\ 
&=
\frac{\rbk{\rbk{u^{x}}^{2} + \rbk{u^{y}}^{2}} \rbk{\varepsilon - \frac{u^{z}}{\sminvtemperature_{c}}}}
{\sminvtemperature_{c} \abs{\physvec{b}}^{2}}
=
A.
\end{aligned}$$
It follows that $\fun{u^{z}}{t}
=
u^{z} + A \rbk{1 - \fun{\cos}{2 \abs{\physvec{b}} t}}$ with $A
\neq
0$ under the stated hypotheses.
The exact conservation in \eqref{eq:exact-precession} gives the constant limit for the true dynamics,
and the two limits differ for all $t
\notin
\frac{\pi}{\abs{\physvec{b}}} \ringratint$.
\end{proof}

\subsection{Weak non-convergence of the evolution operators}\label{weak-non-convergence-of-the-evolution-operators}

Even for configurations aligned in the sense of Definition \ref{def:aligned-configuration}, where Theorem \ref{thm:dynamics} gives the correct Heisenberg dynamics, the evolution operators \(\fnexp{\imunit t \fun{\oarepn_{\physvec{\omega}}}{
\physham_{\Omega} - E_{\physvec{\omega},\Omega}}}\) do not converge to \(\fnexp{\imunit t \physham_{\txtbogoliubov,\physvec{\omega}}}\), where \(\oarepn_{\physvec{\omega}}\) is the product representation of Definition \ref{def:spin-product-representation}. This is the explicit computation of \cite[p. 311]{ThirringWehrl001}, carried out here with full error control. The combinatorial input is a local limit theorem for the symmetric binomial distribution with a quadratic phase.

\begin{lem}[binomial sums with quadratic phase]\label{lem:binomial-clt}
For $2 s
\in
\semigrposint$ let $\fun{P_{s}}{m}
=
2^{- 2 s} \binom{2 s}{s + m}$,
where $m$ runs over $\setone{- s, - s + 1, \dotsc, s}$,
so that $s + m
\in
\monnat$.
Half-integer $s$ is allowed.
Then for all $\alpha, \beta
\in
\fldreal$,
$$\lim_{s
\to
\infty} \sum_{m
=
- s}^{s} \fun{P_{s}}{m} \fnexp{\imunit \alpha \frac{m^{2}}{s} + \imunit \beta \frac{m}{s}}
=
\frac{1}{\sqrt{1 - \imunit \alpha}},$$
with the principal branch of the square root.
\end{lem}

\begin{proof}
We first control the tails.
The entropy function
$$\fun{\hat{h}}{x}
=
- x \fun{\log}{x} - \rbk{1 - x} \fun{\log}{1 - x}$$ satisfies $\fun{\hat{h}}{\frac{1}{2} + u}
\leq
\log 2 - 2 u^{2}$,
since $\hat{h}''
=
- \frac{1}{x \rbk{1 - x}}
\leq
- 4$ on $\rbk{0, 1}$ and $\hat{h}' \rbk{\frac{1}{2}}
=
0$.
By Corollary \ref{cor:entropy-bounds} with $n
=
2 s$,
$a
=
s + m$,
$$\fun{P_{s}}{m}
\leq
\fnexp{2 s \sqbk{\fun{\hat{h}}{\frac{1}{2} + \frac{m}{2 s}} - \log 2} + \frac{3}{2} \fun{\log}{2 s + 1} + 3}
\leq
C \rbk{2 s + 1}^{3 / 2} \fnexp{- \frac{m^{2}}{s}}.$$
The tail probability satisfies
$$\sum_{\substack{m = - s, - s + 1, \dotsc, s \\ \abs{m} > s^{3 / 5}}} \fun{P_{s}}{m}
\leq
C \rbk{2 s + 1}^{5 / 2} \fnexp{- s^{1 / 5}}
\to
0
\quad \rbk{s \to \infty}.$$
Every summand in the oscillatory sum has modulus at most
$\fun{P_{s}}{m}$.
The tail contribution therefore vanishes with the tail probability.

We next derive the bulk asymptotics.
For $\abs{m}
\leq
s^{3 / 5}$,
Proposition \ref{prop:stirling} applied to the three factorials gives
$$\begin{aligned}
\fun{P_{s}}{m}
&=
\frac{1}{\sqrt{\pi s}} \rbk{1 - \frac{m^{2}}{s^{2}}}^{- 1 / 2} \fnexp{- \fun{g_{s}}{m} + \fun{O}{\frac{1}{s}}},
\\ 
\fun{g_{s}}{m}
&=
\rbk{s + m} \fun{\log}{1 + \frac{m}{s}} + \rbk{s - m} \fun{\log}{1 - \frac{m}{s}},
\end{aligned}$$
by the same computation as in Corollary \ref{cor:entropy-bounds} with the square-root factors kept exactly.
The Taylor expansion $\fun{\log}{1 \pm x}
=
\pm x - \frac{x^{2}}{2} \pm \frac{x^{3}}{3} - \dotsb$ gives,
for $\abs{m}
\leq
s^{3 / 5}$,
$$\fun{g_{s}}{m}
=
\frac{m^{2}}{s} + \fun{O}{\frac{m^{4}}{s^{3}}}
=
\frac{m^{2}}{s} + \fun{O}{s^{- 3 / 5}},$$
uniformly in the bulk.
For the square-root factor,
set $y
=
\frac{m^{2}}{s^{2}}$.
The bulk restriction gives
$$\begin{aligned}
0
\leq
y
\leq
s^{- 4 / 5},
\quad
0
\leq
y^{2}
\leq
s^{- 8 / 5}.
\end{aligned}$$
For sufficiently large $s$ one has $y
\leq
\frac{1}{2}$,
and Taylor's formula on $\sqbk{0, \frac{1}{2}}$ gives the uniform expansion
$$\begin{aligned}
\rbk{1 - \frac{m^{2}}{s^{2}}}^{- 1 / 2}
=
1
+\frac{1}{2} \frac{m^{2}}{s^{2}}
+\fun{O}{\frac{m^{4}}{s^{4}}}
=
1
+\fun{O}{s^{- 4 / 5}}.
\end{aligned}$$
The three uniform error factors give
$$\fun{P_{s}}{m}
=
\frac{1}{\sqrt{\pi s}} \fnexp{- \frac{m^{2}}{s}} \rbk{1 + \fun{r_{s}}{m}},
\quad
\sup_{\substack{m = - s, - s + 1, \dotsc, s \\ \abs{m}
\leq
s^{3 / 5}}} \abs{\fun{r_{s}}{m}}
=
\fun{O}{s^{- 3 / 5}}
\to
0
\quad \rbk{s \to \infty}.$$

The bulk sum is then identified as a Riemann sum.
On the bulk,
$\fnexp{\imunit \beta m / s}
=
1 + \fun{O}{s^{- 2 / 5}}$ uniformly,
and with $x_{m}
=
\frac{m}{\sqrt{s}}$,
spacing $\Delta x
=
\frac{1}{\sqrt{s}}$,
$$\sum_{\substack{m = - s, - s + 1, \dotsc, s \\ \abs{m}
\leq
s^{3 / 5}}} \fun{P_{s}}{m} \fnexp{\imunit \alpha \frac{m^{2}}{s} + \imunit \beta \frac{m}{s}}
=
\frac{1}{\sqrt{\pi}} \sum_{\substack{m = - s, - s + 1, \dotsc, s \\ \abs{x_{m}}
\leq
s^{1 / 10}}} \fnexp{- \rbk{1 - \imunit \alpha} x_{m}^{2}} \Delta x + \fun{o}{1}.$$
The error estimate follows because the error terms are summable against $\frac{1}{\sqrt{\pi s}} \fnexp{- m^{2} / s}$.
This Gaussian weight has bounded total mass.
The function $x
\mapsto
\fnexp{- \rbk{1 - \imunit \alpha} x^{2}}$ is continuous,
absolutely integrable with Gaussian decay,
and has derivative bounded on $\fldreal$ growing at most linearly times the Gaussian.
It follows that its Riemann sums with spacing $\frac{1}{\sqrt{s}}$ over $\abs{x}
\leq
s^{1 / 10}$ converge to $\int_{- \infty}^{\infty} \fnexp{- \rbk{1 - \imunit \alpha} x^{2}} \opdmsr{x}$.

It remains to evaluate the Gaussian integral.
For $\alpha
=
0$ the integral is $\sqrt{\pi}$.
Both sides of
$$\frac{1}{\sqrt{\pi}} \int_{- \infty}^{\infty} \fnexp{- \rbk{1 - \imunit \alpha} x^{2}} \opdmsr{x}
=
\frac{1}{\sqrt{1 - \imunit \alpha}}$$
are holomorphic functions of $\zeta
=
1 - \imunit \alpha$ on the half-plane $\opreal \zeta > 0$,
the left side by Morera and dominated convergence,
the right side with the principal branch.
They agree for $\zeta
\in
\rbk{0, \infty}$ by the real Gaussian integral and rescaling.
The identity theorem extends the equality to the whole domain.
\end{proof}

\begin{prop}[the evolution operators do not converge weakly]\label{prop:unitary-failure}
Let the model be degenerate with $\varepsilon
=
0$.
Let $\physvec{\omega}$ be the constant configuration determined by
$\physvec{u}_{p}
=
\hat{\physvec{x}}$ for every $p
\in
\semigrposint$.
This configuration is aligned in the sense of
Definition \ref{def:aligned-configuration} and satisfies the gap equation.
Let $\oarepn_{\physvec{\omega}}$ be its product representation from Definition \ref{def:spin-product-representation}.
Theorem \ref{thm:bogoliubov}(3) gives
$$\physham_{\txtbogoliubov,\physvec{\omega}}
=
\frac{1}{\sminvtemperature_{c}} \sum_{p \in \semigrposint} \fun{\oarepn_{\physvec{\omega}}}{1 - \sigma^{x}_{p}},
\quad
\physham_{\txtbogoliubov,\physvec{\omega}} \xi_{\physvec{\omega}}^{\emptyset}
=
0.$$
The vacuum matrix coefficient has the nontrivial limit:
i.e.,
for $t \neq 0$,
we obtain
$$\lim_{\Omega \to \infty}
\bkt{\xi_{\physvec{\omega}}^{\emptyset}}{
\fnexp{\imunit t \fun{\oarepn_{\physvec{\omega}}}{
\physham_{\Omega} - E_{\physvec{\omega},\Omega}}}
\xi_{\physvec{\omega}}^{\emptyset}}
=
\frac{\fnexp{- \imunit \frac{t}{2 \sminvtemperature_{c}}}}
{\sqrt{1 - \imunit \frac{t}{\sminvtemperature_{c}}}}
\neq
1
=
\bkt{\xi_{\physvec{\omega}}^{\emptyset}}{\fnexp{\imunit t \physham_{\txtbogoliubov,\physvec{\omega}}} \xi_{\physvec{\omega}}^{\emptyset}}.$$
In particular,
the following weak-limit identity fails:
$$\wlim_{\Omega \to \infty}
\fnexp{\imunit t \fun{\oarepn_{\physvec{\omega}}}{
\physham_{\Omega} - E_{\physvec{\omega},\Omega}}}
=
\fnexp{\imunit t \physham_{\txtbogoliubov,\physvec{\omega}}}.$$
The represented generators satisfy only
$\fun{\oarepn_{\physvec{\omega}}}{
\physham_{\Omega} - E_{\physvec{\omega},\Omega}}
\to
\physham_{\txtbogoliubov,\physvec{\omega}}$ in the generalized weak sense as $\Omega \to \infty$ and the Heisenberg dynamics converges strongly.
\end{prop}

\begin{proof}
Write $s
=
\frac{\Omega}{2}$,
assuming $\Omega$ even,
and the odd case is identical.
The vector $\xi_{\physvec{\omega}}^{\emptyset}
=
\bigotimes_{p
\leq
\Omega} \xi_{\hat{\physvec{x}}}^{+}$,
restricted to the first $\Omega$ slots through the factorization,
is the coherent vector with maximal spin along $\hat{\physvec{x}}$.
Expanding each slot $\xi_{\hat{\physvec{x}}}^{+}
=
\frac{1}{\sqrt{2}} \rbk{\ket{\uparrow} + \ket{\downarrow}}$ in the $\sigma^{z}$-basis,
$$\bigotimes_{p
\leq
\Omega} \xi_{\hat{\physvec{x}}}^{+}
=
2^{- s} \sum_{I
\subset
\intint{1..\Omega}} e_{I}
=
\sum_{m
=
- s}^{s} 2^{- s} \sqrt{\binom{2 s}{s + m}} \ket{s, m},$$
where $e_{I}$ has up-spins exactly on $I$
and the vector $\ket{s, m}$ is the normalized sum of the $e_{I}$ with
$\abscard{I}
=
s + m$.
It is the total-spin vector in the maximal multiplet with $\physvec{S}^{2}
=
\fun{s}{s + 1}$,
$S^{z}
=
m$.
With $\varepsilon
=
0$ the expression \eqref{eq:eigenvalues} gives $$\physham_{\Omega} \ket{s, m}
=
\fun{E_{\Omega}}{s, m} \ket{s, m},
\quad
\fun{E_{\Omega}}{s, m}
=
- \frac{2}{\sminvtemperature_{c} \Omega} \rbk{\fun{s}{s + 1} - \fun{m}{m - 1}},$$
and \eqref{eq:e-omega} evaluates to $E_{\physvec{\omega},\Omega}
=
- \frac{\Omega}{2 \sminvtemperature_{c}} - \frac{1}{2 \sminvtemperature_{c}}$,
from $\fun{\oastate[\psi]}{d^{+} d^{-}}
=
\frac{1}{4}$ and $\bar{M}^{\pm}
=
\frac{1}{2}$.
The shifted energy in the maximal-spin sector is
$$\fun{E_{\Omega}}{s, m} - E_{\physvec{\omega},\Omega}
=
\frac{2}{\sminvtemperature_{c} \Omega} \rbk{m^{2} - m} - \frac{1}{2 \sminvtemperature_{c}}
=
\frac{1}{\sminvtemperature_{c} s} \rbk{m^{2} - m} - \frac{1}{2 \sminvtemperature_{c}},$$
where
$$\frac{2}{\sminvtemperature_{c} \Omega} \fun{s}{s + 1}
=
\frac{\Omega}{2 \sminvtemperature_{c}} + \frac{1}{\sminvtemperature_{c}}.$$
The vacuum matrix coefficient is therefore
$$\bkt{\xi_{\physvec{\omega}}^{\emptyset}}{
\fnexp{\imunit t \fun{\oarepn_{\physvec{\omega}}}{
\physham_{\Omega} - E_{\physvec{\omega},\Omega}}}
\xi_{\physvec{\omega}}^{\emptyset}}
=
\fnexp{- \imunit \frac{t}{2 \sminvtemperature_{c}}} \sum_{m
=
- s}^{s} \fun{P_{s}}{m} \fnexp{\imunit \frac{t}{\sminvtemperature_{c} s} \rbk{m^{2} - m}}.$$
Lemma \ref{lem:binomial-clt} with $\alpha
=
\frac{t}{\sminvtemperature_{c}}$,
$\beta
=
- \frac{t}{\sminvtemperature_{c}}$ gives the limit $\fnexp{- \imunit \frac{t}{2 \sminvtemperature_{c}}} \rbk{1 - \imunit \frac{t}{\sminvtemperature_{c}}}^{- 1 / 2}$,
whose modulus is $\rbk{1 + \frac{1}{\sminvtemperature_{c}^{2}} t^{2}}^{- 1 / 4} < 1$ for $t
\neq
0$.
\end{proof}

The moral drawn in \cite[Sections 3--4]{ThirringWehrl001} deserves repetition: weak convergence of unitaries to a unitary limit would upgrade itself to strong convergence and, with uniform boundedness of all time derivatives on a core, to convergence of the generators. Proposition \ref{prop:unitary-failure} shows that no such convergence takes place. It follows that the correct statement of the Bogoliubov--Haag method at zero temperature is the convergence of the Heisenberg dynamics of the observables, Theorem \ref{thm:dynamics}, and not a convergence of Hamiltonians or evolution groups.

\section{Gibbs States of the Degenerate Model and Concentration}\label{sec:thermal}

The thermal analysis begins with the following setting, which remains in force through Section \ref{sec:decomposition}.

\begin{defn}[degenerate thermal setting]\label{def:degenerate-thermal-setting}
Fix $\varepsilon
\in
\fldreal$ and an inverse temperature $\sminvtemperature
>
0$.
The degenerate thermal setting is the BCS model of
\eqref{eq:bcs-hamiltonian} with
$$\varepsilon_{p}
=
\varepsilon
\quad
\rbk{p
\in
\semigrposint}.$$
\end{defn}

All results from the present section through Section \ref{sec:decomposition} use the degenerate thermal setting of Definition \ref{def:degenerate-thermal-setting}. Individual statements impose additional temperature conditions when required. The Hamiltonian \eqref{eq:bcs-hamiltonian} is the function \eqref{eq:hamiltonian-total-spin} of the total spin. Its spectral decomposition produces an exact probability distribution on the total-spin quantum numbers. The free-energy analysis then proves concentration at a unique maximizer, which supplies the order parameter for the limiting Gibbs state.

\subsection{Finite-volume Gibbs states and spectral weights}\label{finite-volume-gibbs-states-and-spectral-weights}

The total-spin decomposition rewrites the finite-volume Gibbs trace as an explicit probability distribution on the spin quantum numbers.

\begin{defn}[finite-volume Gibbs states]
The Gibbs state of $\Omega$ modes and its extension to $\oa{A}$ are
\begin{equation}\label{eq:finite-volume-gibbs-state}
\begin{aligned}
\fun{\oastate[\psi_{\sminvtemperature,\Omega}]}{A}
&=
\frac{\sqfun{\trace}{\fnexp{- \sminvtemperature \physham_{\Omega}} A}}
{\sqfun{\trace}{\fnexp{- \sminvtemperature \physham_{\Omega}}}},
&
\oastate[\widehat{\psi}_{\sminvtemperature,\Omega}]
&=
\oastate[\psi_{\sminvtemperature,\Omega}] \otimes
\oastate[\psi^{\mathrm{tr}}_{> \Omega}],
\end{aligned}
\end{equation}
where $A
\in
\oa{A}_{\Omega}$,
and $\oastate[\psi^{\mathrm{tr}}_{> \Omega}]$ is the product of the normalized trace states on the modes beyond $\Omega$.
\end{defn}

The extension is a convenience: all limits in this section are taken on local elements, which lie in \(\oa{A}_{\Omega}\) eventually. It follows that the choice of tail state is immaterial. Gauge invariance \eqref{eq:gauge-invariance} and permutation invariance give \begin{equation}\label{eq:gibbs-invariance}
\oastate[\psi_{\sminvtemperature,\Omega}] \circ \mathfrak{g}_{\vartheta}
=
\oastate[\psi_{\sminvtemperature,\Omega}],
\quad
\oastate[\psi_{\sminvtemperature,\Omega}] \circ \mathfrak{g}_{\pi}
=
\oastate[\psi_{\sminvtemperature,\Omega}]
\end{equation} for all gauge angles \(\vartheta\) and all permutations \(\pi\) of \(\intint{1..\Omega}\).

\begin{defn}[admissible spin quantum numbers]\label{def:admissible-spin-quantum-numbers}
For $\Omega
\in
\semigrposint$,
a total-spin value $S$ is admissible at volume $\Omega$ if
$S
=
\frac{\Omega}{2} - j$
for some $j
\in
\intint{0..\lfloor \Omega / 2 \rfloor}$.
For such $S$,
a weight $m$ is admissible if
$m
\in
\setone{- S, - S + 1, \dotsc, S}$.
The pair $\rbk{S, m}$ is admissible at volume $\Omega$ when both conditions hold.
\end{defn}

For each total-spin value \(S\) admissible at volume \(\Omega\) in the sense of Definition \ref{def:admissible-spin-quantum-numbers}, define the completely symmetric tensor subspace by \begin{equation}\label{eq:spin-s-space}
V_{S}
=
\setone{\Psi
\in
\rbk{\fldcmp^{2}}^{\otimes 2 S} : U_{\pi} \Psi
=
\Psi \text{ for every permutation } \pi \text{ of } \intint{1..2 S}},
\end{equation} where \(U_{\pi}\) is the unitary that permutes the \(2 S\) tensor factors according to \(\pi\). Lemma \ref{lem:irreps} proves that \(V_{S}\) is the irreducible spin-\(S\) space. For each weight \(m\) admissible for \(S\) in the sense of Definition \ref{def:admissible-spin-quantum-numbers}, let \(\ket{S, m}\) be a normalized vector in the one-dimensional \(S^{z}\)-eigenspace of \(V_{S}\) with eigenvalue \(m\). By Proposition \ref{prop:multiplicity} the trace splits over the joint eigenspaces of \(\rbk{\physvec{S}^{2}_{\Omega}, S^{z}_{\Omega}}\): for any function \(X\) of the total spin operators, \begin{equation}\label{eq:spectral-trace}
\sqfun{\trace}{\fnexp{- \sminvtemperature \physham_{\Omega}} X}
=
\sum_{\substack{j = 0, 1, \dotsc, \lfloor \Omega / 2 \rfloor \\ S = \Omega / 2 - j \\ m = - S, - S + 1, \dotsc, S}} m^{\rbk{\Omega}}_{S} \fnexp{- \sminvtemperature \fun{E_{\Omega}}{S, m}} \bkt{\ket{S, m}}{X \ket{S, m}},
\end{equation} where \(\fun{E_{\Omega}}{S, m}\) is defined in \eqref{eq:eigenvalues}. For the spin-polynomial operators used in this section, Proposition \ref{prop:matrix-element} shows that the diagonal matrix element in \eqref{eq:spectral-trace} has the same value in every spin-\(S\) summand of \(\rbk{\fldcmp^{2}}^{\otimes \Omega}\). The associated two-parameter probability distribution is \begin{equation}\label{eq:p-omega}
\fun{P_{\Omega}}{S, m}
=
\frac{m^{\rbk{\Omega}}_{S} \fnexp{- \sminvtemperature \fun{E_{\Omega}}{S, m}}}{\sum_{\substack{j' = 0, 1, \dotsc, \lfloor \Omega / 2 \rfloor \\ S' = \Omega / 2 - j' \\ m' = - S', - S' + 1, \dotsc, S'}} m^{\rbk{\Omega}}_{S'} \fnexp{- \sminvtemperature \fun{E_{\Omega}}{S', m'}}},
\end{equation}

\subsection{The free-energy function and its maximizer}\label{the-free-energy-function-and-its-maximizer}

The spectral weights determine a limiting free energy with a unique maximizer. This maximizer supplies the order parameter used in the limiting Gibbs state.

\begin{defn}[free-energy function and domain]\label{def:free-energy-function}
The free-energy function $h$ has domain
\begin{equation}\label{eq:free-energy-domain}
\dom h
=
\set{\rbk{\eta, z}}{\abs{z}
\leq
\eta
\leq
1}
\end{equation}
and is defined by
\begin{equation}\label{eq:free-energy}
\begin{aligned}
\fun{h}{\eta, z}
&=
\log 2
+ \fun{\tilde{s}}{\eta}
+ \frac{\sminvtemperature}{2 \sminvtemperature_{c}} \eta^{2}
+ \sminvtemperature \varepsilon z
- \frac{\sminvtemperature}{2 \sminvtemperature_{c}} z^{2},
\\ 
\fun{\tilde{s}}{\eta}
&=
- \frac{1 - \eta}{2} \fun{\log}{1 - \eta} - \frac{1 + \eta}{2} \fun{\log}{1 + \eta},
\end{aligned}
\end{equation}
where $\fun{\tilde{s}}{1}
=
- \log 2$ is defined by continuity.
\end{defn}

The probability distribution \eqref{eq:p-omega} is supported on the lattice points \(\rbk{\eta, z}
=
\rbk{\frac{2 S}{\Omega}, \frac{2 m}{\Omega}}\) in the domain \eqref{eq:free-energy-domain}.

\begin{lem}[uniform free energy asymptotics]\label{lem:free-energy}
There is a constant $C$,
depending only on $\varepsilon, \frac{1}{\sminvtemperature_{c}}, \sminvtemperature$,
such that the following estimate holds uniformly in $\Omega$ and in pairs
$\rbk{S, m}$ admissible at volume $\Omega$ in the sense of
Definition \ref{def:admissible-spin-quantum-numbers}:
$$\abs{\fun{\log}{m^{\rbk{\Omega}}_{S} \fnexp{- \sminvtemperature \fun{E_{\Omega}}{S, m}}} - \Omega \fun{h}{\frac{2 S}{\Omega}, \frac{2 m}{\Omega}}}
\leq
C \rbk{1 + \fun{\log}{\Omega}}.$$
\end{lem}

\begin{proof}
The multiplicity \eqref{eq:multiplicity} is $m^{\rbk{\Omega}}_{S}
=
\binom{\Omega}{\frac{\Omega}{2} + S} \cdot \frac{2 S + 1}{\frac{\Omega}{2} + S + 1}$.
The second factor lies in $\sqbk{\frac{1}{\Omega + 1}, 2}$ and contributes at most $\fun{\log}{\Omega + 1}$ to the error.

For $S$ admissible at volume $\Omega$ in the sense of
Definition \ref{def:admissible-spin-quantum-numbers},
the integer $a
=
\frac{\Omega}{2} + S$ satisfies $0
\leq
a
\leq
\Omega$.
The pair $\rbk{\Omega, a}$ satisfies the hypotheses of Corollary
\ref{cor:entropy-bounds}.
Since $\eta
=
\frac{2 S}{\Omega}$,
the normalized binomial index and the entropy term in that corollary are
$$\begin{aligned}
\frac{a}{\Omega}
&=
\frac{1 + \eta}{2},
\\ 
\fun{s_{\Omega}}{\frac{\Omega}{2} + S}
&=
- \frac{1 + \eta}{2} \fun{\log}{\frac{1 + \eta}{2}}
- \frac{1 - \eta}{2} \fun{\log}{\frac{1 - \eta}{2}}
=
\log 2 + \fun{\tilde{s}}{\eta}.
\end{aligned}$$
The uniform entropy estimate of Corollary \ref{cor:entropy-bounds} now gives
$$\abs{
\log
\binom{\Omega}{\frac{\Omega}{2} + S}
- \Omega \rbk{\log 2 + \fun{\tilde{s}}{\eta}}
}
\leq
\frac{3}{2} \fun{\log}{\Omega + 1} + 3.$$

Equation \eqref{eq:eigenvalues} expresses the logarithm of the Boltzmann factor as
$$\begin{aligned}
-\sminvtemperature \fun{E_{\Omega}}{S, m}
&=
2 \sminvtemperature \varepsilon m + \frac{2 \sminvtemperature}{\Omega \sminvtemperature_{c}} \rbk{S^{2} - m^{2}} + \frac{2 \sminvtemperature}{\Omega \sminvtemperature_{c}} \rbk{S + m}
\\ 
&=
\Omega \sqbk{\sminvtemperature \varepsilon z + \frac{\sminvtemperature}{2 \sminvtemperature_{c}} \rbk{\eta^{2} - z^{2}}} + \frac{2 \sminvtemperature}{\Omega \sminvtemperature_{c}} \rbk{S + m},
\end{aligned}$$
and the last term is bounded by $\frac{2 \sminvtemperature}{\sminvtemperature_{c}}$ in modulus.
The logarithmic error is the sum of the multiplicity-ratio error,
the binomial entropy error,
and the final term in the Boltzmann exponent.
Their bounds are respectively
$\fun{\log}{\Omega + 1}$,
$\frac{3}{2} \fun{\log}{\Omega + 1} + 3$,
and $\frac{2 \sminvtemperature}{\sminvtemperature_{c}}$.
Their sum is bounded by
$C \rbk{1 + \fun{\log}{\Omega}}$
for a constant $C$ depending only on
$\varepsilon$,
$\frac{1}{\sminvtemperature_{c}}$,
and $\sminvtemperature$.
\end{proof}

\begin{prop}[unique maximizer]\label{prop:maximizer}
Fix $\sminvtemperature > 0$.
\begin{enumerate}
\item If $\sminvtemperature > \sminvtemperature_{c}$,
the equation $\eta
=
\fun{\tanh}{\frac{\sminvtemperature}{\sminvtemperature_{c}} \eta}$ has a unique solution $\eta_{0}
=
\fun{\eta_{0}}{\sminvtemperature}
\in
\rbk{0, 1}$ besides $\eta
=
0$.
If moreover $\sminvtemperature_{c} \abs{\varepsilon} < \eta_{0}$,
the function $h$ attains its maximum over the domain
$\dom h$ of \eqref{eq:free-energy-domain} at the unique point
$\rbk{\eta_{0}, z_{0}}$,
$z_{0}
=
\sminvtemperature_{c} \varepsilon$,
which is interior in the sense $\abs{z_{0}} < \eta_{0}$,
and for every $\delta > 0$
$$\sup \set{\fun{h}{\eta, z}}{\rbk{\eta, z}
\in
\dom h, \ \abs{\eta - \eta_{0}} + \abs{z - z_{0}}
\geq
\delta}
\leq
\fun{h}{\eta_{0}, z_{0}} - \fun{\kappa}{\delta}$$
with $\fun{\kappa}{\delta} > 0$.
\item If $\sminvtemperature
\leq
\sminvtemperature_{c}$,
the maximum of $h$ over $\dom h$ is attained at the unique point $\rbk{\eta_{\ast}, z_{\ast}}$ with $\eta_{\ast}
=
\fun{\tanh}{\sminvtemperature \abs{\varepsilon}}$ and $z_{\ast}
=
\fun{\mathrm{sgn}}{\varepsilon} \eta_{\ast}$,
lying on the boundary $\abs{z}
=
\eta$,
with the same quantitative gap statement.
For $\varepsilon
=
0$ the maximizer is $\rbk{0, 0}$.
\end{enumerate}
\end{prop}

\begin{proof}
Completion of the square in $z$ gives
\begin{equation}\label{eq:h-decomposition}
\fun{h}{\eta, z}
=
\fun{G}{\eta} - \frac{\sminvtemperature}{2 \sminvtemperature_{c}} \rbk{z - z_{0}}^{2},
\quad
\fun{G}{\eta}
=
\log 2 + \fun{\tilde{s}}{\eta} + \frac{\sminvtemperature}{2 \sminvtemperature_{c}} \eta^{2} + \frac{\sminvtemperature \sminvtemperature_{c} \varepsilon^{2}}{2},
\end{equation}
with $z_{0}
=
\sminvtemperature_{c} \varepsilon$,
and the derivative of $G$ is
$$\fun{G'}{\eta}
=
\frac{\sminvtemperature}{\sminvtemperature_{c}} \eta - \fun{\mathrm{arctanh}}{\eta},$$
since $\tilde{s}'
=
- \frac{1}{2} \fun{\log}{\frac{1 + \eta}{1 - \eta}}$.
The function $\eta
\mapsto
\frac{\fun{\mathrm{arctanh}}{\eta}}{\eta}$ is continuous on $\rbk{0, 1}$,
strictly increasing,
with limit $1$ at $0$ and $\infty$ at $1$.
Monotonicity follows from the power series $\frac{\fun{\mathrm{arctanh}}{\eta}}{\eta}
=
\sum_{j
\geq
0} \frac{\eta^{2 j}}{2 j + 1}$ with positive coefficients.

(1)
For $\sminvtemperature > \sminvtemperature_{c}$ the ratio crosses the level $\frac{\sminvtemperature}{\sminvtemperature_{c}} > 1$ exactly once,
at some $\eta_{0}
\in
\rbk{0, 1}$.
It follows that $G' > 0$ on $\rbk{0, \eta_{0}}$ and $G' < 0$ on $\rbk{\eta_{0}, 1}$,
and $G$ has the unique maximizer $\eta_{0}$ on $\closedinterval{0}{1}$.
The fixed-point equation $\fun{\mathrm{arctanh}}{\eta_{0}}
=
\frac{\sminvtemperature}{\sminvtemperature_{c}} \eta_{0}$ is equivalent to $\eta_{0}
=
\fun{\tanh}{\frac{\sminvtemperature}{\sminvtemperature_{c}} \eta_{0}}$,
and $\eta
=
0$ is the only other solution by the same monotonicity.
Since $\abs{z_{0}} < \eta_{0}$,
the point $\rbk{\eta_{0}, z_{0}}$ lies in $\dom h$ by
\eqref{eq:free-energy-domain}.
Equation \eqref{eq:h-decomposition} gives
$$\fun{h}{\eta, z}
\leq
\fun{G}{\eta}
\leq
\fun{G}{\eta_{0}}
=
\fun{h}{\eta_{0}, z_{0}}.$$
Equality holds only if $z
=
z_{0}$ and if $\eta
=
\eta_{0}$.
For the quantitative gap,
$h$ is upper semicontinuous and continuous on the compact domain $\dom h$,
with $\fun{\tilde{s}}{1}$-continuity at the edge $\eta
=
1$.
The supremum over the compact set $\dom h \setminus B_{\delta}$ is attained and is strictly less than the unique maximum.

(2)
Assume that $\sminvtemperature
\leq
\sminvtemperature_{c}$.
The ratio $\frac{\fun{\mathrm{arctanh}}{\eta}}{\eta}$ exceeds $1
\geq
\frac{\sminvtemperature}{\sminvtemperature_{c}}$ for all $\eta
\in
\rbk{0, 1}$.
It follows that $G'
\leq
0$ throughout with equality only at $\eta
=
0$,
and $G$ is strictly decreasing on $\closedinterval{0}{1}$.
If $\varepsilon
=
0$,
then
$\fun{h}{\eta, z}
\leq
\fun{G}{\eta}
\leq
\fun{G}{0}
=
\fun{h}{0, 0}$ with equality only at $\rbk{0, 0}$.
It remains to consider $\varepsilon
\neq
0$.
By symmetry,
assume that $\varepsilon > 0$.
On the region $\eta
\geq
\min \rbk{z_{0}, 1}$,
the bound is
$$\fun{h}{\eta, z}
\leq
\fun{G}{\eta}
\leq
\fun{G}{\min \rbk{z_{0}, 1}}.$$
If $\sminvtemperature_{c} \varepsilon
\geq
1$,
this region meets $\dom h$ at most in the segment $\eta
=
1$.
The ensuing case distinction is therefore vacuous on that segment.
On the region $\eta < \min \rbk{z_{0}, 1}$ the $z$-maximum at fixed $\eta$ is on the boundary $z
=
\eta$,
since the quadratic in $z$ increases up to $z_{0} > \eta$.
It follows that there $\fun{h}{\eta, z}
\leq
\fun{\varphi}{\eta}
=
\log 2 + \fun{\tilde{s}}{\eta} + \sminvtemperature \varepsilon \eta$,
and the same bound $\fun{h}{\eta, z}
\leq
\fun{\varphi}{\eta}$ extends to all of $\dom h$ when $\sminvtemperature_{c} \varepsilon
\geq
1$.
Indeed,
the constrained $z$-maximum is at $z
=
\eta$ for every $\eta
\leq
1
\leq
z_{0}$.
The derivative $\varphi'
=
\sminvtemperature \varepsilon - \fun{\mathrm{arctanh}}{\eta}$ vanishes exactly at $\eta_{\ast}
=
\fun{\tanh}{\sminvtemperature \varepsilon}$,
which satisfies $\eta_{\ast} < \min \rbk{z_{0}, 1}$ because $\fun{\tanh}{\sminvtemperature \varepsilon} < \min \rbk{\sminvtemperature \varepsilon, 1}
\leq
\min \rbk{\sminvtemperature_{c} \varepsilon, 1}$.
Thus $\varphi$,
strictly concave since $\varphi''
=
\tilde{s}'' < 0$,
has the unique maximizer $\eta_{\ast}$ on $\closedinterval{0}{1}$ and the boundary point $\rbk{\eta_{\ast}, \eta_{\ast}}$ realizes $\fun{h}{\eta_{\ast}, \eta_{\ast}}
=
\fun{\varphi}{\eta_{\ast}}$.
The comparison of the two regional maxima is needed only for
$\sminvtemperature_{c} \varepsilon < 1$,
and in that case
$\fun{h}{z_{0}, z_{0}}
=
\fun{\varphi}{z_{0}}
\leq
\fun{\varphi}{\eta_{\ast}}$.
It follows that the global maximum is at $\rbk{\eta_{\ast}, \eta_{\ast}}$ in all cases,
unique by the strict concavity of $\varphi$ and the strict monotonicity of $G$.
Negative $\varepsilon$ is symmetric under $z
\mapsto
- z$.
The quantitative gap follows again by compactness.
\end{proof}

\begin{thm}[concentration of the total-spin distribution]\label{thm:concentration}
Let $\rbk{\eta_{\txtmax}, z_{\txtmax}}$ denote the maximizer of Proposition \ref{prop:maximizer} for the given $\rbk{\sminvtemperature, \varepsilon}$,
assuming $\sminvtemperature_{c} \abs{\varepsilon} < \fun{\eta_{0}}{\sminvtemperature}$ when $\sminvtemperature > \sminvtemperature_{c}$.
For every $\delta > 0$ there are $c_{\delta} > 0$ and $\Omega_{\delta}$ with
$$\sum_{\substack{\rbk{S, m}: \abs{\frac{2 S}{\Omega} - \eta_{\txtmax}} + \abs{\frac{2 m}{\Omega} - z_{\txtmax}}
\geq
\delta}} \fun{P_{\Omega}}{S, m}
\leq
\fnexp{- c_{\delta} \Omega}
\quad \rbk{\Omega
\geq
\Omega_{\delta}}.$$
The concentration estimate also implies that,
for every bounded function $g$ on the free-energy domain
$\dom h$ of \eqref{eq:free-energy-domain}
continuous at $\rbk{\eta_{\txtmax}, z_{\txtmax}}$,
$$\lim_{\Omega \to \infty}
\sum_{\substack{j = 0, 1, \dotsc, \lfloor \Omega / 2 \rfloor \\ S = \Omega / 2 - j \\ m = - S, - S + 1, \dotsc, S}} \fun{P_{\Omega}}{S, m} \fun{g}{\frac{2 S}{\Omega}, \frac{2 m}{\Omega}}
=
\fun{g}{\eta_{\txtmax}, z_{\txtmax}}.$$
\end{thm}

\begin{proof}
Write $Z_{\Omega}$ for the denominator of \eqref{eq:p-omega} and $h_{\txtmax}
=
\fun{h}{\eta_{\txtmax}, z_{\txtmax}}$,
with $h$ the free-energy function of
Definition \ref{def:free-energy-function}.
We first derive a lower bound for $Z_{\Omega}$.
The function $h$ is continuous at $\rbk{\eta_{\txtmax}, z_{\txtmax}}$.
It follows that for any prescribed $\epsilon_{0} > 0$ there is $\rho > 0$ with $\fun{h}{\eta, z}
\geq
h_{\txtmax} - \epsilon_{0}$ on the $\rho$-ball.
The lattice of points
$\rbk{\frac{2 S}{\Omega}, \frac{2 m}{\Omega}}$
associated with pairs $\rbk{S, m}$ admissible at volume $\Omega$ in the sense of
Definition \ref{def:admissible-spin-quantum-numbers}
has spacing $\frac{2}{\Omega}$ in each coordinate.
For sufficiently large $\Omega$,
it contains a point of the $\rho$-ball intersected with $\dom h$.
Indeed,
the maximizer lies in $\dom h$.
Moreover,
the constraint $\abs{z}
\leq
\eta$
is preserved when $m$ moves towards $0$ or $S$ moves towards $\frac{\Omega}{2}$ along lattice directions.
A single term together with Lemma \ref{lem:free-energy} gives
$$Z_{\Omega}
\geq
\fnexp{\Omega \rbk{h_{\txtmax} - \epsilon_{0}} - C \rbk{1 + \fun{\log}{\Omega}}}.$$
We next derive an upper bound for the tail.
By Proposition \ref{prop:maximizer} the exponent on the excluded region is at most $h_{\txtmax} - \fun{\kappa}{\delta}$,
and there are at most $\rbk{\Omega + 1}^{2}$ lattice points.
It follows that the numerator of the tail is at most $\rbk{\Omega + 1}^{2} \fnexp{\Omega \rbk{h_{\txtmax} - \fun{\kappa}{\delta}} + C \rbk{1 + \fun{\log}{\Omega}}}$.
Choosing $\epsilon_{0}
=
\frac{\fun{\kappa}{\delta}}{4}$,
the quotient is at most $\fnexp{- \frac{\fun{\kappa}{\delta}}{2} \Omega}$ for $\Omega$ large.
For the second statement split the sum at the $\delta$-ball,
use boundedness on the tail,
continuity on the ball,
and let $\delta \downarrow 0$ along the standard $\frac{\epsilon}{2}$-argument.
\end{proof}

\section{The Limit Gibbs State and the Gauge Average}\label{sec:limitstate}

The symmetry-breaking product states appearing in the limit are now defined, and the limit of the finite-volume Gibbs states \eqref{eq:finite-volume-gibbs-state} is computed and identified. The thermal functional is analyzed first, and its moment limits then determine the limiting state.

\begin{defn}[superconducting thermal regime]\label{def:superconducting-thermal-regime}
A degenerate thermal setting in the sense of
Definition \ref{def:degenerate-thermal-setting}
is in the superconducting thermal regime if
$\sminvtemperature
>
\sminvtemperature_{c}$ and
$\sminvtemperature_{c} \abs{\varepsilon}
<
\eta_{0}$,
where
$\eta_{0}
=
\fun{\eta_{0}}{\sminvtemperature}$
is the positive solution specified in
Proposition \ref{prop:maximizer}(1).
The associated order-parameter coordinates are
\begin{equation}\label{eq:order-parameters}
z_{0}
=
\sminvtemperature_{c} \varepsilon,
\quad
r_{0}
=
\sqrt{\eta_{0}^{2} - z_{0}^{2}}
>
0.
\end{equation}
The corresponding circle of Bloch vectors is
\begin{equation}\label{eq:symmetry-breaking-bloch-vector}
\physvec{m}_{\phi}
=
\rbk{r_{0} \fun{\cos}{\phi}, r_{0} \fun{\sin}{\phi}, z_{0}}
\quad
\rbk{\phi
\in
[0, 2 \pi)}.
\end{equation}
\end{defn}

The principal results from the present section through Section \ref{sec:decomposition} use the superconducting thermal regime of Definition \ref{def:superconducting-thermal-regime}. Statements treating \(\sminvtemperature
\leq
\sminvtemperature_{c}\) identify that alternative temperature range explicitly.

\subsection{Symmetry-breaking product states}\label{symmetry-breaking-product-states}

The self-consistency equation defines a gauge-indexed circle of product states. Their gauge average will be identified with the limiting Gibbs state in Theorem \ref{thm:limit-state}. Their factoriality and KMS properties are proved later without presupposing the central decomposition. The notation is that of Definition \ref{def:superconducting-thermal-regime} and \eqref{eq:order-parameters}--\eqref{eq:symmetry-breaking-bloch-vector}. The case \(\sminvtemperature
\leq
\sminvtemperature_{c}\), where the maximizer has \(\eta
=
\abs{z}\) and the transverse radius vanishes, is carried along in remarks.

\begin{defn}[Bogoliubov product states]
For $\phi
\in
[0, 2 \pi)$ let $\oastate[\psi^{\rbk{1}}_{\sminvtemperature,\phi}]$ be the state of $\spmat{2}{\fldcmp}$ with Bloch vector
$\physvec{m}_{\phi}$ defined by \eqref{eq:symmetry-breaking-bloch-vector},
as parametrized in \eqref{eq:bloch-state}.
The corresponding product state of $\oa{A}$ is defined by
\begin{equation}\label{eq:phase-product-state}
\oastate[\psi_{\sminvtemperature,\phi}]
=
\bigotimes_{p \in \semigrposint}
\oastate[\psi^{\rbk{1}}_{\sminvtemperature,\phi}].
\end{equation}
\end{defn}

The inverse-temperature index is retained because \(\physvec{m}_{\phi}\) depends on \(\sminvtemperature\) through the order parameters \eqref{eq:order-parameters}.

The normalized Haar probability measure on the gauge circle is \(\msrprb_{\mansphere^{1}}\) from \eqref{eq:gauge-haar-measure}.

\begin{prop}[factoriality, self-consistency, and gauge covariance]\label{prop:bogoliubov-product-properties}
For every $\phi
\in
[0, 2 \pi)$,
the Bogoliubov product state \eqref{eq:phase-product-state} has the following properties.
\begin{enumerate}
\item The one-site state $\oastate[\psi^{\rbk{1}}_{\sminvtemperature,\phi}]$ is the Gibbs state of the effective Hamiltonian
$\physham[h]_{\phi}$.
For every $A
\in
\spmat{2}{\fldcmp}$,
its density matrix and effective Hamiltonian satisfy
\begin{equation}\label{eq:self-consistency}
\begin{aligned}
\fun{\oastate[\psi^{\rbk{1}}_{\sminvtemperature,\phi}]}{A}
&=
\sqfun{\trace}{D_{\phi} A},
\\ 
D_{\phi}
&=
\frac{\fnexp{- \sminvtemperature \physham[h]_{\phi}}}{\sqfun{\trace}{\fnexp{- \sminvtemperature \physham[h]_{\phi}}}},
\\ 
\physham[h]_{\phi}
&=
- \frac{1}{\sminvtemperature_{c}} \eta_{0} \physvec{\sigma} \cdot \hat{\physvec{m}}_{\phi}
=
- \varepsilon \sigma^{z} - \frac{2}{\sminvtemperature_{c}} \rbk{m^{-}_{\phi} \sigma^{+} + m^{+}_{\phi} \sigma^{-}},
\\ 
\hat{\physvec{m}}_{\phi}
&=
\frac{\physvec{m}_{\phi}}{\eta_{0}},
\\ 
m^{\pm}_{\phi}
&=
\frac{1}{2} \rbk{m^{x}_{\phi} \pm \imunit m^{y}_{\phi}}
=
\frac{r_{0}}{2} \fnexp{\pm \imunit \phi}.
\end{aligned}
\end{equation}

\item The density matrix $D_{\phi}$ in \eqref{eq:self-consistency} is invertible.
The product state $\oastate[\psi_{\sminvtemperature,\phi}]$ is a factor state,
and its GNS representation is the one given by Proposition \ref{prop:gns-faithful}.

\item The gauge group permutes the family according to
\begin{equation}\label{eq:gauge-covariance-states}
\oastate[\psi_{\sminvtemperature,\phi}] \circ \mathfrak{g}_{\vartheta}
=
\oastate[\psi_{\sminvtemperature,\phi + \vartheta}],
\quad
\oastate[\psi_{\sminvtemperature,\phi}]
=
\oastate[\psi_{\sminvtemperature,0}] \circ \mathfrak{g}_{\phi}.
\end{equation}
\end{enumerate}
\end{prop}

\begin{proof}

(1)
For any $\kappa
\in
\fldreal$ and unit vector $\hat{\physvec{u}}$,
the Pauli relation $\rbk{\physvec{\sigma} \cdot \hat{\physvec{u}}}^{2}
=
1$ gives the exponential identity
\begin{equation*}
\fnexp{\sminvtemperature \kappa \physvec{\sigma} \cdot \hat{\physvec{u}}}
=
\fun{\cosh}{\sminvtemperature \kappa} + \fun{\sinh}{\sminvtemperature \kappa}
\physvec{\sigma} \cdot \hat{\physvec{u}}.
\end{equation*}
The normalized Gibbs matrix has Bloch vector
$\fun{\tanh}{\sminvtemperature \kappa} \hat{\physvec{u}}$.
Set $\kappa
=
\frac{1}{\sminvtemperature_{c}} \eta_{0}$.
The fixed point equation in Proposition \ref{prop:maximizer} gives
$\fun{\tanh}{\frac{\sminvtemperature}{\sminvtemperature_{c}} \eta_{0}}
=
\eta_{0}$.
The resulting Bloch vector is
$\eta_{0} \hat{\physvec{m}}_{\phi}
=
\physvec{m}_{\phi}$.
The Bloch parametrization \eqref{eq:bloch-state} proves the density-matrix identity in
\eqref{eq:self-consistency}.
The second form of $\physham[h]_{\phi}$ is the effective Hamiltonian
\eqref{eq:effective-hamiltonian} of the constant configuration aligned in the
sense of Definition \ref{def:aligned-configuration} in the
$\hat{\physvec{m}}_{\phi}$-direction,
because $\frac{1}{\sminvtemperature_{c}} \eta_{0} \hat{m}^{z}_{\phi}
=
\frac{1}{\sminvtemperature_{c}} z_{0}
=
\varepsilon$.
This is the self-consistency relation of \cite[Eqs. (9)--(10)]{ThirringWalter001}.

(2)
Equation \eqref{eq:symmetry-breaking-bloch-vector} gives
$\abs{\physvec{m}_{\phi}}
=
\eta_{0} < 1$.
The eigenvalue formula \eqref{eq:bloch-density-spectrum} proves that $D_{\phi}$ is invertible.
Proposition \ref{prop:gns-faithful} supplies the GNS representation,
and Theorem \ref{thm:factor} proves that $\oastate[\psi_{\sminvtemperature,\phi}]$ is a factor state.

(3)
For every $p
\in
\semigrposint$,
the one-site expectations satisfy
\begin{equation*}
\begin{aligned}
\fun{\oastate[\psi_{\sminvtemperature,\phi}]}{\fun{\mathfrak{g}_{\vartheta}}{\sigma^{z}_{p}}}
&=
z_{0}
=
\fun{\oastate[\psi_{\sminvtemperature,\phi + \vartheta}]}{\sigma^{z}_{p}},
\\
\fun{\oastate[\psi_{\sminvtemperature,\phi}]}{\fun{\mathfrak{g}_{\vartheta}}{\sigma^{+}_{p}}}
&=
\fnexp{\imunit \vartheta} m^{+}_{\phi}
=
m^{+}_{\phi + \vartheta}
=
\fun{\oastate[\psi_{\sminvtemperature,\phi + \vartheta}]}{\sigma^{+}_{p}}.
\end{aligned}
\end{equation*}
Taking adjoints gives the same equality for $\sigma^{-}_{p}$.
The matrices $1,\sigma^{z},\sigma^{+},\sigma^{-}$ span
$\spmat{2}{\fldcmp}$.
Equality on this spanning set proves that the one-site states agree.
Multiplicativity gives the first identity in \eqref{eq:gauge-covariance-states} on
$\oa{A}_{\txtloc}$,
and continuity extends it to $\oa{A}$.
Its specialization to $\phi
=
0$ and $\vartheta
=
\phi$ gives the second identity.
\end{proof}

\subsection{The thermal functional and its limit}\label{the-thermal-functional-and-its-limit}

The joint generating functional of the intensive spin observables has a locally uniform limit. Its derivatives provide all ordered mean moments needed for convergence of the states.

\begin{thm}[limit of the thermal functionals]\label{thm:functional-limit}
For $r
\geq
1$ and $\physvec{u}_{1}, \dotsc, \physvec{u}_{r}
\in
\fldreal^{3}$,
we set their sum $\physvec{w}
=
\sum_{j
=
1}^{r} \physvec{u}_{j}$.
The vector $\physvec{m}_{\phi}$ in the limiting integral is defined by
\eqref{eq:symmetry-breaking-bloch-vector}.
The joint characteristic function has the limit
$$\begin{aligned}
&\fun{F_{\Omega}}{\physvec{u}_{1}, \dotsc, \physvec{u}_{r}}
=
\fun{\oastate[\psi_{\sminvtemperature,\Omega}]}{\fnexp{\imunit \physvec{u}_{1} \cdot \physvec{M}_{\Omega}} \dotsm \fnexp{\imunit \physvec{u}_{r} \cdot \physvec{M}_{\Omega}}}
\\ 
&\to
\int_{\mansphere^{1}}
\fnexp{\imunit \physvec{w} \cdot \physvec{m}_{\phi}}
\opdmsr{\msrprb_{\mansphere^{1}}}(\phi)
\quad \rbk{\Omega \to \infty}.
\end{aligned}$$
The convergence extends to complex $\physvec{u}_{j}$,
locally uniformly on $\fldcmp^{3 r}$,
with the right side understood as $\fun{G_{\infty}}{\eta_{0}, z_{0}. \physvec{w}}$ of Proposition \ref{prop:bessel-limit},
and all partial derivatives at $\physvec{u}
=
0$ converge accordingly.
For $\sminvtemperature
\leq
\sminvtemperature_{c}$ the same holds with the right side $\fnexp{\imunit \physvec{w} \cdot \physvec{m}_{\sminvtemperature,\ast}}$,
$\physvec{m}_{\sminvtemperature,\ast}
=
\rbk{0, 0, \fun{\tanh}{\sminvtemperature \varepsilon}}$.
\end{thm}

\begin{proof}
Since $\physvec{u} \cdot \physvec{M}_{\Omega}
=
\frac{2}{\Omega} \physvec{u} \cdot \physvec{S}_{\Omega}$,
the product of exponentials is $g_{\Omega}^{\otimes \Omega}$,
where
$$g_{\Omega}
=
\prod_{j
=
1}^{r} \fnexp{\imunit \physvec{u}_{j} \cdot \physvec{\sigma} / \Omega}.$$
This factorization follows by applying the identity
$$\rbk{\fnexp{\physvec{a} \cdot \physvec{\sigma}}}^{\otimes \Omega}
=
\fnexp{2 \physvec{a} \cdot \physvec{S}}$$
from the proof of Proposition \ref{prop:matrix-element} to each factor.
Expanding each factor and collecting,
uniformly for the $\physvec{u}_{j}$ in a compact subset of $\fldcmp^{3 r}$,
$$g_{\Omega}
=
1 + \frac{\imunit}{\Omega} \physvec{w} \cdot \physvec{\sigma} + \fun{O}{\frac{1}{\Omega^{2}}},$$
which is the hypothesis of Proposition \ref{prop:bessel-limit}.
The spectral trace formula \eqref{eq:spectral-trace} and Proposition \ref{prop:matrix-element} give the decomposition
$$\fun{F_{\Omega}}{\physvec{u}}
=
\sum_{\substack{j = 0, 1, \dotsc, \lfloor \Omega / 2 \rfloor \\ S = \Omega / 2 - j \\ m = - S, - S + 1, \dotsc, S}} \fun{P_{\Omega}}{S, m} \fun{G_{\Omega}}{S, m},
\quad
\fun{G_{\Omega}}{S, m}
=
\bkt{\ket{S, m}}{\fun{D^{S}}{g_{\Omega}} \ket{S, m}},$$
where
$\fun{D^{S}}{g}
=
\fnrestr{g^{\otimes 2 S}}{V_{S}}$
for $g
\in
\fun{\liegr{GL}}{2, \fldcmp}$.
Split
$$\begin{aligned}
&\abs{\fun{F_{\Omega}}{\physvec{u}} - \fun{G_{\infty}}{\eta_{\txtmax}, z_{\txtmax}. \physvec{w}}}
\\ 
&\leq
\sum_{\substack{j = 0, 1, \dotsc, \lfloor \Omega / 2 \rfloor \\ S = \Omega / 2 - j \\ m = - S, - S + 1, \dotsc, S}}
\fun{P_{\Omega}}{S, m}
\abs{\fun{G_{\Omega}}{S, m} - \fun{G_{\infty}}{\frac{2 S}{\Omega}, \frac{2 m}{\Omega}. \physvec{w}}}
\\
&\quad+
\abs{\sum_{\substack{j = 0, 1, \dotsc, \lfloor \Omega / 2 \rfloor \\ S = \Omega / 2 - j \\ m = - S, - S + 1, \dotsc, S}}
\fun{P_{\Omega}}{S, m}
\fun{G_{\infty}}
{\frac{2 S}{\Omega}, \frac{2 m}{\Omega}. \physvec{w}}
-\fun{G_{\infty}}{\eta_{\txtmax}, z_{\txtmax}. \physvec{w}}}.
\end{aligned}$$
The first term tends to $0$ uniformly on compacts by Proposition \ref{prop:bessel-limit}.
For the second term,
apply Theorem \ref{thm:concentration} to
$g
=
\fun{G_{\infty}}{\cdot, \cdot. \physvec{w}}$.
This function is continuous on the free-energy domain
$\dom h$ of \eqref{eq:free-energy-domain} and is bounded there locally uniformly in
$\physvec{w}$.
Uniformity for $\physvec{w}$ in a compact set follows from the uniform modulus of continuity of
$G_{\infty}$ on the product of $\dom h$ with that compact set.
The concentration estimate then makes the second term tend to $0$ uniformly on compact sets.
For $\sminvtemperature > \sminvtemperature_{c}$ the maximizer is $\rbk{\eta_{0}, z_{0}}$ with $\eta_{0}^{2} - z_{0}^{2}
=
r_{0}^{2}$,
and Lemma \ref{lem:bessel-average} identifies,
for real $\physvec{w}$,
$$\fun{G_{\infty}}{\eta_{0}, z_{0}. \physvec{w}}
=
\fnexp{\imunit z_{0} w^{z}} \fun{J_{0}}{r_{0} \abs{\physvec{w}_{\perp}}}
=
\int_{\mansphere^{1}} \fnexp{\imunit \physvec{w} \cdot \physvec{m}_{\phi}} \opdmsr{\msrprb_{\mansphere^{1}}}(\phi).$$
For $\sminvtemperature
\leq
\sminvtemperature_{c}$ the maximizer has $\eta_{\txtmax}^{2} - z_{\txtmax}^{2}
=
0$ and $\fun{G_{\infty}}{\eta_{\ast}, z_{\ast}, \physvec{w}}
=
\fnexp{\imunit z_{\ast} w^{z}}
=
\fnexp{\imunit \physvec{w} \cdot \physvec{m}_{\sminvtemperature,\ast}}$,
where $z_{\ast}
=
\fun{\tanh}{\sminvtemperature \varepsilon}$ with sign included.
Finally,
the family $F_{\Omega}$ is holomorphic on $\fldcmp^{3 r}$ and uniformly bounded on compacts,
since $$\norm{\fnexp{\imunit \physvec{u} \cdot \physvec{M}_{\Omega}}}
\leq
\fnexp{\norm{\opimag \physvec{u} \cdot \physvec{M}_{\Omega}}}
\leq
\fnexp{3 \max_{\gamma} \abs{\opimag u^{\gamma}}},$$
It follows that Corollary \ref{cor:vitali-several} upgrades the pointwise convergence to locally uniform convergence with all derivatives.
\end{proof}

The functional \(F_{\Omega}\) replaces the Euler-angle generating function of \cite[Eqs. (4)--(6)]{ThirringWalter001}. The difficulty with generating the \(\sigma^{y}\)-moments at \(\Omega
=
\infty\) observed in the footnote to \cite[Eq. (5)]{ThirringWalter001} does not arise here, the three spin directions entering the exponents symmetrically through the vector parameters.

\subsection{Moments and the limit state}\label{moments-and-the-limit-state}

The moment argument below converts convergence of the thermal functional into convergence on local monomials and thereby identifies the limiting state.

\begin{lem}[collision estimate]\label{lem:collision}
Let $r, k
\geq
0$.
If $k
\geq
1$,
choose distinct fixed sites $p_{1}, \dotsc, p_{k}$ and set
$R
=
\setone{p_{1}, \dotsc, p_{k}}$.
When $k
=
0$,
set $R
=
\emptyset$.
Define $p_{\ast}
=
\max R$ for $k
\geq
1$ and $p_{\ast}
=
0$ for $k
=
0$.
For every $j
\in
\intint{1..k}$,
let $B_{j}
\in
\oa{A}_{\setone{p_{j}}}$ with $\norm{B_{j}}
\leq
1$,
and let $\gamma_{1}, \dotsc, \gamma_{r}
\in
\setone{x, y, z}$.
For $\Omega
\geq
p_{\ast} + r$ define
$$X_{j,\Omega}
=
B_{j}
\quad
\rbk{j
\in
\intint{1..k}},
\quad
X_{k+i,\Omega}
=
M^{\gamma_{i}}_{\Omega}
\quad
\rbk{i
\in
\intint{1..r}},$$
where $M^{\gamma}_{\Omega}
=
M^{\gamma}_{0,\Omega}$ is the mean magnetization
\eqref{eq:collective}.
Fix a permutation $\tau
\in
\grsym{k+r}$.
When $k
=
0$ or $r
=
0$,
the corresponding lists and products are understood as empty.
The product with the order specified by $\tau$ is
\begin{equation}\label{expedition0025003}
\Pi_{\Omega}
=
\prod_{\ell
=
1}^{k+r} X_{\fun{\tau}{\ell},\Omega}
=
X_{\fun{\tau}{1},\Omega} \dotsm X_{\fun{\tau}{k+r},\Omega}.
\end{equation}
For any distinct sites $q_{1}, \dotsc, q_{r}
\in
\intint{1..\Omega} \setminus R$,
$$\abs{\fun{\oastate[\psi_{\sminvtemperature,\Omega}]}{\Pi_{\Omega}} - \fun{\oastate[\psi_{\sminvtemperature,\Omega}]}{B_{1} \dotsm B_{k} \sigma^{\gamma_{1}}_{q_{1}} \dotsm \sigma^{\gamma_{r}}_{q_{r}}}}
\leq
\frac{r \rbk{r - 1} + 2 r k}{\Omega}
\leq
\frac{\rbk{r + k}^{2} 2^{r}}{\Omega}.$$
\end{lem}

\begin{proof}
If $r
=
0$,
the product \eqref{expedition0025003} contains only the operators $B_{j}$.
They act on distinct sites and commute,
so the asserted difference is zero.
Assume below that $r
\geq
1$.

For a tuple $\physvec{q}'
=
\rbk{q_{1}', \dotsc, q_{r}'}
\in
\intint{1..\Omega}^{r}$ define
$$\fun{X_{j}}{\physvec{q}'}
=
B_{j}
\quad
\rbk{j
\in
\intint{1..k}},
\quad
\fun{X_{k+i}}{\physvec{q}'}
=
\sigma^{\gamma_{i}}_{q_{i}'}
\quad
\rbk{i
\in
\intint{1..r}},$$
and preserve the order in \eqref{expedition0025003} by setting
$\fun{\Pi}{\physvec{q}'}
=
\prod_{\ell
=
1}^{k+r} \fun{X_{\fun{\tau}{\ell}}}{\physvec{q}'}$.
Expanding every mean magnetization in \eqref{expedition0025003} gives
\begin{equation}\label{expedition0025004}
\fun{\oastate[\psi_{\sminvtemperature,\Omega}]}{\Pi_{\Omega}}
=
\frac{1}{\Omega^{r}}
\sum_{\physvec{q}'
\in
\intint{1..\Omega}^{r}}
\fun{\oastate[\psi_{\sminvtemperature,\Omega}]}{\fun{\Pi}{\physvec{q}'}}.
\end{equation}

Let $\mathcal{G}_{\Omega}$ be the set of tuples whose entries are pairwise distinct and lie outside $R$,
and let $\mathcal{B}_{\Omega}
=
\intint{1..\Omega}^{r} \setminus \mathcal{G}_{\Omega}$.
For $\physvec{q}'
\in
\mathcal{G}_{\Omega}$,
the factors of $\fun{\Pi}{\physvec{q}'}$ act on distinct sites and commute.
Commutativity gives
$$\fun{\Pi}{\physvec{q}'}
=
B_{1} \dotsm B_{k}
\sigma^{\gamma_{1}}_{q_{1}'} \dotsm
\sigma^{\gamma_{r}}_{q_{r}'}.$$
There is a permutation $\rho$ of $\intint{1..\Omega}$ that fixes every site in $R$ and satisfies
$\fun{\rho}{q_{i}'}
=
q_{i}$ for every $i
\in
\intint{1..r}$.
Permutation invariance \eqref{eq:gibbs-invariance} then gives
$$\fun{\oastate[\psi_{\sminvtemperature,\Omega}]}{\fun{\Pi}{\physvec{q}'}}
=
a_{\Omega},
\quad
a_{\Omega}
=
\fun{\oastate[\psi_{\sminvtemperature,\Omega}]}{
B_{1} \dotsm B_{k}
\sigma^{\gamma_{1}}_{q_{1}} \dotsm
\sigma^{\gamma_{r}}_{q_{r}}}.$$

A tuple belongs to $\mathcal{B}_{\Omega}$ only if some entry lies in $R$ or two entries coincide.
The union bound gives the explicit count
$$\begin{aligned}
\abscard{\mathcal{B}_{\Omega}}
\leq
\sum_{i
=
1}^{r}
k \Omega^{r-1}
+
\sum_{1
\leq
i < j
\leq
r}
\Omega^{r-1}
=
\rbk{r k + \binom{r}{2}} \Omega^{r-1}.
\end{aligned}$$
Since the total number of tuples is $\Omega^{r}$,
\eqref{expedition0025004} and the common value on $\mathcal{G}_{\Omega}$ yield the exact error identity
$$\fun{\oastate[\psi_{\sminvtemperature,\Omega}]}{\Pi_{\Omega}} - a_{\Omega}
=
\frac{1}{\Omega^{r}}
\sum_{\physvec{q}'
\in
\mathcal{B}_{\Omega}}
\rbk{
\fun{\oastate[\psi_{\sminvtemperature,\Omega}]}{\fun{\Pi}{\physvec{q}'}}
- a_{\Omega}
}.$$
Every factor in the products has norm at most $1$.
Both state values in each summand have modulus at most $1$,
and the count above gives
$$\abs{\fun{\oastate[\psi_{\sminvtemperature,\Omega}]}{\Pi_{\Omega}} - a_{\Omega}}
\leq
\frac{2 \abscard{\mathcal{B}_{\Omega}}}{\Omega^{r}}
\leq
\frac{r \rbk{r - 1} + 2 r k}{\Omega}
\leq
\frac{\rbk{r + k}^{2} 2^{r}}{\Omega}.$$
\end{proof}

\begin{thm}[the limit Gibbs state]\label{thm:limit-state}
Assume the degenerate thermal setting of
Definition \ref{def:degenerate-thermal-setting}.
If this setting is in the superconducting thermal regime of
Definition \ref{def:superconducting-thermal-regime},
the limit
$$\fun{\oastate[\psi_{\sminvtemperature}]}{A}
=
\lim_{\Omega
\to
\infty} \fun{\oastate[\widehat{\psi}_{\sminvtemperature,\Omega}]}{A}
\quad \rbk{A
\in
\oa{A}}$$
exists and defines the gauge-invariant state
\begin{equation}\label{eq:limit-state}
\oastate[\psi_{\sminvtemperature}]
=
\int_{\mansphere^{1}} \oastate[\psi_{\sminvtemperature,\phi}] \opdmsr{\msrprb_{\mansphere^{1}}}(\phi),
\end{equation}
meaning $\fun{\oastate[\psi_{\sminvtemperature}]}{A}
=
\int_{\mansphere^{1}} \fun{\oastate[\psi_{\sminvtemperature,\phi}]}{A} \opdmsr{\msrprb_{\mansphere^{1}}}(\phi)$ for every $A
\in
\oa{A}$.
The integrand is continuous in $\phi$.
For $\sminvtemperature
\leq
\sminvtemperature_{c}$ the limit exists and is the product state $\oastate[\psi_{\sminvtemperature,\ast}]
=
\bigotimes_{p \in \semigrposint} \oastate[\psi^{\rbk{1}}_{\physvec{m}_{\sminvtemperature,\ast}}]$ with $\physvec{m}_{\sminvtemperature,\ast}
=
\rbk{0, 0, \fun{\tanh}{\sminvtemperature \varepsilon}}$.
\end{thm}

\begin{proof}
Every element of $\oa{A}_{R}$,
$R
=
\setone{p_{1}, \dotsc, p_{k}}$,
is a linear combination of monomials with one factor per site.
These monomials have the form
$\sigma^{\alpha_{1}}_{p_{1}} \dotsm \sigma^{\alpha_{k}}_{p_{k}}$ with $\alpha_{j}
\in
\setone{1, x, y, z}$.
Indeed,
a product of Pauli matrices at one site reduces to a scalar multiple of a single such factor.
Consider such a monomial with the identity factors dropped,
$r$ nontrivial Pauli factors at the distinct sites $q_{1}, \dotsc, q_{r}$ remaining.
For $r
=
0$ the required comparison is exact.
Assume $r
\geq
1$.
Expanding the mean magnetizations gives
$$\fun{\oastate[\psi_{\sminvtemperature,\Omega}]}{M^{\gamma_{1}}_{\Omega} \dotsm M^{\gamma_{r}}_{\Omega}}
=
\frac{1}{\Omega^{r}}
\sum_{\rbk{p_{1}, \dotsc, p_{r}}
\in
\intint{1..\Omega}^{r}}
\fun{\oastate[\psi_{\sminvtemperature,\Omega}]}{
\sigma^{\gamma_{1}}_{p_{1}} \dotsm \sigma^{\gamma_{r}}_{p_{r}}}.$$
Let $\mathcal{B}^{\rbk{r}}_{\Omega}$ be the set of tuples in this sum with at least two equal entries.
For every tuple outside $\mathcal{B}^{\rbk{r}}_{\Omega}$,
permutation invariance \eqref{eq:gibbs-invariance} gives
$$\fun{\oastate[\psi_{\sminvtemperature,\Omega}]}{
\sigma^{\gamma_{1}}_{p_{1}} \dotsm \sigma^{\gamma_{r}}_{p_{r}}}
=
\fun{\oastate[\psi_{\sminvtemperature,\Omega}]}{
\sigma^{\gamma_{1}}_{q_{1}} \dotsm \sigma^{\gamma_{r}}_{q_{r}}}.$$
The distinct-tuple terms cancel after subtracting the fixed-site expectation:
$$\begin{aligned}
&\fun{\oastate[\psi_{\sminvtemperature,\Omega}]}{M^{\gamma_{1}}_{\Omega} \dotsm M^{\gamma_{r}}_{\Omega}}
-\fun{\oastate[\psi_{\sminvtemperature,\Omega}]}
{\sigma^{\gamma_{1}}_{q_{1}} \dotsm \sigma^{\gamma_{r}}_{q_{r}}}
\\ 
&=
\frac{1}{\Omega^{r}}
\sum_{\rbk{p_{1}, \dotsc, p_{r}}
\in
\mathcal{B}^{\rbk{r}}_{\Omega}}
\rbk{
\fun{\oastate[\psi_{\sminvtemperature,\Omega}]}{
\sigma^{\gamma_{1}}_{p_{1}} \dotsm \sigma^{\gamma_{r}}_{p_{r}}}
-
\fun{\oastate[\psi_{\sminvtemperature,\Omega}]}{
\sigma^{\gamma_{1}}_{q_{1}} \dotsm \sigma^{\gamma_{r}}_{q_{r}}}
}.
\end{aligned}$$
The union bound over the colliding pairs gives
$$\abscard{\mathcal{B}^{\rbk{r}}_{\Omega}}
\leq
\sum_{1
\leq
i < j
\leq
r} \Omega^{r - 1}
=
\binom{r}{2} \Omega^{r - 1}.$$
Each state value in the preceding difference has modulus at most $1$.
It follows that
$$\abs{
\fun{\oastate[\psi_{\sminvtemperature,\Omega}]}{M^{\gamma_{1}}_{\Omega} \dotsm M^{\gamma_{r}}_{\Omega}}
-
\fun{\oastate[\psi_{\sminvtemperature,\Omega}]}{
\sigma^{\gamma_{1}}_{q_{1}} \dotsm \sigma^{\gamma_{r}}_{q_{r}}}}
\leq
\frac{2 \abscard{\mathcal{B}^{\rbk{r}}_{\Omega}}}{\Omega^{r}}
\leq
\frac{r \rbk{r - 1}}{\Omega}.$$
This is the $k
=
0$ specialization of Lemma \ref{lem:collision} and yields
$$\fun{\oastate[\psi_{\sminvtemperature,\Omega}]}{\sigma^{\gamma_{1}}_{q_{1}} \dotsm \sigma^{\gamma_{r}}_{q_{r}}}
=
\fun{\oastate[\psi_{\sminvtemperature,\Omega}]}{M^{\gamma_{1}}_{\Omega} \dotsm M^{\gamma_{r}}_{\Omega}} + \fun{O}{\frac{1}{\Omega}},$$
where the ordered mean moments are derivatives of the functional in Theorem \ref{thm:functional-limit}.
More precisely,
they are
$\rbk{- \imunit}^{r} \partial_{u_{1}^{\gamma_{1}}} \dotsm \partial_{u_{r}^{\gamma_{r}}} \fun{F_{\Omega}}{0}$,
with one derivative for each exponential factor.
By that theorem the derivatives converge to the corresponding derivatives of $\int_{\mansphere^{1}} \fnexp{\imunit \physvec{w} \cdot \physvec{m}_{\phi}} \opdmsr{\msrprb_{\mansphere^{1}}}(\phi)$ at $0$,
namely to $\int_{\mansphere^{1}} \prod_{j
=
1}^{r} m^{\gamma_{j}}_{\phi} \opdmsr{\msrprb_{\mansphere^{1}}}(\phi)$.
Factorization of the product state over the distinct sites gives
$$\fun{\oastate[\psi_{\sminvtemperature,\phi}]}{\sigma^{\gamma_{1}}_{q_{1}} \dotsm \sigma^{\gamma_{r}}_{q_{r}}}
=
\prod_{j = 1}^{r} m^{\gamma_{j}}_{\phi}.$$
It follows that the limit exists on every local monomial and agrees with \eqref{eq:limit-state} there.
Linearity extends it to $\oa{A}_{\txtloc}$.
For a general $A
\in
\oa{A}$ and $\epsilon
>
0$,
pick local $A'$ with $\norm{A - A'} < \epsilon$.
The state norm bound gives
$\limsup_{\Omega,\Omega' \to \infty} \abs{\fun{\oastate[\widehat{\psi}_{\sminvtemperature,\Omega}]}{A} - \fun{\oastate[\widehat{\psi}_{\sminvtemperature,\Omega'}]}{A}}
\leq
2 \epsilon$.
It follows that the limit exists,
and the identity \eqref{eq:limit-state} persists since both sides are norm-continuous in $A$.
The limit of states is a state,
and gauge invariance follows from \eqref{eq:gibbs-invariance} or from \eqref{eq:gauge-covariance-states} and translation invariance of $\msrprb_{\mansphere^{1}}$.
The case $\sminvtemperature
\leq
\sminvtemperature_{c}$ is identical with the constant $\physvec{m}_{\sminvtemperature,\ast}$ in place of the $\phi$-average.
\end{proof}

The state \eqref{eq:limit-state} exhibits off-diagonal long-range order while each fiber does not:

\begin{prop}[long-range order]\label{prop:odlro}
Assume the superconducting thermal regime of
Definition \ref{def:superconducting-thermal-regime}.
For all $p
\neq
q$,
$$\fun{\oastate[\psi_{\sminvtemperature}]}{\sigma^{+}_{p} \sigma^{-}_{q}}
=
\frac{r_{0}^{2}}{4} > 0,
\quad
\fun{\oastate[\psi_{\sminvtemperature}]}{\sigma^{+}_{p}}
=
0,$$
whereas $$\fun{\oastate[\psi_{\sminvtemperature,\phi}]}{\sigma^{+}_{p} \sigma^{-}_{q}}
=
m^{+}_{\phi} m^{-}_{\phi}
=
\frac{r_{0}^{2}}{4},
\quad
\fun{\oastate[\psi_{\sminvtemperature,\phi}]}{\sigma^{+}_{p}}
=
m^{+}_{\phi}
=
\frac{r_{0}}{2} \fnexp{\imunit \phi}
\neq
0.$$
In particular $\oastate[\psi_{\sminvtemperature}]$ is not a product state,
and each $\oastate[\psi_{\sminvtemperature,\phi}]$ breaks the gauge symmetry.
\end{prop}

\begin{proof}
By \eqref{eq:limit-state} and factorization of the fibers,
$\fun{\oastate[\psi_{\sminvtemperature}]}{\sigma^{+}_{p} \sigma^{-}_{q}}
=
\int_{\mansphere^{1}} m^{+}_{\phi} m^{-}_{\phi} \opdmsr{\msrprb_{\mansphere^{1}}}(\phi)
=
\frac{r_{0}^{2}}{4}$
because the integrand is constant.
On the other hand,
$\fun{\oastate[\psi_{\sminvtemperature}]}{\sigma^{+}_{p}}
=
\int_{\mansphere^{1}} \frac{r_{0}}{2} \fnexp{\imunit \phi} \opdmsr{\msrprb_{\mansphere^{1}}}(\phi)
=
0$.
The adjoint relation for a state also gives
$\fun{\oastate[\psi_{\sminvtemperature}]}{\sigma^{-}_{q}}
=
0$.
A product state would therefore make the two-point function vanish,
contrary to its value $\frac{r_{0}^{2}}{4}
>
0$.
\end{proof}

\section{Thermal Green Functions}\label{sec:green}

The finite-volume Gibbs states and Heisenberg dynamics in this section are those defined by \eqref{eq:finite-volume-gibbs-state} and \eqref{eq:finite-volume-dynamics}. The superconducting thermal regime is the one fixed in Definition \ref{def:superconducting-thermal-regime}. The main result of \cite{ThirringWalter001} states that the exact thermal Green functions in this regime converge to the gauge average of the Green functions of the Bogoliubov dynamics. The proof combines the estimates for noncommutative polynomials in Section \ref{sec:polynomials}, the moment convergence of Section \ref{sec:limitstate}, and the Vitali theorem.

\subsection{Degenerate specialization of the finite-volume polynomials}\label{degenerate-specialization-of-the-finite-volume-polynomials}

The degenerate thermal setting of Definition \ref{def:degenerate-thermal-setting} reduces every mean-field polynomial family of Definition \ref{def:mean-field-polynomial-family} to one involving only the unweighted mean-spin operators. Substitution of the fiber means then gives the Bogoliubov derivations used in the Green-function limit.

The assumptions and notation remain those of Definition \ref{def:superconducting-thermal-regime} and \eqref{eq:order-parameters}. The modifications for \(\sminvtemperature
\leq
\sminvtemperature_{c}\) consist in replacing the \(\phi\)-average by evaluation at \(\physvec{m}_{\sminvtemperature,\ast}\) and hold verbatim.

For the weighted means \eqref{eq:weighted-means}, the degenerate model satisfies \(M^{\gamma}_{n,\Omega}
=
\varepsilon^{n} M^{\gamma}_{0, \Omega}\). It follows that the derivation closes on the generators \(\sigma^{\gamma}_{p}\) and \(M^{\gamma}_{0,\Omega}\) alone, with the derivation rules of Lemma \ref{lem:exact-derivation} specialized to \begin{equation}\label{eq:degenerate-derivation}
\begin{aligned}
\fun{\delta_{\Omega}}{M^{+}_{0,\Omega}}
&=
\imunit \rbk{- 2 \varepsilon M^{+}_{0,\Omega} + \frac{2}{\sminvtemperature_{c}} M^{+}_{0,\Omega} M^{z}_{0,\Omega}},
\\ 
\fun{\delta_{\Omega}}{M^{-}_{0,\Omega}}
&=
\imunit \rbk{2 \varepsilon M^{-}_{0,\Omega} - \frac{2}{\sminvtemperature_{c}} M^{z}_{0,\Omega} M^{-}_{0,\Omega}},
\\ 
\fun{\delta_{\Omega}}{M^{z}_{0,\Omega}}
&=
- \frac{4 \imunit}{\sminvtemperature_{c}} \rbk{M^{+}_{0,\Omega} M^{-}_{0,\Omega} - M^{+}_{0,\Omega} M^{-}_{0,\Omega}}
=
0,
\end{aligned}
\end{equation} where the last relation reflects the exact conservation of \(S^{z}_{\Omega}\). The bounds of Lemma \ref{lem:derivation-bounds} hold as stated. For each \(\phi\), let \(F_{\phi}\) be the local element obtained from a mean-field polynomial family \(\seq{F_{\Omega}}{\Omega \geq \Omega_{0}}\) by replacing each \(M^{\gamma}_{0,\Omega}\) in its fixed expression by \(m^{\gamma}_{\phi}\). On local elements define \[\fun{\delta_{\phi}}{B}
=
\imunit
\sum_{p
\in
\fun{\supp}{B}}
\commutator{\physham[h]_{\phi,p}}{B},\] where \(\physham[h]_{\phi,p}\) is the copy of \eqref{eq:self-consistency} at the site \(p\).

\begin{lem}[thermal stationarity]\label{lem:thermal-stationarity}
For every $\phi$,
substitution of $M^{\gamma}_{0,\Omega}$ by $m^{\gamma}_{\phi}$ in
\eqref{eq:degenerate-derivation} gives zero.
For every mean-field polynomial family in the degenerate model,
it holds that
$$\rbk{\fun{\delta_{\Omega}}{F_{\Omega}}}_{\phi}
=
\fun{\delta_{\phi}}{F_{\phi}}.$$
\end{lem}

\begin{proof}
With $m^{z}_{\phi}
=
z_{0}
=
\sminvtemperature_{c} \varepsilon$ in \eqref{eq:degenerate-derivation},
$$\imunit m^{+}_{\phi} \rbk{- 2 \varepsilon + \frac{2}{\sminvtemperature_{c}} z_{0}}
=
0,$$
and the substituted expression for $M^{-}_{0,\Omega}$ vanishes in the same way.
The expression for $M^{z}_{0,\Omega}$ is identically zero.
The identity for a general polynomial family follows from the Leibniz rule,
as in Lemma \ref{lem:stationarity}.
\end{proof}

\begin{lem}[thermal limits of mean-field polynomials]\label{lem:thermal-polynomial-limits}
Fix a finite $R
\subset
\semigrposint$.
Let $\seq{F_{\Omega}}{\Omega \geq \Omega_{0}}$ be a mean-field polynomial family whose local factors have sites in $R$ and whose mean-field factors are
$M^{\gamma}_{0,\Omega}$.
Then we obtain
$$\lim_{\Omega \to \infty}
\fun{\oastate[\widehat{\psi}_{\sminvtemperature,\Omega}]}{F_{\Omega}}
=
\int_{\mansphere^{1}}
\fun{\oastate[\psi_{\sminvtemperature,\phi}]}{F_{\phi}}
\opdmsr{\msrprb_{\mansphere^{1}}}(\phi).$$
\end{lem}

\begin{proof}
By linearity,
it suffices to consider a monomial family
$\seq{F_{\Omega}}{\Omega \geq \Omega_{0}}$.
Multiply out the $\sigma$-factors at coinciding sites,
reducing $F_{\Omega}$ to an interleaved product of one-site elements $B_{1}, \dotsc, B_{k}$ at the distinct sites of $R$,
with $\norm{B_{j}}
\leq
1$ after normalization,
and $r$ factors $M^{\gamma_{i}}_{0,\Omega}$.
Same-site $\sigma$-factors separated by mean-field factors may be moved past them at the cost $\norm{\commutator{\sigma^{\alpha}_{p}}{M^{\gamma}_{0,\Omega}}}
\leq
\frac{2}{\Omega}$ per exchange,
and the finitely many exchanges produce a norm error of order $\frac{1}{\Omega}$.
The same reduction applies to $F_{\phi}$,
whose mean-field factors become the scalars $m^{\gamma_{i}}_{\phi}$ and commute with everything.
Lemma \ref{lem:collision} replaces the repeated mean-field slots by distinct auxiliary sites.
The resulting comparison is
$$\fun{\oastate[\widehat{\psi}_{\sminvtemperature,\Omega}]}{F_{\Omega}}
=
\fun{\oastate[\widehat{\psi}_{\sminvtemperature,\Omega}]}{B_{1} \dotsm B_{k} \sigma^{\gamma_{1}}_{q_{1}} \dotsm \sigma^{\gamma_{r}}_{q_{r}}} + \fun{O}{\frac{1}{\Omega}}$$
for distinct auxiliary sites $q_{i}$ outside $R$.
Expanding each $B_{j}$ in the one-site Pauli basis,
the right side is a finite combination of expectations of distinct-site Pauli monomials,
each of which converges by Theorem \ref{thm:limit-state} to its $\oastate[\psi_{\sminvtemperature}]$-value $\int_{\mansphere^{1}} \fun{\oastate[\psi_{\sminvtemperature,\phi}]}{\cdot} \opdmsr{\msrprb_{\mansphere^{1}}}(\phi)$.
Since all sites in the limiting product are distinct,
the fiber value factorizes as
$$\fun{\oastate[\psi_{\sminvtemperature,\phi}]}
{B_{1} \dotsm B_{k} \sigma^{\gamma_{1}}_{q_{1}} \dotsm \sigma^{\gamma_{r}}_{q_{r}}}
=
\prod_{j = 1}^{k}
\fun{\oastate[\psi^{\rbk{1}}_{\sminvtemperature,\phi}]}{B_{j}}
\prod_{i = 1}^{r} m^{\gamma_{i}}_{\phi}
=
\fun{\oastate[\psi_{\sminvtemperature,\phi}]}{F_{\phi}},$$
because $\fun{\oastate[\psi^{\rbk{1}}_{\sminvtemperature,\phi}]}{\sigma^{\gamma}}
=
m^{\gamma}_{\phi}$.
\end{proof}

\begin{thm}[convergence of the thermal Green functions]\label{thm:green}
Assume the superconducting thermal regime of
Definition \ref{def:superconducting-thermal-regime}.
Let $\tau^{\txtbogoliubov, \phi}$ be the product dynamics of $\oa{A}$ generated by the one-site Hamiltonians $\physham[h]_{\phi, p}$ of \eqref{eq:self-consistency},
defined as in Theorem \ref{thm:dynamics}.
Then for all $k
\geq
1$,
all $B_{1}, \dotsc, B_{k}
\in
\oa{A}_{\txtloc}$,
and all $t_{1}, \dotsc, t_{k}
\in
\fldreal$,
\begin{equation}\label{expedition0025005}
\lim_{\Omega
\to
\infty} \fun{\oastate[\widehat{\psi}_{\sminvtemperature,\Omega}]}{\fun{\tau^{\Omega}_{t_{1}}}{B_{1}} \dotsm \fun{\tau^{\Omega}_{t_{k}}}{B_{k}}}
=
\int_{\mansphere^{1}} \fun{\oastate[\psi_{\sminvtemperature,\phi}]}{\fun{\tau^{\txtbogoliubov, \phi}_{t_{1}}}{B_{1}} \dotsm \fun{\tau^{\txtbogoliubov, \phi}_{t_{k}}}{B_{k}}} \opdmsr{\msrprb_{\mansphere^{1}}}(\phi).
\end{equation}
\end{thm}

\begin{proof}
By linearity assume each $B_{j}$ is a monomial in the local operators
$\sigma^{\gamma}_{p}$ on at most $k_{j}$ sites.
Let $a$ be the constant of Lemma \ref{lem:derivation-bounds} and set $r_{1}
=
\frac{1}{2 a}$.

We first prove convergence under the restriction $\abs{t_{j}} < r_{1}$.
The uniformly convergent series of Lemma \ref{lem:derivation-bounds} gives the product expansion
$$\fun{\oastate[\widehat{\psi}_{\sminvtemperature,\Omega}]}{\prod_{j = 1}^{k} \fun{\tau^{\Omega}_{t_{j}}}{B_{j}}}
=
\sum_{N_{1}, \dotsc, N_{k}
\geq
0} \prod_{j = 1}^{k} \frac{t_{j}^{N_{j}}}{N_{j}!}
\fun{\oastate[\widehat{\psi}_{\sminvtemperature,\Omega}]}{
\prod_{j = 1}^{k} \fun{\delta_{\Omega}^{N_{j}}}{B_{j}}}.$$
This series converges absolutely and uniformly in $\Omega$.
The modulus of its $\rbk{N_{1}, \dotsc, N_{k}}$-term is bounded by
$$\prod_{j = 1}^{k}
\frac{\rbk{a \abs{t_{j}}}^{N_{j}}}{N_{j}!}
\cdot
\frac{\rbk{k_{j} + N_{j} - 1}!}{\rbk{k_{j} - 1}!}
6^{k_{j}}
\norm{B_{j}},$$
and these bounds are summable for $\abs{t_{j}}
<
\frac{1}{a}$.
For each fixed multi-index,
Lemma \ref{lem:thermal-polynomial-limits} and Lemma \ref{lem:thermal-stationarity} give
$$\lim_{\Omega \to \infty}
\fun{\oastate[\widehat{\psi}_{\sminvtemperature,\Omega}]}{
\prod_{j = 1}^{k} \fun{\delta_{\Omega}^{N_{j}}}{B_{j}}}
=
\int_{\mansphere^{1}}
\fun{\oastate[\psi_{\sminvtemperature,\phi}]}{\prod_{j = 1}^{k} \fun{\delta_{\phi}^{N_{j}}}{B_{j}}}
\opdmsr{\msrprb_{\mansphere^{1}}}(\phi),$$
where $\delta_{\phi}$ is the generator of $\tau^{\txtbogoliubov, \phi}$ defined above.
Summation over the multi-indices commutes with the limit $\Omega
\to
\infty$ by the uniform tail bounds,
and with the $\phi$-integral by dominated convergence.
The required uniform estimate is
$$\norm{\fun{\delta_{\phi}^{N}}{B_{j}}}
\leq
\rbk{4 k_{j} \frac{1}{\sminvtemperature_{c}} \eta_{0}}^{N} \norm{B_{j}}.$$
This bound is uniform in $\phi$ because $\norm{\physham[h]_{\phi, p}}
=
\frac{1}{\sminvtemperature_{c}} \eta_{0}$.
The resummed fiber series is $\fun{\oastate[\psi_{\sminvtemperature,\phi}]}{\prod_{j = 1}^{k} \fun{\tau^{\txtbogoliubov, \phi}_{t_{j}}}{B_{j}}}$.

The holomorphic argument extends this convergence to all real times.
The functions
$$\fun{F_{\Omega}}{z_{1}, \dotsc, z_{k}}
=
\fun{\oastate[\widehat{\psi}_{\sminvtemperature,\Omega}]}{\prod_{j = 1}^{k} \fun{\tau^{\Omega}_{z_{j}}}{B_{j}}}$$
are holomorphic on the polydomain $\oa{D}
=
\set{z
\in
\fldcmp^{k}}{\abs{\opimag z_{j}} < r_{1}}$ and uniformly bounded on its compact subsets,
by the strip bounds of Lemma \ref{lem:derivation-bounds} and $\abs{\fun{\oastate[\widehat{\psi}_{\sminvtemperature,\Omega}]}{X}}
\leq
\norm{X}$.
The short-time estimate gives convergence on the real box $\rbk{- r_{1}, r_{1}}^{k}$.
It follows that Corollary \ref{cor:vitali-several} yields locally uniform convergence on $\oa{D}$,
in particular at every real $\rbk{t_{1}, \dotsc, t_{k}}$,
to a holomorphic limit $F$.
The right side of the theorem,
as a function of complex times,
is likewise holomorphic on $\oa{D}$,
indeed entire,
by the norm-convergent fiber series and dominated convergence in $\phi$,
and agrees with $F$ on the real box.
The difference vanishes on an open real box and is holomorphic on the connected $\oa{D}$.
It therefore vanishes identically by the identity theorem applied one variable at a time,
exactly as in the proof of Corollary \ref{cor:vitali-several}.
Evaluating at real times completes the proof.
\end{proof}

Theorem \ref{thm:green} is \cite[Eq. (1)]{ThirringWalter001}.

\begin{cor}[gauge-charge selection rule]\label{cor:green-charge-selection}
Suppose that each $B_{j}$ has gauge charge $q_{j}
\in
\ringratint$ in the sense
$\fun{\mathfrak{g}_{\vartheta}}{B_{j}}
=
\fnexp{\imunit q_{j} \vartheta} B_{j}$.
Set
$q_{\txttot}
=
\sum_{j = 1}^{k} q_{j}$.
Define the fiber Green function by
$$G_{\phi}
=
\fun{\oastate[\psi_{\sminvtemperature,\phi}]}{
\fun{\tau^{\txtbogoliubov,\phi}_{t_{1}}}{B_{1}}
\dotsm
\fun{\tau^{\txtbogoliubov,\phi}_{t_{k}}}{B_{k}}}.$$
Then the gauge average in \eqref{expedition0025005} is
$$\int_{\mansphere^{1}} G_{\phi} \opdmsr{\msrprb_{\mansphere^{1}}}(\phi)
=
\begin{dcases}
G_{0} & \rbk{q_{\txttot}
=
0}, \\
0 & \rbk{q_{\txttot}
\neq
0}.
\end{dcases}$$
Every local insertion has a finite gauge-charge decomposition by
\eqref{eq:gauge-action}.
Hence only the components with $q_{\txttot}
=
0$ contribute.
\end{cor}

\begin{proof}
The covariance of the fiber dynamics and \eqref{eq:gauge-covariance-states} give
$G_{\phi}
=
\fnexp{\imunit \phi q_{\txttot}} G_{0}$.
Integration over $\mansphere^{1}$ proves the formula.
\end{proof}

\begin{rem}[level of convergence]
Equation \eqref{expedition0025005} is convergence of gauge-averaged Gibbs expectations.
The sectorwise strong convergence of the dynamics is instead Theorem \ref{thm:dynamics}.
Proposition \ref{prop:unitary-failure} rules out weak convergence of the implementing evolution operators even in an aligned sector.
The direct integral in the next section assembles the sectorwise dynamics but does not produce a phase-independent operator limit.
\end{rem}

\section{The Direct Integral Decomposition of the Equilibrium State}\label{sec:decomposition}

The limit Gibbs state \eqref{eq:limit-state} is now decomposed. The model parameters and order-parameter coordinates remain those of the superconducting thermal regime in Definition \ref{def:superconducting-thermal-regime}. First, its GNS representation is identified as a direct integral over the gauge circle. The center of the enveloping von Neumann algebra is then computed and shown to be generated by the phase of the order parameter. These results identify the gauge average in Theorem \ref{thm:limit-state} as the central decomposition of the equilibrium state. Its fibers are mutually disjoint factor states, each KMS for its own Bogoliubov dynamics.

\subsection{Fiber representations}\label{fiber-representations}

All symmetry-breaking product states admit realizations on one fixed Hilbert space. The resulting fiber representations are factors and are mutually disjoint for distinct values of \(\phi\).

The assumptions and notation remain those of Definition \ref{def:superconducting-thermal-regime} and \eqref{eq:order-parameters}.

Fix the GNS triple \(\rbk{\sphilb{H}_{0}, \oarepn_{0}, \oagnsvector[\Psi_{0}]}\) of the fiber \(\oastate[\psi_{\sminvtemperature,0}]\), where \(\oastate[\psi_{\sminvtemperature,0}]\) is the product state \eqref{eq:phase-product-state}. The GNS triple is provided explicitly by Proposition \ref{prop:gns-faithful}. The space \(\sphilb{H}_{0}\) is separable by Proposition \ref{prop:flip-basis}, and \(\oa{M}_{0}
=
\oadoublecommutant{\fun{\oarepn_{0}}{\oa{A}}}\) is a factor by Theorem \ref{thm:factor}. By \eqref{eq:gauge-covariance-states} the triple \(\rbk{\sphilb{H}_{0}, \oarepn_{0} \circ \mathfrak{g}_{\phi}, \oagnsvector[\Psi_{0}]}\) is a GNS triple of \(\oastate[\psi_{\sminvtemperature,\phi}]\). Write \(\oarepn_{\phi}
=
\oarepn_{0} \circ \mathfrak{g}_{\phi}\). Because \(\mathfrak{g}_{\phi}\) is an automorphism of \(\oa{A}\), the image algebras coincide, \begin{equation}\label{eq:same-image}
\fun{\oarepn_{\phi}}{\oa{A}}
=
\fun{\oarepn_{0}}{\oa{A}},
\quad
\oadoublecommutant{\fun{\oarepn_{\phi}}{\oa{A}}}
=
\oa{M}_{0},
\quad
\oacommutant{\fun{\oarepn_{\phi}}{\oa{A}}}
=
\oacommutant{\oa{M}_{0}}
\quad \rbk{\phi
\in
[0, 2 \pi)},
\end{equation} although the representations themselves are mutually disjoint:

\begin{prop}[disjointness of the fibers]\label{prop:fiber-disjoint}
For $\phi
\neq
\phi'$ in $[0, 2 \pi)$ the representations $\oarepn_{\phi}$ and $\oarepn_{\phi'}$ are disjoint:
every bounded $V$ on $\sphilb{H}_{0}$ with $V \fun{\oarepn_{\phi}}{A}
=
\fun{\oarepn_{\phi'}}{A} V$ for all $A$ vanishes.
In particular the states $\oastate[\psi_{\sminvtemperature,\phi}]$ are mutually disjoint factor states.
\end{prop}

\begin{proof}
By \eqref{eq:gauge-action},
$\fun{\oarepn_{\phi}}{M^{+}_{\Omega}}
=
\fnexp{\imunit \phi} \fun{\oarepn_{0}}{M^{+}_{\Omega}}$.
By Theorem \ref{thm:lln} applied to $\oastate[\psi_{\sminvtemperature,0}]$,
whose one-site means are $\fun{\oastate[\psi^{\rbk{1}}_{\sminvtemperature,0}]}{\sigma^{\pm}_{p}}
=
m^{\pm}_{0}
=
\frac{r_{0}}{2}$,
one has $\fun{\oarepn_{0}}{M^{+}_{\Omega}}
\to
\frac{r_{0}}{2}$ strongly as $\Omega \to \infty$.
It follows that $\fun{\oarepn_{\phi}}{M^{+}_{\Omega}}
\to
\frac{r_{0}}{2} \fnexp{\imunit \phi}$ strongly as $\Omega \to \infty$.
Intertwining and taking strong limits on a fixed vector,
$\frac{r_{0}}{2} \fnexp{\imunit \phi} V
=
\frac{r_{0}}{2} \fnexp{\imunit \phi'} V$,
and $r_{0} > 0$ with $\fnexp{\imunit \phi}
\neq
\fnexp{\imunit \phi'}$ forces $V
=
0$.
\end{proof}

\subsection{The direct integral representation}\label{the-direct-integral-representation}

The direct integral over the gauge circle assembles the fiber representations into a cyclic representation. The construction is then identified with the GNS representation of the averaged equilibrium state.

\begin{defn}[direct integral]\label{def:direct-integral}
Let $\sphilb{H}_{\sminvtemperature}
=
\fun{\lp^{2}}{\mansphere^{1}, \msrprb_{\mansphere^{1}}. \sphilb{H}_{0}}$ be the Hilbert space of square-integrable $\sphilb{H}_{0}$-valued measurable functions,
with inner product $\bkt{\xi}{\zeta}
=
\int_{\mansphere^{1}} \bkt{\fun{\xi}{\phi}}{\fun{\zeta}{\phi}} \opdmsr{\msrprb_{\mansphere^{1}}}(\phi)$.
Define
$$\funrbk{\fun{\oarepn_{\sminvtemperature}}{A} \xi}{\phi}
=
\fun{\oarepn_{\phi}}{A} \fun{\xi}{\phi}
\quad \rbk{A
\in
\oa{A}},
\quad
\fun{\oagnsvector[\Psi_{\sminvtemperature}]}{\phi}
=
\oagnsvector[\Psi_{0}].$$
A bounded operator $X$ on $\sphilb{H}_{\sminvtemperature}$ is called decomposable when there is a bounded measurable family $\rbk{X_{\phi}}_{0
\leq
\phi
<
2 \pi}$ with $\rbk{X \xi}{\phi}
=
X_{\phi} \fun{\xi}{\phi}$ almost everywhere,
and diagonal when $X_{\phi}
=
\fun{z}{\phi} 1$ for a scalar function $z
\in
\fun{\lp^{\infty}}{\mansphere^{1}, \msrprb_{\mansphere^{1}}}$.
For such $z$,
let
$M_{z}
\in
\opspbddlin{\fun{\lp^{2}}{\mansphere^{1}, \msrprb_{\mansphere^{1}}}}$
be the multiplication operator defined by
$\fun{\rbk{M_{z} g}}{\phi}
=
\fun{z}{\phi} \fun{g}{\phi}$ for
$g
\in
\fun{\lp^{2}}{\mansphere^{1}, \msrprb_{\mansphere^{1}}}$
and almost every $\phi
\in
\mansphere^{1}$.
The algebra of diagonal operators is denoted $\oa{Z}_{\sminvtemperature}$,
and the diagonal operator associated with $z$ is
$M_{z} \otimes \id_{\sphilb{H}_{0}}$.
\end{defn}

For \(A
\in
\oa{A}\) the map \(\phi
\mapsto
\fun{\oarepn_{\phi}}{A} \fun{\xi}{\phi}\) is measurable: for local \(A\) the map \(\phi
\mapsto
\fun{\mathfrak{g}_{\phi}}{A}\) is norm continuous by Section \ref{sec:algebra}. It follows that \(\phi
\mapsto
\fun{\oarepn_{0}}{\fun{\mathfrak{g}_{\phi}}{A}}\) is continuous in norm, and measurability persists under norm limits and multiplication by measurable vector functions. The uniform bound \(\norm{\fun{\oarepn_{\phi}}{A}}
\leq
\norm{A}\) makes \(\fun{\oarepn_{\sminvtemperature}}{A}\) a bounded decomposable operator, and \(\oarepn_{\sminvtemperature}\) is a \(\ast\)-representation.

\begin{thm}[GNS identification]\label{thm:gns-identification}
The triple $\rbk{\sphilb{H}_{\sminvtemperature}, \oarepn_{\sminvtemperature}, \oagnsvector[\Psi_{\sminvtemperature}]}$ is the GNS triple of the limit Gibbs state $\oastate[\psi_{\sminvtemperature}]$ of Theorem \ref{thm:limit-state}.
\end{thm}

\begin{proof}
The expectation values are correct:
$$\bkt{\oagnsvector[\Psi_{\sminvtemperature}]}{\fun{\oarepn_{\sminvtemperature}}{A} \oagnsvector[\Psi_{\sminvtemperature}]}
=
\int_{\mansphere^{1}} \bkt{\oagnsvector[\Psi_{0}]}{\fun{\oarepn_{\phi}}{A} \oagnsvector[\Psi_{0}]} \opdmsr{\msrprb_{\mansphere^{1}}}(\phi)
=
\int_{\mansphere^{1}} \fun{\oastate[\psi_{\sminvtemperature,\phi}]}{A} \opdmsr{\msrprb_{\mansphere^{1}}}(\phi)
=
\fun{\oastate[\psi_{\sminvtemperature}]}{A}.$$
It remains to prove that $\oagnsvector[\Psi_{\sminvtemperature}]$ is cyclic.
First,
the order parameter provides the diagonal phases.
By Proposition \ref{prop:fiber-disjoint} and its proof,
$\fun{\oarepn_{\phi}}{\frac{2}{r_{0}} M^{+}_{\Omega}}
\to
\fnexp{\imunit \phi}$ strongly on $\sphilb{H}_{0}$ as $\Omega \to \infty$ for each $\phi$,
with the uniform bound $\norm{\frac{2}{r_{0}} M^{+}_{\Omega}}
\leq
\frac{2}{r_{0}}$.
For $\xi
\in
\sphilb{H}_{\sminvtemperature}$ dominated convergence in $\phi$ gives
\begin{equation}\label{eq:order-parameter-limit}
\slim_{\Omega \to \infty}
\fun{\oarepn_{\sminvtemperature}}{\frac{2}{r_{0}} M^{+}_{\Omega}}
=
u
=
M_{\fnexp{\imunit \phi}} \otimes \id_{\sphilb{H}_{0}}
\quad \text{on $\sphilb{H}_{\sminvtemperature}$}.
\end{equation}
Iterating with Lemma \ref{lem:strong-products},
$$\slim_{\Omega \to \infty}
\fun{\oarepn_{\sminvtemperature}}{\rbk{\frac{2}{r_{0}} M^{+}_{\Omega}}^{j} \rbk{\frac{2}{r_{0}} M^{-}_{\Omega}}^{l} A}
=
u^{j}
\rbk{\faadj{u}}^{l}
\fun{\oarepn_{\sminvtemperature}}{A}$$
for all $j, l
\geq
0$ and $A
\in
\oa{A}$.
It follows that the closure of $\fun{\oarepn_{\sminvtemperature}}{\oa{A}} \oagnsvector[\Psi_{\sminvtemperature}]$ contains $\fun{f}{\phi}$-modulated vectors
$\rbk{M_{f} \otimes \id_{\sphilb{H}_{0}}}
\fun{\oarepn_{\sminvtemperature}}{A}
\oagnsvector[\Psi_{\sminvtemperature}]$ for every trigonometric polynomial $f$ and every $A
\in
\oa{A}$.
Now approximate an arbitrary $\xi
\in
\sphilb{H}_{\sminvtemperature}$.
Sections of the form $\phi
\mapsto
\sum_{i
=
1}^{I} \fun{1_{E_{i}}}{\phi} \zeta_{i}$,
where $I
\in
\semigrposint$ and $\zeta_{i}
\in
\sphilb{H}_{0}$ are dense by the definition of the Bochner space,
and each such section is $\lp^{2}$-approximated by continuous sections via regularization of the indicators.
It is therefore enough to treat a continuous section $\xi$.
Fix $\epsilon > 0$.
For each $\theta
\in
[0, 2 \pi]$ choose $A_{\theta}
\in
\oa{A}$ with $\norm{\fun{\oarepn_{\theta}}{A_{\theta}} \oagnsvector[\Psi_{0}] - \fun{\xi}{\theta}} < \epsilon$,
possible since $\oagnsvector[\Psi_{0}]$ is cyclic for the GNS representation $\oarepn_{\theta}$.
By the norm continuity of $\phi
\mapsto
\fun{\mathfrak{g}_{\phi}}{A_{\theta}}$ and the continuity of $\xi$ there is an open arc $U_{\theta} \ni \theta$ on which
$$\norm{\fun{\oarepn_{\phi}}{A_{\theta}} \oagnsvector[\Psi_{0}] - \fun{\xi}{\phi}}
\leq
\norm{\fun{\mathfrak{g}_{\phi}}{A_{\theta}} - \fun{\mathfrak{g}_{\theta}}{A_{\theta}}} + \norm{\fun{\oarepn_{\theta}}{A_{\theta}} \oagnsvector[\Psi_{0}] - \fun{\xi}{\theta}} + \norm{\fun{\xi}{\theta} - \fun{\xi}{\phi}} < 3 \epsilon.$$
By compactness finitely many arcs $U_{\theta_{1}}, \dotsc, U_{\theta_{J}}$ cover the circle.
Choose a finite partition $\seq{I_{i}}{i
=
1}^{I}$ of $[0, 2 \pi)$ into half-open intervals,
each contained in one covering arc $U_{\theta_{j \rbk{i}}}$,
and set $A_{i}
=
A_{\theta_{j \rbk{i}}}$.
The section $\zeta
=
\sum_{i
=
1}^{I} 1_{I_{i}} \cdot \rbk{\phi
\mapsto
\fun{\oarepn_{\phi}}{A_{i}} \oagnsvector[\Psi_{0}]}$ then satisfies $\norm{\xi - \zeta}_{\sphilb{H}_{\sminvtemperature}}
\leq
3 \epsilon$.

It remains to show that $\zeta$ lies in the closure of
$\fun{\oarepn_{\sminvtemperature}}{\oa{A}} \oagnsvector[\Psi_{\sminvtemperature}]$.
Lemma \ref{lem:fejer}(2)--(3) provides trigonometric polynomials
$f_{i}
=
\operatorname{Fej}_{N_{i}} 1_{I_{i}}$ with $\norm{f_{i}}_{\infty}
\leq
1$
that converge to $1_{I_i}$ in
$\fun{\lp^{2}}{\mansphere^{1}, \msrprb_{\mansphere^{1}}}$.
The vectors
$\sum_{i
=
1}^{I}
\rbk{M_{f_{i}} \otimes \id_{\sphilb{H}_{0}}}
\fun{\oarepn_{\sminvtemperature}}
{A_{i}} \oagnsvector[\Psi_{\sminvtemperature}]$
belong to the required closure.
Their distance from $\zeta$ is at most
$$\sum_{i
=
1}^{I} \norm{f_{i} - 1_{I_{i}}}_{\lp^{2}} \cdot \sup_{\phi
\in
[0, 2 \pi)} \norm{\fun{\oarepn_{\phi}}{A_{i}} \oagnsvector[\Psi_{0}]}.$$
This bound tends to zero.
The uniqueness clause of the GNS construction identifies this cyclic representation with the GNS representation of $\oastate[\psi_{\sminvtemperature}]$.
\end{proof}

\subsection{The center and the central decomposition}\label{the-center-and-the-central-decomposition}

The center of the direct-integral von Neumann algebra is generated by the phase of the order parameter. Its computation identifies the gauge average as the central decomposition into factor states. The auxiliary decomposable-operator lemmas used in this subsection are collected in Appendix \ref{sec:functional-analytic-appendix}.

\begin{thm}[the center is the order parameter algebra]\label{thm:center}
Let $\oa{M}_{\sminvtemperature}
=
\oadoublecommutant{\fun{\oarepn_{\sminvtemperature}}{\oa{A}}}$.
The center is precisely the order-parameter algebra:
$$\fun{Z}{\oa{M}_{\sminvtemperature}}
=
\oa{M}_{\sminvtemperature} \cap \oacommutant{\oa{M}_{\sminvtemperature}}
=
\oa{Z}_{\sminvtemperature}
\cong
\fun{\lp^{\infty}}{\mansphere^{1}, \msrprb_{\mansphere^{1}}},$$
and the unitary $u
=
M_{\fnexp{\imunit \phi}} \otimes \id_{\sphilb{H}_{0}}$ generating $\oa{Z}_{\sminvtemperature}$ is the strong limit \eqref{eq:order-parameter-limit} of the normalized order parameters $\fun{\oarepn_{\sminvtemperature}}{\frac{2}{r_{0}} M^{+}_{\Omega}}$.
\end{thm}

\begin{proof}
We first prove that the order-parameter algebra is contained in the center:
$\oa{Z}_{\sminvtemperature}
\subset
\fun{Z}{\oa{M}_{\sminvtemperature}}$.
The operator $u$ is the strong limit of elements of $\fun{\oarepn_{\sminvtemperature}}{\oa{A}}$ by \eqref{eq:order-parameter-limit},
and hence in $\oa{M}_{\sminvtemperature}$.
It is unitary,
and it is diagonal with scalar fibers.
It therefore commutes with every $\fun{\oarepn_{\sminvtemperature}}{A}$.
It follows that $u
\in
\fun{Z}{\oa{M}_{\sminvtemperature}}$.

We next show that the von Neumann algebra $\oawstar \setone{u}$ generated by $u$ contains every diagonal operator.
Let $T
\in
\oacommutant{\setone{u, \faadj{u}}}$ and $\xi, \zeta
\in
\sphilb{H}_{\sminvtemperature}$,
and consider the two integrable functions $\fun{k_{1}}{\phi}
=
\bkt{\fun{\rbk{\faadj{T} \xi}}{\phi}}{\fun{\zeta}{\phi}}$ and $\fun{k_{2}}{\phi}
=
\bkt{\fun{\xi}{\phi}}{\fun{\rbk{T \zeta}}{\phi}}$.
For every $j
\in
\ringratint$ the operator $u^{j}$ is the diagonal
$M_{\fnexp{\imunit j \phi}} \otimes \id_{\sphilb{H}_{0}}$,
and commutation with $T$ gives
$$\int_{\mansphere^{1}} \fnexp{\imunit j \phi} \fun{k_{1}}{\phi} \opdmsr{\msrprb_{\mansphere^{1}}}(\phi)
=
\bkt{\faadj{T} \xi}{u^{j} \zeta}
=
\bkt{\xi}{u^{j} T \zeta}
=
\int_{\mansphere^{1}} \fnexp{\imunit j \phi} \fun{k_{2}}{\phi} \opdmsr{\msrprb_{\mansphere^{1}}}(\phi).$$
It follows that all Fourier coefficients of $k_{1} - k_{2}
\in
\fun{\lp^{1}}{\mansphere^{1}, \msrprb_{\mansphere^{1}}}$ vanish and $k_{1}
=
k_{2}$ almost everywhere by Lemma \ref{lem:fejer}(3).
It holds that
$$\bkt{\xi}{\rbk{T \rbk{M_{z} \otimes \id_{\sphilb{H}_{0}}}
-\rbk{M_{z} \otimes \id_{\sphilb{H}_{0}}} T} \zeta}
=
\int_{\mansphere^{1}} z \rbk{k_{2} - k_{1}} \opdmsr{\msrprb_{\mansphere^{1}}}(\phi)
=
0$$ for every $z
\in
\fun{\lp^{\infty}}{\mansphere^{1}, \msrprb_{\mansphere^{1}}}$ and all $\xi, \zeta$,
that is,
$M_{z} \otimes \id_{\sphilb{H}_{0}}
\in
\oadoublecommutant{\setone{u, \faadj{u}}}
=
\oawstar \setone{u}$ by the bicommutant theorem.
Since $u$ is diagonal,
every element of $\oawstar \setone{u}$ is diagonal.
Thus $\oawstar \setone{u}
=
\oa{Z}_{\sminvtemperature}$.
The centrality of $u$ now gives $\oa{Z}_{\sminvtemperature}
\subset
\fun{Z}{\oa{M}_{\sminvtemperature}}$ because the center is a von Neumann algebra containing $u$.

We now prove the reverse inclusion:
$\fun{Z}{\oa{M}_{\sminvtemperature}}
\subset
\oa{Z}_{\sminvtemperature}$.
Let $Z
\in
\fun{Z}{\oa{M}_{\sminvtemperature}}$.
Since $\oa{Z}_{\sminvtemperature}
\subset
\oa{M}_{\sminvtemperature}$ and $Z
\in
\oacommutant{\oa{M}_{\sminvtemperature}}$,
the operator $Z$ commutes with all diagonal operators.
Lemma \ref{lem:diagonal-commutant} shows that $Z$ is decomposable,
$Z
=
\rbk{Z_{\phi}}$.

We first show that almost every fiber $Z_{\phi}$ belongs to
$\oacommutant{\oa{M}_{0}}$.
For a countable norm-dense set $\set{A_{i}}{i
\in
\semigrposint}
\subset
\oa{A}$,
the commutators $\commutator{Z}{\fun{\oarepn_{\sminvtemperature}}{A_{i}}}$ vanish.
Lemma \ref{lem:fiberwise-zero} provides a null set off which $\commutator{Z_{\phi}}{\fun{\oarepn_{\phi}}{A_{i}}}
=
0$ for all $i$.
For fixed such $\phi$ the set $\set{\fun{\mathfrak{g}_{\phi}}{A_{i}}}{i
\in
\semigrposint}$ is norm dense.
It follows that $Z_{\phi}
\in
\oacommutant{\fun{\oarepn_{0}}{\oa{A}}}
=
\oacommutant{\oa{M}_{0}}$ by \eqref{eq:same-image}.

It remains to show that almost every fiber $Z_{\phi}$ also belongs to
$\oa{M}_{0}$.
Choose a sequence $\seq{c_{j}}{j
\geq
1}$ in the unit ball of $\oacommutant{\oa{M}_{0}}$,
dense in that ball for the weak operator topology.
The unit ball of $\opspbddlin{\sphilb{H}_{0}}$ is compact in the weak operator topology by a Tychonoff argument on matrix elements.
Since $\sphilb{H}_{0}$ is separable,
this topology on the unit ball is induced by the metric
$$\fun{d}{S, S'}
=
\sum_{i, l
\geq
1} 2^{- i - l} \abs{\bkt{e_{i}}{\rbk{S - S'} e_{l}}},$$
where $\seq{e_{i}}{i
\geq
1}$ is an orthonormal basis.
The unit ball of $\oacommutant{\oa{M}_{0}}$ is a closed subset of this compact metrizable space.
It is therefore compact metrizable and separable,
which proves the existence of the sequence.
Each constant field $C_{j}
=
\rbk{c_{j}}_{\phi}$ is a bounded decomposable operator,
and $C_{j}
\in
\oacommutant{\fun{\oarepn_{\sminvtemperature}}{\oa{A}}}$ because $c_{j}$ commutes with $\fun{\oarepn_{\phi}}{A}
\in
\fun{\oarepn_{0}}{\oa{A}}$ for every $\phi$ and $A$,
again by \eqref{eq:same-image}.
It follows that $Z
\in
\oa{M}_{\sminvtemperature}
\subset
\oadoublecommutant{\fun{\oarepn_{\sminvtemperature}}{\oa{A}}}
\subset
\oacommutant{\set{C_{j}}{j
\in
\semigrposint}}$,
and Lemma \ref{lem:fiberwise-zero} gives a null set off which $\commutator{Z_{\phi}}{c_{j}}
=
0$ for all $j$.
Commutation passes to weak operator limits.
It follows that $Z_{\phi}$ commutes with the whole ball of $\oacommutant{\oa{M}_{0}}$,
whence $Z_{\phi}
\in
\oadoublecommutant{\oa{M}_{0}}
=
\oa{M}_{0}$.
Therefore $Z_{\phi}
\in
\oa{M}_{0} \cap \oacommutant{\oa{M}_{0}}
=
\fldcmp 1$ almost everywhere,
by Theorem \ref{thm:factor}.
Writing $Z_{\phi}
=
\fun{z}{\phi} 1$,
the function $z$ is a measurable matrix element.
It is bounded by $\norm{Z}$.
It follows that $Z
=
M_{z} \otimes \id_{\sphilb{H}_{0}}
\in
\oa{Z}_{\sminvtemperature}$.
\end{proof}

\begin{prop}[KMS property of the fibers]\label{prop:fiber-kms}
Each $\oastate[\psi_{\sminvtemperature,\phi}]$ is a KMS state at inverse temperature $\sminvtemperature$ for the Bogoliubov dynamics $\tau^{\txtbogoliubov, \phi}$:
for all $A, B
\in
\oa{A}_{\txtloc}$,
$$\fun{\oastate[\psi_{\sminvtemperature,\phi}]}{A \fun{\tau^{\txtbogoliubov, \phi}_{t + \imunit \sminvtemperature}}{B}}
=
\fun{\oastate[\psi_{\sminvtemperature,\phi}]}{\fun{\tau^{\txtbogoliubov, \phi}_{t}}{B} A}
\quad \rbk{t
\in
\fldreal},$$
and the KMS condition in the sense of \cite[Definition 5.3.1, Proposition 5.3.7]{BratteliRobinson004} holds.
\end{prop}

\begin{proof}
Local elements are entire analytic for the product dynamics because the generator is bounded on each $\oa{A}_{\Lambda}$.
It follows that the boundary form of the KMS condition on the norm-dense $\ast$-subalgebra $\oa{A}_{\txtloc}$ of analytic elements implies the KMS property \cite[Proposition 5.3.7]{BratteliRobinson004}.
Fix a finite $\Lambda$ containing the supports of $A$ and $B$.
On $\oa{A}_{\Lambda}$,
$\oastate[\psi_{\sminvtemperature,\phi}]$ restricts to the Gibbs state of
$\physham^{\phi}_{\Lambda}
=
\sum_{p
\in
\Lambda} \physham[h]_{\phi, p}$
at inverse temperature $\sminvtemperature$.
Indeed,
\eqref{eq:self-consistency} gives its density matrix as
$$D_{\phi}^{\otimes \Lambda}
=
\frac{1}{\sqfun{\trace}{\fnexp{- \sminvtemperature \physham^{\phi}_{\Lambda}}}}
\fnexp{- \sminvtemperature \physham^{\phi}_{\Lambda}}.$$
and the complex-time evolution is
$$\fun{\tau^{\txtbogoliubov, \phi}_{z}}{B}
=
\fnexp{\imunit z \physham^{\phi}_{\Lambda}} B \fnexp{- \imunit z \physham^{\phi}_{\Lambda}}.$$
This evolution stays in $\oa{A}_{\Lambda}$ for every complex $z$.
Cyclicity of the trace gives
$$\begin{aligned}
&\sqfun{\trace}
{\fnexp{- \sminvtemperature \physham^{\phi}_{\Lambda}} A \fnexp{- \sminvtemperature \physham^{\phi}_{\Lambda}} \fnexp{\imunit t \physham^{\phi}_{\Lambda}} B \fnexp{- \imunit t \physham^{\phi}_{\Lambda}} \fnexp{\sminvtemperature \physham^{\phi}_{\Lambda}}}
\\ 
&=
\sqfun{\trace}
{\fnexp{- \sminvtemperature \physham^{\phi}_{\Lambda}} \fnexp{\imunit t \physham^{\phi}_{\Lambda}} B \fnexp{- \imunit t \physham^{\phi}_{\Lambda}} A},
\end{aligned}$$
After normalization,
the left side is $\fun{\oastate[\psi_{\sminvtemperature,\phi}]}{A \fun{\tau^{\txtbogoliubov, \phi}_{t + \imunit \sminvtemperature}}{B}}$ by the definition of the analytic continuation.
This proves the desired identity.
\end{proof}

\begin{prop}[phase-extended KMS state]\label{prop:phase-extended-kms}
On the phase-extended algebra
$\oa{C}$
in \eqref{eq:phase-extended-algebra},
the formula
\begin{equation}\label{eq:phase-extended-thermal-dynamics}
\funrbk{\fun{\widetilde{\tau}^{\sminvtemperature}_{t}}{F}}{\phi}
=
\fun{\tau^{\txtbogoliubov,\phi}_{t}}{\fun{F}{\phi}}
\end{equation}
defines a point-norm continuous automorphism group.
The state
\begin{equation}\label{eq:phase-extended-kms-state}
\fun{\oastate[\widetilde{\psi}_{\sminvtemperature}]}{F}
=
\int_{\mansphere^{1}}
\fun{\oastate[\psi_{\sminvtemperature,\phi}]}{\fun{F}{\phi}}
\opdmsr{\msrprb_{\mansphere^{1}}}(\phi)
\end{equation}
is a KMS state at inverse temperature
$\sminvtemperature$
for
$\widetilde{\tau}^{\sminvtemperature}$.
\end{prop}

\begin{proof}
The point-norm argument in the proof of
Proposition \ref{prop:phase-extended-ground-dynamics}
applies to the one-site Hamiltonians
$\physham[h]_{\phi,p}$
from \eqref{eq:self-consistency} and proves that
\eqref{eq:phase-extended-thermal-dynamics}
defines a $\oacstar$-dynamics.

The algebra
$\fun{\conti}{\mansphere^{1}} \odot \oa{A}_{\txtloc}$
is a norm-dense
$\ast$-subalgebra of entire analytic elements.
For two elements
$F$
and
$G$
in this subalgebra,
Proposition \ref{prop:fiber-kms} gives the boundary identity
$$\begin{aligned}
\fun{\oastate[\psi_{\sminvtemperature,\phi}]}
{\fun{F}{\phi}
\fun{\tau^{\txtbogoliubov,\phi}_{t + \imunit \sminvtemperature}}{\fun{G}{\phi}}}
=
\fun{\oastate[\psi_{\sminvtemperature,\phi}]}
{\fun{\tau^{\txtbogoliubov,\phi}_{t}}{\fun{G}{\phi}}
\fun{F}{\phi}}
\end{aligned}$$
for every
$\phi$.
Integration over
$\mansphere^{1}$
proves the KMS boundary identity for the state
\eqref{eq:phase-extended-kms-state}.
The analytic-element criterion in
\cite[Proposition 5.3.7]{BratteliRobinson004}
extends it to
$\oa{C}$.
\end{proof}

The fibers are moreover the only equilibrium states of their own dynamics. This strengthens the KMS statement: the fibers are mutually disjoint factor states, and each is the unique KMS state of its own Bogoliubov dynamics.

\begin{prop}[uniqueness of the fiber KMS states]\label{prop:kms-unique}
For each $\phi$,
the state $\oastate[\psi_{\sminvtemperature,\phi}]$ is the only KMS state of $\rbk{\oa{A}, \tau^{\txtbogoliubov, \phi}}$ at inverse temperature $\sminvtemperature$.
In particular $\oastate[\psi_{\sminvtemperature,\phi}]$ is an extremal KMS state.
\end{prop}

\begin{proof}
Let $\oastate[\psi]$ be a KMS state at $\sminvtemperature$ and fix a finite $\Lambda$.
The dynamics preserves $\oa{A}_{\Lambda}$ and acts there by conjugation with $\fnexp{\imunit t \physham^{\phi}_{\Lambda}}$,
$\physham^{\phi}_{\Lambda}
=
\sum_{p
\in
\Lambda} \physham[h]_{\phi, p}$,
with entire continuation $\fun{\tau^{\txtbogoliubov, \phi}_{\imunit \sminvtemperature}}{B}
=
\fnexp{- \sminvtemperature \physham^{\phi}_{\Lambda}} B \fnexp{\sminvtemperature \physham^{\phi}_{\Lambda}}$ for $B
\in
\oa{A}_{\Lambda}$.
The restriction of $\oastate[\psi]$ to the finite-dimensional $\oa{A}_{\Lambda}$ has a density matrix $\rho_{\Lambda}$ with respect to the matrix trace,
and the KMS boundary identity at $t
=
0$ reads
$$\sqfun{\trace}{\rho_{\Lambda} A \fnexp{- \sminvtemperature \physham^{\phi}_{\Lambda}} B \fnexp{\sminvtemperature \physham^{\phi}_{\Lambda}}}
=
\sqfun{\trace}{\rho_{\Lambda} B A}
\quad \rbk{A, B
\in
\oa{A}_{\Lambda}}.$$
Cyclicity of the trace turns this into $\sqfun{\trace}{\rbk{\fnexp{\sminvtemperature \physham^{\phi}_{\Lambda}} \rho_{\Lambda} A \fnexp{- \sminvtemperature \physham^{\phi}_{\Lambda}} - A \rho_{\Lambda}} B}
=
0$ for all $B$,
and the nondegeneracy of the finite-dimensional trace gives the KMS identity
$$\fnexp{\sminvtemperature \physham^{\phi}_{\Lambda}} \rho_{\Lambda} A \fnexp{- \sminvtemperature \physham^{\phi}_{\Lambda}}
=
A \rho_{\Lambda}
\quad \rbk{A
\in
\oa{A}_{\Lambda}}.$$
With $\sigma
=
\fnexp{\sminvtemperature \physham^{\phi}_{\Lambda}} \rho_{\Lambda}$ this says $\sigma \rbk{A \fnexp{- \sminvtemperature \physham^{\phi}_{\Lambda}}}
=
\rbk{A \fnexp{- \sminvtemperature \physham^{\phi}_{\Lambda}}} \sigma$,
and the mapping $A
\mapsto
A \fnexp{- \sminvtemperature \physham^{\phi}_{\Lambda}}$ is a bijection of $\oa{A}_{\Lambda}$.
It follows that $\sigma$ commutes with the full matrix algebra and is a scalar.
Normalizing yields
$\rho_{\Lambda}
=
\fnexp{- \sminvtemperature \physham^{\phi}_{\Lambda}} / \sqfun{\trace}{\fnexp{- \sminvtemperature \physham^{\phi}_{\Lambda}}}$,
which is the density matrix of $\fnrestr{\oastate[\psi_{\sminvtemperature,\phi}]}{\oa{A}_{\Lambda}}$ by \eqref{eq:self-consistency}.
The marginals of $\oastate[\psi]$ and $\oastate[\psi_{\sminvtemperature,\phi}]$ agree on every $\oa{A}_{\Lambda}$.
They agree on the dense subalgebra $\oa{A}_{\txtloc}$,
and continuity extends the equality to $\oa{A}$.
Extremality follows since the KMS states at a fixed inverse temperature form a convex set of which $\oastate[\psi_{\sminvtemperature,\phi}]$ is the only element.
\end{proof}

\begin{thm}[central decomposition of the equilibrium state]\label{thm:central-decomposition}
Assume the superconducting thermal regime of
Definition \ref{def:superconducting-thermal-regime}.
The limit Gibbs state $\oastate[\psi_{\sminvtemperature}]$ of the degenerate BCS model admits the direct integral decomposition
$$\oastate[\psi_{\sminvtemperature}]
=
\int_{\mansphere^{1}} \oastate[\psi_{\sminvtemperature,\phi}] \opdmsr{\msrprb_{\mansphere^{1}}}(\phi),$$
with the following properties.
\begin{enumerate}
\item The GNS representation of $\oastate[\psi_{\sminvtemperature}]$ is the direct integral $\sphilb{H}_{\sminvtemperature}
=
\int^{\oplus}_{\mansphere^{1}} \sphilb{H}_{0} \opdmsr{\msrprb_{\mansphere^{1}}}(\phi)$ with $\oarepn_{\sminvtemperature}
=
\int^{\oplus}_{\mansphere^{1}} \oarepn_{\phi} \opdmsr{\msrprb_{\mansphere^{1}}}(\phi)$,
in the sense of Theorem \ref{thm:gns-identification}.
\item The center of $\oa{M}_{\sminvtemperature}
=
\oadoublecommutant{\fun{\oarepn_{\sminvtemperature}}{\oa{A}}}$ is exactly the algebra $\fun{\lp^{\infty}}{\mansphere^{1}, \msrprb_{\mansphere^{1}}}$ of diagonal operators,
generated by the unitary phase $u$ of the order parameter,
which is the strong limit of the local order parameters $\fun{\oarepn_{\sminvtemperature}}{\frac{2}{r_{0}} M^{+}_{\Omega}}$.
The decomposition is the central decomposition of $\oastate[\psi_{\sminvtemperature}]$.
The states $\oastate[\psi_{\sminvtemperature,\phi}]$ are obtained from the minimal projections of the center in the following sense.
For every measurable $E
\subset
\mansphere^{1}$ of positive measure,
$$\frac{
\bkt{\oagnsvector[\Psi_{\sminvtemperature}]}{
\rbk{M_{1_{E}} \otimes \id_{\sphilb{H}_{0}}}
\fun{\oarepn_{\sminvtemperature}}{A}
\oagnsvector[\Psi_{\sminvtemperature}]}
}{
\bkt{\oagnsvector[\Psi_{\sminvtemperature}]}{
\rbk{M_{1_{E}} \otimes \id_{\sphilb{H}_{0}}}
\oagnsvector[\Psi_{\sminvtemperature}]}
}
=
\frac{1}{\fun{\msrprb_{\mansphere^{1}}}{E}} \int_{E} \fun{\oastate[\psi_{\sminvtemperature,\phi}]}{A} \opdmsr{\msrprb_{\mansphere^{1}}}(\phi).$$
\item The fiber states $\oastate[\psi_{\sminvtemperature,\phi}]$ are mutually disjoint factor states,
invariant under their own Bogoliubov dynamics $\tau^{\txtbogoliubov, \phi}$,
and the unique extremal KMS states of $\tau^{\txtbogoliubov, \phi}$ at inverse temperature
$\sminvtemperature$.
The thermal Green functions of the model converge to the corresponding fiber averages.
\item The gauge group acts covariantly:
$\oastate[\psi_{\sminvtemperature,\phi}] \circ \mathfrak{g}_{\vartheta}
=
\oastate[\psi_{\sminvtemperature,\phi + \vartheta}]$,
the state $\oastate[\psi_{\sminvtemperature}]$ is gauge invariant,
the automorphisms $\mathfrak{g}_{\vartheta}$ are unitarily implemented on $\sphilb{H}_{\sminvtemperature}$ by $\rbk{U_{\vartheta} \xi}{\phi}
=
\fun{\xi}{\phi + \vartheta}$ with $U_{\vartheta} \oagnsvector[\Psi_{\sminvtemperature}]
=
\oagnsvector[\Psi_{\sminvtemperature}]$,
and the induced action on the center is the rotation of the circle.
In each fiber the gauge symmetry is broken,
$\fun{\oastate[\psi_{\sminvtemperature,\phi}]}{\sigma^{+}_{p}}
=
\frac{r_{0}}{2} \fnexp{\imunit \phi}
\neq
0$.
\item For $\sminvtemperature
\leq
\sminvtemperature_{c}$ the limit Gibbs state is the factor state $\oastate[\psi_{\sminvtemperature,\ast}]$,
its GNS representation has trivial center,
and the decomposition is trivial.
\end{enumerate}
\end{thm}

\begin{proof}

(1)
The first assertion is Theorem \ref{thm:gns-identification}.

(2)
The second assertion up to the formula is Theorem \ref{thm:center}.
The formula itself is the computation
$$\bkt{\oagnsvector[\Psi_{\sminvtemperature}]}{
\rbk{M_{1_{E}} \otimes \id_{\sphilb{H}_{0}}}
\fun{\oarepn_{\sminvtemperature}}{A}
\oagnsvector[\Psi_{\sminvtemperature}]}
=
\int_{E} \fun{\oastate[\psi_{\sminvtemperature,\phi}]}{A} \opdmsr{\msrprb_{\mansphere^{1}}}(\phi),
\quad
\bkt{\oagnsvector[\Psi_{\sminvtemperature}]}{
\rbk{M_{1_{E}} \otimes \id_{\sphilb{H}_{0}}}
\oagnsvector[\Psi_{\sminvtemperature}]}
=
\fun{\msrprb_{\mansphere^{1}}}{E},$$
both immediate from Definition \ref{def:direct-integral}.
That the decomposition deserves the name central is the combination of these formulas with (2).
The states subordinate to $\oastate[\psi_{\sminvtemperature}]$ through positive elements of the center are exactly the averages of the $\oastate[\psi_{\sminvtemperature,\phi}]$ over measurable sets.
More generally,
if $z
\in
\fun{\lp^{\infty}}{\mansphere^{1}, \msrprb_{\mansphere^{1}}}$ is nonnegative and
$\int_{\mansphere^{1}} z \opdmsr{\msrprb_{\mansphere^{1}}}(\phi)
=
1$,
the corresponding subordinate state is
$$A
\mapsto
\int_{\mansphere^{1}} z\rbk{\phi} \fun{\oastate[\psi_{\sminvtemperature,\phi}]}{A}
\opdmsr{\msrprb_{\mansphere^{1}}}(\phi).$$
These states realize the decomposition of $\oastate[\psi_{\sminvtemperature}]$ along $\fun{Z}{\oa{M}_{\sminvtemperature}}
=
\fun{\lp^{\infty}}{\mansphere^{1}, \msrprb_{\mansphere^{1}}}$ with fiber states $\oastate[\psi_{\sminvtemperature,\phi}]$,
in the sense of the decomposition theory of \cite[Section 4.2]{BratteliRobinson003}.

(3)
The third assertion collects Proposition \ref{prop:fiber-disjoint},
Theorem \ref{thm:factor},
Propositions \ref{prop:fiber-kms} and \ref{prop:kms-unique},
and Theorem \ref{thm:green}.
Invariance of $\oastate[\psi_{\sminvtemperature,\phi}]$ under its dynamics follows from the KMS property at $t$ real,
or directly from the invariance of $\oastate[\psi^{\rbk{1}}_{\sminvtemperature,\phi}]$ under $\fun{\Ad}{\fnexp{\imunit t \physham[h]_{\phi}}}$.

(4)
The covariance and invariance are \eqref{eq:gauge-covariance-states} and Theorem \ref{thm:limit-state}.
For the implementation,
$U_{\vartheta}$ is unitary because the measure is rotation invariant.
It fixes the constant section $\oagnsvector[\Psi_{\sminvtemperature}]$.
Its covariance is expressed by
$$\funrbk{U_{\vartheta} \fun{\oarepn_{\sminvtemperature}}{A} \xi}{\phi}
=
\fun{\oarepn_{\phi + \vartheta}}{A} \fun{\xi}{\phi + \vartheta}
=
\funrbk{\fun{\oarepn_{\sminvtemperature}}{\fun{\mathfrak{g}_{\vartheta}}{A}} U_{\vartheta} \xi}{\phi},$$
where $\oarepn_{\phi + \vartheta}
=
\oarepn_{\phi} \circ \mathfrak{g}_{\vartheta}$.
On the center,
$U_{\vartheta}
\rbk{M_{z} \otimes \id_{\sphilb{H}_{0}}}
\faadj{U_{\vartheta}}
=
M_{z \rbk{\cdot + \vartheta}} \otimes \id_{\sphilb{H}_{0}}$ by the same computation,
the rotation of the circle.
Symmetry breaking in the fibers is Proposition \ref{prop:odlro}.

(5)
For $\sminvtemperature
\leq
\sminvtemperature_{c}$ Theorem \ref{thm:limit-state} gives $\oastate[\psi_{\sminvtemperature}]
=
\oastate[\psi_{\sminvtemperature,\ast}]$,
a faithful product state,
and Theorem \ref{thm:factor} gives factoriality.
It follows that the center of its GNS representation is trivial.
\end{proof}

\begin{rem}[equilibrium decomposition and the dynamics]\label{rem:equilibrium-fixed-dynamics}
The extremal decomposition of a KMS state for a fixed $\oacstar$-dynamical system keeps that dynamics fixed
\cite[Theorem 5.3.30]{BratteliRobinson004}.
Theorem \ref{thm:central-decomposition} is not such a decomposition on $\oa{A}$.
Each fiber is KMS for a phase-dependent $\tau^{\txtbogoliubov,\phi}$,
whose generator depends on $\phi$ when $r_{0}>0$.
Proposition \ref{prop:phase-extended-kms} restores one fixed dynamics on the enlarged algebra
$\oa{C}$.
The phase is a central classical variable.
The fiber family is obtained by evaluating
$\widetilde{\tau}^{\sminvtemperature}$
at each phase,
and \eqref{eq:phase-extended-kms-state} is KMS for this single extended dynamics.
The general mean-field variational decomposition of limiting Gibbs states into pure phases is established in \cite{RaggioWerner002}.
The theorem here realizes the same pure-phase decomposition mechanism explicitly as the gauge-circle orbit of the degenerate model and records the associated fiber dynamics.
\end{rem}

Theorem \ref{thm:central-decomposition} completes the program: above the critical inverse temperature the equilibrium state of the BCS model is not a factor state. Its central decomposition is a genuine direct integral over the phase circle of the order parameter. The disjoint fibers are factor KMS states. Each fiber carries its own Bogoliubov dynamics and is the unique KMS product state of the corresponding self-consistent mean field. The gauge symmetry is unbroken in the average, but it is broken in every fiber and transitively permutes the fibers. This is the operator-algebraic content of the Bogoliubov--Haag treatment of the BCS model \cite{BogoliubovNikolai002,RudolfHaag005,ThirringWehrl001,ThirringWalter001}, in the decomposition-theoretic language of \cite{HaagHugenholtzWinnink001,EmchGuenin001,BratteliRobinson003,BratteliRobinson004}. The analogy with the gauge-orbit decomposition of the condensed phases of the Bose gas, and with the general theory of ergodic decompositions of invariant equilibrium states, is developed in the companion note on the hard-core Bose gas.

The finite-temperature decomposition has the ground-state decomposition of Theorem \ref{thm:ground-state-direct-integral} as its zero-temperature limit. The order of limits in the following result is the thermodynamic limit followed by the zero-temperature limit.

\begin{cor}[zero-temperature limit of the central decomposition]\label{cor:zero-temperature-central-decomposition}
Assume the degenerate strong-coupling condition
$\sminvtemperature_{c}\abs{\varepsilon}
<
1$ from Corollary \ref{cor:strong-coupling}.
Restore the inverse-temperature dependence in
\eqref{eq:order-parameters}--\eqref{eq:symmetry-breaking-bloch-vector} by writing
$$\begin{aligned}
\eta_{0,\sminvtemperature}
&=
\fun{\eta_{0}}{\sminvtemperature},
\\
r_{0,\sminvtemperature}
&=
\sqrt{\eta_{0,\sminvtemperature}^{2}
-
\sminvtemperature_{c}^{2}\varepsilon^{2}},
\\
\physvec{m}_{\sminvtemperature,\phi}
&=
\rbk{
r_{0,\sminvtemperature}\fun{\cos}{\phi},
r_{0,\sminvtemperature}\fun{\sin}{\phi},
\sminvtemperature_{c}\varepsilon
}.
\end{aligned}$$
The following conclusions hold as
$\sminvtemperature
\to
\infty$.
\begin{enumerate}
\item The superconducting thermal regime of
Definition \ref{def:superconducting-thermal-regime} holds for every sufficiently large
$\sminvtemperature$,
and the order parameters satisfy
\begin{equation}\label{eq:zero-temperature-order-parameters}
\begin{aligned}
\eta_{0,\sminvtemperature}
\to
1,
\quad
r_{0,\sminvtemperature}
\to
r
=
\sqrt{1 - \sminvtemperature_{c}^{2}\varepsilon^{2}}.
\end{aligned}
\end{equation}
The thermal Bloch vectors converge uniformly on the gauge circle to the ground-state vectors of
Corollary \ref{cor:strong-coupling}:
\begin{equation}\label{eq:zero-temperature-bloch-vectors}
\sup_{\phi
\in
[0, 2 \pi)}
\abs{
\physvec{m}_{\sminvtemperature,\phi}
-
\fun{\physvec{u}}{\phi}
}
=
\abs{
r_{0,\sminvtemperature}
-
r
}
\to
0.
\end{equation}
\item For every $A
\in
\oa{A}$,
the thermal fiber states converge to the pure product ground states uniformly in the gauge angle:
\begin{equation}\label{eq:zero-temperature-fiber-states}
\sup_{\phi
\in
[0, 2 \pi)}
\abs{
\fun{\oastate[\psi_{\sminvtemperature,\phi}]}{A}
-
\fun{\oastate[\psi_{\fun{\physvec{\omega}}{\phi}}]}{A}
}
\to
0.
\end{equation}
The one-site effective Hamiltonians satisfy
\begin{equation}\label{eq:zero-temperature-effective-hamiltonians}
\sup_{\phi
\in
[0, 2 \pi)}
\norm{
\physham[h]_{\sminvtemperature,\phi}
+
\frac{1}{\sminvtemperature_{c}}
\physvec{\sigma}
\cdot
\fun{\physvec{u}}{\phi}
}
\to
0,
\end{equation}
where
$\physham[h]_{\sminvtemperature,\phi}$
denotes the Hamiltonian
$\physham[h]_{\phi}$ of \eqref{eq:self-consistency} with its dependence on
$\sminvtemperature$ restored.
The limiting one-site generator differs by the scalar
$1 / \sminvtemperature_{c}$
from the positive summand of
$\physham_{\txtbogoliubov,\fun{\physvec{\omega}}{\phi}}$
in Corollary \ref{cor:strong-coupling}.
\item The limit Gibbs states of Theorem \ref{thm:limit-state} converge in the weak-$\ast$ topology to the gauge-averaged ground-sector state:
\begin{equation}\label{eq:zero-temperature-averaged-state}
\lim_{\sminvtemperature
\to
\infty}
\fun{\oastate[\psi_{\sminvtemperature}]}{A}
=
\fun{\oastate[\psi_{\txtgs}]}{A}.
\end{equation}
Equivalently,
the finite-volume Gibbs states have the iterated limit
\begin{equation}\label{eq:iterated-zero-temperature-limit}
\lim_{\sminvtemperature
\to
\infty}
\lim_{\Omega
\to
\infty}
\fun{\oastate[\widehat{\psi}_{\sminvtemperature,\Omega}]}{A}
=
\fun{\oastate[\psi_{\txtgs}]}{A}.
\end{equation}
\item Let
$\Phi_{\sminvtemperature}$ and
$\Phi_{\txtgs}$ be the maps from $\mansphere^{1}$ to
$\oaspstate_{\oa{A}}$ defined by
$$\begin{aligned}
\fun{\Phi_{\sminvtemperature}}{\phi}
&=
\oastate[\psi_{\sminvtemperature,\phi}],
\\
\fun{\Phi_{\txtgs}}{\phi}
&=
\oastate[\psi_{\fun{\physvec{\omega}}{\phi}}].
\end{aligned}$$
The central decomposition measures converge weakly on
$\oaspstate_{\oa{A}}$:
\begin{equation}\label{eq:zero-temperature-central-measures}
\wlim_{\sminvtemperature \to \infty}
\pushout{\rbk{\Phi_{\sminvtemperature}}}
\msrprb_{\mansphere^{1}}
=
\pushout{\rbk{\Phi_{\txtgs}}}
\msrprb_{\mansphere^{1}}.
\end{equation}
The barycenters in \eqref{eq:zero-temperature-central-measures} are respectively the finite-temperature state of
Theorem \ref{thm:central-decomposition} and the ground-sector state of
Theorem \ref{thm:ground-state-direct-integral}.
The finite-temperature central decomposition into factor states with invertible one-site density matrices converges to the central decomposition into pure product ground states.
\end{enumerate}
\end{cor}

\begin{proof}

(1)
The fixed-point equation of Proposition \ref{prop:maximizer}(1) is
\begin{equation}\label{eq:zero-temperature-fixed-point}
\frac{\fun{\mathrm{arctanh}}{\eta_{0,\sminvtemperature}}}
{\eta_{0,\sminvtemperature}}
=
\frac{\sminvtemperature}{\sminvtemperature_{c}}.
\end{equation}
If \eqref{eq:zero-temperature-order-parameters} failed in its first assertion,
there would be a $\delta
>
0$ and a sequence
$\sminvtemperature_{j}
\to
\infty$ for which
$\eta_{0,\sminvtemperature_{j}}
\leq
1
-
\delta$.
The left side of
\eqref{eq:zero-temperature-fixed-point} would remain bounded,
whereas its right side would tend to infinity.
This is impossible.
The remaining assertions in
\eqref{eq:zero-temperature-order-parameters} and
\eqref{eq:zero-temperature-bloch-vectors} follow from their defining formulas.
The inequality
$\sminvtemperature_{c}\abs{\varepsilon}
<
\eta_{0,\sminvtemperature}$
holds eventually because
$\sminvtemperature_{c}\abs{\varepsilon}
<
1$.
This proves (1).

(2)
Let
$\oastate[\psi^{\rbk{1}}_{\fun{\physvec{\omega}}{\phi}}]$
denote the one-site pure state with Bloch vector
$\fun{\physvec{u}}{\phi}$.
The Bloch parametrization
\eqref{eq:bloch-state} gives
$$\norm{
\oastate[\psi^{\rbk{1}}_{\sminvtemperature,\phi}]
-
\oastate[\psi^{\rbk{1}}_{\fun{\physvec{\omega}}{\phi}}]
}
=
\abs{
\physvec{m}_{\sminvtemperature,\phi}
-
\fun{\physvec{u}}{\phi}
}
=
\abs{
r_{0,\sminvtemperature}
-
r
}.$$
For $A
\in
\oa{A}_{\Lambda}$,
the telescoping expansion of the two product states gives
\begin{equation}\label{eq:zero-temperature-local-bound}
\sup_{\phi
\in
[0, 2 \pi)}
\abs{
\fun{\oastate[\psi_{\sminvtemperature,\phi}]}{A}
-
\fun{\oastate[\psi_{\fun{\physvec{\omega}}{\phi}}]}{A}
}
\leq
\abscard{\Lambda}
\abs{
r_{0,\sminvtemperature}
-
r
}
\norm{A}.
\end{equation}
The density of $\oa{A}_{\txtloc}$ in $\oa{A}$ and the unit norm of every state extend
\eqref{eq:zero-temperature-local-bound} to the convergence
\eqref{eq:zero-temperature-fiber-states}.
Equation \eqref{eq:self-consistency} gives
$$\physham[h]_{\sminvtemperature,\phi}
=
-
\frac{1}{\sminvtemperature_{c}}
\physvec{\sigma}
\cdot
\physvec{m}_{\sminvtemperature,\phi}.$$
Equation \eqref{eq:zero-temperature-effective-hamiltonians} follows from
\eqref{eq:zero-temperature-bloch-vectors}.
This proves (2).

(3)
The central decomposition in
Theorem \ref{thm:central-decomposition} and
the uniform convergence
\eqref{eq:zero-temperature-fiber-states} give
$$\begin{aligned}
\abs{
\fun{\oastate[\psi_{\sminvtemperature}]}{A}
-
\fun{\oastate[\psi_{\txtgs}]}{A}
}
&\leq
\int_{\mansphere^{1}}
\abs{
\fun{\oastate[\psi_{\sminvtemperature,\phi}]}{A}
-
\fun{\oastate[\psi_{\fun{\physvec{\omega}}{\phi}}]}{A}
}
\opdmsr{\msrprb_{\mansphere^{1}}}(\phi)
\\
&\leq
\sup_{\phi
\in
[0, 2 \pi)}
\abs{
\fun{\oastate[\psi_{\sminvtemperature,\phi}]}{A}
-
\fun{\oastate[\psi_{\fun{\physvec{\omega}}{\phi}}]}{A}
}
\to
0.
\end{aligned}$$
Theorem \ref{thm:limit-state} identifies the inner limit in
\eqref{eq:iterated-zero-temperature-limit}.
This proves (3).

(4)
The algebra $\oa{A}$ is separable.
Its state space is therefore compact and metrizable in the weak-$\ast$ topology.
Equation \eqref{eq:zero-temperature-fiber-states} implies that the maps
$\Phi_{\sminvtemperature}$
converge uniformly to
$\Phi_{\txtgs}$ in a metric inducing that topology.
Every continuous function on the compact state space is uniformly continuous.
Its integrals against the two pushforward measures in
\eqref{eq:zero-temperature-central-measures} therefore converge.
This is weak convergence of the central decomposition measures and proves (4).
\end{proof}

Corollary \ref{cor:zero-temperature-central-decomposition} concerns the state-valued central decompositions and their barycenters. It does not assert unitary convergence of the GNS triples: the finite-temperature fibers are faithful factor representations, whereas the limiting fibers are pure irreducible representations. It also does not interchange the thermodynamic and zero-temperature limits in \eqref{eq:iterated-zero-temperature-limit}.

\appendix

\section{Appendix}\label{appendix}

The appendices supply some results and proofs invoked in the main text.

\subsection{Functional-Analytic Lemmas}\label{sec:functional-analytic-appendix}

Elementary Hilbert-space and direct-integral facts used in the main text are collected here. Throughout the appendices, \(\msrprb_{\mansphere^{1}}
=
\frac{\opdmsr{\phi}}{2 \pi}\) denotes the normalized Haar probability measure on \(\mansphere^{1}\).

\subsubsection{Strong convergence}\label{strong-convergence}

The following lemmas record the stability properties of strong convergence used in the main limiting arguments.

\begin{lem}[products of strongly convergent sequences]\label{lem:strong-products}
Let $A_{\Omega}
\to
A$ and $B_{\Omega}
\to
B$ strongly as $\Omega \to \infty$ on a Hilbert space with $\sup_{\Omega \geq 1} \norm{A_{\Omega}} < \infty$.
Then $A_{\Omega} B_{\Omega}
\to
A B$ strongly as $\Omega \to \infty$.
\end{lem}

\begin{proof}
For a vector $\Psi$,
$\norm{\rbk{A_{\Omega} B_{\Omega} - A B} \Psi}
\leq
\norm{A_{\Omega}} \norm{\rbk{B_{\Omega} - B} \Psi} + \norm{\rbk{A_{\Omega} - A} B \Psi}
\to
0$ as $\Omega \to \infty$.
\end{proof}

\begin{lem}[exponentials of strongly convergent sequences]\label{lem:strong-exponentials}
Let $X_{\Omega}
\to
X$ strongly as $\Omega \to \infty$ with all $X_{\Omega}, X$ self-adjoint and $\sup_{\Omega \geq 1} \norm{X_{\Omega}}
=
C < \infty$.
Then $\fnexp{\imunit t X_{\Omega}}
\to
\fnexp{\imunit t X}$ strongly as $\Omega \to \infty$,
uniformly for $t$ in compact subsets of $\fldreal$.
\end{lem}

\begin{proof}
Note first $\norm{X}
\leq
C$,
since $\norm{X \Psi}
=
\lim \norm{X_{\Omega} \Psi}
\leq
C \norm{\Psi}$.
The Duhamel formula
$$\fnexp{\imunit t X_{\Omega}} - \fnexp{\imunit t X}
=
\imunit \int_{0}^{t} \fnexp{\imunit \rbk{t - s} X_{\Omega}} \rbk{X_{\Omega} - X} \fnexp{\imunit s X} \opdmsr{s}$$
follows by integrating the derivative of $s
\mapsto
\fnexp{\imunit \rbk{t - s} X_{\Omega}} \fnexp{\imunit s X}$.
The integrand is norm continuous in $s$ because both exponential families are norm continuous for bounded generators.
Applied to a vector $\Psi$,
$$\norm{\rbk{\fnexp{\imunit t X_{\Omega}} - \fnexp{\imunit t X}} \Psi}
\leq
\int_{0}^{\abs{t}} \norm{\rbk{X_{\Omega} - X} \fnexp{\pm \imunit s X} \Psi} \opdmsr{s},$$
with the sign of $t$.
The integrand tends to $0$ for each $s$ and is bounded by $2 C \norm{\Psi}$.
It follows that the integral tends to $0$ by dominated convergence,
uniformly for $t$ in compacts by monotonicity of the bound in $\abs{t}$.
\end{proof}

\subsubsection{Decomposable operators}\label{decomposable-operators}

Let \(\sphilb{H}_0\) be a separable complex Hilbert space. The following results characterize decomposable operators and diagonal multiplication operators on the direct integral used in Section \ref{sec:decomposition}.

\begin{lem}[commutant of the diagonal algebra]\label{lem:diagonal-commutant}
A bounded operator on $\sphilb{H}_{\sminvtemperature}
=
\fun{\lp^{2}}{\mansphere^{1}, \msrprb_{\mansphere^{1}}} \otimes \sphilb{H}_{0}$ commutes with all diagonal operators
$M_{z} \otimes \id_{\sphilb{H}_{0}}$,
$z
\in
\fun{\lp^{\infty}}{\mansphere^{1}, \msrprb_{\mansphere^{1}}}$,
if and only if it is decomposable.
\end{lem}

\begin{proof}
Decomposable operators commute with diagonal ones fiberwise.
Conversely,
let $X
\in
\opspbddlin{\sphilb{H}_{\sminvtemperature}}$ commute with
$M_{z} \otimes \id_{\sphilb{H}_{0}}$ for every
$z
\in
\fun{\lp^{\infty}}{\mansphere^{1}, \msrprb_{\mansphere^{1}}}$,
and fix an orthonormal basis $\seq{e_{j}}{j
\geq
1}$ of $\sphilb{H}_{0}$.
For $j, l
\in
\semigrposint$,
and $g
\in
\fun{\lp^{2}}{\mansphere^{1}, \msrprb_{\mansphere^{1}}}$,
the vector $X \rbk{g \otimes e_{l}}$ is an
$\sphilb{H}_{0}$-valued square-integrable function of $\phi$.
Define its $e_{j}$-component $X_{j l} g
\in
\fun{\lp^{2}}{\mansphere^{1}, \msrprb_{\mansphere^{1}}}$ by
$$\fun{\rbk{X_{j l} g}}{\phi}
=
\bkt{e_{j}}{\fun{\rbk{X \rbk{g \otimes e_{l}}}}{\phi}}_{\sphilb{H}_{0}}
\quad \rbk{\phi
\in
\mansphere^{1} \text{ almost everywhere}}.$$
This definition gives a bounded operator
$X_{j l}
\colon
\fun{\lp^{2}}{\mansphere^{1}, \msrprb_{\mansphere^{1}}}
\to
\fun{\lp^{2}}{\mansphere^{1}, \msrprb_{\mansphere^{1}}}$
because
$$\int_{\mansphere^{1}}
\abs{\fun{\rbk{X_{j l} g}}{\phi}}^{2}
\opdmsr{\msrprb_{\mansphere^{1}}}(\phi)
\leq
\norm{X \rbk{g \otimes e_{l}}}^{2}_{\sphilb{H}_{\sminvtemperature}}
\leq
\norm{X}^{2}
\int_{\mansphere^{1}}
\abs{\fun{g}{\phi}}^{2}
\opdmsr{\msrprb_{\mansphere^{1}}}(\phi).$$
For $z
\in
\fun{\lp^{\infty}}{\mansphere^{1}, \msrprb_{\mansphere^{1}}}$,
the commutation of $X$ with
$M_{z} \otimes \id_{\sphilb{H}_{0}}$ gives
$$\begin{aligned}
\fun{\rbk{X_{j l} \rbk{z g}}}{\phi}
&=
\bkt{e_{j}}{
\fun{
\rbk{
X
\rbk{M_{z} \otimes \id_{\sphilb{H}_{0}}}
\rbk{g \otimes e_{l}}
}
}{\phi}
}_{\sphilb{H}_{0}}
\\
&=
\bkt{e_{j}}{
\fun{
\rbk{
\rbk{M_{z} \otimes \id_{\sphilb{H}_{0}}}
X
\rbk{g \otimes e_{l}}
}
}{\phi}
}_{\sphilb{H}_{0}}
\\
&=
\fun{z}{\phi}
\bkt{e_{j}}{\fun{\rbk{X \rbk{g \otimes e_{l}}}}{\phi}}_{\sphilb{H}_{0}}
=
\fun{z}{\phi} \fun{\rbk{X_{j l} g}}{\phi}
\end{aligned}$$
almost everywhere.
Thus $X_{j l}$ commutes with every multiplication operator on
$\fun{\lp^{2}}{\mansphere^{1}, \msrprb_{\mansphere^{1}}}$ induced by a function in
$\fun{\lp^{\infty}}{\mansphere^{1}, \msrprb_{\mansphere^{1}}}$.
The commutant of the multiplication algebra on $\fun{\lp^{2}}{\mansphere^{1}, \msrprb_{\mansphere^{1}}}$ is the multiplication algebra itself.
To verify this fact,
let
$Y
\in
\opspbddlin{\fun{\lp^{2}}{\mansphere^{1}, \msrprb_{\mansphere^{1}}}}$
commute with $M_{z}$ for every
$z
\in
\fun{\lp^{\infty}}{\mansphere^{1}, \msrprb_{\mansphere^{1}}}$,
and
set $y
=
\fun{Y}{1}
\in
\fun{\lp^{2}}{\mansphere^{1}, \msrprb_{\mansphere^{1}}}$.
Then $\fun{Y}{z}
=
Y M_{z} 1
=
M_{z} y
=
z y$ for $z
\in
\fun{\lp^{\infty}}{\mansphere^{1}, \msrprb_{\mansphere^{1}}}$,
and $\norm{z y}_{2}
\leq
\norm{Y} \norm{z}_{2}$ for all $z
\in
\fun{\lp^{\infty}}{\mansphere^{1}, \msrprb_{\mansphere^{1}}}$ forces $\abs{y}
\leq
\norm{Y}$ almost everywhere,
testing with $z
=
1_{\set{\phi}{\phi
\in
\mansphere^{1}, \ \abs{\fun{y}{\phi}} > \norm{Y} + \epsilon}}$.
It follows that $Y
=
M_{y}$ on the dense subspace
$\fun{\lp^{\infty}}{\mansphere^{1}, \msrprb_{\mansphere^{1}}}$.
Continuity extends the equality to $\fun{\lp^{2}}{\mansphere^{1}, \msrprb_{\mansphere^{1}}}$.
It follows that $X_{j l}
=
M_{x_{j l}}$ with $x_{j l}
\in
\fun{\lp^{\infty}}{\mansphere^{1}, \msrprb_{\mansphere^{1}}}$,
$\norm{x_{j l}}_{\infty}
\leq
\norm{X}$.
For each finite $J$ and rational-coefficient vectors $\zeta
=
\sum_{j
\in
J} c_{j} e_{j}$,
$\zeta'
=
\sum_{l
\in
J} c'_{l} e_{l}$,
consider the scalar function
$\sum_{j, l
\in
J} \cmpconj{c_{j}} c'_{l} x_{j l}$.
For every measurable $E$,
the vector $1_{E} \otimes \zeta'$ is mapped by $X$ with norm control.
It follows that the scalar function is bounded almost everywhere by
$\norm{X} \norm{\zeta} \norm{\zeta'}$.
Taking the countable union of null sets over all rational data,
for almost every $\phi$ the matrix $\rbk{\fun{x_{j l}}{\phi}}$ defines a bounded operator $X_{\phi}$ on $\sphilb{H}_{0}$ with $\norm{X_{\phi}}
\leq
\norm{X}$.
The family is measurable,
and $X$ agrees with the decomposable operator of the family on the total set of vectors $g \otimes e_{l}$ with $g
\in
\fun{\lp^{\infty}}{\mansphere^{1}, \msrprb_{\mansphere^{1}}}$ and $l
\in
\semigrposint$.
Both operators are bounded,
which makes them equal everywhere.
\end{proof}

\begin{lem}[vanishing fiberwise]\label{lem:fiberwise-zero}
A decomposable operator $X
=
\rbk{X_{\phi}}$ vanishes if and only if $X_{\phi}
=
0$ almost everywhere.
Two decomposable operators commuting as operators have commuting fibers almost everywhere.
\end{lem}

\begin{proof}
Identify $\mansphere^{1}$ with the phase interval
$\rightopeninterval{0}{2 \pi}$,
and let
$$\mathcal{I}
=
\setone{\emptyset, \mansphere^{1}}
\cup
\set{
\rightopeninterval{2 \pi r}{2 \pi s}
}{
r, s
\in
\fldrat \cap \unitclosedinterval,
\quad r < s
}.$$
This is a countable $\pi$-system that generates the Borel subsets of
$\mansphere^{1}$.
Fix a countable total family $\seq{\zeta_{i}}{i
\geq
1}$ of $\sphilb{H}_{0}$.
For $i, i'
\in
\semigrposint$,
define the integrable function
$$\fun{f_{i i'}}{\phi}
=
\bkt{\zeta_{i}}{X_{\phi} \zeta_{i'}}_{\sphilb{H}_{0}}
\quad \rbk{\phi
\in
\mansphere^{1} \text{ almost everywhere}}.$$
If $X
=
0$,
then every $I
\in
\mathcal{I}$ satisfies
$$0
=
\bkt{1_{I} \otimes \zeta_{i}}{
X \rbk{1_{\mansphere^{1}} \otimes \zeta_{i'}}
}_{\sphilb{H}_{\sminvtemperature}}
=
\int_{I}
\fun{f_{i i'}}{\phi}
\opdmsr{\msrprb_{\mansphere^{1}}}(\phi).$$
Define a finite complex measure on the Borel subsets of $\mansphere^{1}$ by
$$\fun{\nu_{i i'}}{E}
=
\int_{E}
\fun{f_{i i'}}{\phi}
\opdmsr{\msrprb_{\mansphere^{1}}}(\phi).$$
The preceding identity says that $\nu_{i i'}$ vanishes on $\mathcal{I}$.
The Borel sets $E$ satisfying
$\fun{\nu_{i i'}}{E}
=
0$
form a Dynkin system:
they contain $\mansphere^{1}$,
are closed under complements in $\mansphere^{1}$,
and are closed under countable disjoint unions.
This Dynkin system contains the generating $\pi$-system $\mathcal{I}$.
The $\pi$--$\lambda$ theorem therefore shows that $\nu_{i i'}$ vanishes on every Borel set.
Uniqueness of the Radon--Nikodym derivative now gives
$f_{i i'}
=
0$ almost everywhere.

The pairs $\rbk{i, i'}$
form a countable set.
Outside the union of their null sets,
all matrix elements
$\bkt{\zeta_{i}}{X_{\phi} \zeta_{i'}}_{\sphilb{H}_{0}}$
vanish.
Since $\seq{\zeta_{i}}{i
\geq
1}$ is total and $X_{\phi}$ is bounded,
this implies $X_{\phi}
=
0$ almost everywhere.

Conversely,
if $X_{\phi}
=
0$ almost everywhere,
then every $\xi
\in
\sphilb{H}_{\sminvtemperature}$ satisfies
$$\norm{X \xi}^{2}_{\sphilb{H}_{\sminvtemperature}}
=
\int_{\mansphere^{1}}
\norm{X_{\phi} \fun{\xi}{\phi}}^{2}_{\sphilb{H}_{0}}
\opdmsr{\msrprb_{\mansphere^{1}}}(\phi)
=
0.$$
Thus $X
=
0$.

If $X$ and $Y$ are decomposable,
then $\commutator{X}{Y}$ is decomposable with fibers
$\commutator{X_{\phi}}{Y_{\phi}}$.
Applying the equivalence just proved to this commutator proves the final assertion.
\end{proof}

\subsection{Operator-Algebraic Setup}\label{sec:oa-appendix}

The quasi-local observable algebra, its local embeddings, and the operator-algebraic conventions used throughout the paper are fixed here.

\subsubsection{The quasi-local algebra and spin notation}\label{the-quasi-local-algebra-and-spin-notation}

The kinematical arena of the quasi-spin BCS model is one two-dimensional degree of freedom for each pair mode \(p\). In the electron picture the mode \(p\) stands for the pair \(\rbk{k \uparrow, -k \downarrow}\) of single-electron states, and the two-dimensional space is spanned by the unoccupied and the doubly occupied pair state. The reduction of the electron algebra to this pair subspace is the quasi-spin formalism. In this note the quasi-spin algebra is taken as the definition of the model.

\begin{defn}[quasi-spin algebra]\label{def:quasilocal}
For a finite subset $\Lambda
\subset
\semigrposint$ the local algebra is
$$\oa{A}_{\Lambda}
=
\bigotimes_{p
\in
\Lambda} \spmat{2}{\fldcmp},$$
and for $\Lambda
\subset
\Lambda'$ the algebra $\oa{A}_{\Lambda}$ is embedded into $\oa{A}_{\Lambda'}$ by tensoring with the identity on $\Lambda' \setminus \Lambda$.
The quasi-spin algebra $\oa{A}$ is the $\oacstar$-inductive limit
$$\oa{A}
=
\gtclos{\bigcup_{\Lambda
\subset
\semigrposint, \abscard{\Lambda} < \infty} \oa{A}_{\Lambda}}^{\norm{\cdot}},$$
and $\oa{A}_{\txtloc}
=
\bigcup_{\Lambda} \oa{A}_{\Lambda}$ is the algebra of local elements.
For $\Omega
\in
\semigrposint$ the abbreviation $\oa{A}_{\Omega}
=
\oa{A}_{\setone{1, \dotsc, \Omega}}$ is used.
\end{defn}

The algebra \(\oa{A}\) is the UHF algebra of type \(2^{\infty}\). It is simple, unital, and separable \cite[Section 2.6]{BratteliRobinson003}. Simplicity is used in the main text only through the fact that every representation of \(\oa{A}\) is faithful, and for completeness a short proof is recorded.

\begin{lem}[simplicity]\label{lem:simple}
Every nonzero closed two-sided ideal of $\oa{A}$ equals $\oa{A}$.
In particular every representation of $\oa{A}$ on a Hilbert space with $\fun{\oarepn}{1}
\neq
0$ is isometric.
\end{lem}

\begin{proof}
Let $\oa{J}
\neq
\oa{A}$ be a closed two-sided ideal and $q \colon \oa{A}
\to
\oa{A} / \oa{J}$ the quotient map.
Since $\oa{J}
\neq
\oa{A}$,
the quotient is a unital $\oacstar$-algebra and $\fun{q}{1}
=
1
\neq
0$.
Each local algebra $\oa{A}_{\Lambda}$ is isomorphic to a full matrix algebra and therefore simple.
The kernel of the restriction of $q$ to $\oa{A}_{\Lambda}$ is a proper two-sided ideal that does not contain $1$.
Simplicity therefore makes this kernel $\setone{0}$.
An injective $\ast$-homomorphism between $\oacstar$-algebras is isometric.
It follows that $q$ is isometric on $\oa{A}_{\txtloc}$.
Continuity makes $q$ isometric on $\oa{A}$,
and $\oa{J}
=
\Ker q
=
\setone{0}$.
For a representation $\oarepn$ with $\fun{\oarepn}{1}
\neq
0$ the kernel is a closed two-sided ideal not containing $1$.
The kernel is therefore zero,
and a faithful representation is isometric by the same cited fact.
\end{proof}

The permutation automorphisms are used for the thermal analysis. For a permutation \(\pi\) of \(\semigrposint\) moving only finitely many points, the assignment \(\sigma_{p}^{\gamma}
\mapsto
\sigma_{\fun{\pi}{p}}^{\gamma}\) extends, exactly as in Definition \ref{def:gauge}, to a \(\ast\)-automorphism \(\mathfrak{g}_{\pi}\) of \(\oa{A}\).

\subsection{One-Site Spin Calculus and the Bloch Sphere}\label{sec:bloch}

The product-representation calculations in the main text use the following one-site spin geometry. For complex vectors \(\physvec{a}, \physvec{b}
\in
\fldcmp^{3}\) write \(\physvec{a} \cdot \physvec{b}
=
\sum_{\gamma \in \setone{x, y, z}} a^{\gamma} b^{\gamma}\) without complex conjugation, and \(\physvec{\sigma} \cdot \physvec{a}
=
\sum_{\gamma \in \setone{x, y, z}} a^{\gamma} \sigma^{\gamma}\).

\begin{lem}[product identity]\label{lem:pauli-product}
For all $\physvec{a}, \physvec{b}
\in
\fldcmp^{3}$, it holds that
$$\rbk{\physvec{\sigma} \cdot \physvec{a}} \rbk{\physvec{\sigma} \cdot \physvec{b}}
=
\physvec{a} \cdot \physvec{b}
+\imunit \physvec{\sigma} \cdot \rbk{\physvec{a} \times \physvec{b}}.$$
\end{lem}

\begin{proof}
The Pauli matrices satisfy $\sigma^{\alpha} \sigma^{\beta}
=
\delta^{\alpha \beta} + \imunit \sum_{\gamma \in \setone{x, y, z}} \epsilon^{\alpha \beta \gamma} \sigma^{\gamma}$,
which is checked entry by entry for the nine pairs $\rbk{\alpha, \beta}$.
Multiplying by $a^{\alpha} b^{\beta}$ and summing gives the assertion,
because both sides are bilinear in $\rbk{\physvec{a}, \physvec{b}}$.
\end{proof}

\begin{lem}[Rodrigues rotation formula]\label{lem:rodrigues}
Let $\physvec{q}
\in
\mansphere^{2}$ and define the linear map $K_{\physvec{q}}$ on $\fldreal^{3}$ by
$$\fun{K_{\physvec{q}}}{\physvec{a}}
=
\physvec{q} \times \physvec{a}.$$
For $\theta
\in
\fldreal$,
the rotation $R_{\physvec{q},\theta}
=
\fnexp{\theta K_{\physvec{q}}}$ satisfies
\begin{equation}\label{eq:rodrigues-rotation}
\fun{R_{\physvec{q},\theta}}{\physvec{a}}
=
\rbk{\physvec{q} \cdot \physvec{a}} \physvec{q}
+
\sqbk{
\physvec{a}
-
\rbk{\physvec{q} \cdot \physvec{a}} \physvec{q}
}
\fun{\cos}{\theta}
+
\rbk{\physvec{q} \times \physvec{a}}
\fun{\sin}{\theta}.
\end{equation}
It fixes the axis $\fldreal \physvec{q}$ and rotates $\physvec{q}^{\perp}$ through the oriented angle $\theta$.
\end{lem}

\begin{proof}
The vector triple-product identity gives
$$\begin{aligned}
\fun{K_{\physvec{q}}^{2}}{\physvec{a}}
=
\physvec{q} \times \rbk{\physvec{q} \times \physvec{a}}
=
\rbk{\physvec{q} \cdot \physvec{a}} \physvec{q}
-\physvec{a},
\quad
K_{\physvec{q}}^{3}
=
- K_{\physvec{q}}.
\end{aligned}$$
Separating the odd and even powers in the exponential series yields
$$\fnexp{\theta K_{\physvec{q}}}
=
1
+
\fun{\sin}{\theta} K_{\physvec{q}}
+
\rbk{1 - \fun{\cos}{\theta}} K_{\physvec{q}}^{2}.$$
Applying this identity to $\physvec{a}$ gives \eqref{eq:rodrigues-rotation}.
The formula fixes $\physvec{q}$.
For $\physvec{a}
\in
\physvec{q}^{\perp}$,
it reduces to
$$\fun{R_{\physvec{q},\theta}}{\physvec{a}}
=
\physvec{a} \fun{\cos}{\theta}
+
\rbk{\physvec{q} \times \physvec{a}} \fun{\sin}{\theta},$$
which is the oriented planar rotation through $\theta$.
\end{proof}

For a real unit vector \(\physvec{u}
\in
\mansphere^{2}\), the operator \(\physvec{\sigma} \cdot \physvec{u}\) is self-adjoint and satisfies \(\rbk{\physvec{\sigma} \cdot \physvec{u}}^{2}
=
1\) by Lemma \ref{lem:pauli-product} and \(\sqfun{\trace}{\physvec{\sigma} \cdot \physvec{u}}
=
0\). Its spectrum is therefore \(\setone{+1, -1}\) with one-dimensional eigenspaces and spectral projections \begin{equation}\label{eq:bloch-projection}
P_{\physvec{u}}
=
\frac{1}{2} \rbk{1 + \physvec{\sigma} \cdot \physvec{u}},
\quad
1 - P_{\physvec{u}}
=
P_{- \physvec{u}}.
\end{equation} Choose a unit vector \(\xi_{\physvec{u}}^{+}
\in
\fldcmp^{2}\) with \(P_{\physvec{u}} \xi_{\physvec{u}}^{+}
=
\xi_{\physvec{u}}^{+}\) for each real unit vector \(\physvec{u}\) that occurs. The choice is unique up to a phase. The space \(\fldcmp^{2}\) is equipped with its standard Hermitian inner product \(\bkt{u}{v}
=
\faadj{u} v\), linear in the second variable.

\begin{lem}[transition probability]\label{lem:overlap}
For unit vectors $\physvec{u}, \physvec{u}'
\in
\mansphere^{2}$, it follows that
$$\abs{\bkt{\xi_{\physvec{u}'}^{+}}{\xi_{\physvec{u}}^{+}}}^{2}
=
\sqfun{\trace}{P_{\physvec{u}'} P_{\physvec{u}}}
=
\frac{1 + \physvec{u} \cdot \physvec{u}'}{2}.$$
\end{lem}

\begin{proof}
For rank-one projections onto unit vectors,
the product trace equals the squared modulus of the inner product of the spanning unit vectors.
The squared overlap is
$$\abs{\bkt{\xi_{\physvec{u}'}^{+}}{\xi_{\physvec{u}}^{+}}}^{2}
=
\sqfun{\trace}{P_{\physvec{u}'} P_{\physvec{u}}}
=
\frac{1}{4} \sqfun{\trace}{1 + \physvec{\sigma} \cdot \physvec{u} + \physvec{\sigma} \cdot \physvec{u}' + \rbk{\physvec{\sigma} \cdot \physvec{u}'} \rbk{\physvec{\sigma} \cdot \physvec{u}}}
=
\frac{2 + 2 \physvec{u} \cdot \physvec{u}'}{4},$$
where the last equality uses $\sqfun{\trace}{\sigma^{\gamma}}
=
0$ and Lemma \ref{lem:pauli-product} in the last step.
\end{proof}

Every unit vector \(\physvec{u}
\in
\mansphere^{2}\) is completed to a right-handed orthonormal frame \(\rbk{\physvec{e}_{1}, \physvec{e}_{2}, \physvec{u}}\) of \(\fldreal^{3}\). The choice of frame is arbitrary but fixed once and for all for each \(\physvec{u}\) that occurs. The complex flip vectors and flip operators are \begin{equation}\label{eq:flip-vectors}
\physvec{u}^{\pm}
=
\frac{1}{2} \rbk{\physvec{e}_{1} \pm \imunit \physvec{e}_{2}}
\in \fldreal^{3},
\quad
f^{\pm}
=
\physvec{\sigma} \cdot \physvec{u}^{\pm}
\in \spmat{2}{\fldcmp},
\end{equation} so that \(\physvec{u}^{-}
=
\cmpconj{\physvec{u}^{+}}\) componentwise and \(\faadj{\rbk{f^{+}}}
=
f^{-}\).

\begin{lem}[frame relations]\label{lem:frame}
The frame relations associated with \eqref{eq:flip-vectors} are:
$$\rbk{f^{\pm}}^{2}
=
0,
\quad
f^{+} f^{-}
=
P_{\physvec{u}},
\quad
f^{-} f^{+}
=
P_{- \physvec{u}},
\quad
\commutator{\physvec{\sigma} \cdot \physvec{u}}{f^{\pm}}
=
\pm 2 f^{\pm}.$$
The vector $\xi_{\physvec{u}}^{-}
=
f^{-} \xi_{\physvec{u}}^{+}$ is a unit eigenvector of $\physvec{\sigma} \cdot \physvec{u}$ for the eigenvalue $-1$,
and the flip operators act on the two eigenvectors by
$$f^{+} \xi_{\physvec{u}}^{-}
=
\xi_{\physvec{u}}^{+},
\quad
f^{+} \xi_{\physvec{u}}^{+}
=
0,
\quad
f^{-} \xi_{\physvec{u}}^{-}
=
0.$$
For every $\gamma
\in
\setone{x, y, z}$, we obtain
\begin{equation}\label{eq:frame-decomposition}
\sigma^{\gamma}
=
u^{\gamma} \rbk{\physvec{\sigma} \cdot \physvec{u}} + 2 \rbk{u^{-}}^{\gamma} f^{+} + 2 \rbk{u^{+}}^{\gamma} f^{-}.
\end{equation}
\end{lem}

\begin{proof}
Lemma \ref{lem:pauli-product} reduces every operator product to a scalar product and a cross product of frame vectors.
The orthonormality of the frame gives
$$\begin{aligned}
\physvec{u}^{\pm} \cdot \physvec{u}^{\pm}
=
\frac{1}{4} \rbk{\physvec{e}_{1} \cdot \physvec{e}_{1} - \physvec{e}_{2} \cdot \physvec{e}_{2} \pm 2 \imunit \physvec{e}_{1} \cdot \physvec{e}_{2}}
=
0,
\quad
\physvec{u}^{\pm} \times \physvec{u}^{\pm}
=
0.
\end{aligned}$$
The Pauli product formula therefore yields $\rbk{f^{\pm}}^{2}
=
0$.

The mixed products retain the orientation of the frame:
$$\begin{aligned}
\physvec{u}^{+} \cdot \physvec{u}^{-}
=
\frac{1}{2},
\quad
\physvec{u}^{+} \times \physvec{u}^{-}
=
\frac{1}{4} \rbk{\physvec{e}_{1} + \imunit \physvec{e}_{2}}
\times
\rbk{\physvec{e}_{1} - \imunit \physvec{e}_{2}}
=
-\frac{\imunit}{2} \physvec{u}.
\end{aligned}$$
Substitution into Lemma \ref{lem:pauli-product} gives the symmetric product and both ordered products:
$$\begin{gathered}
f^{+} f^{-}
=
\frac{1}{2} + \imunit \physvec{\sigma} \cdot \rbk{- \frac{\imunit}{2} \physvec{u}}
=
P_{\physvec{u}},
\quad
f^{-} f^{+}
=
1 - P_{\physvec{u}}
=
P_{- \physvec{u}},
\quad
f^{+} f^{-} + f^{-} f^{+}
=
1.
\end{gathered}$$

For the commutator,
the frame geometry gives
$\physvec{u} \cdot \physvec{u}^{\pm}
=
0$ and $\physvec{u} \times \physvec{u}^{\pm}
=
\mp \imunit \physvec{u}^{\pm}$.
Applying Lemma \ref{lem:pauli-product} in the two possible orders yields
$\rbk{\physvec{\sigma} \cdot \physvec{u}} f^{\pm}
=
\pm f^{\pm}$ and $f^{\pm} \rbk{\physvec{\sigma} \cdot \physvec{u}}
=
\mp f^{\pm}$.
Their difference is the asserted commutator relation.

The relation for the minus sign and the identity $f^{+} f^{-}
=
P_{\physvec{u}}$ show that
$$\begin{aligned}
\rbk{\physvec{\sigma} \cdot \physvec{u}} f^{-} \xi_{\physvec{u}}^{+}
=
- f^{-} \xi_{\physvec{u}}^{+},
\quad
\norm{f^{-} \xi_{\physvec{u}}^{+}}^{2}
=
\bkt{\xi_{\physvec{u}}^{+}}{P_{\physvec{u}} \xi_{\physvec{u}}^{+}}
=
1.
\end{aligned}$$
Thus $\xi_{\physvec{u}}^{-}
=
f^{-} \xi_{\physvec{u}}^{+}$ is a unit eigenvector for the eigenvalue $-1$.
The product relations and nilpotency now give all three flip identities at once:
$$\begin{aligned}
f^{+} \xi_{\physvec{u}}^{-}
=
f^{+} f^{-} \xi_{\physvec{u}}^{+}
=
\xi_{\physvec{u}}^{+},
\quad
f^{+} \xi_{\physvec{u}}^{+}
=
\rbk{f^{+}}^{2} \xi_{\physvec{u}}^{-}
=
0,
\quad
f^{-} \xi_{\physvec{u}}^{-}
=
\rbk{f^{-}}^{2} \xi_{\physvec{u}}^{+}
=
0.
\end{aligned}$$

It remains to prove \eqref{eq:frame-decomposition}.
The orthonormal frame decomposes the coordinate Pauli matrix as
$$\sigma^{\gamma}
=
u^{\gamma} \physvec{\sigma} \cdot \physvec{u}
+ \rbk{\physvec{e}_{1}}^{\gamma} \physvec{\sigma} \cdot \physvec{e}_{1}
+ \rbk{\physvec{e}_{2}}^{\gamma} \physvec{\sigma} \cdot \physvec{e}_{2}.$$
The definitions of $\physvec{u}^{\pm}$ and $f^{\pm}$ identify its transverse part as
$$\rbk{\physvec{e}_{1}}^{\gamma} \physvec{\sigma} \cdot \physvec{e}_{1}
+ \rbk{\physvec{e}_{2}}^{\gamma} \physvec{\sigma} \cdot \physvec{e}_{2}
=
2 \rbk{u^{-}}^{\gamma} f^{+}
+ 2 \rbk{u^{+}}^{\gamma} f^{-}.$$
Combining these two identities proves \eqref{eq:frame-decomposition}.
\end{proof}

The states of \(\spmat{2}{\fldcmp}\) are parametrized by the closed unit ball of \(\fldreal^{3}\): every state is of the form \begin{equation}\label{eq:bloch-state}
\fun{\oastate[\psi^{\rbk{1}}_{\physvec{m}}]}{A}
=
\sqfun{\trace}{\frac{1}{2} \rbk{1 + \physvec{m} \cdot \physvec{\sigma}} A},
\quad
\abs{\physvec{m}}
\leq
1,
\end{equation} with \(\fun{\oastate[\psi^{\rbk{1}}_{\physvec{m}}]}{\sigma^{\gamma}}
=
m^{\gamma}\), and \(\oastate[\psi^{\rbk{1}}_{\physvec{m}}]\) is pure exactly for \(\abs{\physvec{m}}
=
1\), where the density matrix \eqref{eq:bloch-state} is the projection \eqref{eq:bloch-projection}. The vector \(\physvec{m}\) is called the Bloch vector of the state. The spectral decomposition of the density matrix in \eqref{eq:bloch-state} gives the eigenvalues \begin{equation}\label{eq:bloch-density-spectrum}
\frac{1}{2} \rbk{1 \pm \abs{\physvec{m}}}.
\end{equation} In particular, \eqref{eq:bloch-density-spectrum} proves the preceding purity criterion.

A frequently used consequence of \eqref{eq:frame-decomposition} concerns centered one-site operators. For a unit vector \(\physvec{u}\) and \(\gamma
\in
\setone{x, y, z}\) let \(d^{\gamma}
=
\sigma^{\gamma} - u^{\gamma}\) denote the fluctuation of \(\sigma^{\gamma}\) around its expectation in the pure state defined by \(\xi_{\physvec{u}}^{+}\). Then \eqref{eq:frame-decomposition} together with \(\rbk{\physvec{\sigma} \cdot \physvec{u}} \xi_{\physvec{u}}^{+}
=
\xi_{\physvec{u}}^{+}\) gives \begin{equation}\label{eq:centered-action}
d^{\gamma} \xi_{\physvec{u}}^{+}
=
2 \rbk{u^{+}}^{\gamma} \xi_{\physvec{u}}^{-},
\quad
\bkt{\xi_{\physvec{u}}^{+}}{d^{\gamma} \xi_{\physvec{u}}^{+}}
=
0,
\end{equation} It follows that a centered one-site operator maps the reference vector into the flipped vector, with amplitude bounded by \(2 \abs{\rbk{u^{+}}^{\gamma}}
\leq
1\).

\subsection{Analytic Preliminaries}\label{sec:analytic}

The thermal analysis of the degenerate model rests on a handful of classical facts: Wallis' formula, the Stirling bounds with explicit error control, the integral representation of the Bessel function \(J_{0}\), the Fejér kernel, and Vitali's theorem on families of holomorphic functions. The original papers take the corresponding steps from \cite{WatsonWhittaker001} and from the theory of hypergeometric functions. Complete proofs are given here.

\subsubsection{Wallis' formula and the Stirling--Robbins bounds}\label{wallis-formula-and-the-stirlingrobbins-bounds}

Wallis' formula and the Stirling--Robbins bounds yield the explicit factorial and entropy estimates used in the total-spin asymptotics.

\begin{lem}[Wallis integrals]\label{lem:wallis}
For $k
\in
\monnat$ let $W_{k}
=
\int_{0}^{\pi / 2} \sin^{k} t \opdmsr{t}$,
where the integral is the Riemann integral.
Then $W_{2 k}
=
\frac{\pi}{2} \cdot \frac{\rbk{2 k}!}{4^{k} \rbk{k!}^{2}}$ and $W_{2 k + 1}
=
\frac{4^{k} \rbk{k!}^{2}}{\rbk{2 k + 1}!}$,
and the Wallis limit is
$$\lim_{k
\to
\infty} \frac{1}{\sqrt{k}} \cdot \frac{4^{k} \rbk{k!}^{2}}{\rbk{2 k}!}
=
\sqrt{\pi}.$$
\end{lem}

\begin{proof}
Integration by parts gives $W_{k}
=
\frac{k - 1}{k} W_{k - 2}$ for $k
\geq
2$,
and $W_{0}
=
\frac{\pi}{2}$,
$W_{1}
=
1$.
Induction yields the two closed formulas.
Since $0
\leq
\sin t
\leq
1$ the sequence $\seq{W_{k}}{k
\geq
1}$ is nonincreasing,
and the ratio estimate is
$$1
\leq
\frac{W_{2 k}}{W_{2 k + 1}}
\leq
\frac{W_{2 k - 1}}{W_{2 k + 1}}
=
\frac{2 k + 1}{2 k}
\to
1
\quad \rbk{k \to \infty}.$$
Inserting the closed formulas,
$\frac{W_{2 k}}{W_{2 k + 1}}
=
\frac{\pi}{2} \rbk{2 k + 1} \rbk{\frac{\rbk{2 k}!}{4^{k} \rbk{k!}^{2}}}^{2}
\to
1$ as $k
\to
\infty$,
which rearranges to the stated limit.
\end{proof}

\begin{prop}[Stirling--Robbins bounds]\label{prop:stirling}
For every $n
\in
\semigrposint$ there is $r_{n}$ with
$$n!
=
\sqrt{2 \pi n} \rbk{\frac{n}{e}}^{n} \fnexp{r_{n}},
\quad
\frac{1}{12 n + 1} < r_{n} < \frac{1}{12 n}.$$
\end{prop}

\begin{proof}
Define $a_{n}
=
\log n! - \rbk{n + \frac{1}{2}} \log n + n$.
The consecutive difference of $a_{n}$ is
$$a_{n} - a_{n + 1}
=
\rbk{n + \frac{1}{2}} \log \frac{n + 1}{n} - 1.$$
With $x
=
\frac{1}{2 n + 1}
\in
\rbk{0, 1}$ one has $\frac{n + 1}{n}
=
\frac{1 + x}{1 - x}$ and the elementary expansion
$$\rbk{n + \frac{1}{2}} \log \frac{n + 1}{n}
=
\frac{1}{2 x} \log \frac{1 + x}{1 - x}
=
\sum_{j
=
0}^{\infty} \frac{x^{2 j}}{2 j + 1}
=
1 + \frac{x^{2}}{3} + \frac{x^{4}}{5} + \dotsb,$$
where the series converges absolutely for $\abs{x} < 1$.
It follows that $0 < a_{n} - a_{n + 1}$,
with the two-sided bounds
$$\frac{x^{2}}{3}
\leq
a_{n} - a_{n + 1}
=
\sum_{j
\geq
1} \frac{x^{2 j}}{2 j + 1}
\leq
\frac{1}{3} \sum_{j
\geq
1} x^{2 j}
=
\frac{1}{3} \cdot \frac{x^{2}}{1 - x^{2}}
=
\frac{1}{3} \cdot \frac{1}{\rbk{2 n + 1}^{2} - 1}
=
\frac{1}{12} \rbk{\frac{1}{n} - \frac{1}{n + 1}}.$$
Set $b_{n}
=
a_{n} - \frac{1}{12 n}$ and $c_{n}
=
a_{n} - \frac{1}{12 n + 1}$.
Note that the upper bound above is strict for $x > 0$,
the terms with $j
\geq
2$ strictly smaller than the geometric comparison,
and the lower bound is strict as well.
It follows that the monotonicities of $b_{n}$ and $c_{n}$ are strict.
The upper bound gives $b_{n} - b_{n + 1}
=
\rbk{a_{n} - a_{n + 1}} - \frac{1}{12 n \rbk{n + 1}}
\leq
0$.
It follows that $b_{n}$ is nondecreasing.
The lower bound gives
$$c_{n} - c_{n + 1}
\geq
\frac{1}{3 \rbk{2 n + 1}^{2}} - \rbk{\frac{1}{12 n + 1} - \frac{1}{12 n + 13}}
=
\frac{12}{36 \rbk{2 n + 1}^{2}} - \frac{12}{\rbk{12 n + 1} \rbk{12 n + 13}}
\geq
0,$$
because,
for $n
\in
\semigrposint$,
it holds that
$$\begin{aligned}
&\rbk{12 n + 1} \rbk{12 n + 13} - 36 \rbk{2 n + 1}^{2}
=
\rbk{144 n^{2} + 168 n + 13} - \rbk{144 n^{2} + 144 n + 36}
\\ 
&=
24 n - 23 > 0.
\end{aligned}$$
It follows that $c_{n}$ is nonincreasing.
Since $c_{n} - b_{n}
=
\frac{1}{12 n} - \frac{1}{12 n + 1}
\to
0$ as $n
\to
\infty$,
the two monotone sequences converge to a common limit $a_{\infty}$,
and $b_{n} < a_{\infty} < c_{n}$ rearranges to
$$\frac{1}{12 n + 1} < a_{n} - a_{\infty} < \frac{1}{12 n}.$$
It remains to identify $\fnexp{a_{\infty}}
=
\sqrt{2 \pi}$.
From the definition of $a_{n}$,
$$\frac{4^{k} \rbk{k!}^{2}}{\rbk{2 k}! \sqrt{k}}
=
\fnexp{2 a_{k} - a_{2 k}} \cdot \frac{4^{k} k^{2 k + 1} e^{- 2 k}}{\rbk{2 k}^{2 k + \frac{1}{2}} e^{- 2 k} \sqrt{k}}
=
\fnexp{2 a_{k} - a_{2 k}} \cdot \frac{1}{\sqrt{2}},$$
and by Lemma \ref{lem:wallis} the left side tends to $\sqrt{\pi}$.
It follows that $\fnexp{a_{\infty}}
=
\sqrt{2} \sqrt{\pi}
=
\sqrt{2 \pi}$.
\end{proof}

The following consequence is the form in which Proposition \ref{prop:stirling} enters the concentration argument. It replaces the appeal to Binet's second formula in \cite[Eq. (21)]{ThirringWalter001}, of which it is the two-sided integrated version with the same error order.

\begin{cor}[uniform entropy bounds]\label{cor:entropy-bounds}
For $n
\in
\semigrposint$ and $a
\in
\monnat$ with $a
\leq
n$ write $\fun{s_{n}}{a}
=
-\frac{a}{n} \log \frac{a}{n}
-\frac{n - a}{n} \log \frac{n - a}{n}$ with the convention $0 \log 0
=
0$.
Then for all $n
\in
\semigrposint$ and $a
\in
\monnat$ with $a
\leq
n$, it holds that
$$\abs{\log \binom{n}{a} - n \fun{s_{n}}{a}}
\leq
\frac{3}{2} \fun{\log}{n + 1} + 3.$$
\end{cor}

\begin{proof}
For $a
\in
\setone{0, n}$ both terms vanish and the bound is trivial.
For $a
\in
\semigrposint$ with $a
\leq
n - 1$ apply Proposition \ref{prop:stirling} to the three factorials in $\binom{n}{a}
=
\frac{n!}{a! \rbk{n - a}!}$:
$$\log \binom{n}{a}
=
n \fun{s_{n}}{a}
+\frac{1}{2}
\log \frac{n}{2 \pi a \rbk{n - a}} + r_{n} - r_{a} - r_{n - a}.$$
The error terms satisfy $\abs{r_{n} - r_{a} - r_{n - a}}
\leq
\frac{1}{12} + \frac{1}{12} + \frac{1}{12} < 1$,
and from $1
\leq
a \rbk{n - a}
\leq
n^{2}$ one gets $\frac{1}{n}
\leq
\frac{n}{a \rbk{n - a}}
\leq
n$.
It follows that $$\abs{\frac{1}{2} \log \frac{n}{2 \pi a \rbk{n - a}}}
\leq
\frac{1}{2} \fun{\log}{2 \pi} + \frac{1}{2} \log n
\leq
\frac{3}{2} \fun{\log}{n + 1} + 2.$$
\end{proof}

\subsubsection{\texorpdfstring{The Bessel function \(J_{0}\)}{The Bessel function J\_\{0\}}}\label{the-bessel-function-j_0}

The Bessel function enters as the gauge average of a plane wave over the circle. Both the series and the integral representation are needed.

\begin{defn}[Bessel function]\label{def:bessel}
For $z
\in
\fldcmp$, we define $J_0$ as
$$\fun{J_{0}}{z}
=
\sum_{k
=
0}^{\infty} \frac{\rbk{-1}^{k}}{\rbk{k!}^{2}} \rbk{\frac{z}{2}}^{2 k}.$$
\end{defn}

The series in Definition \ref{def:bessel} converges on \(\fldcmp\) and defines an entire function of \(z^{2}\).

\begin{lem}[circle average]\label{lem:bessel-average}
For all $v_{1}, v_{2}
\in
\fldcmp$,
$$\int_{\mansphere^{1}}
\fnexp{\imunit \rbk{v_{1} \fun{\cos}{\phi} + v_{2} \fun{\sin}{\phi}}}
\opdmsr{\msrprb_{\mansphere^{1}}}(\phi)
=
\fun{J_{0}}{\sqrt{v_{1}^{2} + v_{2}^{2}}},$$
where the right side means the entire function $\sum_{k \geq 0}
\frac{\rbk{-1}^{k}}{\rbk{k!}^{2} 4^{k}}
\rbk{v_{1}^{2} + v_{2}^{2}}^{k}$ of $\rbk{v_{1}, v_{2}}$.
\end{lem}

\begin{proof}
Both sides are entire functions of $\rbk{v_{1}, v_{2}}
\in
\fldcmp^{2}$.
The left side because the integrand is entire in $\rbk{v_{1}, v_{2}}$ locally uniformly in $\phi$,
so that the integral is entire by Morera and Fubini.
It follows that it suffices to prove the identity for real $\rbk{v_{1}, v_{2}}$,
where it extends by the identity theorem,
applied one variable at a time.
For real $\rbk{v_{1}, v_{2}}$ write $v_{1}
=
R \fun{\cos}{\alpha}$,
$v_{2}
=
R \fun{\sin}{\alpha}$ with $R
=
\sqrt{v_{1}^{2} + v_{2}^{2}}
\geq
0$.
Then $v_{1} \fun{\cos}{\phi} + v_{2} \fun{\sin}{\phi}
=
R \fun{\cos}{\phi - \alpha}$ and by $2 \pi$-periodicity the integral equals $\int_{\mansphere^{1}} \fnexp{\imunit R \fun{\cos}{\phi}} \opdmsr{\msrprb_{\mansphere^{1}}}(\phi)$.
Expanding the exponential,
which converges uniformly in $\phi$,
$$\int_{\mansphere^{1}} \fnexp{\imunit R \fun{\cos}{\phi}} \opdmsr{\msrprb_{\mansphere^{1}}}(\phi)
=
\sum_{m
=
0}^{\infty} \frac{\rbk{\imunit R}^{m}}{m!} \cdot \int_{\mansphere^{1}} \cos^{m} \phi \opdmsr{\msrprb_{\mansphere^{1}}}(\phi).$$
The odd moments vanish,
and for $m
=
2 k$ the moment equals $\frac{2}{\pi} \cdot 2 W_{2 k} \cdot \frac{1}{2}
=
\frac{\rbk{2 k}!}{4^{k} \rbk{k!}^{2}}$ by Lemma \ref{lem:wallis} and the symmetry of $\cos^{2 k}$ over the four quarter-periods.
It follows that the integral equals $\sum_{k \geq 0} \frac{\rbk{-1}^{k} R^{2 k}}{\rbk{2 k}!} \cdot \frac{\rbk{2 k}!}{4^{k} \rbk{k!}^{2}}
=
\fun{J_{0}}{R}$.
\end{proof}

\subsubsection{The Fejér kernel}\label{the-fejuxe9r-kernel}

Section \ref{sec:decomposition} uses two approximation statements on the circle, both consequences of the Fejér kernel. Uniformly bounded trigonometric approximations of indicator functions in \(\lp^{2}\) establish cyclicity of the direct-integral representation, whereas uniqueness of Fourier coefficients for integrable functions is used to compute the center.

\begin{lem}[Fejér approximation]\label{lem:fejer}
For $N
\geq
0$ and a $2 \pi$-periodic integrable $f$ let
$$\fun{K_{N}}{\phi}
=
\sum_{j
=
- N}^{N} \rbk{1 - \frac{\abs{j}}{N + 1}} \fnexp{\imunit j \phi},
\quad
\fun{\rbk{\operatorname{Fej}_{N} f}}{\theta}
=
\int_{\mansphere^{1}} \fun{f}{\theta - \phi} \fun{K_{N}}{\phi} \opdmsr{\msrprb_{\mansphere^{1}}}(\phi),$$
a trigonometric polynomial of degree at most $N$.
\begin{enumerate}
\item $K_{N}
\geq
0$,
$\int_{\mansphere^{1}} K_{N}
\opdmsr{\msrprb_{\mansphere^{1}}}(\phi)
=
1$,
and $\sup_{\delta
\leq
\abs{\phi}
\leq
\pi} \fun{K_{N}}{\phi}
\leq
\frac{\pi^{2}}{\rbk{N + 1} \delta^{2}}$ for $0 < \delta < \pi$.
\item $\norm{\operatorname{Fej}_{N} f}_{\infty}
\leq
\norm{f}_{\infty}$,
and $\operatorname{Fej}_{N} f
\to
f$ uniformly as $N \to \infty$ when $f$ is continuous.
\item $\norm{\operatorname{Fej}_{N} f - f}_{\lp^{1}}
\to
0$ as $N \to \infty$ for every $f
\in
\lp^{1}$.
An integrable function whose Fourier coefficients all vanish is zero almost everywhere,
and for bounded measurable $f$ also $\norm{\operatorname{Fej}_{N} f - f}_{\lp^{2}}
\to
0$ as $N \to \infty$.
\end{enumerate}
\end{lem}

\begin{proof}

(1)
The geometric sum gives
$$\abs{\sum_{j
=
0}^{N} \fnexp{\imunit j \phi}}^{2}
=
\sum_{j, l
=
0}^{N} \fnexp{\imunit \rbk{j - l} \phi}
=
\sum_{r
=
- N}^{N} \rbk{N + 1 - \abs{r}} \fnexp{\imunit r \phi}
=
\rbk{N + 1} \fun{K_{N}}{\phi},$$
It follows that $K_{N}
\geq
0$,
and for $\phi
\notin
2 \pi \ringratint$ the closed form of the geometric sum gives
$$\fun{K_{N}}{\phi}
=
\frac{1}{N + 1} \rbk{\frac{\sin \frac{\rbk{N + 1} \phi}{2}}{\sin \frac{\phi}{2}}}^{2}.$$
The mean is the coefficient of $\fnexp{\imunit 0 \phi}$,
namely $1$.
For $\delta
\leq
\abs{\phi}
\leq
\pi$ the concavity of $\sin$ on $\closedinterval{0}{\frac{\pi}{2}}$ gives $\abs{\fun{\sin}{\frac{\phi}{2}}}
\geq
\frac{\abs{\phi}}{\pi}
\geq
\frac{\delta}{\pi}$,
and the numerator is at most $1$.

(2)
The operator $\operatorname{Fej}_{N}$ averages translates of $f$ against a nonnegative kernel of total mass one.
It follows that $\norm{\operatorname{Fej}_{N} f}_{\infty}
\leq
\norm{f}_{\infty}$.
Let $f$ be continuous,
and fix $\delta
>
0$.
Uniform continuity gives
$$\abs{\fun{\rbk{\operatorname{Fej}_{N} f}}{\theta} - \fun{f}{\theta}}
\leq
\int_{\abs{\phi} < \delta} \fun{K_{N}}{\phi} \abs{\fun{f}{\theta - \phi} - \fun{f}{\theta}} \opdmsr{\msrprb_{\mansphere^{1}}}(\phi) + 2 \norm{f}_{\infty} \cdot \frac{\pi^{2}}{\rbk{N + 1} \delta^{2}},$$
and the first term is bounded by the modulus of continuity of $f$ at scale $\delta$.
Let $N
\to
\infty$,
then $\delta
\to
0$.

(3)
Translation is continuous in $\lp^{1}$.
Given $f
\in
\lp^{1}$ and $\epsilon > 0$,
choose a continuous $g$ with $\norm{f - g}_{\lp^{1}} < \epsilon$,
possible by the density of continuous functions in $\lp^{1}$.
Then $$\norm{\fun{f}{\cdot - \phi} - f}_{\lp^{1}}
\leq
2 \epsilon + \norm{\fun{g}{\cdot - \phi} - g}_{\lp^{1}}
\leq
2 \epsilon + \norm{\fun{g}{\cdot - \phi} - g}_{\infty},$$
small for small $\abs{\phi}$ by uniform continuity of $g$.
The triangle inequality under the integral gives
$$\norm{\operatorname{Fej}_{N} f - f}_{\lp^{1}}
\leq
\int_{\mansphere^{1}} \fun{K_{N}}{\phi} \norm{\fun{f}{\cdot - \phi} - f}_{\lp^{1}} \opdmsr{\msrprb_{\mansphere^{1}}}(\phi),$$
and the same $\delta$-splitting as in (2) with the $\lp^{1}$-modulus of continuity gives the convergence.
The Fourier coefficients of $\operatorname{Fej}_{N} f$ are $\rbk{1 - \frac{\abs{j}}{N + 1}} \fun{\hat{f}}{j}$.
If all $\fun{\hat{f}}{j}$ vanish,
then $\operatorname{Fej}_{N} f
=
0$ for every $N$ and $f
=
0$ in $\lp^{1}$.
For bounded $f$,
the remaining convergence is
$$\norm{\operatorname{Fej}_{N} f - f}_{\lp^{2}}^{2}
\leq
\norm{\operatorname{Fej}_{N} f - f}_{\infty} \norm{\operatorname{Fej}_{N} f - f}_{\lp^{1}}
\leq
2 \norm{f}_{\infty} \norm{\operatorname{Fej}_{N} f - f}_{\lp^{1}}
\to
0
\quad \rbk{N \to \infty}.$$
\end{proof}

\subsubsection{Vitali's theorem}\label{vitalis-theorem}

Vitali's theorem and its several-variable corollary upgrade pointwise convergence of generating functions to convergence of all derivatives.

\begin{prop}[Vitali's theorem]\label{prop:vitali}
Let $D
\subset
\fldcmp$ be a connected open set,
$\seq{f_{\Omega}}{\Omega
\geq
1}$ a sequence of holomorphic functions on $D$ with $\sup_{\Omega \geq 1} \sup_{z
\in
K} \abs{\fun{f_{\Omega}}{z}} < \infty$ for every compact $K
\subset
D$,
and suppose $\fun{f_{\Omega}}{t}$ converges for every $t$ in a subset $A
\subset
D$ having an accumulation point in $D$.
Then $f_{\Omega}$ converges locally uniformly on $D$ to a holomorphic function $f$,
and every derivative $f_{\Omega}^{\rbk{k}}$ converges locally uniformly to $f^{\rbk{k}}$.
\end{prop}

\begin{proof}
Local boundedness and the Cauchy integral formula give a local uniform Lipschitz bound.
For a closed disc $\gtclos{\fun{D}{a, 2 r}}
\subset
D$ with $\sup_{\Omega \geq 1} \sup_{\fun{D}{a, 2 r}} \abs{f_{\Omega}}
=
C$ and $z, w
\in
\fun{D}{a, r}$,
$$\abs{\fun{f_{\Omega}}{z} - \fun{f_{\Omega}}{w}}
=
\abs{\frac{1}{2 \pi \imunit} \oint_{\abs{\zeta - a}
=
2 r} \fun{f_{\Omega}}{\zeta} \rbk{\frac{1}{\zeta - z} - \frac{1}{\zeta - w}} \opdmsr{\zeta}}
\leq
\frac{2 C}{r} \abs{z - w}.$$
It follows that $\rbk{f_{\Omega}}$ is locally equicontinuous and locally bounded,
The Arzelà--Ascoli theorem gives uniform convergence on each compact set after passage to a subsequence.
A diagonal argument over a countable compact exhaustion gives a locally uniformly convergent subsubsequence.
Its limit is holomorphic by the uniform limit theorem of Weierstrass.
Any two such limit functions agree on $A$.
Since $A$ has an accumulation point in $D$,
the identity theorem makes the limit functions equal on $D$.
A sequence in a metrizable space,
here the space of continuous functions on compacta with the topology of locally uniform convergence,
converges as soon as every subsequence has a subsubsequence converging to one and the same limit.
This gives $f_{\Omega}
\to
f$ locally uniformly as $\Omega \to \infty$.
The convergence of derivatives follows from the Cauchy formula for $f^{\rbk{k}}$ over circles contained in $D$.
\end{proof}

\begin{cor}[several variables, real convergence set]\label{cor:vitali-several}
Let $r
\geq
1$,
let $B
\subset
\fldreal^{r}$ be an open box,
and let $D
=
\set{z
\in
\fldcmp^{r}}{\opreal z
\in
B, \ \abs{\opimag z_{j}} < \delta}$ for some $\delta > 0$.
Let $\seq{F_{\Omega}}{\Omega
\geq
1}$ be a sequence of holomorphic functions on $D$,
uniformly bounded on compact subsets of $D$,
and convergent at every point of a product set $A
=
A_{1} \times \dotsb \times A_{r}$,
where each $A_{j}
\subset
B_{j}$ has an accumulation point in $B_{j}$.
For example $A
=
B$.
Then $F_{\Omega}$ converges locally uniformly on $D$ to a holomorphic limit $F$,
together with all partial derivatives.
\end{cor}

\begin{proof}
We prove the assertion by induction on $r$.
For $r
=
1$ this is Proposition \ref{prop:vitali} with $A
=
A_{1}$.
For the step let $z'
=
\rbk{z_{1}, \dotsc, z_{r - 1}}$ and fix $z_{r}
=
t_{r}
\in
A_{r}$ real.
The functions $z'
\mapsto
\fun{F_{\Omega}}{z', t_{r}}$ satisfy the inductive hypotheses in $r - 1$ variables with the convergence set $A_{1} \times \dotsb \times A_{r - 1}$.
It follows that $F_{\Omega}$ converges pointwise on $D' \times A_{r}$ where $D'$ is the corresponding complex box.
Now fix $z'
\in
D'$.
The functions $z_{r}
\mapsto
\fun{F_{\Omega}}{z', z_{r}}$ are holomorphic on the strip over $B_{r}$,
uniformly bounded on compacta,
and converge for $z_{r}
\in
A_{r}$,
which has an accumulation point in the strip.
Proposition \ref{prop:vitali} shows that they converge locally uniformly in $z_{r}$,
and moreover the limit is locally uniform in $z'$ as well.
On a compact product set $K' \times K_{r}$,
the family $\rbk{z', z_{r}}
\mapsto
\fun{F_{\Omega}}{z', z_{r}}$ is equicontinuous.
This follows by applying the Cauchy--Lipschitz bound of the previous proof in each variable.
The two preceding steps also give pointwise convergence at every point of $K' \times K_{r}$.
Equicontinuity upgrades pointwise to uniform convergence on compacta.
The limit is holomorphic,
because it is a locally uniform limit of holomorphic functions.
It is represented by the iterated Cauchy integral over a product of circles,
and the integral representation passes to the locally uniform limit.
Derivative convergence again follows from Cauchy's formula in one variable at a time.
\end{proof}

\subsection{Infinite Tensor Products}\label{sec:itp}

The Hilbert spaces on which the BCS model lives at infinite volume are the incomplete infinite tensor product spaces. The construction is developed from scratch for a sequence of Hilbert spaces in the two forms needed in the main text. For pure product-state representations each factor is \(\fldcmp^{2}\). For GNS representations of thermal product states each factor is the Hilbert--Schmidt space of \(\spmat{2}{\fldcmp}\).

\subsubsection{Infinite products of complex numbers}\label{infinite-products-of-complex-numbers}

The convergence and tail estimates for unordered complex products provide the scalar input for the infinite tensor-product inner products. The lemmas below characterize nonzero limits and control finite tails.

Throughout, products \(\prod_{p \in \semigrposint} z_{p}\) over \(p
\in
\semigrposint\) are understood as limits of the net \(\net{\prod_{p \in \Lambda} z_{p}}{\Lambda \nearrow \semigrposint}\) ordered by inclusion. This unordered convergence is the correct notion for tensor products, where no ordering of the factors is distinguished.

\begin{lem}[absolutely convergent products]\label{lem:products}
Let $\seq{z_{p}}{p
\in
\semigrposint}$ be complex numbers with $$\sum_{p \in \semigrposint} \abs{1 - z_{p}}
=
\delta < \infty.$$
Then the net
$\net{\prod_{p \in \Lambda} z_{p}}{\Lambda \nearrow \semigrposint}$
converges to a limit $\prod_{p \in \semigrposint} z_{p}$ with $\abs{\prod_{p \in \semigrposint} z_{p}}
\leq
\fnexp{\delta}$,
the limit vanishes if and only if some factor vanishes,
and for every finite $\Lambda_{0}$,
$$\abs{\prod_{p \in \semigrposint} z_{p} - \prod_{p
\in
\Lambda_{0}} z_{p}}
\leq
\fnexp{2 \delta} \sum_{p
\notin
\Lambda_{0}} \abs{1 - z_{p}}.$$
\end{lem}

\begin{proof}
Note that
$\abs{z_{p}}
\leq
1 + \abs{1 - z_{p}}
\leq
\fnexp{\abs{1 - z_{p}}}$.
This will be used to bound all partial products by $\fnexp{\delta}$.
For finite $\Lambda
\subset
\Lambda'$,
the telescoping sum over the elements $q_{1}, \dotsc, q_{m}$ of $\Lambda' \setminus \Lambda$ gives
$$\abs{\prod_{p \in \Lambda'} z_{p} - \prod_{p \in \Lambda} z_{p}}
\leq
\sum_{j
=
1}^{m} \abs{z_{q_{1}} \dotsm z_{q_{j - 1}}} \cdot \abs{1 - z_{q_{j}}} \cdot \abs{\prod_{p \in \Lambda} z_{p}}
\leq
\fnexp{\delta} \cdot \fnexp{\delta} \sum_{q
\in
\Lambda' \setminus \Lambda} \abs{1 - z_{q}}.$$
The right side is small once $\Lambda$ contains a large finite set.
It follows that the net is Cauchy and the tail estimate follows with $\Lambda_{0}
=
\Lambda$.
Boundedness is clear.
If no factor vanishes,
choose $\Lambda_{1}$ with $\sum_{p
\notin
\Lambda_{1}} \abs{1 - z_{p}}
\leq
\frac{1}{2}$.
The telescoping bound applied from the empty set gives,
for every finite $G$ with $G \cap \Lambda_{1}
=
\emptyset$,
$$\abs{\prod_{q
\in
G} z_{q} - 1}
\leq
\fnexp{1 / 2} \sum_{q
\in
G} \abs{1 - z_{q}}
\leq
\frac{\fnexp{1 / 2}}{2} < \frac{5}{6},$$
It follows that $\abs{\prod_{q
\in
G} z_{q}}
\geq
\frac{1}{6}$ for all such $G$.
This implies that $$\abs{\prod_{p \in \semigrposint} z_{p}}
=
\lim_{\Lambda \uparrow \semigrposint}
\abs{\prod_{p
\in
\Lambda \setminus \Lambda_{1}} z_{p}} \cdot \abs{\prod_{p
\in
\Lambda_{1}} z_{p}}
\geq
\frac{1}{6} \abs{\prod_{p
\in
\Lambda_{1}} z_{p}} > 0.$$
If some factor vanishes the net is eventually $0$.
\end{proof}

\begin{lem}[unordered convergence forces absolute convergence]\label{lem:unordered}
Let $\seq{z_{p}}{p
\in
\semigrposint}$ satisfy $\abs{z_{p}}
\leq
1$ for all $p$ and suppose the net
$\net{\prod_{p \in \Lambda} z_{p}}{\Lambda \nearrow \semigrposint}$
converges to a limit $L
\neq
0$.
Then $\sum_{p \in \semigrposint} \abs{1 - z_{p}} < \infty$.
\end{lem}

\begin{proof}
Since $\abs{L} > 0$ and the moduli of the partial products are nonincreasing in $\Lambda$,
all partial products satisfy $\abs{\prod_{p \in \Lambda} z_{p}}
\geq
\abs{L} > 0$.
In particular no factor vanishes.
Choose $\Lambda_{1}$ with $\abs{\prod_{p \in \Lambda} z_{p} - L}
\leq
\frac{\abs{L}}{8}$ for all $\Lambda
\supseteq
\Lambda_{1}$.
For every finite $G$ with $G \cap \Lambda_{1}
=
\emptyset$,
the Cauchy bounds for $\Lambda_{1}$ and $\Lambda_{1} \cup G$ give $\abs{\prod_{p
\in
\Lambda_{1} \cup G} z_{p} - \prod_{p
\in
\Lambda_{1}} z_{p}}
\leq
\frac{\abs{L}}{4}$,
and dividing by $\abs{\prod_{p
\in
\Lambda_{1}} z_{p}}
\geq
\abs{L} - \frac{\abs{L}}{8}
\geq
\frac{\abs{L}}{2}$,
\begin{equation}\label{eq:tail-product-near-one}
\abs{\prod_{q
\in
G} z_{q} - 1}
\leq
\frac{\abs{L} / 4}{\abs{L} / 2}
=
\frac{1}{2}.
\end{equation}
We first control the moduli.
For every finite $G \cap \Lambda_{1}
=
\emptyset$,
\eqref{eq:tail-product-near-one} gives $\prod_{q
\in
G} \abs{z_{q}}
\geq
\frac{1}{2}$.
It follows that $$\sum_{q
\in
G} \rbk{1 - \abs{z_{q}}}
\leq
\sum_{q
\in
G} \rbk{- \fun{\log}{\abs{z_{q}}}}
\leq
\log 2.$$
Since $G$ is arbitrary we obtain
$\sum_{q
\notin
\Lambda_{1}} \rbk{1 - \abs{z_{q}}}
\leq
\log 2$.

We next control the phases.
Write $z_{q}
=
\abs{z_{q}} \fnexp{\imunit \theta_{q}}$ with $\theta_{q}
\in
\rbk{- \pi, \pi}$.
Applying \eqref{eq:tail-product-near-one} to the singleton $G
=
\setone{q}$ gives $\abs{z_{q} - 1}
\leq
\frac{1}{2}$.
It follows that $\abs{\theta_{q}}
\leq
\frac{\pi}{3}$.
Suppose $\sum_{q
\notin
\Lambda_{1}, \theta_{q} > 0} \theta_{q}
=
\infty$.
Since each term is at most $\frac{\pi}{3}$,
the increasing partial sums over an enumeration of these $q$ hit the interval $\sqbk{\frac{\pi}{2}, \frac{\pi}{2} + \frac{\pi}{3}}$.
It follows that there is a finite $G$ with $\Theta
=
\sum_{q
\in
G} \theta_{q}
\in
\sqbk{\frac{\pi}{2}, \frac{5 \pi}{6}}$.
For this $G$,
$$\abs{\prod_{q
\in
G} z_{q} - 1}
\geq
\abs{\fun{\opreal}{\prod_{q
\in
G} z_{q}} - 1}
=
1 - \cos \Theta \prod_{q
\in
G} \abs{z_{q}}
\geq
1,$$
because $\cos \Theta
\leq
0$ on $\sqbk{\frac{\pi}{2}, \frac{5 \pi}{6}}$.
It follows that the real part of the product,
$\cos \Theta \prod_{q
\in
G} \abs{z_{q}}$,
is nonpositive.
This contradicts \eqref{eq:tail-product-near-one}.
The same argument applies to the negative phases.
It follows that $\sum_{q
\notin
\Lambda_{1}} \abs{\theta_{q}} < \infty$.
The modulus and phase estimates give
$$\abs{1 - z_{q}}
\leq
\rbk{1 - \abs{z_{q}}} + \abs{z_{q}} \cdot \abs{1 - \fnexp{\imunit \theta_{q}}}
\leq
\rbk{1 - \abs{z_{q}}} + \abs{\theta_{q}},$$
and both series converge.
\end{proof}

\begin{lem}[divergence to zero]\label{lem:divergence-zero}
Let $\abs{z_{p}}
\leq
1$ for all $p$.
If $\sum_{p \in \semigrposint} \rbk{1 - \abs{z_{p}}}
=
\infty$ then
$\prod_{p \in \Lambda} z_{p}
\to
0$ as $\Lambda \nearrow \semigrposint$.
If $\sum_{p \in \semigrposint} \rbk{1 - \abs{z_{p}}} < \infty$ then $\inf_{\substack{\Lambda \subset \semigrposint \\ \abscard{\Lambda} < \infty}} \abs{\prod_{p \in \Lambda} z_{p}} > 0$ or some factor vanishes.
\end{lem}

\begin{proof}
From $\abs{z}
\leq
\fnexp{- \rbk{1 - \abs{z}}}$,
valid for $0
\leq
\abs{z}
\leq
1$,
one gets
$$\abs{\prod_{p \in \Lambda} z_{p}}
\leq
\fnexp{- \sum_{p \in \Lambda} \rbk{1 - \abs{z_{p}}}}
\to
0
\quad
\rbk{\Lambda \nearrow \semigrposint}$$ in the first case.
In the second case,
if no factor vanishes,
we obtain
$$\abs{\prod_{p \in \Lambda} z_{p}}
=
\prod_{p \in \Lambda} \abs{z_{p}}
\geq
\prod_{p \in \semigrposint} \abs{z_{p}} > 0$$ by Lemma \ref{lem:products} applied to the moduli.
\end{proof}

\subsubsection{Sequences of unit vectors and their equivalence classes}\label{sequences-of-unit-vectors-and-their-equivalence-classes}

Equivalence and weak equivalence label the incomplete tensor-product sectors and encode the effect of slotwise phase changes.

Let \(\seq{h_{p}}{p
\geq
1}\) be a sequence of Hilbert spaces with \(2
\leq
\dim h_{p}
\leq
\infty\), each separable. For the quasi-spin model \(h_{p}
=
\fldcmp^{2}\) throughout Sections \ref{sec:intensive}--\ref{sec:time}, and \(h_{p}
=
\spmat{2}{\fldcmp}\) with the Hilbert--Schmidt inner product in the thermal sections.

Adapted sequences and the equivalence relations are defined in Definition \ref{def:c-sequences}.

\begin{ex}[adapted and non-adapted sequences]
Let $h_{p}
=
\fldcmp^{2}$ for every $p
\in
\semigrposint$,
and fix an orthonormal basis $e_{0},e_{1}$ of $\fldcmp^{2}$.

The sequence $\xi
=
\seq{\xi_{p}}{p \in \semigrposint}$ defined by
$$\xi_{1}
=
2 e_{0},
\quad
\xi_{p}
=
e_{0}
\quad
\rbk{p
\geq
2}$$
is adapted because $\norm{\xi_{p}}
=
1$ for every $p
\geq
2$.
It is not an adapted sequence with unit entries because $\norm{\xi_{1}}
=
2$.

The nonconstant sequence $\eta
=
\seq{\eta_{p}}{p \in \semigrposint}$ defined by
$$\eta_{p}
=
\begin{cases}
e_{0}, & p \text{ is odd},\\
e_{1}, & p \text{ is even}
\end{cases}
\quad
\rbk{p
\in
\semigrposint}$$
is an adapted sequence with unit entries because $\norm{\eta_{p}}
=
1$ for every $p
\in
\semigrposint$.

Finally,
the nonconstant sequence $\zeta
=
\seq{\zeta_{p}}{p \in \semigrposint}$ defined by
$$\zeta_{p}
=
\rbk{1 + \frac{1}{p}} e_{0}
\quad
\rbk{p
\in
\semigrposint}$$
is not adapted.
Indeed,
$\norm{\zeta_{p}}
=
1 + 1 / p
\neq
1$ for every $p
\in
\semigrposint$,
even though $\norm{\zeta_{p}}
\to
1$ as $p
\to
\infty$.
\end{ex}

Both relations ignore finitely many slots. It follows that finite modifications of an adapted sequence stay in its class. That the relations are indeed equivalence relations requires an argument only for transitivity.

\begin{lem}[transitivity]\label{lem:transitivity}
The relations $
\approx
$ and $
\approx
_{w}$ are equivalence relations on adapted sequences,
and $\xi
\approx
\eta$ implies $\xi
\approx
_{w} \eta$.
If $\xi
\approx
\eta$ and $\xi
\approx
\zeta$,
then the stronger conclusion
$\sum_{p \in \semigrposint} \abs{1 - \bkt{\eta_{p}}{\zeta_{p}}} < \infty$
holds.
\end{lem}

\begin{proof}
Reflexivity and symmetry are clear from $\bkt{\eta_{p}}{\xi_{p}}
=
\cmpconj{\bkt{\xi_{p}}{\eta_{p}}}$,
and the implication from $\abs{1 - \abs{w}}
\leq
\abs{1 - w}$.
For transitivity it suffices,
ignoring the finitely many slots where a norm differs from $1$,
to consider unit vectors.
For unit vectors $u, v, w
\in
h_{p}$,
the decomposition defined by the projection $P_{v}$ gives
$$\bkt{u}{w}
=
\bkt{u}{v} \bkt{v}{w} + \bkt{u}{\rbk{1 - P_{v}} w},$$
and by the Cauchy--Schwarz inequality
$$\abs{\bkt{u}{\rbk{1 - P_{v}} w}}
\leq
\norm{\rbk{1 - P_{v}} u} \norm{\rbk{1 - P_{v}} w}
=
\sqrt{1 - \abs{\bkt{u}{v}}^{2}} \sqrt{1 - \abs{\bkt{v}{w}}^{2}}.$$
With $a
=
\abs{1 - \bkt{u}{v}}$ and $b
=
\abs{1 - \bkt{v}{w}}$,
the first square-root factor satisfies
$$1 - \abs{\bkt{u}{v}}^{2}
\leq
2 \abs{1 - \bkt{u}{v}}
=
2 a,$$
and the analogous bound holds with $b$.
The product term obeys
$$\abs{1 - \bkt{u}{v} \bkt{v}{w}}
\leq
a + b.$$
Indeed,
for $\abs{z}
\leq
1$,
$$\abs{1 - z w}
\leq
\abs{1 - z} + \abs{z} \abs{1 - w}
\leq
\abs{1 - z} + \abs{1 - w}.$$
The elementary inner-product estimate is
$$\abs{1 - \bkt{u}{w}}
\leq
a + b + 2 \sqrt{a b}
\leq
2 \rbk{a + b},$$
where the last inequality follows from $2 \sqrt{a b}
\leq
a + b$.
Summing over $p$ proves transitivity of $\approx$ and the final statement.
For the relation $\approx_{w}$,
we argue as follows.
Replacing $u, v, w$ by suitable unimodular multiples makes $\bkt{u}{v}
\geq
0$ and $\bkt{v}{w}
\geq
0$ without changing any modulus,
and then $\abs{1 - \bkt{u}{v}}
=
1 - \abs{\bkt{u}{v}}$,
$\abs{1 - \bkt{v}{w}}
=
1 - \abs{\bkt{v}{w}}$,
while $1 - \abs{\bkt{u}{w}}
\leq
\abs{1 - \bkt{u}{w}}$.
\end{proof}

The distinction between the two equivalence relations is already visible for sequences contained in a fixed one-dimensional subspace.

\begin{ex}[weak equivalence without equivalence]
Let $h_{p}
=
\fldcmp^{2}$ for every $p
\in
\semigrposint$,
fix a unit vector $e
\in
\fldcmp^{2}$,
and consider the adapted sequences with unit entries
$$\xi
=
\seq{\xi_{p}}{p \in \semigrposint},
\quad
\eta
=
\seq{\eta_{p}}{p \in \semigrposint},
\quad
\xi_{p}
=
e,
\quad
\eta_{p}
=
- e
\quad
\rbk{p \in \semigrposint}.$$
For every $p
\in
\semigrposint$,
the phase-free defect is
$$\abs{1 - \abs{\bkt{\xi_{p}}{\eta_{p}}}}
=
\abs{1 - \abs{- 1}}
=
0
\quad
\rbk{p \in \semigrposint}.$$
It follows that $\xi
\approx
_{w} \eta$.
In contrast, the calculation
$$\begin{aligned}
\abs{1 - \bkt{\xi_{p}}{\eta_{p}}}
&=
\abs{1 - \rbk{- 1}}
=
2
\quad
\rbk{p \in \semigrposint},
\\ 
\sum_{p \in \semigrposint} \abs{1 - \bkt{\xi_{p}}{\eta_{p}}}
&=
\infty.
\end{aligned}$$
shows that the absolute defects are constant and their sum diverges.
Thus $\xi
\not\approx
\eta$.
\end{ex}

The equivalences take the following explicit form for spin configurations. A spin configuration is a sequence \(\physvec{\omega}
=
\seq{\physvec{u}_{p}}{p
\geq
1}\) of unit vectors in \(\fldreal^{3}\), with associated unit vectors \(\xi_{p}
=
\xi_{\physvec{u}_{p}}^{+}
\in
\fldcmp^{2}\).

\begin{lem}[geometric form of weak equivalence]\label{lem:geometric-equivalence}
For two spin configurations $\physvec{\omega}, \physvec{\omega}'$,
write $\xi'_{p}
=
\xi_{\physvec{u}'_{p}}^{+}$.
The weak-equivalence criterion is:
$$\seq{\xi_{p}}{p
\in
\semigrposint}
\approx
_{w} \seq{\xi'_{p}}{p \in \semigrposint}
\iff
\sum_{p \in \semigrposint} \abs{\physvec{u}_{p} - \physvec{u}'_{p}}^{2} < \infty.$$
\end{lem}

\begin{proof}
Lemma \ref{lem:overlap} implies that
$$\abs{\bkt{\xi'_{p}}{\xi_{p}}}
=
\sqrt{x_{p}},
\quad
x_{p}
=
\frac{1 + \physvec{u}_{p} \cdot \physvec{u}'_{p}}{2}
=
1 - \frac{\abs{\physvec{u}_{p} - \physvec{u}'_{p}}^{2}}{4}.$$
For $x
\in
\closedinterval{0}{1}$ one has $\frac{1 - x}{2}
\leq
1 - \sqrt{x}
\leq
1 - x$.
It follows that $1 - \abs{\bkt{\cdot}{\cdot}}$ satisfies two-sided bounds by constant multiples of
$\abs{\physvec{u}_{p} - \physvec{u}'_{p}}^{2}$,
and the two summability conditions coincide.
\end{proof}

\subsubsection{The incomplete tensor product space}\label{the-incomplete-tensor-product-space}

The complete product space of Definition \ref{def:complete} is the orthogonal direct sum of sectors associated with distinct equivalence classes. The incomplete tensor product space constructed below is one such sector, obtained by fixing a single equivalence class.

\begin{lem}[incomplete tensor product form]\label{lem:itps-form}
Fix an adapted sequence $\xi$.
Its equivalence class $\fun{C}{\xi}$ is the one defined in Definition \ref{def:c-sequences}.
For every $\eta, \zeta
\in
\fun{C}{\xi}$,
the product
$$\fun{\opform{q}_{\fun{C}{\xi}}}{\eta,\zeta}
=
\prod_{p \in \semigrposint} \bkt{\eta_{p}}{\zeta_{p}}_{h_{p}}$$
converges absolutely,
where $\bkt{\cdot}{\cdot}_{h_{p}}$ denotes the inner product of $h_{p}$.
The formula extends uniquely to a positive semidefinite sesquilinear form on the free complex vector space over $\fun{C}{\xi}$.
\end{lem}

\begin{proof}
For convergence,
Lemma \ref{lem:transitivity} applied with the reference $\xi$ shows that
every $\eta, \zeta
\in
\fun{C}{\xi}$ satisfies
$$\sum_{p \in \semigrposint} \abs{1 - \bkt{\eta_{p}}{\zeta_{p}}_{h_{p}}}
<
\infty.$$
It follows that Lemma \ref{lem:products} applies.

For positivity,
fix finitely many $\eta^{\rbk{1}}, \dotsc, \eta^{\rbk{m}}
\in
\fun{C}{\xi}$ and a finite set $\Lambda$.
For $j, k
\in
\intint{1..m}$ define
$$G^{\Lambda}_{j k}
=
\prod_{p
\in
\Lambda} \bkt{\eta^{\rbk{j}}_{p}}{\eta^{\rbk{k}}_{p}}_{h_{p}}.$$
This is the Gram matrix of the vectors $\bigotimes_{p
\in
\Lambda} \eta^{\rbk{j}}_{p}$ in the finite tensor product $\bigotimes_{p
\in
\Lambda} h_{p}$.
The matrix $G^{\Lambda}$ is therefore positive semidefinite.
As $\Lambda$ increases,
the entries satisfy
$$\net{G^{\Lambda}_{j k}}{\Lambda \nearrow \semigrposint}
\to
G_{j k}
=
\fun{\opform{q}_{\fun{C}{\xi}}}{\eta^{\rbk{j}},\eta^{\rbk{k}}}
\quad \rbk{\Lambda \nearrow \semigrposint}$$
by Lemma \ref{lem:products}.
The cone of positive semidefinite matrices is closed.
It follows that every Gram matrix of $\opform{q}_{\fun{C}{\xi}}$ is positive semidefinite.
This proves positivity of $\opform{q}_{\fun{C}{\xi}}$ on the free vector space.
\end{proof}

\begin{defn}[incomplete tensor product space]\label{def:itps}
Fix an adapted sequence $\xi$,
and let $\fun{C}{\xi}$ be its equivalence class from Definition \ref{def:c-sequences}.
On the free complex vector space over $\fun{C}{\xi}$,
let $\opform{q}_{\fun{C}{\xi}}$ be the positive semidefinite sesquilinear form of Lemma \ref{lem:itps-form}.
The incomplete tensor product space $\sphilb{H}_{\xi}$ is the completion of the quotient of this space by
$\Ker \opform{q}_{\fun{C}{\xi}}$.
The image of $\eta
\in
\fun{C}{\xi}$ in $\sphilb{H}_{\xi}$ is written $\bigotimes_{p \in \semigrposint} \eta_{p}$ and called a product vector.
\end{defn}

\begin{lem}[slotwise linearity]\label{lem:slotwise}
Let $\eta
\in
\fun{C}{\xi}$,
let $q$ be a slot,
and let $\eta_{q}
=
\alpha u + \beta v$ with $u, v
\in
h_{q}$.
Denote by $\eta\sqbk{u}$ and $\eta\sqbk{v}$ the sequences obtained from $\eta$ by replacing the slot $q$ entry by $u$ respectively $v$.
Both are adapted and equivalent to $\xi$.
The product vector decomposes linearly in $\sphilb{H}_{\xi}$:
$$\bigotimes_{p \in \semigrposint} \eta_{p}
=
\alpha \bigotimes_{p \in \semigrposint} \eta\sqbk{u}_{p} + \beta \bigotimes_{p \in \semigrposint} \eta\sqbk{v}_{p}.$$
\end{lem}

\begin{proof}
Both replacement sequences are finite modifications of $\eta$.
They therefore lie in $\fun{C}{\xi}$.
The difference vector $d
=
\eta - \alpha \eta\sqbk{u} - \beta \eta\sqbk{v}$
is understood in the free vector space.
For any $\zeta
\in
\fun{C}{\xi}$,
the form evaluates to
$$\fun{\opform{q}_{\fun{C}{\xi}}}{\zeta,d}
=
\rbk{\bkt{\zeta_{q}}{\eta_{q}}_{h_{q}} - \alpha \bkt{\zeta_{q}}{u}_{h_{q}} - \beta \bkt{\zeta_{q}}{v}_{h_{q}}} \prod_{p
\neq
q} \bkt{\zeta_{p}}{\eta_{p}}_{h_{p}}
=
0,$$
where the products over $p
\neq
q$ converge by Lemma \ref{lem:products} and the first factor vanishes by linearity of the inner product of $h_{q}$ in its second argument.
It follows that $\fun{\opform{q}_{\fun{C}{\xi}}}{d,d}
=
0$ after expansion in the second argument.
This proves that $d
\in
\Ker \opform{q}_{\fun{C}{\xi}}$.
\end{proof}

\begin{prop}[flip basis and separability]\label{prop:flip-basis}
Fix an adapted sequence $\xi$ with $\norm{\xi_{p}}
=
1$ for all $p$,
and for each $p$ fix an orthonormal basis $\fun{e_{p}}{0}, \fun{e_{p}}{1}, \dotsc$ of $h_{p}$ with $\fun{e_{p}}{0}
=
\xi_{p}$.
Let $\mathcal{F}$ be the set of maps $F
\colon
\semigrposint
\to
\monnat$ with finite support.
For $F
\in
\mathcal{F}$,
let $\xi^{F}$ be the sequence with entries $\fun{e_{p}}{\fun{F}{p}}$ on the support of $F$ and $\xi_{p}$ elsewhere.
Then $\set{\bigotimes_{p \in \semigrposint} \xi^{F}_{p}}{F
\in
\mathcal{F}}$ is an orthonormal basis of $\sphilb{H}_{\xi}$.
In particular $\sphilb{H}_{\xi}$ is separable when each $h_{p}$ is.
\end{prop}

\begin{proof}
Suppose first that $F
\neq
F'$.
The two sequences take distinct basis vectors in at least one slot,
which gives a zero factor in the product defining their inner product.
If $F
=
F'$,
every factor equals one.
This proves that the stated family is orthonormal.

It remains to prove totality.
Let $\eta
\in
\fun{C}{\xi}$,
and use the same symbol for the corresponding product vector.
For a finite $\Lambda$ containing all slots where $\norm{\eta_{p}}
\neq
1$,
let $\eta^{\rbk{\Lambda}}$ be the sequence equal to $\eta$ on $\Lambda$ and to $\xi$ outside.
The approximation error satisfies
$$\begin{aligned}
&\norm{\eta - \eta^{\rbk{\Lambda}}}^{2}
=
\norm{\eta}^{2} + \norm{\eta^{\rbk{\Lambda}}}^{2} - 2 \opreal \bkt{\eta^{\rbk{\Lambda}}}{\eta}
\\ 
&=
\norm{\eta}^{2}
+\prod_{p \in \Lambda} \norm{\eta_{p}}^{2}
-2 \fun{\opreal}{\prod_{p \in \Lambda}
\norm{\eta_{p}}^{2}
\prod_{p \notin \Lambda}
\bkt{\xi_{p}}{\eta_{p}}}.
\end{aligned}$$
As $\Lambda$ increases,
$\prod_{p
\in
\Lambda} \norm{\eta_{p}}^{2}$ stabilizes at $\norm{\eta}^{2}$.
The tail-product net
$\prod_{p \notin \Lambda} \bkt{\xi_{p}}{\eta_{p}}
\to
1$ as $\Lambda \nearrow \semigrposint$ by the tail estimate of Lemma \ref{lem:products},
because $\sum_{p \in \semigrposint} \abs{1 - \bkt{\xi_{p}}{\eta_{p}}}
<
\infty$.
It follows that
$\norm{\eta - \eta^{\rbk{\Lambda}}}
\to
0$ as $\Lambda \nearrow \semigrposint$.

Fix such a finite set $\Lambda$.
For $p
\in
\Lambda$,
write
$\eta_{p}
=
\sum_{j \geq 0} c_{p j} \fun{e_{p}}{j}$
and set
$v_{p,N}
=
\sum_{j = 0}^{N} c_{p j} \fun{e_{p}}{j}$.
Let $\eta^{\rbk{\Lambda,N}}$ be equal to $v_{p,N}$ on $\Lambda$ and to $\xi_{p}$ outside.
Repeated application of Lemma \ref{lem:slotwise} shows that its product vector is a finite linear combination of the $\xi^{F}$.
The finite-dimensional approximants satisfy
$$\begin{aligned}
\norm{\eta^{\rbk{\Lambda}} - \eta^{\rbk{\Lambda,N}}}^{2}
=
\prod_{p \in \Lambda} \norm{\eta_{p}}^{2}
+\prod_{p \in \Lambda} \norm{v_{p,N}}^{2}
-2 \fun{\opreal}{\prod_{p \in \Lambda} \bkt{v_{p,N}}{\eta_{p}}}
\to
0
\quad \rbk{N \to \infty}.
\end{aligned}$$
Indeed,
$v_{p,N}
\to
\eta_{p}$ in $h_{p}$ as $N \to \infty$ for every $p
\in
\Lambda$,
and only finitely many factors occur.
It follows that every $\eta^{\rbk{\Lambda}}$ belongs to the closed linear span of the $\xi^{F}$.
The preceding tail approximation gives the same conclusion for every product vector associated with an element of $\fun{C}{\xi}$.
These product vectors span a dense subspace of $\sphilb{H}_{\xi}$ by Definition \ref{def:itps},
which proves totality.

If each $h_{p}$ is separable,
the flip vectors are indexed by a countable union of countable sets.
This makes the orthonormal basis countable,
and $\sphilb{H}_{\xi}$ is separable.
\end{proof}

\begin{prop}[factorization]\label{prop:factorization}
Let $\xi$ be an adapted sequence with unit entries,
and let $\Lambda$ be a finite set of slots.
Let $\fnrestr{\xi}{\Lambda^{c}}$ denote the restricted sequence,
an adapted sequence for $\seq{h_{p}}{p
\notin
\Lambda}$,
with incomplete tensor product space $\sphilb{H}_{\fnrestr{\xi}{\Lambda^{c}}}$.
There is a unique unitary
$$U_{\Lambda} \colon \rbk{\bigotimes_{p
\in
\Lambda} h_{p}} \otimes \sphilb{H}_{\fnrestr{\xi}{\Lambda^{c}}}
\to
\sphilb{H}_{\xi}$$
mapping $\rbk{\bigotimes_{p
\in
\Lambda} v_{p}} \otimes \rbk{\bigotimes_{p
\notin
\Lambda} \zeta_{p}}$ to the product vector of the combined sequence,
for all $v_{p}
\in
h_{p}$ and all adapted $\zeta
\approx
\fnrestr{\xi}{\Lambda^{c}}$.
\end{prop}

\begin{proof}
Combined sequences of the stated form are adapted and equivalent to $\xi$,
since equivalence ignores the finitely many slots in $\Lambda$.
On elementary tensors of the stated form the inner products agree by definition of the product form.
Both sides factor into $\prod_{p
\in
\Lambda} \bkt{v_{p}}{v'_{p}} \cdot \prod_{p
\notin
\Lambda} \bkt{\zeta_{p}}{\zeta'_{p}}$.
The span of the elementary tensors of the stated form is dense on the left,
because product vectors are dense in each factor by construction.
Therefore the map extends to an isometry.
Its range contains all product vectors of sequences equivalent to $\xi$.
Given $\eta
\in
\fun{C}{\xi}$,
the combined-sequence form with $v_{p}
=
\eta_{p}$ on $\Lambda$ and $\zeta
=
\fnrestr{\eta}{\Lambda^{c}}$ reproduces $\eta$,
where $\fnrestr{\eta}{\Lambda^{c}}
\approx
\fnrestr{\xi}{\Lambda^{c}}$ because equivalence is insensitive to the removed finite set.
Product vectors are total in $\sphilb{H}_{\xi}$.
It follows that the isometry is onto.
\end{proof}

The factorization is compatible with the bases of Proposition \ref{prop:flip-basis}. Under \(U_{\Lambda}\), the tensor product of a basis vector in \(\bigotimes_{p \in \Lambda} h_{p}\) and a tail flip vector becomes a flip vector in \(\sphilb{H}_{\xi}\). Every flip vector in \(\sphilb{H}_{\xi}\) arises in this way.

Define the dense subspace with finite modification \begin{equation}\label{expedition0025006}
\begin{aligned}
\sphilb{D}_{\xi}
=
\fun{\splinspan}
{\set{\bigotimes_{p \in \semigrposint} \eta_{p}}
{\eta
\in
\fun{C}{\xi},
\abscard{\set{p
\in
\semigrposint}{\eta_{p}
\neq
\xi_{p}}}
<
\infty}}.
\end{aligned}
\end{equation} Proposition \ref{prop:flip-basis} shows that \(\sphilb{D}_{\xi}\) is dense in \(\sphilb{H}_{\xi}\).

\begin{lem}[centered infinite tensor sums]\label{lem:centered-tensor-sum}
Let $\xi$ be an adapted sequence with unit entries,
and let $a_{p}$ be a bounded self-adjoint operator on $h_{p}$ for every
$p
\in
\semigrposint$.
Suppose that
$$\bkt{\xi_{p}}{a_{p} \xi_{p}}_{h_{p}}
=
0
\quad
\rbk{p
\in
\semigrposint},
\quad
\sum_{p
\in
\semigrposint}
\norm{a_{p} \xi_{p}}_{h_{p}}^{2}
<
\infty.$$
On $\sphilb{D}_{\xi}$ in \eqref{expedition0025006},
let $A^{0}$ be the norm-convergent tensor sum
$A^{0}
=
\sum_{p
\in
\semigrposint} a_{p}$.
Then $A^{0}$ is essentially self-adjoint.
Its closure is the Stone generator of the strongly continuous infinite tensor product
of the map
$\fldreal
\ni
t
\mapsto
\bigotimes_{p
\in
\semigrposint}
\fnexp{\imunit t a_{p}}$.
\end{lem}

\begin{proof}
The spectral theorem and the scalar estimate
$\abs{\fnexp{\imunit x} - 1 - \imunit x}
\leq
\frac{x^{2}}{2}$
give
$$\abs{1 - \bkt{\xi_{p}}{\fnexp{\imunit t a_{p}} \xi_{p}}_{h_{p}}}
\leq
\frac{t^{2}}{2}
\norm{a_{p} \xi_{p}}_{h_{p}}^{2}.$$
The infinite product criterion of Lemma \ref{lem:products} therefore defines the stated tensor product unitary.
The same estimate,
first on the reference product vector and then on the generating product vectors of $\sphilb{D}_{\xi}$,
proves strong continuity.
Stone's theorem supplies a self-adjoint generator $A$.

For a generating product vector of $\sphilb{D}_{\xi}$,
the tail vectors obtained by applying distinct $a_{p}$ are orthogonal because
$\bkt{\xi_{p}}{a_{p}\xi_{p}}_{h_{p}}
=
0$.
The square-summability assumption therefore makes the tensor sum norm convergent on
$\sphilb{D}_{\xi}$.
Differentiating the finite tensor products and passing to this norm limit shows that
$A$ extends $A^{0}$.

It remains to prove that $\sphilb{D}_{\xi}$ is a core for $A$.
For every finite $\Lambda
\subset
\semigrposint$,
Proposition \ref{prop:factorization} identifies the closed subspace
$$\sphilb{K}_{\Lambda}
=
\fun{U_{\Lambda}}
{\rbk{\bigotimes_{p \in \Lambda} h_{p}}
\otimes
\fun{\splinspan}{\setone{\bigotimes_{p \notin \Lambda} \xi_{p}}}}
\subset
\sphilb{H}_{\xi}.$$
Let $A_{\Lambda}$ be the bounded self-adjoint operator on $\sphilb{K}_{\Lambda}$ obtained by applying the finite tensor sum
$\sum_{p \in \Lambda} a_{p}$
to the first factor in this factorization.
The product vectors in the defining set for $\sphilb{D}_{\xi}$ whose modified slots lie in $\Lambda$ span a norm-dense subspace of $\sphilb{K}_{\Lambda}$.
The centering assumption makes the summands in the complementary tail orthogonal.
On this dense span,
the tail sum satisfies
$$\norm{\rbk{A - A_{\Lambda}} \Psi}^{2}
=
\norm{\Psi}^{2}
\sum_{p \notin \Lambda} \norm{a_{p} \xi_{p}}_{h_{p}}^{2}.$$
The operator $A_{\Lambda}$ is bounded,
and $A$ is closed.
The tail identity therefore extends to every $\Psi
\in
\sphilb{K}_{\Lambda}$ and shows that $\sphilb{K}_{\Lambda}
\subset
\dom A$.

Fix $\Phi
\in
\dom A$ and set $\Theta
=
\rbk{A - \imunit} \Phi$.
The union of the subspaces $\sphilb{K}_{\Lambda}$ over finite $\Lambda
\subset
\semigrposint$ contains $\sphilb{D}_{\xi}$ and is dense in $\sphilb{H}_{\xi}$.
Choose the following sequences.
\begin{itemize}
\item
Choose sequences
$\seq{\Theta_{n}}{n \in \semigrposint}$
and
$\seq{F_{n}}{n \in \semigrposint}$
such that each $F_{n}
\subset
\semigrposint$ is finite,
$\Theta_{n}
\in
\sphilb{K}_{F_{n}}$,
and
$\lim_{n \to \infty}
\Theta_{n}
= \Theta$.

\item
Choose sequences
$\seq{\Lambda_{n}}{n \in \semigrposint}$
and $\seq{r_{n}}{n \in \semigrposint}$
such that each $\Lambda_{n}
\subset \semigrposint$ is finite,
$\Lambda_{n}
\supseteq
F_{n}$ and
$r_{n}^{2}
=
\sum_{p \notin \Lambda_{n}} \norm{a_{p} \xi_{p}}_{h_{p}}^{2}
\leq
\frac{1}{n^{2}}$.

\item
Define the sequence
$\seq{\widetilde{\Phi}_{n}}{n \in \semigrposint}$ by
$\widetilde{\Phi}_{n}
=
\opfnresolvent{A_{\Lambda_{n}} - \imunit} \Theta_{n}
\in
\sphilb{K}_{\Lambda_{n}}$.
\end{itemize}
The resolvent bound and the tail identity give
$$\begin{aligned}
\norm{\rbk{A - \imunit} \widetilde{\Phi}_{n} - \Theta}
\leq
\norm{\Theta_{n} - \Theta}
+\norm{\rbk{A - A_{\Lambda_{n}}} \widetilde{\Phi}_{n}}
&\leq
\norm{\Theta_{n} - \Theta}
+
r_{n} \norm{\Theta_{n}}
\to
0
\quad \rbk{n \to \infty}.
\end{aligned}$$
For a self-adjoint $A$,
the norm induced by $\Psi
\mapsto
\rbk{A - \imunit} \Psi$ is the graph norm because
$$\norm{\rbk{A - \imunit} \Psi}^{2}
=
\norm{A \Psi}^{2}
+
\norm{\Psi}^{2}
\quad \rbk{\Psi
\in
\dom A}.$$
It follows that $\widetilde{\Phi}_{n}
\to
\Phi$ in the graph norm of $A$ as $n
\to
\infty$.

The same product-vector span is norm dense in $\sphilb{K}_{\Lambda_{n}}$.
For every $\Psi
\in
\sphilb{K}_{\Lambda_{n}}$,
the tail identity gives the graph-norm bound
$$\norm{A \Psi}
\leq
\rbk{\norm{A_{\Lambda_{n}}} + r_{n}} \norm{\Psi}.$$
The density and this bound allow an approximation of $\widetilde{\Phi}_{n}$ in the graph norm by a vector $\Phi_{n}
\in
\sphilb{D}_{\xi}
\cap
\sphilb{K}_{\Lambda_{n}}$.
Choosing the approximation error below $1 / n$ gives a sequence
$\seq{\Phi_{n}}{n \in \semigrposint}$
that converges to $\Phi$ in the graph norm as $n
\to
\infty$.
This approximation proves that $\sphilb{D}_{\xi}$ is a core for $A$,
and $\opclos{A^{0}}
=
A$.
\end{proof}

The next lemma is the engine behind all irreducibility and factoriality statements: no nontrivial operator commutes with enough finite blocks.

\begin{lem}[triviality of the tail]\label{lem:tail}
Let $\xi$ be an adapted sequence with unit entries,
and let $T
\in
\opspbddlin{\sphilb{H}_{\xi}}$ be such that for every finite $\Lambda$ there is $T_{\Lambda}
\in
\opspbddlin{\sphilb{H}_{\fnrestr{\xi}{\Lambda^{c}}}}$ with $T
=
U_{\Lambda} \rbk{1 \otimes T_{\Lambda}} \faadj{U_{\Lambda}}$.
Then $T$ is a scalar.
\end{lem}

\begin{proof}
Let $\xi^{F}, \xi^{G}$ be flip vectors as in Proposition \ref{prop:flip-basis}
and choose a finite $\Lambda$ containing the supports of $F$ and $G$.
Under $U_{\Lambda}$ these vectors factor as $\xi^{F}
=
u_{F} \otimes \tau$ and $\xi^{G}
=
u_{G} \otimes \tau$,
where $u_{F}, u_{G}$ are the corresponding basis vectors of $\bigotimes_{p \in \Lambda} h_{p}$
and $\tau$ is the unflipped tail product vector.
The matrix coefficient of $T$ is
$$\bkt{\xi^{F}}{T \xi^{G}}
=
\bkt{u_{F} \otimes \tau}{\rbk{1 \otimes T_{\Lambda}} \rbk{u_{G} \otimes \tau}}
=
\bkt{u_{F}}{u_{G}} \bkt{\tau}{T_{\Lambda} \tau}
=
\delta_{F G} \bkt{\tau}{T_{\Lambda} \tau}.$$
For $F
=
G
=
\emptyset$ this shows $\bkt{\tau}{T_{\Lambda} \tau}
=
\bkt{\xi^{\emptyset}}{T \xi^{\emptyset}}
=
c$ for every $\Lambda$,
a constant independent of $\Lambda$.
Therefore $\bkt{\xi^{F}}{T \xi^{G}}
=
c \delta_{F G}$ for all $F, G$,
and since the flip vectors form an orthonormal basis,
$T
=
c$.
\end{proof}

\subsubsection{The complete product space and von Neumann's commutant theorem}\label{the-complete-product-space-and-von-neumanns-commutant-theorem}

The complete product space is defined in Definition \ref{def:complete}. Its inner product and von Neumann's commutant theorem are established below.

\begin{prop}[consistency of the inner product]\label{prop:complete-inner}
Let $\eta, \zeta
\in
\mathcal{W}$ be adapted sequences with unit entries.
The inner product of the associated product vectors in $\widehat{\sphilb{H}}_{\mathcal{W}}$
equals $\prod_{p \in \semigrposint} \bkt{\eta_{p}}{\zeta_{p}}$
whenever this unordered product converges,
and the product vectors of two distinct classes are orthogonal.
When the unordered product fails to converge,
which happens exactly when $\eta
\approx
_{w} \zeta$ but $\eta \not\approx \zeta$ and no factor vanishes,
the inner product is $0$.
\end{prop}

\begin{proof}
If $\eta
\approx
\zeta$ the product converges absolutely and equals the inner product by Definition \ref{def:itps}.

If $\eta \not\approx_{w} \zeta$ then $\sum_{p \in \semigrposint} \rbk{1 - \abs{\bkt{\eta_{p}}{\zeta_{p}}}}
=
\infty$ and the product converges to $0$ by Lemma \ref{lem:divergence-zero}.
This agrees with the orthogonality of distinct summands.

It remains to consider the case $\eta
\approx
_{w} \zeta$ but $\eta \not\approx \zeta$.
If some factor vanishes,
the product converges to $0$,
in agreement with the orthogonality of distinct summands.
Suppose that no factor vanishes.
The unordered product cannot have a nonzero limit,
because Lemma \ref{lem:unordered} would then give $\sum_{p
\in
\semigrposint} \abs{1 - \bkt{\eta_{p}}{\zeta_{p}}}
<
\infty$.
This would contradict $\eta \not\approx \zeta$.
The unordered product cannot have limit $0$ either.
Such a limit would imply
$\lim_{\Lambda \nearrow \semigrposint}
\prod_{p \in \Lambda} \abs{\bkt{\eta_{p}}{\zeta_{p}}}
=
0$,
whereas the second part of Lemma \ref{lem:divergence-zero} gives a nonzero limit for the product of the moduli.
\end{proof}

From here to the end of the section the slot spaces are \(h_{p}
=
\fldcmp^{2}\) and the algebra is the quasi-spin algebra \(\oa{A}\) of Definition \ref{def:quasilocal}.

\begin{defn}[product representation]\label{def:product-rep}
Let $\xi$ be an adapted sequence with unit entries in $\fldcmp^{2}$.
The product representation $\oarepn_{\xi}$ of $\oa{A}$ on $\sphilb{H}_{\xi}$ is defined on $A
\in
\oa{A}_{\Lambda}$ by
$$\fun{\oarepn_{\xi}}{A}
=
U_{\Lambda} \rbk{A \otimes 1} \faadj{U_{\Lambda}},$$
where $A$ acts on $\bigotimes_{p
\in
\Lambda} \fldcmp^{2}$,
and extended to $\oa{A}$ by continuity.
\end{defn}

The definition is consistent: for \(\Lambda
\subset
\Lambda'\) the two prescriptions agree on \(\oa{A}_{\Lambda}\), as one checks on product vectors, where both act by applying \(A\) to the slots in \(\Lambda\) and leaving the remaining entries unchanged. For every local observable \(A\), the representation satisfies the norm identity \[\norm{\fun{\oarepn_{\xi}}{A}}
=
\norm{A \otimes 1}
=
\norm{A}.\] It follows that the extension exists and \(\oarepn_{\xi}\) is an isometric unital \(\ast\)-representation. On product vectors, for \(A
=
\bigotimes_{p
\in
\Lambda} a_{p}\), it holds that \begin{equation}\label{eq:product-action}
\fun{\oarepn_{\xi}}{A} \bigotimes_{p \in \semigrposint} \eta_{p}
=
\bigotimes_{p \in \semigrposint} \eta'_{p},
\quad
\eta'_{p}
=
\begin{dcases}
a_{p} \eta_{p} & \rbk{p \in \Lambda}, \\
\eta_{p} & \rbk{p \notin \Lambda}.
\end{dcases}
\end{equation}

\begin{prop}[irreducibility]\label{prop:irreducible}
Every product representation $\oarepn_{\xi}$ is irreducible:
$\oacommutant{\fun{\oarepn_{\xi}}{\oa{A}}}
=
\fldcmp 1$.
The vector state of every product vector is a pure state of $\oa{A}$.
\end{prop}

\begin{proof}
We first prove the commutant identity used below:
for Hilbert spaces $K_{1}, K_{2}$,
it holds that
$\oacommutant{\rbk{\opspbddlin{K_{1}} \otimes 1}}
=
1 \otimes \opspbddlin{K_{2}}$ in $\opspbddlin{K_{1} \otimes K_{2}}$.
Indeed,
for a unit vector $u
\in
K_{1}$ let $V_{v} \colon K_{2}
\to
K_{1} \otimes K_{2}$,
$\zeta
\mapsto
v \otimes \zeta$.
If $T$ commutes with $\opspbddlin{K_{1}} \otimes 1$ then for all $v, w
\in
K_{1}$,
the identities $\faadj{V_{v}} \rbk{\ketbra{w}{u} \otimes 1}
=
\bkt{v}{w} \faadj{V_{u}}$ and the commutation with $\ketbra{w}{u} \otimes 1$
give
$$\faadj{V_{v}} T V_{w}
=
\faadj{V_{v}} T \rbk{\ketbra{w}{u} \otimes 1} V_{u}
=
\faadj{V_{v}} \rbk{\ketbra{w}{u} \otimes 1} T V_{u}
=
\bkt{v}{w} \faadj{V_{u}} T V_{u}.$$
It follows that all matrix blocks of $T$ are proportional to $S
=
\faadj{V_{u}} T V_{u}$ with the scalar $\bkt{v}{w}$,
which says exactly $T
=
1 \otimes S$.
The reverse inclusion is clear.

Now let $T
\in
\oacommutant{\fun{\oarepn_{\xi}}{\oa{A}}}$ and let $\Lambda$ be finite.
The image $\fun{\oarepn_{\xi}}{\oa{A}_{\Lambda}}
=
U_{\Lambda} \rbk{\opspbddlin{\bigotimes_{p \in \Lambda} \fldcmp^{2}} \otimes 1} \faadj{U_{\Lambda}}$ is a full matrix block.
The commutant lemma gives $\faadj{U_{\Lambda}} T U_{\Lambda}
=
1 \otimes T_{\Lambda}$ for some bounded $T_{\Lambda}$ on the tail space.
This holds for every $\Lambda$.
It follows that Lemma \ref{lem:tail} gives $T
\in
\fldcmp 1$.

It remains to prove purity of the product-vector state.
Let $\oastate[\psi]
=
\bkt{\oagnsvector[\Psi_{\eta}]}{\fun{\oarepn_{\xi}}{\cdot} \oagnsvector[\Psi_{\eta}]}$
with a unit product GNS vector $\oagnsvector[\Psi_{\eta}]$,
and let $\oastate[\psi]
=
\lambda \oastate[\psi_{1}] + \rbk{1 - \lambda} \oastate[\psi_{2}]$ with states $\oastate[\psi_{i}]$ and $0 < \lambda < 1$.
Then $\oastate[\psi_{1}]
\leq
\lambda^{-1} \oastate[\psi]$.
Define the sesquilinear form
$\opform{q}_{\oastate[\psi_{1}],\eta}$ on
$\fun{\oarepn_{\xi}}{\oa{A}} \oagnsvector[\Psi_{\eta}]$ by
$$\fun{\opform{q}_{\oastate[\psi_{1}],\eta}}{
\fun{\oarepn_{\xi}}{A} \oagnsvector[\Psi_{\eta}],
\fun{\oarepn_{\xi}}{B} \oagnsvector[\Psi_{\eta}]}
=
\fun{\oastate[\psi_{1}]}{\faadj{A} B}.$$
The domination $\oastate[\psi_{1}]
\leq
\lambda^{-1} \oastate[\psi]$ gives
$$0
\leq
\fun{\opform{q}_{\oastate[\psi_{1}],\eta}}{
\fun{\oarepn_{\xi}}{A} \oagnsvector[\Psi_{\eta}],
\fun{\oarepn_{\xi}}{A} \oagnsvector[\Psi_{\eta}]}
\leq
\lambda^{-1}
\norm{\fun{\oarepn_{\xi}}{A} \oagnsvector[\Psi_{\eta}]}^{2}.$$
This estimate and the Cauchy--Schwarz inequality for positive forms show that
$\opform{q}_{\oastate[\psi_{1}],\eta}$ is well defined and bounded.
Its domain is dense because applying local elements to the cyclic vector
$\oagnsvector[\Psi_{\eta}]$ produces all finite modifications of
$\oagnsvector[\Psi_{\eta}]$,
whose span is dense by the argument of Proposition \ref{prop:flip-basis}.
The representation theorem for bounded positive forms gives a unique operator $T$ with
$$0
\leq
T
\leq
\lambda^{-1},
\quad
\fun{\oastate[\psi_{1}]}{\faadj{A} B}
=
\bkt{\fun{\oarepn_{\xi}}{A} \oagnsvector[\Psi_{\eta}]}{
T \fun{\oarepn_{\xi}}{B} \oagnsvector[\Psi_{\eta}]}
\quad
\rbk{A, B
\in
\oa{A}}.$$
For $A, B, C
\in
\oa{A}$,
the defining identity gives
$$\begin{aligned}
&\bkt{\fun{\oarepn_{\xi}}{A} \oagnsvector[\Psi_{\eta}]}{
T \fun{\oarepn_{\xi}}{C B} \oagnsvector[\Psi_{\eta}]}
=
\fun{\oastate[\psi_{1}]}{\faadj{A} C B}
\\ 
&=
\bkt{\fun{\oarepn_{\xi}}{\faadj{C} A} \oagnsvector[\Psi_{\eta}]}{
T \fun{\oarepn_{\xi}}{B} \oagnsvector[\Psi_{\eta}]}
=
\bkt{\fun{\oarepn_{\xi}}{A} \oagnsvector[\Psi_{\eta}]}{
\fun{\oarepn_{\xi}}{C} T
\fun{\oarepn_{\xi}}{B} \oagnsvector[\Psi_{\eta}]}.
\end{aligned}$$
The density of
$\fun{\oarepn_{\xi}}{\oa{A}} \oagnsvector[\Psi_{\eta}]$
implies that $T$ commutes with every
$\fun{\oarepn_{\xi}}{C}$.
Irreducibility therefore gives
$T
=
t$ for some $t
\in
\closedinterval{0}{\lambda^{-1}}$.
Normalization determines this scalar:
$$1
=
\fun{\oastate[\psi_{1}]}{1}
=
\bkt{\oagnsvector[\Psi_{\eta}]}{
T \oagnsvector[\Psi_{\eta}]}
=
t.$$
It follows that
$T
=
1$,
which gives $\oastate[\psi_{1}]
=
\oastate[\psi]$.
The original convex-decomposition identity then gives
$\oastate[\psi_{2}]
=
\oastate[\psi]$.
\end{proof}

\begin{prop}[equivalence and disjointness of product representations]\label{prop:disjointness}
Let $\xi, \eta$ be adapted sequences with unit entries in $\fldcmp^{2}$.
\begin{enumerate}
\item
If $\xi
\approx
_{w} \eta$,
the representations $\oarepn_{\xi}$ and $\oarepn_{\eta}$ are unitarily equivalent.

\item
If $\xi \not\approx_{w} \eta$,
the representations are disjoint,
i.e.,
every bounded operator $V \colon \sphilb{H}_{\xi}
\to
\sphilb{H}_{\eta}$ with $V \fun{\oarepn_{\xi}}{A}
=
\fun{\oarepn_{\eta}}{A} V$ for all $A
\in
\oa{A}$ vanishes.
\end{enumerate}
\end{prop}

\begin{proof}

(1)
Choose phases $\phi_{p}$ with $\fnexp{\imunit \phi_{p}} \bkt{\xi_{p}}{\eta_{p}}
=
\abs{\bkt{\xi_{p}}{\eta_{p}}}$ and set $\eta'_{p}
=
\fnexp{\imunit \phi_{p}} \eta_{p}$,
so that $$\bkt{\eta'_{p}}{\xi_{p}}
=
\cmpconj{\fnexp{\imunit \phi_{p}} \bkt{\xi_{p}}{\eta_{p}}}
=
\abs{\bkt{\eta_{p}}{\xi_{p}}}
\geq
0.$$
Then it holds that
$\sum_{p \in \semigrposint} \abs{1 - \bkt{\eta'_{p}}{\xi_{p}}}
=
\sum_{p \in \semigrposint} \rbk{1 - \abs{\bkt{\eta_{p}}{\xi_{p}}}} < \infty$,
and this implies $\eta'
\in
\fun{C}{\xi}$.
The map sending the product vector of $\zeta
\in
\fun{C}{\eta}$ to the product vector of $\rbk{\fnexp{\imunit \phi_{p}} \zeta_{p}}
\in
\fun{C}{\eta'}
=
\fun{C}{\xi}$ preserves all inner products,
since the phases cancel between the two arguments,
maps a total set onto a total set.
It therefore extends to a unitary $W \colon \sphilb{H}_{\eta}
\to
\sphilb{H}_{\xi}$.
By \eqref{eq:product-action} it satisfies $W \fun{\oarepn_{\eta}}{A} \faadj{W}
=
\fun{\oarepn_{\xi}}{A}$ first for elementary tensors $A$,
both sides multiplying the slot entries by $a_{p}$ and the phases commuting past the $a_{p}$ as scalars attached to the vectors,
then for all $A$ by linearity and continuity.

(2)
Since $\xi \not\approx_{w} \eta$,
$\sum_{p \in \semigrposint} \rbk{1 - \abs{\bkt{\eta_{p}}{\xi_{p}}}}
=
\infty$.
For $\Omega
\in
\semigrposint$ let
$$T_{\Omega}
=
\bigotimes_{p
=
1}^{\Omega} P_{\eta_{p}}
\in
\oa{A}_{\Omega},$$
the product of the rank-one projections onto the $\eta_{p}$.
We first consider its action on $\sphilb{H}_{\eta}$.
For a flip vector $\eta^{F}$,
the projection $P_{\eta_{p}}$ fixes the unflipped entries and annihilates the flipped ones.
The action $\fun{\oarepn_{\eta}}{T_{\Omega}} \eta^{F}$ equals $\eta^{F}$ if $\fun{\supp}{F} \cap \intint{1..\Omega}
=
\emptyset$ and $0$ otherwise.
For every flip vector,
this action converges in norm to that of
$\ketbra{\eta^{\emptyset}}{\eta^{\emptyset}}$.
Since the family is uniformly bounded by $1$ and the flip vectors form an orthonormal basis,
the convergence extends strongly to the whole space:
$$\slim_{\Omega \to \infty}
\fun{\oarepn_{\eta}}{T_{\Omega}}
=
P_{\eta^{\emptyset}}.$$

We next consider the action on $\sphilb{H}_{\xi}$.
For a flip vector $\xi^{F}$,
$$\begin{aligned}
&\norm{\fun{\oarepn_{\xi}}{T_{\Omega}} \xi^{F}}
=
\prod_{p
\leq
\Omega, p
\notin
\fun{\supp}{F}} \abs{\bkt{\eta_{p}}{\xi_{p}}} \cdot \prod_{p
\in
\fun{\supp}{F}, p
\leq
\Omega} \norm{P_{\eta_{p}} \xi^{F}_{p}}
\\ 
&\leq
\prod_{p
\leq
\Omega, p
\notin
\fun{\supp}{F}} \abs{\bkt{\eta_{p}}{\xi_{p}}}
\to
0
\quad \rbk{\Omega \to \infty},
\end{aligned}$$
by Lemma \ref{lem:divergence-zero},
since removing the finitely many flipped slots does not affect the divergence $\sum_{p
\in
\semigrposint} \rbk{1 - \abs{\bkt{\eta_{p}}{\xi_{p}}}}
=
\infty$.
By uniform boundedness,
$\fun{\oarepn_{\xi}}{T_{\Omega}}
\to
0$ strongly as $\Omega \to \infty$.
Now intertwining gives $V \fun{\oarepn_{\xi}}{T_{\Omega}}
=
\fun{\oarepn_{\eta}}{T_{\Omega}} V$,
and taking strong limits on both sides,
$0
=
P_{\eta^{\emptyset}} V$,
that is,
$\bkt{\eta^{\emptyset}}{V \Psi}
=
0$ for all $\Psi$.
Repeating the argument with the reference sequence of $\sphilb{H}_{\eta}$ replaced by an arbitrary $\zeta
\in
\fun{C}{\eta}$,
which is again not weakly equivalent to $\xi$ by Lemma \ref{lem:transitivity},
gives $\bkt{\zeta}{V \Psi}
=
0$ for every product vector $\zeta$ of the class.
Product vectors are total.
It follows that $V
=
0$.
Disjointness in the standard sense follows since a nonzero intertwiner between irreducible representations would even be a multiple of a unitary.
\end{proof}

The theorem of von Neumann quoted as (I) in \cite[Section 2]{ThirringWehrl001} describes the commutant of the algebra generated by all quasi-spins in the complete product space over a weak equivalence class.

\begin{defn}[phase families]\label{def:phase-families}
A phase family is a sequence $\phi
=
\seq{\phi_{p}}{p
\in
\semigrposint}$ of real numbers.
It acts on an adapted sequence $\eta$ with unit entries by
$$\phi \cdot \eta
=
\seq{\fnexp{\imunit \phi_{p}} \eta_{p}}{p
\in
\semigrposint}.$$
For an adapted sequence $\xi$ with unit entries,
the equivalence classes of $\xi$ and its phase shift are denoted by
$\fun{C}{\xi}$ and $\fun{C}{\phi \cdot \xi}$,
respectively,
in the notation of
Definition \ref{def:c-sequences}.
\end{defn}

\begin{lem}[phase action on product sectors]\label{lem:phase-unitary}
Let $\phi$ be a phase family in the sense of
Definition \ref{def:phase-families},
and let $\mathcal{W}$ be a weak equivalence class of adapted sequences with unit entries.
For every $\xi
\in
\mathcal{W}$,
the phase shift satisfies
$$\phi \cdot \xi
\approx
_{w} \xi,
\quad
\fun{C_{w}}{\phi \cdot \xi}
=
\fun{C_{w}}{\xi}
=
\mathcal{W}.$$
Its equivalence class is
$$\fun{C}{\phi \cdot \xi}
=
\set{\phi \cdot \eta}{\eta
\in
\fun{C}{\xi}}.$$
The two equivalence classes coincide precisely under the condition
$$\fun{C}{\phi \cdot \xi}
=
\fun{C}{\xi}
\iff
\sum_{p
\in
\semigrposint}
\abs{1 - \fnexp{\imunit \phi_{p}}}
<
\infty.$$
The map on product vectors given by
$$\bigotimes_{p
\in
\semigrposint} \eta_{p}
\mapsto
\bigotimes_{p
\in
\semigrposint} \fnexp{\imunit \phi_{p}} \eta_{p}$$
extends uniquely to a unitary
$U_{\phi}$ on $\widehat{\sphilb{H}}_{\mathcal{W}}$.
It maps the summand $\sphilb{H}_{\fun{C}{\xi}}$ onto
$\sphilb{H}_{\fun{C}{\phi \cdot \xi}}$.
Its restriction to this summand satisfies
$$\fnrestr{U_{\phi}}{\sphilb{H}_{\fun{C}{\xi}}}
\fun{\oarepn_{\fun{C}{\xi}}}{A}
=
\fun{\oarepn_{\fun{C}{\phi \cdot \xi}}}{A}
\fnrestr{U_{\phi}}{\sphilb{H}_{\fun{C}{\xi}}}
\quad
\rbk{A
\in
\oa{A}}.$$
\end{lem}

\begin{proof}
For adapted sequences $\eta$ and $\zeta$ with unit entries,
the slotwise inner products are unchanged:
$$\bkt{\fnexp{\imunit \phi_{p}} \eta_{p}}{
\fnexp{\imunit \phi_{p}} \zeta_{p}}
=
\bkt{\eta_{p}}{\zeta_{p}}
\quad
\rbk{p
\in
\semigrposint}.$$
The two defect sums in Definition \ref{def:c-sequences} are unchanged.
This proves that the phase action preserves equivalence and weak equivalence.
The identity
$$\abs{\bkt{\xi_{p}}{\fnexp{\imunit \phi_{p}} \xi_{p}}}
=
1
\quad
\rbk{p
\in
\semigrposint}$$
gives $\phi \cdot \xi
\approx
_{w} \xi$.
Lemma \ref{lem:transitivity} identifies their weak equivalence classes and gives
$\fun{C_{w}}{\phi \cdot \xi}
=
\fun{C_{w}}{\xi}
=
\mathcal{W}$.
Applying the phase actions $\phi$ and $- \phi$ to equivalent sequences gives
$$\fun{C}{\phi \cdot \xi}
=
\set{\phi \cdot \eta}{\eta
\in
\fun{C}{\xi}}.$$
The corresponding inner product without its modulus is
$$\bkt{\xi_{p}}{\fnexp{\imunit \phi_{p}} \xi_{p}}
=
\fnexp{\imunit \phi_{p}}
\quad
\rbk{p
\in
\semigrposint}.$$
Definition \ref{def:c-sequences} and Lemma \ref{lem:transitivity} now give the stated criterion for
$\fun{C}{\phi \cdot \xi}
=
\fun{C}{\xi}$.

The product formula of Lemma \ref{lem:itps-form} shows that the map in the statement preserves inner products of product vectors.
These vectors are total in every summand by Definition \ref{def:itps}.
The map extends to an isometry from
$\sphilb{H}_{\fun{C}{\xi}}$ onto
$\sphilb{H}_{\fun{C}{\phi \cdot \xi}}$,
and its inverse is induced by the phase family $- \phi$.
Taking the orthogonal direct sum over the equivalence classes in
$\mathcal{W}$ gives the asserted unitary $U_{\phi}$.
For a local elementary tensor $A$ and a product vector associated with $\eta$,
the two compositions
$\fnrestr{U_{\phi}}{\sphilb{H}_{\fun{C}{\xi}}}
\fun{\oarepn_{\fun{C}{\xi}}}{A}$
and
$\fun{\oarepn_{\fun{C}{\phi \cdot \xi}}}{A}
\fnrestr{U_{\phi}}{\sphilb{H}_{\fun{C}{\xi}}}$
produce the sequence with entries
$\fnexp{\imunit \phi_{p}} a_{p} \eta_{p}$ on the support of $A$ and
$\fnexp{\imunit \phi_{p}} \eta_{p}$ elsewhere.
Totality of the product vectors and continuity in $A$ give the intertwining identity.
\end{proof}

\begin{lem}[phase transitivity on product sectors]\label{lem:phase-transitivity}
Let $\mathcal{W}$ be a weak equivalence class of adapted sequences with unit entries,
and let $C, C'$
be equivalence classes contained in $\mathcal{W}$.
Choose reference sequences $\xi_{C'}
\in
C'$ and $\zeta_{C}
\in
C$.
There is a phase family $\phi_{C,C'}$ such that
$$\phi_{C,C'} \cdot \xi_{C'}
\approx
\zeta_{C},
\quad
\fun{C}{\phi_{C,C'} \cdot \xi_{C'}}
=
\fun{C}{\zeta_{C}}
=
C.$$
The restriction
$$V_{C C'}
=
\fnrestr{U_{\phi_{C,C'}}}{\sphilb{H}_{C'}}
\colon
\sphilb{H}_{C'}
\to
\sphilb{H}_{C}$$
is unitary and satisfies
$$V_{C C'} \fun{\oarepn_{C'}}{A}
=
\fun{\oarepn_{C}}{A} V_{C C'}
\quad
\rbk{A
\in
\oa{A}}.$$
\end{lem}

\begin{proof}
The representatives $\xi_{C'}$ and $\zeta_{C}$ belong to the same weak equivalence class.
Definition \ref{def:c-sequences} gives
$\sum_{p
\in
\semigrposint}
\rbk{1 - \abs{\bkt{\xi_{C',p}}{\zeta_{C,p}}}}
<
\infty$.
For each $p
\in
\semigrposint$,
choose $\phi_{C,C',p}
\in
\fldreal$ such that
$$\bkt{\fnexp{\imunit \phi_{C,C',p}} \xi_{C',p}}{\zeta_{C,p}}
=
\abs{\bkt{\xi_{C',p}}{\zeta_{C,p}}}.$$
The equivalence criterion of Definition \ref{def:c-sequences} gives
$$\sum_{p
\in
\semigrposint}
\abs{1 -
\bkt{\fnexp{\imunit \phi_{C,C',p}} \xi_{C',p}}{\zeta_{C,p}}}
=
\sum_{p
\in
\semigrposint}
\rbk{1 - \abs{\bkt{\xi_{C',p}}{\zeta_{C,p}}}}
<
\infty.$$
This proves
$\phi_{C,C'} \cdot \xi_{C'}
\approx
\zeta_{C}$ and identifies its equivalence class with $C$.
Lemma \ref{lem:phase-unitary} supplies the unitary restriction
$V_{C C'}$ and its intertwining identity.
\end{proof}

\begin{thm}[von Neumann's commutant theorem]\label{thm:vn-commutant}
Let $\mathcal{W}$ be a single weak equivalence class of adapted sequences with unit entries in $\fldcmp^{2}$,
and let $\mathfrak{C}\rbk{\mathcal{W}}$ denote the set of equivalence classes contained in $\mathcal{W}$.
Let $\widehat{\sphilb{H}}_{\mathcal{W}}
=
\bigoplus_{C
\in
\mathfrak{C}\rbk{\mathcal{W}}} \sphilb{H}_{C}$,
and let $\oarepn
=
\bigoplus_{C
\in
\mathfrak{C}\rbk{\mathcal{W}}} \oarepn_{C}$ be the direct sum of the product representations.
Denote by $P_{C}$ the orthogonal projection onto the summand $\sphilb{H}_{C}$.
Let $\Phi
=
\fldreal^{\semigrposint}$ be the set of all phase families.
For $\phi
\in
\Phi$,
let $U_{\phi}$ be the unitary supplied by
Lemma \ref{lem:phase-unitary}.
Define the family of generators by
$$\mathcal{G}_{\mathcal{W}}
=
\set{P_{C}}{C
\in
\mathfrak{C}\rbk{\mathcal{W}}}
\cup
\set{U_{\phi}}{\phi
\in
\Phi}.$$
The generator family determines both commutants:
$$\oacommutant{\fun{\oarepn}{\oa{A}}}
=
\oadoublecommutant{\mathcal{G}_{\mathcal{W}}},
\quad
\oadoublecommutant{\fun{\oarepn}{\oa{A}}}
=
\oacommutant{\mathcal{G}_{\mathcal{W}}}.$$
A bounded operator belongs to the weak closure $\oadoublecommutant{\fun{\oarepn}{\oa{A}}}$ if and only if it commutes with every member of $\mathcal{G}_{\mathcal{W}}$,
that is,
with every class projection $P_{C}$ and every phase unitary $U_{\phi}$.
\end{thm}

\begin{proof}
For every $C
\in
\mathfrak{C}\rbk{\mathcal{W}}$,
$P_{C}$ commutes with $\fun{\oarepn}{\oa{A}}$ because the summands are invariant.
Lemma \ref{lem:phase-unitary} shows that every $U_{\phi}$ commutes with
the direct sum $\fun{\oarepn}{\oa{A}}$.
It follows that $\mathcal{G}_{\mathcal{W}}
\subset
\oacommutant{\fun{\oarepn}{\oa{A}}}$.
Taking commutants now gives
$\oadoublecommutant{\fun{\oarepn}{\oa{A}}}
\subset
\oacommutant{\mathcal{G}_{\mathcal{W}}}$.

For the converse inclusion $\oacommutant{\fun{\oarepn}{\oa{A}}}
\subset
\oadoublecommutant{\mathcal{G}_{\mathcal{W}}}$,
let $T
\in
\oacommutant{\fun{\oarepn}{\oa{A}}}$ and consider its blocks $T_{C C'}
=
P_{C} T P_{C'}$,
bounded operators $\sphilb{H}_{C'}
\to
\sphilb{H}_{C}$ intertwining $\oarepn_{C'}$ and $\oarepn_{C}$.
The representations are irreducible by Proposition \ref{prop:irreducible}.
For each ordered pair $\rbk{C,C'}$,
fix the unitary intertwiner $V_{C C'}$ supplied by
Lemma \ref{lem:phase-transitivity}.
We use the following standard form of Schur's lemma.
For an intertwiner $V$ between irreducible representations,
$\faadj{V} V$ lies in the commutant of the source.
Irreducibility gives $\faadj{V} V
=
c$ for a scalar $c
\geq
0$,
and either $c
=
0$ and $V
=
0$,
or $c^{- 1 / 2} V$ is a unitary intertwiner.
If $V$ and $V'$ are unitary intertwiners,
then $V \faadj{\rbk{V'}}$ belongs to the commutant of the target representation.
Irreducibility gives
$V
=
\lambda V'$.
Applying these observations to $T_{C C'}$ and the fixed unitary
$V_{C C'}$ gives
$$T_{C C'}
=
\lambda_{C C'} V_{C C'}$$
for a scalar $\lambda_{C C'}$.
Every block $P_{C} T P_{C'}
=
\lambda_{C C'} P_{C} U_{\phi_{C,C'}} P_{C'}$ lies in the von Neumann algebra $\oa{N}
=
\oadoublecommutant{\mathcal{G}_{\mathcal{W}}}$.
Finally $T$ is the weak-operator limit of the net
$$\net{
\sum_{C, C' \in \oa{F}} P_{C} T P_{C'}
=
\rbk{\sum_{C \in \oa{F}} P_{C}} T \rbk{\sum_{C' \in \oa{F}} P_{C'}}
}{
\oa{F} \nearrow \mathfrak{C}\rbk{\mathcal{W}}
},$$
since
$\sum_{C \in \oa{F}} P_{C}
\to
1$ strongly as $\oa{F}$ increases through the finite subsets of $\mathfrak{C}\rbk{\mathcal{W}}$.
Each member of the net lies in $\oa{N}$,
and hence we obtain $T
\in
\oa{N}$.
The bicommutant theorem shows that taking commutants makes the two identities equivalent.
\end{proof}

Theorem \ref{thm:vn-commutant} is the precise form of the statements of \cite[Section 2]{ThirringWehrl001}: the weak closure \(\oadoublecommutant{\fun{\oarepn}{\oa{A}}}\) does not lead out of an equivalence class. Its restriction to each class acts irreducibly. The sector representations associated with all equivalence classes in one weak class are unitarily equivalent. Moreover, the operators in the weak closure are insensitive to the phase factors \(\phi_{p}\). These phase factors are precisely the residual data that distinguish equivalence classes within a weak class.

\subsubsection{Product states and their GNS representations}\label{product-states-and-their-gns-representations}

Mixed product states are realized as incomplete tensor-product GNS representations and shown to be factorial. The resulting representations supply the thermal fibers used in the main text.

\begin{defn}[product state]
Let $\physvec{\rho}
=
\seq{\oastate[\psi_{p}]}{p
\in
\semigrposint}$ be a sequence of states of $\spmat{2}{\fldcmp}$.
The product state $\oastate[\psi_{\physvec{\rho}}]
=
\bigotimes_{p \in \semigrposint} \oastate[\psi_{p}]$ is the unique state of $\oa{A}$ with
$$\fun{\oastate[\psi_{\physvec{\rho}}]}{\bigotimes_{p
\in
\Lambda} a_{p}}
=
\prod_{p
\in
\Lambda} \fun{\oastate[\psi_{p}]}{a_{p}}$$
for all finite $\Lambda$ and $a_{p}
\in
\spmat{2}{\fldcmp}$.
\end{defn}

The prescription defines a positive normalized functional on each \(\oa{A}_{\Lambda}\), namely the state with density matrix \(\bigotimes_{p \in \Lambda} D_{p}\) where \(D_{p}\) is the density matrix of \(\oastate[\psi_{p}]\). The prescriptions are compatible with the embeddings, and the resulting functional on \(\oa{A}_{\txtloc}\) is bounded by \(1\) and extends to \(\oa{A}\). Positivity persists in the limit. When every \(\oastate[\psi_{p}]\) is pure with unit eigenvector \(\xi_{p}\), set \(\oagnsvector[\Psi_{\xi}]
=
\bigotimes_{p \in \semigrposint} \xi_{p}\). The product state is the vector state of \(\oagnsvector[\Psi_{\xi}]\) in \(\oarepn_{\xi}\), which is pure by Proposition \ref{prop:irreducible}. Its GNS triple is \(\rbk{\sphilb{H}_{\xi}, \oarepn_{\xi}, \oagnsvector[\Psi_{\xi}]}\), cyclicity having been noted in the proof of Proposition \ref{prop:irreducible}. The opposite regime is a faithful product state, where every \(D_{p}\) is invertible. This is the case for Gibbs states at positive temperature.

\begin{prop}[GNS representation of a faithful product state]\label{prop:gns-faithful}
Let $\oastate[\psi_{\physvec{\rho}}]$ be a product state whose density matrices $D_{p}$ are all invertible.
Equip $h_{p}
=
\spmat{2}{\fldcmp}$ with the Hilbert--Schmidt inner product $\bkt{x}{y}
=
\sqfun{\trace}{\faadj{x} y}$ and let $\Xi_{p}
=
D_{p}^{1 / 2}
\in
h_{p}$,
a unit vector.
On the incomplete tensor product space $\sphilb{H}_{\Xi}$ of the sequence $\Xi
=
\seq{\Xi_{p}}{p
\in
\semigrposint}$ define $\oarepn_{\psi_{\physvec{\rho}}}$ by letting $a
\in
\spmat{2}{\fldcmp}$ act on the slot $p$ by left multiplication $x
\mapsto
a x$,
in the manner of Definition \ref{def:product-rep}.
Let $\oagnsvector[\Psi]
=
\bigotimes_{p \in \semigrposint} \Xi_{p}$.
Then $\rbk{\sphilb{H}_{\Xi}, \oarepn_{\psi_{\physvec{\rho}}}, \oagnsvector[\Psi]}$ is the GNS triple of $\oastate[\psi_{\physvec{\rho}}]$.
\end{prop}

\begin{proof}
Left multiplication $L_{a} x
=
a x$ on $h_{p}$ is bounded with $\norm{L_{a}}
\leq
\norm{a}$,
and $a
\mapsto
L_{a}$ is a unital $\ast$-homomorphism of $\spmat{2}{\fldcmp}$ into $\opspbddlin{h_{p}}$,
with adjoint relation
$$\bkt{a x}{y}
=
\sqfun{\trace}{\faadj{x} \faadj{a} y}
=
\bkt{x}{\faadj{a} y}.$$
Exactly as in Definition \ref{def:product-rep},
acting by $L_{a_{p}}$ on finitely many slots defines the isometric representation $\oarepn_{\psi_{\physvec{\rho}}}$ of $\oa{A}$ on $\sphilb{H}_{\Xi}$.
Isometry follows from Lemma \ref{lem:simple} and does not require the factorization argument.
The vector $\oagnsvector[\Psi]$ gives the correct expectation values:
$$\bkt{\oagnsvector[\Psi]}{\fun{\oarepn_{\psi_{\physvec{\rho}}}}{\bigotimes_{p \in \Lambda} a_{p}} \oagnsvector[\Psi]}
=
\prod_{p
\in
\Lambda} \sqfun{\trace}{D_{p}^{1 / 2} a_{p} D_{p}^{1 / 2}}
=
\prod_{p
\in
\Lambda} \sqfun{\trace}{D_{p} a_{p}}
=
\fun{\oastate[\psi_{\physvec{\rho}}]}{\bigotimes_{p \in \Lambda} a_{p}}.$$

It remains to prove that $\oagnsvector[\Psi]$ is cyclic.
Applying local elements to $\oagnsvector[\Psi]$ yields the product vectors with entries $a_{p} D_{p}^{1 / 2}$ on a finite set and $\Xi_{p}$ elsewhere.
Since $D_{p}^{1 / 2}$ is invertible,
we must obtain
$\set{a D_{p}^{1 / 2}}{a
\in
\spmat{2}{\fldcmp}}
=
h_{p}$.
It follows that these product vectors comprise all finite modifications of $\Xi$,
whose span is dense by the argument of Proposition \ref{prop:flip-basis}.
Uniqueness of the GNS triple completes the proof.
\end{proof}

\begin{thm}[factoriality of product states]\label{thm:factor}
Every product state $\oastate[\psi_{\physvec{\rho}}]$ of $\oa{A}$ with all $D_{p}$ invertible is a factor state:
the center of $\oa{M}_{\physvec{\rho}}
=
\oadoublecommutant{\fun{\oarepn_{\psi_{\physvec{\rho}}}}{\oa{A}}}$ is trivial.
The same holds for pure product states,
where $\oa{M}
=
\opspbddlin{\sphilb{H}_{\xi}}$.
\end{thm}

\begin{proof}
The pure case is Proposition \ref{prop:irreducible}.
For the faithful case work in the realization of Proposition \ref{prop:gns-faithful}.
For $b
\in
\spmat{2}{\fldcmp}$ and a slot $q$ let $R^{q}_{b}$ be the operator acting on the slot $q$ by right multiplication $x
\mapsto
x b$ and trivially elsewhere,
bounded via the factorization of Proposition \ref{prop:factorization} with $\Lambda
=
\setone{q}$.
Right multiplications commute with left multiplications at the same slot and act on different slots for different $q$.
It follows that $R^{q}_{b}
\in
\oacommutant{\fun{\oarepn}{\oa{A}}}$,
as one verifies on the total set of product vectors.
Let $Z
\in
\fun{Z}{\oa{M}_{\physvec{\rho}}}
=
\oa{M}_{\physvec{\rho}} \cap \oacommutant{\oa{M}_{\physvec{\rho}}}$ and let $\Lambda$ be finite.
Membership in $\oacommutant{\oa{M}_{\physvec{\rho}}}$ shows that $Z$ commutes with $\fun{\oarepn}{\oa{A}_{\Lambda}}$.
It also commutes with all $R^{q}_{b}$,
$q
\in
\Lambda$,
because $Z$ belongs to $\oa{M}_{\physvec{\rho}}
\subset
\oadoublecommutant{\fun{\oarepn}{\oa{A}}}$.
Under the factorization $U_{\Lambda}$ the operators $\fun{\oarepn}{\oa{A}_{\Lambda}}$ and $R^{q}_{b}$,
$q
\in
\Lambda$,
generate $\opspbddlin{\bigotimes_{p \in \Lambda} h_{p}} \otimes 1$.
At a single slot,
consider the maps $x
\mapsto
e_{i j} x e_{k l}$ defined by matrix units.
Their action on the matrix-unit basis of $h_q$ is
$$e_{i j} e_{a b} e_{k l}
=
\delta_{j a} \delta_{b k} e_{i l}.$$
Thus $x
\mapsto
e_{i j} x e_{k l}$ maps the basis vector $e_{j k}$ to $e_{i l}$ and annihilates the other matrix units.
The linear span of these maps is therefore all of $\opspbddlin{h_q}$.
Tensoring over the slots of $\Lambda$ gives all of $\opspbddlin{\bigotimes_{p \in \Lambda} h_{p}}$,
whose generated von Neumann algebra is itself.
It follows that,
by the commutant lemma in the proof of Proposition \ref{prop:irreducible},
$\faadj{U_{\Lambda}} Z U_{\Lambda}
\in
\oacommutant{\opspbddlin{\bigotimes_{p \in \Lambda} h_{p}} \otimes 1}
=
1 \otimes \opspbddlin{\sphilb{H}_{\fnrestr{\Xi}{\Lambda^{c}}}}$.
This holds for every finite $\Lambda$.
It follows that Lemma \ref{lem:tail} gives $Z
\in
\fldcmp 1$.
\end{proof}

\subsection{\texorpdfstring{Representation Theory of \(\Omega\) Spins}{Representation Theory of \textbackslash Omega Spins}}\label{sec:spin}

The thermal analysis of the degenerate model requires the total-spin decomposition of \(\rbk{\fldcmp^{2}}^{\otimes \Omega}\) and the multiplicity formula quoted from \cite{EugeneWigner002} in \cite[Section 4]{ThirringWalter001}. It also requires a closed formula for diagonal matrix elements of group elements in each irreducible component, replacing the hypergeometric expressions of \cite[Eq. (19)]{ThirringWalter001}. Everything is proved from scratch. Throughout this section \(\Omega\) is fixed, \(\physvec{S}
=
\physvec{S}_{\Omega}\) are the collective operators \eqref{eq:collective} on \(V^{\otimes \Omega}\) with \(V
=
\fldcmp^{2}\), and \(e_{\uparrow}, e_{\downarrow}\) is the standard basis of \(V\) diagonalizing \(\sigma^{z}\).

\subsubsection{Weights, irreducible representations, and multiplicities}\label{weights-irreducible-representations-and-multiplicities}

The ladder relations construct the irreducible spin representations and determine their multiplicities in the tensor power.

The operators \(S^{z}, S^{\pm}\) satisfy \begin{equation}\label{eq:ladder-relations}
\commutator{S^{z}}{S^{\pm}}
=
\pm S^{\pm},
\quad
\commutator{S^{+}}{S^{-}}
=
2 S^{z},
\quad
\physvec{S}^{2}
=
S^{-} S^{+} + S^{z} \rbk{S^{z} + 1},
\end{equation} verified from the one-site relations \(\commutator{\sigma^{z}}{\sigma^{\pm}}
=
\pm 2 \sigma^{\pm}\), \(\commutator{\sigma^{+}}{\sigma^{-}}
=
\sigma^{z}\) and bilinearity, and \(\physvec{S}^{2}\) commutes with \(S^{z}, S^{\pm}\). For \(m
\in
\frac{1}{2} \ringratint\) let \(W_{m}
=
\fun{\Ker}{S^{z} - m}\) be the weight space. It is spanned by the vectors \(e_{I}
=
\bigotimes_{p = 1}^{\Omega} e_{\uparrow / \downarrow}\) with up-spins exactly on \(I
\subset
\intint{1..\Omega}\) and \(\abscard{I}
=
\frac{\Omega}{2} + m\), and the weight-space dimension is \begin{equation}\label{eq:weight-dimension}
N_{m}
=
\dim W_{m}
=
\binom{\Omega}{\frac{\Omega}{2} + m},
\quad \abs{m}
\leq
\frac{\Omega}{2},
\end{equation} with the convention \(N_{m}
=
0\) for \(\abs{m} > \frac{\Omega}{2}\).

\begin{lem}[irreducible spin representations]\label{lem:irreps}
For $2 S
\in
\monnat$,
let $V_{S}$ be the completely symmetric tensor subspace
\eqref{eq:spin-s-space} of $V^{\otimes 2 S}$.
It is invariant under the collective operators of $2 S$ sites.
Then $\dim V_{S}
=
2 S + 1$.
The operator $S^{z}$ on $V_{S}$ has the simple eigenvalues $m
=
- S, - S + 1, \dotsc, S$,
whereas $\physvec{S}^{2}$ acts as the scalar $\fun{S}{S + 1}$.
The space $V_{S}$ is irreducible under the algebra generated by
$S^{\pm}$ and $S^{z}$.

The same ladder relations give a converse characterization.
Let a finite-dimensional Hilbert space carry a unitary realization of
\eqref{eq:ladder-relations},
and let $v$ be a joint eigenvector of $\physvec{S}^{2}$ and $S^{z}$ such that
$S^{+} v
=
0$ and $S^{z} v
=
S v$.
Then the subspace spanned by
$v, S^{-} v, \dotsc, \rbk{S^{-}}^{2 S} v$
is invariant and isomorphic to $V_{S}$.
\end{lem}

\begin{proof}
For $- S
\leq
m
\leq
S$ define
\begin{equation}\label{eq:symmetric-weight-vectors}
u_{m}
=
\sum_{\substack{I
\subset
\intint{1..2 S} \\ \abscard{I}
=
S + m}} e_{I}
\in
V^{\otimes 2 S}.
\end{equation}
These vectors are symmetric,
nonzero,
and form a basis of $V_{S}$.
Indeed,
a symmetric tensor is determined by its coefficients on the symmetric-group orbits,
which are labeled by $\abscard{I}$.
This gives $\dim V_{S}
=
2 S + 1$ and the simplicity of the $S^{z}$-eigenvalues on $V_{S}$.
The ladder operators act by
$$S^{+} u_{m}
=
\rbk{S + m + 1} u_{m + 1},
\quad
S^{-} u_{m}
=
\rbk{S - m + 1} u_{m - 1},$$
because
$$S^{+} e_{I}
=
\sum_{q
\notin
I} e_{I \cup \setone{q}}.$$
Each $e_{J}$ with $\abscard{J}
=
S + m + 1$ arises from those $I
\subset
J$ with $\abscard{J \setminus I}
=
1$.
There are $S + m + 1$ such subsets.
Similarly,
$$S^{-} e_{I}
=
\sum_{q
\in
I} e_{I \setminus \setone{q}},$$
and each $e_{J}$ with $\abscard{J}
=
S + m - 1$ arises from the $I
\supset
J$ with $\abscard{I \setminus J}
=
1$.
The number of such supersets is
$2 S - \rbk{S + m - 1}
=
S - m + 1$.
Only the fact that $S^{\pm} u_{m}$ is a nonzero multiple of $u_{m \pm 1}$ for $m \pm 1$ in range is used in the irreducibility argument.
Any nonzero invariant subspace is $S^{z}$-invariant and therefore contains some
$u_{m}$ by the simple spectrum.
Repeated application of $S^{\pm}$ then places every $u_{m}$ in that subspace,
which proves irreducibility.
The scalar action of $\physvec{S}^{2}$ follows from \eqref{eq:ladder-relations}.
On the highest vector $u_{S}$,
$S^{+} u_{S}
=
0$ and $S^{z} u_{S}
=
S u_{S}$ give $\physvec{S}^{2} u_{S}
=
\fun{S}{S + 1} u_{S}$,
and $\physvec{S}^{2}$ commutes with $S^{-}$.
It follows that the scalar propagates to the basis.
For the last statement,
set $v_{j}
=
\rbk{S^{-}}^{j} v$.
The relations give inductively $S^{z} v_{j}
=
\rbk{S - j} v_{j}$ and $S^{+} v_{j}
=
j \rbk{2 S - j + 1} v_{j - 1}$.
Moreover,
$$\norm{v_{j + 1}}^{2}
=
\bkt{v_{j}}{S^{+} S^{-} v_{j}}
=
\rbk{\fun{S}{S + 1} - \rbk{S - j} \rbk{S - j - 1}} \norm{v_{j}}^{2}.$$
The last equality uses \eqref{eq:ladder-relations} in the form
$S^{+} S^{-}
=
\physvec{S}^{2} - S^{z} \rbk{S^{z} - 1}$.
It follows that $v_{j}
\neq
0$ for $j
\leq
2 S$ and $v_{2 S + 1}
=
0$.
The span of $v_{0}, \dotsc, v_{2 S}$ is invariant,
and the assignment
$v_{j}
\mapsto
\rbk{S^{-}}^{j} u_{S}$ extends to an isomorphism intertwining the three generators.
This extension is possible because the ladder formulas give the same structure constants on both sides.
\end{proof}

\begin{prop}[Clebsch--Gordan decomposition of $\Omega$ spins]\label{prop:multiplicity}
The space $V^{\otimes \Omega}$ decomposes under the commuting self-adjoint operators $\physvec{S}^{2}, S^{z}$ into joint eigenspaces.
The eigenvalues of $\physvec{S}^{2}$ are $\fun{S}{S + 1}$ with $S
\in
\setone{\frac{\Omega}{2}, \frac{\Omega}{2} - 1, \dotsc}$,
$S
\geq
0$.
The tensor power $V^{\otimes \Omega}$ is the orthogonal direct sum of
irreducible subspaces,
and the irreducible representation of spin $S$ occurs with multiplicity
\begin{equation}\label{eq:multiplicity}
m^{\rbk{\Omega}}_{S}
=
N_{S} - N_{S + 1}
=
\frac{\Omega! \rbk{2 S + 1}}{\rbk{\frac{\Omega}{2} - S}! \rbk{\frac{\Omega}{2} + S + 1}!}.
\end{equation}
\end{prop}

\begin{proof}
We first prove complete reducibility.
The generators are contained in the $\ast$-algebra generated by the self-adjoint $S^{x}, S^{y}, S^{z}$.
It follows that the orthogonal complement of an invariant subspace is invariant,
and induction on the dimension decomposes $V^{\otimes \Omega}$ into irreducibles.
Every irreducible finite-dimensional invariant subspace has a highest weight vector,
namely an eigenvector $v$ of $S^{z}$ with maximal eigenvalue in the subspace.
It necessarily satisfies $S^{+} v
=
0$.
The eigenspaces of $\physvec{S}^{2}$ are invariant under the generators.
We may therefore take $v$ to be an $\physvec{S}^{2}$-eigenvector.
The ladder relation \eqref{eq:ladder-relations} gives
$\physvec{S}^{2} v
=
\fun{S}{S + 1} v$ with $S$ the $S^{z}$-eigenvalue of $v$.
The eigenvalue satisfies $S
\geq
0$ because $\norm{S^{-} v}^{2}
=
\rbk{\fun{S}{S + 1} - S \rbk{S - 1}} \norm{v}^{2}
=
2 S \norm{v}^{2}
\geq
0$,
and $2 S
\in
\monnat$ because all $S^{z}$-eigenvalues on $V^{\otimes \Omega}$ lie in $\frac{\Omega}{2} - \monnat$.
Lemma \ref{lem:irreps} then identifies the subspace with $V_{S}$.
In $V_{S}$ each weight $m$ with $\abs{m}
\leq
S$ has multiplicity one.
The weight multiplicities \eqref{eq:weight-dimension} in the two decompositions of $V^{\otimes \Omega}$ give
$$N_{m}
=
\sum_{S
\geq
\abs{m}} m^{\rbk{\Omega}}_{S}.$$
It follows that $m^{\rbk{\Omega}}_{S}
=
N_{S} - N_{S + 1}$,
and with $a
=
\frac{\Omega}{2} + S$,
$$\begin{aligned}
&N_{S} - N_{S + 1}
=
\frac{\Omega!}{a! \rbk{\Omega - a}!} - \frac{\Omega!}{\rbk{a + 1}! \rbk{\Omega - a - 1}!}
\\ 
&=
\frac{\Omega!}{a! \rbk{\Omega - a - 1}!} \sqbk{\frac{1}{\Omega - a} - \frac{1}{a + 1}}
=
\frac{\Omega! \rbk{2 a + 1 - \Omega}}{\rbk{a + 1}! \rbk{\Omega - a}!},
\end{aligned}$$
and $2 a + 1 - \Omega
=
2 S + 1$ gives \eqref{eq:multiplicity}.
\end{proof}

\subsubsection{Diagonal matrix elements of group elements}\label{diagonal-matrix-elements-of-group-elements}

The weight basis from \eqref{eq:symmetric-weight-vectors} reduces diagonal matrix elements of tensor-power group elements to an explicit finite sum. The resulting formula is later scaled in the Bessel-limit argument.

For \(g
\in
\fun{\liegr{GL}}{2, \fldcmp}\) the operator \(g^{\otimes 2 S}\) preserves \(V_{S}\). Write \(\fun{D^{S}}{g}
=
\fnrestr{g^{\otimes 2 S}}{V_{S}}\). With \(u_{m}\) defined by \eqref{eq:symmetric-weight-vectors}, the normalized weight basis of \(V_{S}\) is \(\ket{S, m}
=
\binom{2 S}{S + m}^{- 1 / 2} u_{m}\), orthonormal because the \(e_{I}\) are orthonormal and \(u_{m}\) contains \(\binom{2 S}{S + m}\) of them.

\begin{prop}[closed formula for diagonal matrix elements]\label{prop:matrix-element}
For $g
=
\begin{pmatrix} g_{11} & g_{12} \\ g_{21} & g_{22} \end{pmatrix}
\in
\fun{\liegr{GL}}{2, \fldcmp}$ and $\abs{m}
\leq
S$,
\begin{equation}\label{eq:matrix-element}
\bkt{\ket{S, m}}{\fun{D^{S}}{g} \ket{S, m}}
=
\sum_{k
=
0}^{\min \rbk{S + m, S - m}} \binom{S + m}{k} \binom{S - m}{k} g_{11}^{S + m - k} g_{22}^{S - m - k} \rbk{g_{12} g_{21}}^{k}.
\end{equation}
The same diagonal matrix element determines the corresponding matrix element
in every copy of the spin-$S$ representation.
More precisely, for every unit joint eigenvector $\Psi$ of
$\rbk{\physvec{S}^{2}, S^{z}}$ in $V^{\otimes \Omega}$ with eigenvalues
$\rbk{\fun{S}{S + 1}, m}$,
and every finite product $g
=
\fnexp{\physvec{a}_{1} \cdot \physvec{\sigma}} \dotsm \fnexp{\physvec{a}_{r} \cdot \physvec{\sigma}}$ of exponentials of traceless matrices,
$\physvec{a}_{j}
\in
\fldcmp^{3}$,
$$\bkt{\Psi}{g^{\otimes \Omega} \Psi}
=
\bkt{\ket{S, m}}{\fun{D^{S}}{g} \ket{S, m}}.$$
\end{prop}

\begin{proof}
For the closed formula,
expand $u_{m}
=
\sum_{\substack{I
\subset
\intint{1..2 S} \\ \abscard{I}
=
S + m}} e_{I}$ and $g^{\otimes 2 S} e_{I}
=
\bigotimes_{p = 1}^{2 S} \fun{g}{e_{\uparrow / \downarrow}}$.
The coefficient of $e_{J}$ in $g^{\otimes 2 S} e_{I}$ is
$$g_{11}^{\abscard{I \cap J}} g_{21}^{\abscard{I \setminus J}} g_{12}^{\abscard{J \setminus I}} g_{22}^{\abscard{\rbk{I \cup J}^{c}}},$$
reading $g e_{\uparrow}
=
g_{11} e_{\uparrow} + g_{21} e_{\downarrow}$ and $g e_{\downarrow}
=
g_{12} e_{\uparrow} + g_{22} e_{\downarrow}$ slotwise.
Summing the slotwise coefficients over the symmetrized basis gives
$$\bkt{u_{m}}{g^{\otimes 2 S} u_{m}}
=
\sum_{\substack{I, J
\subset
\intint{1..2 S} \\ \abscard{I}
=
\abscard{J}
=
S + m}} g_{11}^{\abscard{I \cap J}} g_{21}^{\abscard{I \setminus J}} g_{12}^{\abscard{J \setminus I}} g_{22}^{2 S - \abscard{I \cup J}}.$$
Since $\abscard{I}
=
\abscard{J}$,
the sets $I \setminus J$ and $J \setminus I$ have a common cardinality $k$.
We now count the pairs $\rbk{I, J}$ with fixed $k$.
Choose $I$ in $\binom{2 S}{S + m}$ ways,
then $I \setminus J
\subset
I$ in $\binom{S + m}{k}$ ways,
then $J \setminus I$ in the complement of $I$ in $\binom{S - m}{k}$ ways.
Dividing by the normalization $\binom{2 S}{S + m}$ of $\bkt{u_{m}}{u_{m}}$ yields \eqref{eq:matrix-element}.
We next transport the formula to $V^{\otimes \Omega}$.
For $\physvec{a}
\in
\fldcmp^{3}$ one has $\rbk{\fnexp{\physvec{a} \cdot \physvec{\sigma}}}^{\otimes \Omega}
=
\fnexp{2 \physvec{a} \cdot \physvec{S}}$,
where $\physvec{a} \cdot \physvec{S}
=
\sum_{\gamma \in \setone{x, y, z}} a^{\gamma} S^{\gamma}$ is the complex-linear combination of the collective operators.
Indeed $\frac{d}{d t} \rbk{\fnexp{t \physvec{a} \cdot \physvec{\sigma}}}^{\otimes \Omega}
=
\rbk{\sum_{p
\leq
\Omega} \physvec{a} \cdot \physvec{\sigma}_{p}} \rbk{\fnexp{t \physvec{a} \cdot \physvec{\sigma}}}^{\otimes \Omega}$ and both sides solve the same linear differential equation with value $1$ at $t
=
0$.
The same identity holds in $V_{S}$ with its collective operators.
Set $v
=
\rbk{S^{+}}^{S - m} \Psi$.
Iterating $S^{-} S^{+}
=
\physvec{S}^{2} - S^{z} \rbk{S^{z} + 1}$,
which follows from \eqref{eq:ladder-relations},
on the weight vectors $\rbk{S^{+}}^{j - 1} \Psi$ of weight $m + j - 1$ inside the $\physvec{S}^{2}$-eigenspace of $\Psi$ gives
$$\rbk{S^{-}}^{j} \rbk{S^{+}}^{j} \Psi
=
c_{j} \Psi,
\quad
c_{j}
=
\prod_{i
=
1}^{j} \rbk{\fun{S}{S + 1} - \rbk{m + i - 1} \rbk{m + i}} > 0
\quad \rbk{j
\leq
S - m}.$$
Every factor is positive because $m + i
\leq
S$.
In particular $\norm{v}^{2}
=
\bkt{\Psi}{\rbk{S^{-}}^{S - m} \rbk{S^{+}}^{S - m} \Psi}
=
c_{S - m} \norm{\Psi}^{2} > 0$,
and $S^{+} v
=
0$,
$S^{z} v
=
S v$,
since $v$ has the maximal weight of its $\physvec{S}^{2}$-eigenspace.
Lemma \ref{lem:irreps} therefore provides the invariant subspace spanned by
$v, S^{-} v, \dotsc, \rbk{S^{-}}^{2 S} v$.
This subspace is isomorphic to $V_S$ through a unitary intertwiner $U_0$ for $S^{\pm}$ and $S^z$.
The intertwiner carries unit vectors proportional to $\rbk{S^{-}}^{j}v$ to phase multiples of the corresponding weight-basis vectors of $V_S$.
The ladder relations also give
$$\rbk{S^{-}}^{S - m} v
=
\rbk{S^{-}}^{S - m} \rbk{S^{+}}^{S - m} \Psi
=
c_{S - m} \Psi.$$
It follows that $U_{0}$ carries $\Psi$ to a phase multiple of $\ket{S, m}$,
because the weight spaces of $V_{S}$ are one-dimensional.
Complex linearity extends the intertwining to $\physvec{a} \cdot \physvec{S}$ for every $\physvec{a}
\in
\fldcmp^{3}$.
Since the invariant subspace is finite dimensional and invariant under all three generators,
the power series extends the intertwining to their exponentials and then to finite products of exponentials.
It follows that $U_{0} \rbk{\fnexp{2 \physvec{a}_{1} \cdot \physvec{S}} \dotsm \fnexp{2 \physvec{a}_{r} \cdot \physvec{S}}}
=
\fun{D^{S}}{g} U_{0}$ on the subspace.
The subspace is invariant under every factor.
Thus the diagonal matrix element at $\Psi$ equals the one at $\ket{S, m}$,
since the phase cancels.
\end{proof}

The specialization of \eqref{eq:matrix-element} to a rotation \(g
=
\fnexp{\imunit \frac{\beta}{2} \sigma^{y}}\), with \(g_{11}
=
g_{22}
=
\fun{\cos}{\frac{\beta}{2}}\), \(g_{12} g_{21}
=
- \sin^{2} \frac{\beta}{2}\), recovers Wigner's formula for the diagonal \(d\)-functions, \[\begin{aligned}
&\fun{d^{S}_{m m}}{\beta}
=
\sum_{k \geq 0}
\rbk{-1}^{k} \binom{S + m}{k} \binom{S - m}{k}
\fun{\cos}{\frac{\beta}{2}}^{2 S - 2 k}
\fun{\sin}{\frac{\beta}{2}}^{2 k}
\\ 
&=
\fun{\cos}{\frac{\beta}{2}}^{2 S}
\sum_{k \geq 0}
\rbk{-1}^{k}
\binom{S + m}{k}
\binom{S - m}{k}
\fun{\tan}{\frac{\beta}{2}}^{2 k},
\end{aligned}\] which is the form used in \cite[Eq. (19)]{ThirringWalter001}.

\subsubsection{The Bessel limit of the matrix elements}\label{the-bessel-limit-of-the-matrix-elements}

The thermal computation requires the limit of \eqref{eq:matrix-element} along group elements approaching the identity at rate \(\frac{1}{\Omega}\), with quantum numbers growing proportionally to \(\Omega\). The limit is uniform over \(2 S
\in
\monnat\) with \(0
\leq
S
\leq
\frac{\Omega}{2}\) and over \(m
\in
\setone{- S, - S + 1, \dotsc, S}\), which is what the concentration argument needs.

\begin{prop}[uniform Bessel limit]\label{prop:bessel-limit}
For $\physvec{v}
\in
\fldcmp^{3}$ set $v^{\pm}
=
v^{x} \pm \imunit v^{y}$ and define on the free-energy domain
$\dom h$ of \eqref{eq:free-energy-domain}
$$\fun{G_{\infty}}{\eta, z. \physvec{v}}
=
\fnexp{\imunit z v^{z}} \sum_{k
=
0}^{\infty} \frac{1}{\rbk{k!}^{2}} \rbk{- \frac{\rbk{\eta^{2} - z^{2}} v^{+} v^{-}}{4}}^{k},$$
a jointly continuous function,
entire in $\physvec{v}$,
equal to $\fnexp{\imunit z v^{z}} \fun{J_{0}}{\abs{\physvec{v}_{\perp}} \sqrt{\eta^{2} - z^{2}}}$ for real $\physvec{v}$ with $\physvec{v}_{\perp}
=
\rbk{v^{x}, v^{y}}$.
Let $\seq{g_{\Omega}}{\Omega \in \semigrposint}$ be $2 \times 2$ matrices with
$$g_{\Omega}
=
1 + \frac{\imunit}{\Omega} \begin{pmatrix} v^{z} & v^{-} \\ v^{+} & - v^{z} \end{pmatrix} + \fun{Q_{\Omega}}{\physvec{v}},
\quad
\norm{\fun{Q_{\Omega}}{\physvec{v}}}
\leq
\frac{A}{\Omega^{2}},$$
uniformly for $\physvec{v}$ in a compact set $\oa{V}
\subset
\fldcmp^{3}$.
The diagonal matrix elements converge uniformly:
$$\sup_{\substack{0
\leq
S
\leq
\Omega / 2 \\ 2 S \in \monnat \\ m = - S, - S + 1, \dotsc, S}} \sup_{\physvec{v}
\in
\oa{V}} \abs{\bkt{\ket{S, m}}{\fun{D^{S}}{g_{\Omega}} \ket{S, m}} - \fun{G_{\infty}}{\frac{2 S}{\Omega}, \frac{2 m}{\Omega}. \physvec{v}}}
\to
0
\quad \rbk{\Omega
\to
\infty}.$$
\end{prop}

\begin{proof}
Write $x
=
S + m$,
$y
=
S - m$.
It follows that $0
\leq
x, y$ and $x + y
=
2 S
\leq
\Omega$,
and let $G_{\Omega}
=
\bkt{\ket{S, m}}{\fun{D^{S}}{g_{\Omega}} \ket{S, m}}$,
given by \eqref{eq:matrix-element}.
Fix the compact $\oa{V}$ and let $c$ bound $\abs{v^{z}}, \abs{v^{\pm}}, A$ on $\oa{V}$.
The entries obey $\abs{g_{11}}, \abs{g_{22}}
\leq
1 + \frac{2 c}{\Omega}$ and $\abs{g_{12} g_{21}}
\leq
\frac{W}{\Omega^{2}}$ with $W
=
\rbk{c + \frac{c}{\Omega}}^{2} \cdot 4
\leq
16 c^{2}$ for $\Omega
\geq
1$.
the bounds $\binom{x}{k}
\leq
\frac{x^{k}}{k!}$ and $x y
\leq
\frac{\Omega^{2}}{4}$ give the termwise estimate
\begin{equation}\label{eq:term-domination}
\abs{\binom{x}{k} \binom{y}{k} g_{11}^{x - k} g_{22}^{y - k} \rbk{g_{12} g_{21}}^{k}}
\leq
\fnexp{4 c} \frac{1}{\rbk{k!}^{2}} \rbk{\frac{W}{4}}^{k},
\end{equation}
a summable envelope independent of $\Omega, S, m, \physvec{v}
\in
\oa{V}$.
The same envelope dominates the terms of $G_{\infty}$.
Given $\epsilon > 0$ choose $k_{0}$ so that the envelope \eqref{eq:term-domination} satisfies $\sum_{k > k_{0}} \fnexp{4 c} \frac{\rbk{W / 4}^{k}}{\rbk{k!}^{2}} < \epsilon$.
It remains to show that each term with $k
\leq
k_{0}$ converges to the corresponding term of $G_{\infty}$,
uniformly in $\rbk{S, m}$ and $\physvec{v}
\in
\oa{V}$.

First suppose that $\min \rbk{x, y}
\leq
\sqrt{\Omega}$.
Then $\frac{x y}{\Omega^{2}}
\leq
\frac{1}{\sqrt{\Omega}}$.
It follows that for $1
\leq
k
\leq
k_{0}$ both the $G_{\Omega}$-term and the $G_{\infty}$-term are bounded by $\fnexp{4 c} \frac{\rbk{W / \sqrt{\Omega}}^{k}}{\rbk{k!}^{2}}
\to
0$ as $\Omega \to \infty$.
For $k
=
0$ the terms are $g_{11}^{x} g_{22}^{y}$ and $\fnexp{\imunit z v^{z}}$,
and the principal-logarithm estimate below gives their uniform convergence without a restriction on
$x$ and $y$.

Next suppose that $x, y > \sqrt{\Omega}$.
For $k
\leq
k_{0}$,
$$\binom{x}{k} \binom{y}{k}
=
\frac{x^{k} y^{k}}{\rbk{k!}^{2}} \rbk{1 + \theta},
\quad \abs{\theta}
\leq
\frac{2 k^{2}}{\sqrt{\Omega}} \quad \rbk{\Omega \text{ large}},$$
from $\prod_{j < k} \rbk{1 - \frac{j}{x}}
=
1 + \fun{O}{\frac{k^{2}}{x}}$,
and the off-diagonal factor satisfies
$$\rbk{g_{12} g_{21}}^{k}
=
\rbk{- \frac{v^{+} v^{-}}{\Omega^{2}}}^{k} \rbk{1 + \theta'},
\quad \abs{\theta'}
\leq
\frac{C_{k_{0}}}{\Omega \abs{v^{+} v^{-}}}
\leq
C_{k_{0}} \Omega^{- 1 / 2} \text{ on } \set{\physvec{v}
\in
\oa{V}}{\abs{v^{+} v^{-}}
\geq
\Omega^{- 1 / 2}},$$
because $g_{12} g_{21}
=
- \frac{v^{+} v^{-}}{\Omega^{2}} + \fun{O}{\Omega^{- 3}}$ from the hypothesis on $Q_{\Omega}$.
It follows that the relative error of each factor is at most a constant multiple of $\frac{1}{\Omega \abs{v^{+} v^{-}}}$.
For $\abs{v^{+} v^{-}} < \Omega^{- 1 / 2}$ both the $G_{\Omega}$-term and the $G_{\infty}$-term with $k
\geq
1$ are bounded by $C' \Omega^{- k / 2}$ uniformly,
as in Case 1.
It follows that their difference tends to $0$ uniformly there without any multiplicative comparison.
On the complementary region,
the relative estimates show that for $k
\leq
k_{0}$,
with $\abs{v^{+} v^{-}}
\geq
\Omega^{- 1 / 2}$,
the $k$-th term of $G_{\Omega}$ equals
$$\frac{1}{\rbk{k!}^{2}} \rbk{- \frac{x y v^{+} v^{-}}{\Omega^{2}}}^{k} g_{11}^{x - k} g_{22}^{y - k} \rbk{1 + \fun{O}{\Omega^{- 1 / 2}}},$$
uniformly,
and $\frac{x y}{\Omega^{2}}
=
\frac{\rbk{\eta + z} \rbk{\eta - z}}{4}
=
\frac{\eta^{2} - z^{2}}{4}$ with $\eta
=
\frac{2 S}{\Omega}$,
$z
=
\frac{2 m}{\Omega}$.
The diagonal factors are controlled by the principal logarithm,
which gives the following estimates uniformly on $\oa{V}$:
$$\begin{aligned}
g_{11}^{x - k}
&=
\fnexp{\rbk{x - k} \rbk{\frac{\imunit v^{z}}{\Omega} + \fun{O}{\Omega^{- 2}}}}
=
\fnexp{\imunit v^{z} \frac{x}{\Omega}} \rbk{1 + \fun{O}{\frac{k + 1}{\Omega}}},
\\
g_{22}^{y - k}
&=
\fnexp{- \imunit v^{z} \frac{y}{\Omega}} \rbk{1 + \fun{O}{\frac{k + 1}{\Omega}}},
\end{aligned}$$
and $\fnexp{\imunit v^{z} \rbk{x - y} / \Omega}
=
\fnexp{\imunit z v^{z}}$ exactly.
Collecting,
each term with $k
\leq
k_{0}$ converges uniformly to the corresponding term of $G_{\infty}$,
the number of exceptional small-denominator regimes having been absorbed into vanishing bounds,
and the tail is at most $2 \epsilon$.
This proves the uniform convergence.
Continuity and entireness of $G_{\infty}$ are clear from local uniform convergence of the series,
and for real $\physvec{v}$ one has $v^{+} v^{-}
=
\abs{\physvec{v}_{\perp}}^{2}$.
Definition \ref{def:bessel} identifies the series with $\fun{J_{0}}{\abs{\physvec{v}_{\perp}} \sqrt{\eta^{2} - z^{2}}}$.
\end{proof}

\bibliography{myref.bib}

\end{document}

%% file: mycommands.tex
\makeatletter
\providecommand*{\dashv}{\mathrel{\mathpalette\@Dashv\vDash}}
\newcommand*{\@dashv}[2]{\reflectbox{$\m@th#1#2$}}
\makeatother
\newcommand{\abscard}[1]{\abs{#1}} 
\newcommand{\abs}[1]{\left| #1 \right|} 
\NewDocumentCommand{\weaknorm}{O{\dbk} m}{#1{#2}} 
\newcommand{\norm}[1]{\left\Vert #1 \right\Vert} 

\newcommand{\bkt}[2]{\left\langle #1,\,#2 \right\rangle} 
\newcommand{\rbkt}[2]{\left( #1,\,#2 \right)} 
\newcommand{\slim}{\mathrm{s} \hyphen \lim} 
\newcommand{\wlim}{\mathrm{w} \hyphen \lim} 
\newcommand{\isomto}{\mathrel{\rightarrowtail\kern-1.9ex\twoheadrightarrow}} 
\newcommand{\cbk}[1]{\left\{ #1 \right\}} 
\newcommand{\dbk}[1]{\left\langle #1 \right\rangle} 
\newcommand{\rbk}[1]{\left( #1 \right)} 
\newcommand{\sqbk}[1]{\left[ #1 \right]} 
\newcommand{\vecbk}[1]{\rbk{#1}} 
\newcommand{\funrbk}[2]{\fun{\rbk{#1}}{#2}} 
\newcommand{\fun}[2]{#1 \rbk{#2}} 
\newcommand{\sqfun}[2]{#1 \sqbk{#2}} 
\newcommand{\closedinterval}[2]{\sqbk{#1,\,#2}} 
\newcommand{\intint}[1]{\sqbk{#1}} 
\newcommand{\rightopeninterval}[2]{\left[#1, \, #2 \right)} 
\newcommand{\unitclosedinterval}{\closedinterval{0}{1}} 
\newcommand{\commutator}[2]{\sqbk{#1,\,#2}} 
\newcommand{\bra}[1]{\left\langle #1 \right|} 
\newcommand{\ketbra}[2]{\ket{#1} \!\! \bra{#2}} 
\newcommand{\ket}[1]{\left| #1 \right\rangle} 

\NewDocumentCommand{\imunit}{O{\mathsf{i}}}{#1} 
\NewDocumentCommand{\placeholder}{O{\bullet}}{#1} 
\NewDocumentCommand{\trace}{O{\operatorname{Tr}}}{#1} 
\newcommand{\Ad}{\operatorname{Ad}} 
\newcommand{\Ker}{\operatorname{Ker}} 
\newcommand{\Ran}{\operatorname{Ran}} 
\newcommand{\cmpconj}[1]{\overline{#1}} 
\newcommand{\dom}{\operatorname{dom}} 
\newcommand{\eqcsq}[1]{\sqbk{#1}} 
\newcommand{\hyphen}{\hbox{-}} 
\newcommand{\id}{\mathrm{id}} 
\newcommand{\invrbk}[1]{\rbk{#1}^{-1}} 
\newcommand{\net}[2]{\rbk{#1}_{#2}} 
\newcommand{\opchern}[1]{\operatorname{ch}} 
\newcommand{\opimag}{\operatorname{Im}} 
\newcommand{\opod}[1]{\frac{d}{d #1}} 
\newcommand{\opreal}{\operatorname{Re}} 
\newcommand{\setSymbolDownLeft}[2]{{\vphantom{#2}}_{#1}{#2}} 
\newcommand{\setSymbolUpLeft}[2]{{\vphantom{#2}}^{#1}{#2}} 
\newcommand{\supp}{\operatorname{supp}} 
\NewDocumentCommand{\agvariety}{O{\mathcal}}{#1} 
\NewDocumentCommand{\cmdrel}{O{\omega}}{#1} 
\NewDocumentCommand{\dfsp}{O{A}}{#1} 
\NewDocumentCommand{\eqcpointed}{O{\eqcsq} m}{#1{#2}_{\ast}} 
\NewDocumentCommand{\fnheaviside}{O{H}}{#1} 
\NewDocumentCommand{\fthol}{O{\mathcal{O}}}{#1} 
\NewDocumentCommand{\ftmero}{O{\mathcal{M}}}{#1} 
\NewDocumentCommand{\grcentralizer}{O{Z}}{#1} 
\NewDocumentCommand{\grmetform}{O{2} m m}{\grmet[#1] \! \rbkt{#2}{#3}} 
\NewDocumentCommand{\grmet}{O{2}}{\setSymbolDownLeft{#1}{g}} 
\NewDocumentCommand{\grnormalizer}{O{N}}{#1} 
\NewDocumentCommand{\gropasym}{O{A}}{#1} 
\NewDocumentCommand{\gropsym}{O{S}}{#1} 
\NewDocumentCommand{\grpermorderedpair}{O{\mathcal{P}}}{#1} 
\NewDocumentCommand{\grsym}{O{\mathfrak{S}} m}{#1_{#2}} 
\NewDocumentCommand{\gtbase}{O{\mathcal}}{#1} 
\NewDocumentCommand{\gtfilter}{O{\mathcal}}{#1} 
\NewDocumentCommand{\gtfmlclosed}{O{\mathcal}}{#1} 
\NewDocumentCommand{\gtfmlopen}{O{\mathcal}}{#1} 
\NewDocumentCommand{\gtopenball}{O{U}}{#1} 
\NewDocumentCommand{\gtopencover}{O{\mathcal}}{#1} 
\NewDocumentCommand{\gtopennbh}{O{\mathcal}}{#1} 
\NewDocumentCommand{\gtpreopencover}{O{\mathcal}}{#1} 
\NewDocumentCommand{\gtsubbase}{O{\mathcal}}{#1} 
\NewDocumentCommand{\gtvicinity}{O{\mathcal}}{#1} 
\NewDocumentCommand{\lasp}{O{\mathcal}}{#1} 
\NewDocumentCommand{\latprightrbk}{O{\top} m}{\rbk{#2}^{#1}} 
\NewDocumentCommand{\latpright}{O{\top} m}{#2^{#1}} 
\NewDocumentCommand{\latp}{O{t} m}{\setSymbolUpLeft{#1}{#2}} 
\NewDocumentCommand{\lpdistribution}{O{\mu} m}{#2_{\ast,#1}} 
\NewDocumentCommand{\lpmollifier}{O{\rho}}{#1} 
\NewDocumentCommand{\lpofpositive}{O{\chi}}{#1} 
\NewDocumentCommand{\manliederiv}{O{L}}{#1} 
\NewDocumentCommand{\mansmoothnbh}{O{\mathcal}}{#1} 
\NewDocumentCommand{\mblfmldsysgenerated}{O{d} m}{\fun{#1}{#2}} 
\NewDocumentCommand{\mblfmlgenerated}{O{\sigma} m}{\fun{#1}{#2}} 
\NewDocumentCommand{\oacorrfn}{O{\Gamma}}{#1} 
\NewDocumentCommand{\oagnsvector}{O{\Omega}}{#1} 
\NewDocumentCommand{\oaideal}{O{\mathcal}}{#1} 
\NewDocumentCommand{\oanumberoperator}{O{A}}{#1} 
\NewDocumentCommand{\oaposcone}{O{\mathcal{P}}}{#1} 
\NewDocumentCommand{\oapressure}{O{P}}{#1} 
\NewDocumentCommand{\oarepn}{O{\pi}}{#1} 
\NewDocumentCommand{\oaspnormalstate}{O{N}}{#1} 
\NewDocumentCommand{\oasppurestate}{O{P}}{#1} 
\NewDocumentCommand{\oaspstate}{O{E}}{#1} 
\NewDocumentCommand{\oastatevector}{O{\Omega}}{#1} 
\NewDocumentCommand{\oastate}{O{\omega}}{#1} 
\NewDocumentCommand{\opdilation}{O{\delta}}{#1} 
\NewDocumentCommand{\opdmat}{O{\rho}}{#1} 
\NewDocumentCommand{\opfockan}{O{a}}{#1} 
\NewDocumentCommand{\opfockcran}{O{a}}{#1^{\#}} 
\NewDocumentCommand{\opfockcrdagger}{O{a}}{#1^{\dagger}} 
\NewDocumentCommand{\opfockcr}{O{a}}{#1^{\ast}} 
\NewDocumentCommand{\opfocknumber}{O{N}}{#1} 
\NewDocumentCommand{\opfocksegalconj}{O{\pi}}{#1} 
\NewDocumentCommand{\opfocksegal}{O{\phi}}{#1} 
\NewDocumentCommand{\opspecmeas}{O{E}}{#1} 
\NewDocumentCommand{\opspec}{O{} m}{\fun{\sigma_{#1}}{#2}} 
\NewDocumentCommand{\optransl}{O{\tau}}{#1} 
\NewDocumentCommand{\physaction}{O{\mathcal{A}}}{#1} 
\NewDocumentCommand{\physcharge}{O{e}}{\mathrm{#1}} 
\NewDocumentCommand{\physcplconst}{O{\mathsf{g}}}{#1} 
\NewDocumentCommand{\physelectrostaticcapasity}{O{\mathrm{Cap}}}{#1} 
\NewDocumentCommand{\physenergy}{O{E}}{#1} 
\NewDocumentCommand{\physgse}{O{E}}{#1_{0}} 
\NewDocumentCommand{\physham}{O{H}}{#1} 
\NewDocumentCommand{\physlagdensity}{O{\mathcal{L}}}{#1} 
\NewDocumentCommand{\physlag}{O{L}}{#1} 
\NewDocumentCommand{\physliouvilean}{O{L}}{#1} 
\NewDocumentCommand{\physmass}{O{m}}{#1} 
\NewDocumentCommand{\prbbrownmv}{O{B}}{#1} 
\NewDocumentCommand{\prbcharfun}{O{\chi}}{#1} 
\NewDocumentCommand{\prbdist}{O{\mathcal{P}}}{#1} 
\NewDocumentCommand{\prbgaussianmeasure}{O{\msrcal{N}}}{#1} 
\NewDocumentCommand{\prbnormaldist}{O{N}}{#1} 
\NewDocumentCommand{\prbpoissonprocess}{O{N}}{#1} 
\NewDocumentCommand{\prbprocess}{O{X}}{#1} 
\NewDocumentCommand{\prbqspace}{O{\mathcal{Q}}}{#1} 
\NewDocumentCommand{\prbspsample}{O{\Omega}}{#1} 
\NewDocumentCommand{\psh}{O{\mathfrak}}{#1} 
\NewDocumentCommand{\qtquantumchannel}{O{\mathcal{L}}}{#1} 
\NewDocumentCommand{\repn}{O{\pi}}{#1} 
\NewDocumentCommand{\schattencls}{O{\mathbb{K}}}{#1} 
\NewDocumentCommand{\setfmlcylinder}{O{\mathcal{C}}}{#1} 
\NewDocumentCommand{\setfml}{O{\mathcal}}{#1} 
\NewDocumentCommand{\setindex}{O{\mathcal} m}{#1{#2}} 
\NewDocumentCommand{\setlattice}{O{\Gamma}}{#1} 
\NewDocumentCommand{\setspecial}{O{\mathcal} m}{#1{#2}} 
\NewDocumentCommand{\shdiffform}{O{\sheaf{A}}}{#1} 
\NewDocumentCommand{\sheaf}{O{\mathfrak}}{#1} 
\NewDocumentCommand{\smchemicalpotential}{O{\mu}}{#1} 
\NewDocumentCommand{\smenergydensity}{O{\varrho}}{#1} 
\NewDocumentCommand{\smfluctuationwithdmat}{O{\beta} m}{\smuncertaintywithdmat[#1]{#2}^2} 
\NewDocumentCommand{\sminvtemperature}{O{\beta}}{#1} 
\NewDocumentCommand{\smlocaldensityoperator}{O{\rho}}{#1} 
\NewDocumentCommand{\smmicrocanonicalstate}{O{\beta} m}{\physmean{#2}_{#1}} 
\NewDocumentCommand{\smnumberdensity}{O{\rho}}{#1} 
\NewDocumentCommand{\smooth}{O{\mathcal{E}}}{#1} 
\NewDocumentCommand{\smparticlenumber}{O{N}}{#1} 
\NewDocumentCommand{\smpressure}{O{p}}{#1} 
\NewDocumentCommand{\smspecificfreeenergy}{O{\bar{f}}}{#1} 
\NewDocumentCommand{\smthermalvac}{O{\beta}}{\Omega_{#1}} 
\NewDocumentCommand{\smuncertaintywithdmat}{O{\beta} m}{\rbk{\triangle #2}_{#1}} 
\NewDocumentCommand{\sphilb}{O{\mathcal}}{#1} 
\NewDocumentCommand{\splowerhalf}{O{\mathbb{H}}}{#1_{\txtneg}} 
\NewDocumentCommand{\spupperhalf}{O{\mathbb{H}}}{#1_{\txtnonneg}} 
\NewDocumentCommand{\topmetric}{O{d}}{#1} 
\NewDocumentCommand{\vaoutnormal}{O{\widehat}}{#1} 
\newcommand{\category}[1]{\mathop{\mathsf{#1}}} 

\newcommand{\catpresheaf}[1]{\category{PSh}} 
\newcommand{\conti}{C} 
\newcommand{\dggrgauge}{\mathcal{G}} 
\newcommand{\faadj}[1]{#1^{\ast}} 
\newcommand{\fldcmp}{\fld{C}} 
\newcommand{\fldrat}{\fld{Q}} 
\newcommand{\fldreal}{\fld{R}} 
\newcommand{\fld}[1]{\mathbb{#1}} 
\newcommand{\fnexp}[1]{\fun{\exp}{#1}} 
\newcommand{\fnrestr}[2]{\left. #1 \right|_{#2}} 
\newcommand{\grauto}{\operatorname{Aut}} 
\newcommand{\gtclos}[1]{\overline{#1}} 
\newcommand{\liegr}[1]{\mathrm{#1}} 
\newcommand{\lp}{L} 
\newcommand{\mansphere}{\mathbb{S}} 
\newcommand{\monnat}{\mathbb{N}} 
\newcommand{\msrcal}[1]{\mathcal{#1}} 
\newcommand{\msrprb}{\mathrm{Pr}} 

\newcommand{\oacommutant}[1]{#1^{\prime}} 
\newcommand{\oacstar}{C^{\ast}} 
\newcommand{\oadoublecommutant}[1]{#1^{\prime \prime}} 
\newcommand{\oawstar}{W^{\ast}} 
\newcommand{\oa}[1]{\mathcal{#1}} 
\newcommand{\opclos}[1]{\overline{#1}} 
\newcommand{\opdmsr}[1]{\mathop{d #1}} 
\newcommand{\opfnresolvent}[1]{\invrbk{#1}} 
\newcommand{\opform}[1]{\mathsf{#1}} 
\newcommand{\opspbddlin}[1]{\fun{\mathbb{B}}{#1}} 
\newcommand{\opspecint}[1]{\mathcal{E}} 
\newcommand{\physmean}[1]{\dbk{#1}} 
\newcommand{\physvec}[1]{\boldsymbol{#1}} 
\newcommand{\pushout}[1]{#1_{\ast}} 
\newcommand{\ringratint}{\mathbb{Z}} 
\newcommand{\semigrposint}{\monnat_1} 

\newcommand{\seq}[2]{\if\relax\detokenize{#1}\relax \rbk{#1} \else \rbk{#1}_{#2} \fi} 
\newcommand{\setisomorphism}[1]{\operatorname{Iso}} 
\newcommand{\setone}[1]{\cbk{#1}}
\newcommand{\set}[2]{\left\{#1 \, \middle| \, #2\right\}}
\newcommand{\splinspan}{\operatorname{span}} 
\newcommand{\spmat}[2]{\fun{M_{#1}}{{#2}}} 
\newcommand{\txtbogoliubov}{\mathrm{Bog}} 
\newcommand{\txteff}{\mathrm{eff}} 
\newcommand{\txtexternal}{\mathrm{ext}} 
\newcommand{\txtgs}{\mathrm{gs}} 
\newcommand{\txtloc}{\mathrm{loc}} 
\newcommand{\txtmax}{\mathrm{max}} 
\newcommand{\txtneg}{\mathrm{-}} 
\newcommand{\txtnonneg}{\mathrm{+}} 
\newcommand{\txttot}{\mathrm{tot}} 